\documentclass[magisterska]{pracamgr}
\usepackage{graphicx}
\usepackage[latin2]{inputenc}
\usepackage[toc,page]{appendix}
\usepackage{tikz}
\usepackage{graphicx}
\usepackage{amssymb}
\usepackage{amsthm}
\usepackage{amsmath}
\usepackage{amsfonts}
\usepackage{bbold}
\usepackage{subfig}
\usepackage{url} 
\usepackage{multirow}
\usepackage{subfig}
\usepackage{url} 

\author{Tomasz Badowski}
\nralbumu{260542}

\title{Variance-based sensitivity analysis and orthogonal approximations for stochastic models}
\tytulang{Oparta na wariancji analiza wra\.zliwo\'sci i aproksymacje ortogonalne modeli stochastycznych}

\kierunek{Physics} 

\zakres{Mathematical and Computer Modeling of Physical Processes}

\opiekun{prof. dr hab. Bogdan Lesyng\\
  Department of Biophysics, Institute of Experimental Physics,\\
Faculty of Physics, University of Warsaw
}

\date{Warsaw, August 2013}

\dziedzina{ 
13.2 Physics\\ 
}

\klasyfikacja{65C05 Monte Carlo methods}

\keywords{reakcja chemiczna, estymator, proces Markowa, Monte Carlo, rzut ortogonalny, analiza wra\.zliwo\'sci, 
  statystyka, proces stochastyczny, symulacje stochastyczne}
 
\keywordsEng{chemical reaction, estimator, Markov process,
 Monte Carlo, orthogonal projection, sensitivity analysis, statistics, stochastic process, stochastic simulations}

 \streszczang{
Wprowadzamy nowe estymatory dla pewnych wielko\'sci zdefiniowanych dla funk\-cji moment\'ow warunkowych,
takich jak wariancje i \'srednie warunkowe, funk\-cji dw\'och niezale\.znych zmiennych losowych pod warunkiem pierwszej zmiennej,
w tym dla pewnych wynik\'ow modeli stochastycznych pod warunkiem ich parametr\'ow.
Do estymowanych wielko\'sci nale\.z\c{a} wsp\'o\l{}czynniki wra\.zliwo\'sci oparte na wariancji, \'sredniokwadratowe b\l{}\c{e}dy przybli\.zenia
funk\-cjami pierwszej zmiennej, wsp\'o\l{}czynniki rzutu ortogonalnego oraz nowo zdefiniowane wsp\'o\l{}\-czyn\-ni\-ki nieliniowo\'sci.
Definiujemy powy\.zsze estymatory i analizujemy ich wydajno\'s\'c w procedurach Monte Carlo
u\.zywaj\c{a}c uog\'olnionego poj\c{e}cia schematu do estymacji i sta\l{}ej charakteryzuj\c{a}cej jego nieefektywno\'s\'c.
W symulacjach numerycznych uk\l{}ad\'ow reakcji che\-micz\-nych przy u\.zyciu metod \emph{Gillespie's direct} i
\emph{random time change} nowe schematy dla wsp\'o\l{}\-czyn\-ni\-k\'ow wra\.zliwo\'sci warunkowych warto\'sci oczekiwanych by\l{}y w pewnych przypadkach
wydajniejsze od tych zaproponowanych poprzednio, a wariancje niekt\'orych estymator\'ow znacz\c{a}co zale\.za\l{}y od zastosowanej
 metody symulacji.}

\newcommand{\Z}{\mathbb{Z}}
\newcommand{\N}{\mathbb{N}}

\newcommand{\R}{\mathbb{R}}
\newcommand{\I}{\mathbb{1}}
\newcommand{\PR}{\mathbb{P}}

\newcommand{\wt}[1]{\widetilde{{#1}}}
\newcommand{\wh}[1]{\widehat{{#1}}}

\newcommand{\mc}[1]{\mathcal{{#1}}}
\newtheorem{theorem}{Theorem}
 
\newtheorem{lemma}[theorem]{Lemma}
\newtheorem{defin}[theorem]{Definition}
\newtheorem{constr}{Construction}
\DeclareMathOperator{\A}{A}

\DeclareMathOperator{\Var}{Var}

\DeclareMathOperator{\id}{id}
\DeclareMathOperator{\tr}{tr}
\DeclareMathOperator{\Sym}{Sym}
\DeclareMathOperator{\hist}{hist}

\newcommand{\E}{\mathbb{E}}

\DeclareMathOperator{\U}{U}
\DeclareMathOperator{\Exp}{Exp}
\DeclareMathOperator{\ND}{N}

\DeclareMathOperator{\corr}{corr}
\DeclareMathOperator{\corrL}{corrL}
\DeclareMathOperator{\Cov}{Cov}
\DeclareMathOperator{\err}{err}

\DeclareMathOperator{\ave}{Ave}

\DeclareMathOperator{\msd}{msd}

\begin{document}
\maketitle
\begin{abstract}
We develop new unbiased estimators of a number of quantities defined for functions of conditional moments,
like conditional expectations and variances, of functions of two independent random variables given the first variable,
including certain outputs of stochastic models given the models parameters.
These quantities include variance-based sensitivity indices, mean squared
error of approximation with functions of the first variable, orthogonal projection
coefficients, and newly defined nonlinearity coefficients.
We define the above estimators and analyze their performance in Monte Carlo
procedures using generalized concept of an estimation scheme and its inefficiency constant.
In numerical simulations of chemical reaction networks,
using the Gillespie's direct and random time change methods,
the new schemes for sensitivity indices of conditional expectations in some cases outperformed the
ones proposed previously, and variances of some estimators significantly
depended on the simulation method being applied.

\end{abstract}
\tableofcontents
\chapter*{Introduction}
\addcontentsline{toc}{chapter}{Introduction}
Stochastic models have proven to be useful for describing a variety of physical systems, like 
chemical reaction networks involving few particle numbers of certain species \cite{Pahle2009, vanKampen_B07}, 
including gene regulatory networks \cite{Arkin1999, Rao_Wolf_Arkin_2002} and signaling pathways \cite{Lipniacki2007, Tay_2010}. 
A popular stochastic model for a well-stirred chemical reaction network is a continuous-time  
Markov chain model of reaction network dynamics (MR) \cite{mcquarrie}, which can be simulated for example using the 
 Gillespie's direct (GD) method \cite{Gillespie1976}. A number of other 
stochastic formalisms have also been used to model chemical reactions, like chemical Langevin equation, 
$\tau$-leaping, or hybrid stochastic-deterministic methods \cite{Pahle2009}. 

Sensitivity analysis is a procedure yielding sensitivity indices, which can be thought of as   
certain measures of importance of arguments in influencing the values of functions. 
As such functions one often takes outputs of deterministic 
models whose arguments are model parameters, e. g. in ordinary differential equation models of chemical kinetics \cite{Pahle2009, goutsias_2010} 
one can consider concentrations of different chemical species 
at a given moment of time in function of kinetic rates and initial concentrations of the species. In stochastic models, the 
stochastic outputs for given parameters are not constants but random variables with distribution specified by the parameters.  
For instance for an MR a stochastic output can be the number of particles of selected species at a given 
moment of time, and 
the parameters can be the initial particle numbers and kinetic rates. 
Thus as functions for the sensitivity analysis in stochastic models one usually 
considers parameters of conditional distribution of stochastic outputs,
like conditional expectation \cite{Rathinam_2010}, variance, \cite{Barmassa} or histograms \cite{Degasperi2008}, given the model parameters. 
Sensitivity analysis has been used in a variety of fields, including 
chemical kinetics 
\cite{Rabitz_Kramer_Dacol_1983, Turanyi_1990,Saltelli2005}, biochemical reaction networks \cite{Van_Riel_2006, Zhang2009}, 
nuclear safety \cite{Iooss2008}, environmental science \cite{Tarantola_Giglioli_Jesinghaus_Saltelli_2002}, 
and molecular dynamics \cite{Cooke_Schmidler_2008}, applications in chemical kinetics including 
parameter estimation \cite{Kim_Spencer_Albeck_Burke_Sorger_Gaudet_Kim_2010, Juillet2009} 
and model simplification \cite{Cristaldi_2011, Okino_Mavrovouniotis_1998, Liu_Swihart_Neelamegham_2005, Degenring2004}.

Variance-based sensitivity analysis (VBSA) is a well-established sensitivity analysis method, 
dating back to applications in chemical kinetics in the seventies \cite{Saltelli2008, Cukier1973}. \mbox{Variance-ba-}
sed sensitivity indices 
provide quantitative 
answers to questions like \textit{what average reduction of uncertainty of the model output, measured by its variance, can be achieved 
if some uncertain model parameters are determined e. g. in an experiment}, or \textit{what average error is caused by 
fixing a parameter for instance to simplify the model} \cite{Sobol2007, badowski2011}. Recently, VBSA has supported  
parameter estimation in a linear compartmental biochemical model \cite{Juillet2009} and 
simplification of a model of synthesis of an antiparasitic drug Ivermectin \cite{Cristaldi_2011}. 

In many applications outputs of physical models are being approximated by linear combinations of functions of few model parameters. 
Such approximations are used for example  
 for potential energies in molecular dynamics simulations \cite{lelievre2010free}. 
A number of approximations of this type 
has become known under the common name of high dimensional model representations \cite{rabitz_1999}.
Approximations using linear combinations of products of the first few 
orthogonal functions of model parameters, like polynomials or trigonometric functions, 
have proven to accurately imitate outputs of a number of complex models 
with many parameters, and the coefficients of the approximating linear combinations have been used to estimate
variance-based sensitivity indices  as an alternative strategy to their 
Monte Carlo estimation on which we focus on in this work \cite{Cukier1973, Li2002b, Li2002, blatman_2010, goutsias_2010}.
In the case of orthogonal polynomials being used, such approximations have also been called 
polynomial chaos expansions \cite{Khammash_2012, blatman_2010}. 
Correlation coefficient is a popular measure of the strength of linear relationship between random variables 
\cite{terrell2010mathematical,  bickel2000mathematical}. 

With a few exceptions, the above mentioned indices and coefficients have predominantly 
been used for the analysis of  input-output relationships in deterministic models, but 
they also have a potential for analogous applications to stochastic models. 
For instance VBSA can be useful for determining parameters whose measurement 
would on average most reduce the uncertainty of a given parameter of conditional distribution of the model output, like conditional 
expectation or variance, while polynomial approximations 
and correlations can provide useful information about the relationship  between variance or mean particle numbers and model parameters, 
which has been of great interest e. g. in the analysis of gene circuits \cite{Stelling2010, Becksei_2000, vanOud2002}, 

Outputs of stochastic models used in computer simulations can be represented as functions of two independent random variables - 
the first being random vector of parameters of the model and the second a 
sequence of random variables used to generate the random trajectories of its stochastic
process. 
In the first work \cite{Degasperi2008} in which VBSA of parameters of conditional distribution of stochastic outputs was considered, 
conditional histograms were chosen as outputs for sensitivity analysis and 
a grid-based method was used providing no error estimates
of the results \cite{badowski2011}. 

In our Master's Thesis in Computer Science \cite{badowski2011} we proposed unbiased estimators of variance-based sensitivity 
indices of conditional expectations of functions of two independent random variables given the first random variable, 
which can be used in a MC procedure yielding error estimates. 
In a numerical experiment such procedure led to lower mean squared error of approximation of the indices
 than a grid-based method analogous to that in \cite{Degasperi2008}. 

In this thesis we provide for the first time unbiased estimators of variance-based sensitivity indices of 
a large class of functions of conditional moments, including all conditional moments and central moments, like conditional variance, of 
functions of two independent random variables given the first variable. 
We also introduce 
new unbiased estimators of sensitivity indices of conditional expectations.   
Furthermore, we derive first unbiased estimators of means of functions of conditional moments, of products and 
covariances of these moments 
with functions of the first variable, and estimators of 
normalized sensitivity indices and correlation coefficients 
of functions of conditional moments and the first variable. 
We also introduce different unbiased estimators of coefficients 
of orthogonal projection of functions of conditional moments onto 
linear combinations of orthogonal functions of the first variable. 
We show that in a Hilbert space, 
 squared error of approximation with a linear combination of orthonormal elements using 
unbiased estimates of orthogonal projection coefficients, 
averaged over distribution of the estimates, is a sum of
variances of the estimators plus squared
error of the approximation with orthogonal projection onto span of the elements.
We use this fact to numerically compare average
mean squared errors of approximation 
of conditional expectations and variances of stochastic models outputs by linear combinations of orthogonal functions of model parameters 
with coefficients obtained using different estimators. 
We also provide unbiased estimators of mean squared errors of approximation of functions of conditional moments using functions
of the first variable, which can be used for the above approximations with linear combinations of orthogonal functions with fixed coefficients.
Approximations of conditional expectations and variances of outputs of stochastic models using orthogonal polynomials of model parameters
have already been constructed in \cite{Khammash_2012} using double-loop sampling
and convex optimization techniques. 
As we discuss in a detail in Conclusions,  
an interesting idea for the future research would be to compare the error 
 of different methods of approximation 
of functions of conditional moments 
using orthogonal functions of the first variable, like those from \cite{Khammash_2012} and this work. 
We also define nonlinearity coefficients of random independent arguments of a function, which 
can be used for obtaining lower bounds 
on probabilities of certain localizations of functions values changes, corresponding to some perturbations 
of their independent arguments. We also provide unbiased estimators of these coefficients.  

In \cite{badowski2011} we introduced the concept of an estimation scheme which is useful for defining 
generalized estimators acting not only on random variables, but also on functions, as the ones discussed in this work.
We also defined inefficiency constant of such a scheme, 
equal to the product of variance of the corresponding estimator and the number of function evaluations needed to compute it, 
so that the ratio of such constants for different schemes is equal to the ratio of variances of the final MC estimators 
for the same number of 
function evaluations carried out in MC procedures using the schemes. Thus it  
can be used for quantifying the inefficiency of using unbiased estimation schemes in MC procedures 
if function evaluations are the most time-consuming elements of the procedures. 

Here we formalize and generalize the above concepts of a scheme and its inefficiency constant to be useful for defining 
and comparing efficiency of the corresponding estimators of estimands 
depending on many functions and having vector-valued outputs, like vector of coefficients of orthogonal projection of a function 
of a random variable onto orthogonal functions of the variable. 
One of the defined schemes, called $SVar$, allows for simultaneous estimation of  
most of the above mentioned sensitivity indices and coefficients for conditional expectation and variance, 
including all sensitivity indices with respect to individual coordinates of the first variable 
and orthogonal projection coefficients onto these coordinates and constant vectors. 
We derive 
a number of inequalities between the inefficiency constants of the introduced schemes and schemes from \cite{badowski2011}. 
For instance we show that subschemes of $SVar$ for estimation of variance-based sensitivity indices 
of conditional expectation 
can have no more than four times higher and three times lower inefficiency constants than the best schemes 
for these indices from \cite{badowski2011}. 

We carried out numerical experiments testing estimators introduced in this work for the case of
conditional expectations and variances of particle numbers at a given moment of time in a 
MR simulated using 
GD \cite{Gillespie1976} and random time change (RTC) \cite{Rathinam_2010} methods. In some of our numerical experiments 
the subschemes of scheme $SVar$ for estimation of sensitivity indices of conditional expectation with respect to certain model parameters  
 had over two times lower 
inefficiency constants than the best schemes from \cite{badowski2011}. Furthermore,  
the order of estimators of orthogonal projection coefficients 
with respect to the average mean squared errors of approximations of conditional expectations and variances constructed using them 
varied from model to model. 
The numerical experiments also demonstrated significant dependence of variances of some of the introduced estimators 
on whether the GD or the RTC method is used, and on the order of reactions in the GD method. 
We discuss the relationship of these effects with 
analogous phenomena observed for different estimators in \cite{badowski2011} and \cite{Rathinam_2010}.  

The structure of this work is as follows. In Chapter \ref{overview} we give an overwiev of less common definitions 
and results from the literature, mostly from our previous master's thesis \cite{badowski2011}. Throughout this chapter we frequently
make improvements in the definitions, correct errors in the constructions or theorems, generalize the latter,
and provide more precise and comprehensive descriptions than in \cite{badowski2011}.
Chapter \ref{chapOwnRes} presents our new results. 
More common mathematical definitions and complex proofs or calculations are provided in the appendices. 
Readers not acquainted with probability theory are referred to standard textbooks like \cite{Durrett} and \cite{billingsley1979}.
Some basic definitions from this theory are also provided in Appendix \ref{appMath}.  

\chapter{Literature overview \label{overview}}

\section{Chemical reaction network}\label{secCRN}
In this section we repeat selected definitions  
from Section 1.1 of \cite{badowski2011}, improving some of them, in particular simplifying the formal 
definition of a chemical reaction network $RN$ and 
specifying the domain of reaction rates.
We shall model the time evolution of a reaction network by a 
continuous-time Markov chain defined in the next section. 
Suppose that we are given $N \in \N_+$ chemical species with symbols $X_1, \ldots, X_N$. The state of $RN$  
at a given moment of time is described by a vector of natural numbers 
$x = (x_1,\ldots,x_N)$ from the state space $E = \N^N$, 
where the $i$th coordinate of $x$ describes the number of particles of 
the $i$th species. 
$L$ chemical reactions $(R_1,\ldots,R_L)$ can occur, the $l$th reaction being described by a stoichiometric formula 
\begin{equation}\label{reac} 
\underline{s}_{l,1} X_{1} + ... + \underline{s}_{l,N} X_{N} \rightarrow  \overline{s}_{l,1} X_{1} + ... + \overline{s}_{l,N} X_{N}. 
\end{equation}
We call vector $\underline{s}_l = (\underline{s}_{l,i})_{i=1}^{N}$ 
the stoichiometric vector of reactants and $ \overline{s}_l = (\overline{s}_{l,i})_{i=1}^{N}$ of products 
of the $l$th reaction. 
In this whole work for $n \in \N_+$ we denote $I_n = \{1,\ldots, n\}$ 
and for vectors $v,w \in \R^n$, we write $v \geq w$ if $v_i \geq w_i$, $i \in I_n$ 
(where the last notation means for each $i \in I_n$). 
We require that $\underline{s}_l \geq 0 $ and $\overline{s}_l \geq 0$, where by $0$ we mean here $(0,\ldots,0) \in \R^N$. 
We define transition vector of the $l$th reaction as $s_{l} =  \overline{s_l} -  \underline{s_l}$. 
In the model of dynamics of reaction network 
discussed in the next section occurrence of the $l$th reaction will make the system at state $x$ move to state $x +  s_l$. 
There is given a set  $B_{RN} \in \R^m$
for some $m \in \N_+$, called the set of admissible reaction rates of $RN$. 
We define a measurable space $\mc{S}_{RN}=(B_{RN},\mc{B}(B_{RN}))$ (see Section \ref{appMath}).
For each $l \in I_l$, there is given a real 
nonnegative function $a_l$ measurable on $\mc{S}_{RN,E} = (B_{RN,E},\mc{B}_{RN,E}) = \mc{S}_{RN}\otimes\mc{S}(E)$,
called reaction rate of the $l$th reaction. 
Intuitively speaking, in the mathematical model we discuss in the next section $a_l(k,x)$ 
describes how quickly the $l$th reaction is proceeding in the state $x$ and for the rate constants $k$. 
We require that $a_l(k,x)  = 0$ if for some $i\in I_N$, $x_i < \underline{s}_i $, which  means that 
there are too few particles of a certain reactant in 
the system for the reaction to occur. For example in the stochastic version of mass action kinetics \cite{Kurtz1986}, 
for each $k = (k_i)_{i=1}^{L} \geq 0$ and $l \in I_L$, 
\begin{equation}\label{alkx} 
a_l(k,x) = k_l\prod_{i=1}^N{x_i \choose \underline{s}_{l,i} }, 
\end{equation} 
which is rate constant of the $l$th reaction times 
the number of possible ways in which the reactants can collide for the $l$th reaction to happen. 
Formally, we define the chemical reaction $R_l$, $l \in I_L$, to be a triple 
\begin{equation} 
R_l(k) = (a_l(k,\cdot), \underline{s}_l, \overline{s}_l), 
\end{equation} 
and the chemical reaction network $RN$ is defined as a sequence of reactions 
\begin{equation}\label{NTfun} 
RN(k) = (R_l(k))_{l=1}^L, 
\end{equation} 
both being functions of rate constants $k \in B_{RN}$. 

\section{\label{secMRCP}Continuous-time Markov chain model of reaction network dynamics with constant parameters (MRCP)}
See Appendix \ref{appHMC} for an introduction to stochastic processes, including continuous-time homogeneous Markov chains (HMC). 
Below we repeat the definition of a discrete stochastic chemical reaction network with constant parameters from our
previous work \cite{badowski2011}, 
calling it this time continuous-time Markov chain model of reaction network dynamics with constant parameters or shortly MRCP. 
Let the reaction network $RN$ and othernotations be as in Section \ref{secCRN}, let  $p = (k,c) \in B_{RN, E}$,
and $T= [0,\infty)$.
\begin{defin}
MRCP corresponding to $RN$ and $p$ is defined as  
a nonexplosive HMC on $E$ with times $T=[0,\infty)$, with deterministic initial distribution $\delta_{c}$ and 
$Q$-matrix with intensities equal to, for each $x,y \in E$, $x \neq y,$ 
\begin{equation}\label{qxyEq}
q_{x,y} = \sum_{l: \  y = x + s_l}\ a_l(k,x), 
\end{equation}
where we used the convention that sum over an empty set is zero. 
\end{defin}
Unfortunately, for some $p$ and $RN$ as above such nonexplosive HMC may not exist \cite{KurtzReview2010}. 
A useful criterion for its existence shall be provided in Section \ref{secMRNew}. 
We denote the distribution of a nonexplosive MRCP corresponding to $p = (k, c)$ and chemical reaction network $RN$ for which it exists 
as $\mu_{MRCP}(RN(k), c)$. 
Let $RN$ be some reaction network as in Section \ref{secCRN} and $p = (k,c) \in B_{RN, E}$ be some its parameters. 
Below we describe two constructions of processes which yield MRCPs corresponding to $RN$ and $p$, if any such  
MRCP exists. The description of these constructions  is similar as in \cite{badowski2011} in Section 1.3, 
but we do it in a more formal way and 
correct a number of oversights we made in \cite{badowski2011}, like 
overlooking the case when the set of reactions with positive rates is empty in the second construction. 
The first construction corresponds to the GD method for simulating MRCP introduced in \cite{Gillespie1976}, while the 
second to the RTC algorithm from 
\cite{Rathinam_2010}, and is a special case of the random time change representation of Markov processes due to Kurtz 
(\cite{Kurtz1986}, Section 6.4). 
In the below constructions we inductively define the initial jump chain $(Z_n)_{n \geq 0} $ and initial jump 
times $(J_n)_{n \geq 0}$. 
Let the initial explosion time $\zeta$ be defined as in (\ref{expTime}) in Appendix \ref{appHMC} using the initial jump times. 
We assume without explicitly writing this in the constructions that $Z_0 = c$ and $J_0 = 0$, and that 
each construction ends by changing, for some arbitrary $c_1 \in E$, and on each elementary event $\omega$ for
 which $\zeta(\omega) < \infty$, all the initial jump chain and times variables with positive indices 
to $c_1$ and $\infty$, respectively, so 
that we receive final jump chain and times $Z_n'$, $J_n'$, that are jump times and chain of some unique 
nonexplosive process $Y$. If there exists 
any MRCP corresponding to $RN$ and $p$, then $\zeta = \infty$ a. s. and $Y$ is such an MRCP. 
For convenience in the constructions below the dependence on the elementary event $\omega$ is omitted. 
\begin{constr}[GD construction]\label{GCon} 
Let $U_1,U_2,\ldots $ be independent identically distributed (i. i. d.),  $U_1 \sim \U(0,1)$, 
and $E_1, E_2,\ldots$ i. i. d., $E_1 \sim \Exp(1)$ (see Appendix \ref{appMath}). 
Suppose that $Z_i$ and $J_i$ have been defined for $i\in \N$. Let 
\begin{equation}
q = \sum_{l=1}^L a_l(k,Z_i).
\end{equation}
If $q = 0$, then we define 
\begin{equation}
 J_{i+1} = \infty,\ Z_{i+1} = Z_{i},
\end{equation}
otherwise we take 
$J_{i+1} = J_i + \frac{E_{i+1}}{q}$, 
and for 
\begin{equation}
 l = \min\{m \in I_{L}: \frac{1}{q}\sum_{n=1}^{m} a_{n}(k)(Z_i) \geq \U_{i}\},
\end{equation}
we set 
\begin{equation}
Z_{i + 1} = x + s_l.
\end{equation}
\end{constr} 
For $k \in B_{RN}$ and $x \in E$, we denote 
$B(k, x) = \{l \in I_L:\ a_l(k, x) > 0\}$ - the 
set of indices of reactions with positive rates in state $x$ and for the rate constants $k$. 
\begin{constr}[RTC construction]\label{RTCCon} 
Let us consider $L$ independent Poisson processes $(N_l)_{l=1}^L$ with unit rates (see Appendix \ref{appHMC}).  
The initial jump times and chain in this construction are the jump times and chain of 
any right-continuous process $Y$ satisfying  
\begin{equation}\label{intEqu} 
Y_t =  c + \sum_{l=1}^L s_l N_l(\int_0^t \! a_l(k, Y_s) \, \mathrm{d}s), 
\end{equation} 
for $t < \zeta$ (\cite{Kurtz1986} Section 6 Theorem 4.1 a)). 
Suppose that the $i$th call of function $N_l.next$ returns the $i$th holding time of $N_l$. 
We set for $l\in I_L$ 
\begin{equation}
\tau_{0,l} = N_l.next. 
\end{equation}
Let us assume that  $Z_{i}$, $J_{i}$ and $\{\tau_{i ,l} \}_{l \in I_L}$  have been defined for $i \in \N$. 
If $B(k,Z_i)$ is empty, then we set 
\begin{equation} 
J_{i+1} = \infty,\ Z_{i+1} = Z_i,
\end{equation} 
and finish the inductive step. 
Otherwise, we set 
\begin{equation} 
J_{i+1} = J_i + \min_{l \in B(k,Z_i)} \left\{ \frac{\tau_{i,l}}{a_l(k, Z_i)}  \right\}, 
\end{equation} 
and for a certain $l$ realizing the above minimum,
\begin{equation} 
 Z_{i+1} = Z_i + s_l,\quad \tau_{i+1,l} = N_l.next.
\end{equation} 
Furthermore, for each $m \in B(k,Z_i),\ m \neq l$, we set
\begin{equation}
 \tau_{m,i+1} = \tau_{m,i} - a_m(k, Z_i) (J_{i+1} - J_i)
\end{equation}
and for reaction indices $l \notin B(k, Z_i)$, 
\begin{equation}
 \tau_{l,i+1} = \tau_{l,i}.  
\end{equation}
\end{constr} 

As we discussed in \cite{badowski2011}, 
in all constructions of processes corresponding to some stochastic simulation algorithm one
uses a random variable $R$, which in the algorithm is generated e. g. using a random number generator, 
to build the random trajectories of the process. For instance  
for the first construction of MRCP above we have $R = (U_i, E_i)_{i \geq 0}$, while for the second one $R =  (N_i)_{i=1}^L$. 
As in \cite{badowski2011}, we call $R$ artificial noise variable or simply noise variable. 
For some construction of MRCP as above, let $\mc{S}_R = (B_R,\mc{B}_R)$ be the measurable space 
of possible values of the noise variable from this construction. 
We define a function $h$ from $B_{RN,E}\times B_R$ to $E^T$ to be such that for each $p \in B_{RN, E}$ and $r \in B_R$, 
\begin{equation}\label{MRCPFun}
h(p,r)
\end{equation}
is equal to the trajectory $Y(\omega)$ of the process built in the construction using 
parameters $p$ and for $\omega\in \Omega$ and noise variable $R$ such that 
 $R(\omega)=r$. In particular, process $Y$ created in the construction using some $p$ and $R$ 
is equal to $h(p,R)$ on $\Omega$. 

\section{\label{genParSec}Models with random parameters and their outputs} 
As we discussed in Section 2.1 of \cite{badowski2011}, there are many situations when one may want to treat parameters of a model as 
random variables $(P_i)_{i=1}^N$ rather than constants. Shortly, when the parameters represent uncertain quantities, there are 
two types of such variables distinguished in the literature - stochastic and epistemic ones. Stochastic variables are changeable in the
modelled system, 
like particle numbers in equilibrium distribution of a reaction network, and their 
uncertainty, measured e. g. by their variance, cannot be reduced by gaining further knowledge about the system. Epistemic variables 
are constants in the modelled system, whose exact values are unknown, which is often the case for reaction rates. 
Distribution of epistemic variables reflects our best judgement about their possible values, based e. g. 
on the uncertainty estimates of experimental measurements, and uncertainty of these variables 
can be reduced by gaining further knowledge about the system, 
like performing more precise experiments. 
The definition of a continuous-time Markov chain model of reaction network dynamics (MR) and its construction 
we provide below are more precise and 
general than the ones we proposed in \cite{badowski2011}, e. g. because we specify the domain of 
distribution of parameters in the definition and 
do not require the existence of an MRCP for each value of the parameters. 
Let $RN$, $E$ and $T$ be as in the previous sections and let $\mu_0$ be some probability distribution on $\mc{S}(E^T)$ 
and $\nu$ on $\mc{S}_{RN,E}$. 
Let $\wt{\mu}: B_{RN, E} \times \mc{B}(E^T)\rightarrow \R$ be such that for each 
$p=(k,c) \in B_{RN,E}$ for which an MRCP corresponding 
to $RN$ and $p$ exists $\wt{\mu}(p,\cdot)$ is equal to $\mu_{MRCP}(RN(k), c)$ i. e. the distribution of such MRCP,  
and for other values of $p$ it is equal to $\mu_0$.
\begin{defin}\label{MRdef} 
We say that a pair $M = (P,Y)$ is an MR corresponding
to a chemical reaction network $RN$ and (distribution of parameters) $\nu$,
if for $\nu$ almost every (a. e.) $p$, a
 MRCP corresponding to $p$ and $RN$ exists, $P$ is a random vector with $\mu_P=\nu$,
$Y$ is a right-continuous nonexplosive process on $E$ with times $T$, 
and $\wt{\mu}$ is a conditional distribution (see Definition \ref{defMu}) of $Y$ given $P$. 
$P$ is called the parameters and $Y$ the process of $M$. 
\end{defin}
If for $\nu$ a. e. $p =(k,c)$ MRCP corresponding to $RN$ and parameters $p$ exists, 
then MR $(P,Y)$ corresponding to $\nu$ and $RN$ can be constructed similarly as for 
the previous less general definition of MR in \cite{badowski2011}. 
For some $P = (K, C) \sim \nu$ independent of the artificial noise variable $R$ used by one of the constructions 
of MRCP from the previous section, one sets $k= K(\omega)$ and $c = C(\omega)$ at the beginning of this construction and then proceeds with it. 
Using function $h$ (\ref{MRCPFun}) corresponding to the construction of MRCP, the process of MR we defined above can be written as 
\begin{equation}\label{formP} 
Y = h(P,R). 
\end{equation} 
$Y$ conforms to the definition of a process of MR with parameters $P$ due to Theorem 
\ref{aveSecFin} in Appendix \ref{appMath}. 

Analogously we can define constant and random parameter versions $M=(P,Y)$, as well as constructions 
in form of a function $h(P,R)$ of independent parameters $P$ and some noise term $R$ of other 
stochastic or deterministic models used in computer simulations, where for deterministic 
models $R$ can be chosen constant. In particular this applies to models used for simulation of chemical kinetics, like   
Euler-Maruyama approximation of solutions of the
chemical Langevin equation \cite{Wilkinson2006}, Euler scheme for ordinary differential equations of chemical kinetics, 
or hybrid stochastic-deterministic methods \cite{Pahle2009}. 

By an output of such a model $M = (P,Y)$ 
we mean a random variable $g(M)$ for some function $g$,
measurable from the product measurable space of the image of $M$ to $\R^n$, for some $n \in \N_+$.
For an MR the output can be e. g. the number of particles of the $i$th species at the moment $t$, 
while for a deterministic model of chemical kinetics this can be concentration of some species at a given time. 
One can also consider vector-valued outputs, like vectors of particle numbers of different species
or single-sample histograms of numbers of a given particle which we define below. 
As discussed in the introduction, for outputs of stochastic models, like particle numbers at a given time for MR,
we shall be interested in their certain parameters of conditional distribution, like conditional expectation 
given the model parameters. According to the above definition of an output,
such parameters of conditional distribution are themselves model outputs, 
which can be expressed as functions of only the model parameters. 

Conditional expectation of an integrable random variable $Z$ given a random variable $X$, denoted as $\E(Z|X)$, 
is a random variable $f(X)$ for a certain function $f$, where $f(x)$ can be informally thought of as the mean value of $Z$ on the set $X = x$ 
(see Appendix \ref{appMath} for a precise Definition \ref{condDef}). 
As in \cite{badowski2011}, for $p>0$ we define
 $L^p_n(\mu)$ to be the space of classes of equivalence of the relation of being 
equal $\mu$ a. e. considered on random vectors $X= (X_i)_{i=1}^n$ such that $X_i \in L^p(\mu)$ for each $i \in {I_n}$
 (see Appendix \ref{appMath} 
for more details on $L^p(\mu)$ spaces including the associated notational conventions, which we by analogy extend to $L^p_n(\mu)$ spaces, 
in particular if $\mu=\PR$ is the implicit probability measure, then $L^p_n(\PR)$ is denoted simply as $L^p_n$).  
For an $\R^n$-valued random vector $Z = (Z_i)_{i=1}^n$, we define 
$\E(Z) = (\E(Z_i))_{i=1}^n$. If $Z\in L^1_n$ and $X$ is a random variable, then we define the conditional expectation of $Z$ given $X$ as 
\begin{equation} \label{genCond} 
\E(Z|X) = (\E(Z_i|X))_{i=1}^n. 
\end{equation} 
From the fact that conditional expectation is contraction in $L^p$ (see Theorem \ref{contrac} in Appendix \ref{appMath}) 
it follows that $\E(Z|X) \in L^p_n$ if $Z \in L^p_n$.
For $k \in \N_{+}$ and numbers  $x_{min}$, $x_{max}$ such that $L =  x_{max} -x_{min} >0$, let 
$B_i = \left[x_{min}+ \frac{(i-1)L}{k}, x_{min}+ \frac{iL}{k}\right)$, $i \in I_k$. 
The corresponding histogram function $\hist$ is defined for $x \in \R$ as 
\begin{equation} 
\hist(x) = \left(\I_{B_i}(x)\right)_{i=1}^k. 
\end{equation} 
A (single-sample) histogram corresponding to a real-valued random variable $Z$ is defined as  $\hist(Z)$. 
Note that $\hist(Z) \in L^p_{k}$ for each $p > 0$.  
Conditional histogram of $Z$ given some random variable $X$ is defined as $\E(\hist(Z)|X)$ and mean histogram as $\E(\hist(Z))$. 
For a random vector $X = (X_1,\ldots, X_N)$ and any $J \subset I_N$, we denote $X_J = (X_i)_{i \in J}$.  
For $J=\emptyset$, we define $X_J=\emptyset$. 
For $Z$ integrable, for each  $J \subset K \subset I_N$, we have the following iterated expectation property \cite{Durrett}
\begin{equation}\label{doubleCond}
\E(\E(Z|X_K)|X_J) = \E(Z|X_J),
\end{equation}
where for $Y \in L^1$, by $\E(Y|\emptyset)$ we mean $\E(Y)$.
For a stochastic model $M= (P,Y)$ whose process $Y$ 
has form $h(P,R)$ for some variable $R$ independent of $P$ as in (\ref{formP}), a stochastic
output $g(M)$ is equal to $f(P,R)$ where $f$ is defined by formula 
\begin{equation}\label{obsForm}
f(p,r) = g(p, h(p,r)),\quad p \in B_{RN, E}, r \in B_R.
\end{equation}
In such case, thanks to Theorem \ref{indepCond} we have 
\begin{equation}
\E((f(P,R)|P) = \E(f(p,R))_{p=P}.
\end{equation}

\section{Variance for random vectors}\label{secOrthog} 
In this Section we mainly reformulate some theory from Section 2.5 of \cite{badowski2011}. 
Reader not acquainted with Hilbert space theory is referred to Appendix \ref{appHilb}. 
An example of a Hilbert space is $L^2$ with scalar product given by 
\begin{equation}\label{scall2} 
(X,Y) = \E(XY). 
\end{equation} 
As in \cite{badowski2011}, we denote the norm it induces as $||\cdot ||$ and the metric $d$. 
Let $n \in \N_{+}$, 
$<,>$ be a scalar product in $\R^n$, and $(a_{ij})_{i,j \in I_n}$ be 
the real numbers such that for each $x,y \in \R^n$, 
\begin{equation}\label{scalAij} 
<x,y> = \sum_{i,j\in I_n} a_{ij} x_iy_j. 
\end{equation} 
For the standard scalar product we have $a_{ij} = \delta_{ij}$ ($\delta_{ij}$ being the Kronecker delta). 
The norm induced by $<,>$ is denoted as $|\cdot|$ and the distance as $\wt{d}$. 
$L^2_n$ with scalar product $(,)_{n}$, defined for $X,Y \in L^2_n$ as
\begin{equation}\label{scalarFun} 
(X,Y)_n = \E(<X,Y>)= \sum_{i,j\in I_n}a_{ij}(X_i,Y_j), 
\end{equation} 
is a Hilbert space equal to the direct sum of $L^2$ given by $<,>$  (see definition in Theorem \ref{genDSpace} in Appendix \ref{appHilb}) and 
denoted as $\bigoplus_{<,>}L^2$. 
As in \cite{badowski2011}, the norm induced by the scalar product $(,)_n$ is denoted as  $||\cdot ||_{n}$ and the metric as $d_{n}$. 
For some random variable $X$ and $p > 0$, 
let $L^p_{X}$ be the subspace of $L^p$ consisting of all its classes of random variables 
containing an element $f(X)$ for some measurable real-valued
 function $f$, and $L^p_{n,X}$ be an analogous subspace of $L^p_n$ but for functions 
$f$ with values in $\R^n$. For $p \geq 1$, $L^p_{X}$ is a closed subspace of $L^p$, because from the change of variable Theorem 
\ref{thchvar} the map  $[f(X)]_{\PR} \rightarrow [f]_{\mu_X}$ (see Appendix \ref{appMath}) 
is a linear isometry between $L^p_{X}$ and the complete space $L^p(\mu_X)$. 
In particular $L^2_{X}$ is a Hilbert space and $L^2_{n,X}$ with scalar product $(,)_n$ is equal to the 
direct sum of $L^2_{X}$ given by $<,>$. 
Conditional expectation $\E(\cdot|X)$ is an 
orthogonal projection from $L^2(\PR)$ onto $L^2_{X}$ (see Lemma \ref{condort} in Appendix \ref{appHilb}), 
so that from Theorem \ref{projDS} in Appendix \ref{appHilb} 
it follows that the generalized conditional expectation $\E(\cdot|X)$ given by (\ref{genCond}) 
is orthogonal projection from $L^2_n$ to $L^2_{n,X}$. 
In particular, $\E(Z|X)$ is the best approximation of $Z$ in $L^2_{n,X}$ and the squared error of 
this approximation fulfills 
\begin{equation}\label{condError}
d^2_n(Z, \E(Z|X)) =  ||Z||_n^2 - ||\E(Z|X)||_n^2. 
\end{equation}
Similarly as in \cite{badowski2011}, we define variance of a random vector $Z \in L^2_n$, $n \geq 2$, as follows, using for it
informally the same notation as for one-dimensional variance,
\begin{equation}\label{varGen}
\begin{split}
\Var(Z) &= d^2_n(Z, \E(Z)) = ||Z||^2_n - ||\E(Z)||^2_n \\
&= \E(|Z|^2) - |\E(Z)|^2.
\end{split}
\end{equation}
When the probability measure considered is $\mu$ rather than $\PR$, we write $\Var_{\mu}(Z)$ instead of $\Var(Z)$. 
Standard deviation of $Z$ is defined as  
\begin{equation}
\sigma(Z) = \sqrt{\Var(Z)}. 
\end{equation}
As in \cite{badowski2011}, conditional variance of $Z$ given $X$ is defined as
\begin{equation}\label{condGen}
\begin{split}
\Var(Z|X) &= \E(\wt{d}^2_n(Z, \E(Z|X))|X) \\
&= \E(|Z|^2 + |\E(Z|X)|^2 - 2<Z, \E(Z|X)>|X) \\
&=\E(|Z|^2|X) - |\E(Z|X)|^2,
\end{split}
\end{equation}
where in the second equality we used the fact that $\E(<Z, \E(Z|X)>|X) = |\E(Z|X)|^2$, which follows from Theorem \ref{condexpX} 
from Appendix \ref{appMath} and from (\ref{scalAij}). 
We have
\begin{equation}\label{d2n}
\E(\Var(Z|X)) =  ||Z||^2_n - ||\E(Z|X)||^2_n = d^2_n(Z, \E(Z|X)), 
\end{equation}
where in the first equality we used the iterated expectation property (\ref{doubleCond}) applied to the last term in (\ref{condGen}), 
and in the second equality from (\ref{condError}). 
From (\ref{d2n}) and the third term in (\ref{varGen}) we receive a formula already derived in \cite{badowski2011},
\begin{equation} \label{aveVarError}
\Var(Z) = \E(\Var(Z|X)) + \Var(\E(Z|X)).  
\end{equation}
As we shall prove in Section \ref{secCondMoms} for $f$ measurable such that $f(Z) \in L^2_n$, 
\begin{equation}\label{varmuzx}
\Var(f(Z)|X) = \Var_{\mu_{Z|X}(X,\cdot)}(f). 
\end{equation}

\section{\label{secVBSA}ANOVA decomposition and variance-based sensitivity indices}
In this section we mainly reformulate some definitions and theorems from sections 3.1-3.3 of \cite{badowski2011}. 
Let $X=(X_i)_{i=1}^N$ be a random vector with $N \in \N_{+}$ independent coordinates. 
Let $I = I_N$ and $J \subset I$. We define $\sim J = I \setminus J$, 
and $X_J$ as in Section \ref{genParSec}. 
For $J \neq \emptyset$, we denote $\mu_J=\mu_{X_J}$, and define $L^2_{n,X_J}$ as in Section \ref{secOrthog}. 
For $J=\{i\}$ we write $i$ rather than $\{i\}$ in the above and below introduced notations.
$L^2_{n,X_{\emptyset}}$ is defined to consist of classes from $L^2_{n}$ containing 
constant $\R^n$-valued random vectors. 
For each $J \subset I$, we define $L^2_{n,J}$ to be the subspace of $L^2_{n,X_J}$ 
consisting of its classes containing variables $Z$ such that for each $i \in J$, 
\begin{equation}\label{zeroExp} 
\E(Z|X_{\sim i}) = 0. 
\end{equation} 
Note that $L^2_{n,\emptyset} = L^2_{n,X_{\emptyset}}$ and that 
due to Theorem \ref{indepCond}, for $Z = g(X_J)$ for a measurable function $g$, (\ref{zeroExp}) is equivalent to 
\begin{equation}\label{zeroInt} 
\int \! g(X_{J \setminus \{i\}}, x_i) \, d\mu_i = 0, 
\end{equation}
where we used a convenient notation for integrating $X_i$ out over its distribution. 
From (\ref{zeroExp}) and iterated expectation property it follows that for each nonempty $ J \subset I$ and variable $Z \in L^2_{n,J}$, it holds 
\begin{equation}\label{zj0} 
\E(Z) = 0. 
\end{equation} 
In \cite{badowski2011} we proved as Theorem 5 the following theorem (see Definition \ref{defHilb} of a direct sum in a 
Hilbert space). 
\begin{theorem}
For $n \in \N_+$ and Hilbert space $L^2_{n,X}$ with certain scalar product $(,)_n$   
defined as in Section \ref{secOrthog}, it holds 
\begin{equation}
L^2_{n,X} = \bigoplus_{J \subset I} L^2_{n,J}. 
\end{equation}
\end{theorem}
In the proof of Theorem 5 in \cite{badowski2011} we also showed that if for some measurable $f$, 
$Z=f(X)\in L^2_{n,X}$, then there exist measurable functions $f_J$ 
such that $f_J(X_J) \in L^2_{n,J}$, $J \subset I$, and 
\begin{equation}\label{anovaDec}
f(X) = \sum_{J\subset I}f_J(X_J).
\end{equation}
Random variables $f_J(X)$, $J \subset I$, are uniquely determined a. s. and we call $(f_J(X_J))_{J \subset I}$ the 
ANOVA decomposition of $f(X)$, as such decompositions 
for the case of $n=1$ were used under this name in the literature (see \cite{badowski2011} for references). 
From (\ref{zeroInt}), (\ref{indepCond}), and Fubini's theorem, for each $K \subset I$ we have 
\begin{equation}\label{subAnova}
\E(f(X)|X_K) = \sum_{J\subset K}f_J(X_J).
\end{equation}
Denoting for $J \subset I$, 
\begin{equation}
V_J = \Var(f_J(X_J)),
\end{equation}
and using (\ref{zj0}), (\ref{anovaDec}), and orthogonality of the elements of ANOVA decomposition, we receive for $D= \Var(f(X))$,
\begin{equation}\label{sumvar}
D  =  \sum_{K \subset I} V_K.
\end{equation} 
For $|J|>1$, $V_J$ has been called an interaction index between the  
variables with indices in $J$ in the literature \cite{Saltelli2005}, and as we proved in \cite{badowski2011} 
$V_J$ it can be interpreted as difference of squared 
errors of the best approximation of $f(X)$ using linear combinations of functions of proper subvectors of $X_J$,
 and of the whole vector $X_J$. 
We define Sobol's indices $S_J =\frac{V_J}{D}$, $J \subset I$. 
We have
\begin{equation}\label{sumSob}
1  =  \sum_{K \subset I} S_K \geq \sum_{i \subset I} S_i,
\end{equation} 
equality in the rhs inequality meaning that 
\begin{equation}
f(X)= \sum_{i=1}^N f_i(X_i).
\end{equation}

For some $n \in \N_{+}$, let $Z \in L^2_n$, $D=\Var(Z) > 0$, and let now $X = (X_i)_{i=1}^N$ be 
a random vector (with not necessarily independent coordinates). 
The main sensitivity index of $Z$ given $X_J$ is defined as 
\begin{equation}\label{VXJ} 
V_{X_J} = \Var(\E(Z|X_J)). 
\end{equation} 
From (\ref{d2n})  and (\ref{aveVarError})
it follows that $D - V_{X_J}$ is equal to the squared error of the best approximation of $Z$ in $L^2_{n,X_J}$. 
Suppose that  $Z= f(X)$  for a certain measurable function $f$. 
The total sensitivity index of $f(X)$ with respect to $X_J$ is defined  as
\begin{equation}\label{VXJtot} 
V_{X_J}^{tot} = D - V_{X_{\sim J}}. 
\end{equation} 
From (\ref{aveVarError}),
\begin{equation}\label{evarz}
 V_{X_J}^{tot} = \E(\Var(f(X)|X_{\sim J})),
\end{equation}
so using further  (\ref{d2n}) we receive that 
$V_{X_J}^{tot}$ is the squared error of the best approximation of $f(X)$ in $L^2_{n,X_{\sim J}}$. 
Sensitivity indices $V_{X_J}$ and $V_{X_J}^{tot}$ divided by $D$ are 
called Sobol's main and total sensitivity indices or normalized sensitivity indices, 
and denoted $S_{X_J}$ and $S_{X_J}^{tot}$. 
Let us assume that coordinates of $X$ are independent so that we can apply the ANOVA decomposition. 
Then $V_i = V_{X_i}$, $i \in I$, and using (\ref{subAnova}) 
and (\ref{sumvar}) we receive that $V_{X_J}$ is 
a sum of all main and interaction indices $V_{K}$, $K \subset J$, and  $V_{X_J}^{tot}$ 
is a sum of indices $V_K$, $K \cap J \neq \emptyset$, which 
provides some intuition for the words main and total in the names of the indices and from which it follows that 
\begin{equation}
0\leq S_{X_J} \leq S_{X_J}^{tot} \leq  1.
\end{equation}
Furthermore, we then have from (\ref{evarz}), (\ref{condGen}), and Theorem \ref{indepCond} that
 \begin{equation}
V_{X_J}^{tot} = \E((\Var(f(X_J,z)))_{z=X_{\sim J}}), 
 \end{equation}
so in a sense given by this formula $V_{X_J}^{tot}$ can be thought of as an average variance of $f(X)$ with respect to $X_J$.

Let us consider an output $Z = g(M) \in L^2_n$ of an MR $M=(P,Y)$ with parameters 
$P = (P_i)_{i=1}^{N_P}$ and corresponding to a reaction network $RN$. 
Note that conditional distribution of $M$ given $P$ is specified by Definition \ref{MRdef} 
and thus from formula (\ref{condCond}) and iterated expectation property 
the distributions of $\E(Z|P_J)$ for different $J \subset I$
are specified by $RN$, $g$, and $\mu_P$. 
Thus the main sensitivity index with respect to $P_J$, denoted as $V_{P_J}$,
 $D = \Var(Z)$, and $Ave=\E(Z)$ are all determined by this data. 
As discussed in Section \ref{genParSec}, 
for a given construction of an MR using the noise variable $R$ one can provide construction of $g(M)$ of form $f(P,R)$ 
for which some further sensitivity indices can be considered, like 
\begin{equation}\label{VRTot} 
V_{R}^{tot} = D - V_P.
\end{equation} 
Its value, by inspection of the rhs of (\ref{VRTot}), is also determined by $RN$, $g$, and $\mu_P$, 
and from (\ref{aveVarError}) it is equal to $AveVar = \E(\Var(f(P,R)|P))$.  
We denote the main sensitivity index with respect to $P_J$ of conditional expectation 
$\wt{g}(P) = \E(Z|P)$,
as $VE_{P_J}$ or $\wt{V}_{P_J}$ 
and such total sensitivity index as $VE_{P_J}^{tot}$ or $\wt{V}_{P_J}^{tot}$. 
From the iterated expectation property it follows that $\E(\wt{g}(P)|P_J) = \E(Z|P_J)$, and therefore
\begin{equation}\label{tVCJ}
\wt{V}_{P_J} = \Var(\E (\wt{g}(P)|P_J)) = V_{P_J}
\end{equation}
and
\begin{equation}\label{tildeVP}
\wt{V}_{P_J}^{tot} = \wt{V}_{P} - \wt{V}_{P_{\sim J}} = V_{P} - V_{P_{\sim J}}. 
\end{equation}
For the special case of $J = \{i\}$, we often write $i$ in place of $P_J$ in the above notations. 
Analogous observations about sensitivity indices can be made and notations introduced also for other types of stochastic models. 

\section{Application of VBSA to selection of parameters for determination}
\label{varRedSec}
Certain possible applications of VBSA were described in our previous work \cite{badowski2011} and include 
identifying parameters which can be fixed in order to simplify the model, 
computing measures of average dispersion of stochastic models, as well as planning experiments, but, as discussed in the introduction, 
the indices have been used also for other purposes, like to assist the process of parameter estimation. 
In this section we describe in a detailed and novel way the possibility of application 
of main sensitivity indices to comparing the average decreases of the model output uncertainty
resulting from determination of values of uncertain model parameters, e. g. through a measurement,
which can be useful in planning of experiments.  
%
See Section 3.5 in \cite{badowski2011} or \cite{Saltelli2008}
for alternative descriptions. 
For some $n, N \in \N_+$, let us consider some model $M=(P,Y)$ whose output 
is $g(M) \in L^2_n$ for some measurable function $g$. 
The uncertainty of model output can be quantified using the output variance $D = \Var(g(M))$ for some variance for random vectors as in 
Section \ref{secOrthog}. 
Let us assume that the subvector $P_J$ of parameters $P$ consists of epistemic parameters of the model and we can determine their values exactly, 
for instance by measuring them, which can be a useful idealisation when the uncertainty of these parameters after
 the measurement is negligibly small. 
 For some conditional distribution $\mu_{M|P_J}$ of $M$ given $P_J$, if the determined value of $P_J$ 
is $p_J$, we update $M$ to a new model $M'$ with distribution equal to $\mu_{M|P_J}(p_J, \cdot)$.
In case of $M$ being an MR 
we can take  $M'=(P',Y')$ to be an MR with the same reaction network but distribution of parameters 
 $\mu(P|P_J)(p_J,\cdot)$, which for $P$ with independent parameters can be taken to be the distribution of $P' = (p_J, P_{\sim J})$.  
The variance of output of the new model fulfills 
\begin{equation} 
\Var(g(M')) = \Var_{\mu_{M|P_J}(p_J, \cdot)}(g). 
\end{equation} 
We received $p_J$ as an outcome of determination, e. g. through a measurement, 
of value of the initially uncertain random vector $P_J$, so the expected decrease of variance from the initial one
$D$ can be obtained by averaging over such possible outcomes as follows
\begin{equation}
\E(D - \Var_{\mu_{M|P_J}(P_J, \cdot)}(g)) = D - \E(\Var(g(M)|P_J)) = V_{P_J},
\end{equation}
where in the first equality we used (\ref{varmuzx}) and in the last (\ref{aveVarError}). We received the main sensitivity
index of $g(M)$ given $P_J$, thus 
$S_{P_J}$ tells by what fraction the model output variance is reduced on average if we determine the value of $P_J$. 
Note that since the main sensitivity indices of the output of a model and of its conditional expectation given the parameters are the same 
(see (\ref{tVCJ})),
then so are their average decreases of variances. 
Note also that 
for stochastic outputs which are not functions of the parameters, 
like particle numbers in an MR, even if all the parameters are epistemic and are determined 
there will still be remaining average output variance $D-V_P$. 
From comparing values of $V_{P_J}$ or $S_{P_J}$ for different subvectors $P_J$ consisting of epistemic parameters one can get to know 
determining which of them leads on average to higher reduction of variance of the model output.  
This knowledge can 
assist the decision what parameters should be determined next, e. g. in an experiment, if the goal is to improve the 
precision of the model predictions. 
After some parameters are determined, the above procedure can be repeated with the updated model $M'$ as above. 
For an $\R^n$-valued output $g(M)$ for $n >1$, like a vector of 
different particle numbers or their conditional expectations at a given time, 
the scalar product $<,>$ used in the definition of its variance as in Section \ref{secOrthog} 
can be  given for example by numbers $a_{ij} = c_{i}\delta_{ij}$ as in (\ref{scalAij}), 
where $c_{i}$ is a weight describing how important it is to be able to predict
the $i$th coordinate of $g(M)$ more precisely using the model. 

\section{\label{secSchemesPrev}Estimands on pairs and their unbiased estimation schemes}
Let us recall certain concepts from Section 4.3 of \cite{badowski2011} like 
generalized estimands, which we call here estimands on pairs, 
and their unbiased estimation schemes. We make numerous changes to correct errors 
in the previous definitions and 
to increase their compatibility with future generalizations in Section \ref{secUnbiased}. 
See Appendix \ref{appStatMC} for an introduction to statistics including standard definitions of estimands and estimators.  

For a measure $\mu$ we denote its measurable space as
\begin{equation}\label{smudef}
\mc{S}_{\mu}= (B_\mu, \mc{B}_\mu). 
\end{equation}
Let $N \in \N_+$. For a sequence of measures 
$\mu = (\mu_i)_{i=1}^N$, we denote $\mc{S}_{\mu} = (\mc{S}_{\mu_i})_{i=1}^N$, 
$B_{\mu} = (B_{\mu_i})_{i=1}^N$, and $\mc{B}_{\mu} = (\mc{B}_{\mu_i})_{i=1}^N$.
Furthermore, for a vector $v \in \N_+^{N}$ and a sequence of sets $B = (B_i)_{i=1}^N$, we define
$B^v = \prod_{i=1}^NB_i^{v_i}$, sequence of 
probability distributions $\mu = (\mu_i)_{i=1}^N$, 
$\mu^v = \bigotimes_{i=1}^N\mu_i^{v_i}$, 
and of measurable spaces $S = (S_i)_{i=1}^N$, $\mc{S}^v = \bigotimes_{i=1}^N \mc{S}_i^{v_i}$, 
For a vector $x=((x_{i,j})_{j=1}^{v_i})_{i=1}^N$ we often use a C-like notation
$x_{i,j} = x_i[j-1]$, $j \in I_{v_i}$, $i \in I_N$.
For $N \in \N_+$ let $\mathcal{R}_N$ be the class of all pairs $(\mu,f)$
such that $\mu = (\mu_i)_{i=1}^N$ is a sequence of probability measures
and $f$ is a measurable real-valued function on $\bigotimes_{i=1}^NS_{\mu_i}$.
Subsets $\mc{V}\subset\mathcal{R}_N$ are called admissible pairs with $N$ distributions.   
Set $\mc{V}_1$ 
is defined to consist of all $\mu$ such that there exists an $f$ such that $(\mu,f) \in \mc{V}$ 
and $\mc{V}_2$ is defined to consist of all $f$ for which there exists a $\mu$ such that $(\mu,f) \in \mc{V}$. 
By an estimand on $\mc{V}$ we mean a real-valued function $G$ on it. 
Note that in fact $\mathcal{R}_N$ is too large to be a set - it is a class so that $\mc{V}$ as above  
may also not be a set and thus $G$ may not be a function in the set 
theoretic sense but rather an operation, but we further on ignore such disctinction. In particular we use
 notation $\mc{V} = D_G$ as for domain of a function. 
For some $N\in \N_+$ and $K \subset I_N$,
let us consider the total sensitivity index $V_{X_K}^{tot}$ defined in Section \ref{secVBSA} for $Z=f(X) \in L^2(\PR)$, 
for a random vector $X = (X_i)_{i=1}^N$ with independent coordinates such that $X_i \sim \mu_i$, $i \in I_N$. 
Value of $V_{X_K}^{tot}$ is determined by $\mu= (\mu_i)_{i=1}^N$ and $f$, and thus we can an shall treat $V_{X_K}^{tot}$ 
as an estimand on $\mc{V} = \{((\mu_i)_{i=1}^N,f)\in \mc{R}_N: f \in L^2(\bigotimes_{i=1}^N\mu_i)\}$. 
We analogously define estimands corresponding to the main sensitivity index $V_{X_K}$ 
of $f(X)$ with respect to $X_K$ or variance $D= \Var(f(X))$, both being defined on the same admissible pairs 
as the total sensitivity index, and estimand $Ave = \E(f(X))$ on such pairs but with a less restrictive condition 
$f \in L^1(\bigotimes_{i=1}^N\mu_i)$ in their definition. 
Let $\mc{V}$ be some admissible pairs with $N$ distributions.
Let us define a new as compared to \cite{badowski2011} helper concept of a real-valued statistic $\phi$ for $\mc{V}$
with dimensions of arguments $v \in \N_+^N$. 
Such $\phi$ is defined as a function on $\mc{V}_2$, such that 
for each $\alpha=(\mu, f) \in \mc{V}$, $\phi(f)$ is a real-valued measurable function on $\mc{S}_{\mu}^v$. 
We denote 
\begin{equation}
Q_\alpha(\phi)= Q_{\mu^v}(\phi(f)) 
\end{equation}
for $Q=\E$ or $Q=\Var$ whenever these expressions make sense. 
Let $G$ be an estimand on $\mc{V}$. An unbiased estimator of $G$ is a statistic for 
$\mc{V}$ with some dimensions of arguments $v$ such that for each $\alpha = (\mu,f) \in \mc{V}$,
\begin{equation}
 \E_{\mu^v}(\phi(f)) = G(\alpha),
\end{equation}
i. e. $\phi(f)$ is an unbiased estimator of $G(\alpha)$ for $\mu^v$.  
Let $w \in \N_+^N$.  
We define $I_w = \prod_{i=1}^NI_{w_i}$. For each $x = ((x_{i,j})_{j=1}^{w_i})_{i=1}^N \in B^{w}$ and $v= (v_i)_{i=1}^N \in I_w$, we denote
\begin{equation}
x_v = (x_{i,v_i})_{i=1}^N.
\end{equation}
Let $A$ be nonempty subset of $\N_+^N$, called set of evaluation vectors for $N$. We define 
\begin{equation}
(A)_i = \{j_i: j \in A\},
\end{equation}
\begin{equation}\label{nAi}
n_{A,i} = \max\{k: k \in (A)_i\},
\end{equation}
and $n_A = (n_{A,i})_{i=1}^N$. For $v \in A$ we define evaluation operator or simply evaluation
$g_{\mc{V},A,v}$ to be a real-valued statistic for $\mc{V}$ with dimensions of arguments $n_A$  such that 
for each $(\mu,f) \in \mc{V}$ and $x \in B_\mu^{n_A}$,
\begin{equation}\label{gjDef}
 g_{\mc{V},A,v}(f)(x)  = f(x_v).
\end{equation}
For a nonempty $I \subset \N_+$ and a finite nonempty set $D \subset \N_+^I$, 
let for $j \in I_{|D|}$, $\psi_{D}(j)$ denote the lexicographically $j$th element of $D$. 
Since for $I =\{1\}$ we identify $\N_+^I$ with $\N_+$, in such case $D \subset \N_+$. 
For each set $C$ and its finite subset indexed by $D$, 
$\{y_v\in C: v \in D\}$, we define a vector from $C^{|D|}$ as follows 
\begin{equation}\label{ordnot}
(y_v)_{|v \in D} = (y_{\psi_D(j)})_{j=1}^{|D|}. 
\end{equation} 
We define 
\begin{equation}
g_{\mc{V},A} = (g_{\mc{V},A,j})_{|j \in A}.
\end{equation}
For $\mc{V}$ and $A$ being known from the context, we denote $g_{\mc{V}, A, v}$ shortly as $g_{v}$ 
or using a convenient C-array like notation
\begin{equation}\label{clikegprev} 
g[v_1-1]\ldots[v_l-1]. 
\end{equation} 
A scheme for $N$ is a pair $\kappa=(t, A)$ for some 
set of evaluation vectors $A$ for $N$ as above and $t$ being a real-valued measurable function on $\R^{|A|}$.
A statistic given by $\kappa$ and $\mc{V}$ is defined as
\begin{equation}\label{phiAF} 
\phi_{\kappa,\mc{V}} =  t(g_{\mc{V},A}).
\end{equation} 
Let $G$ be an estimand on $\mc{V}$. $\kappa$ is called an unbiased estimation scheme for $G$
if $\phi_{\kappa,\mc{V}}$ is unbiased estimator of $G$.
Let $G = (G_i)_{i=1}^n$ be a sequence of estimands, each on some (possibly different) admissible pairs but all with the same 
number of distributions $N$.  
Let us assume that $\kappa_i$ is an unbiased estimation scheme for $G_i$, $i \in I_n$, 
in which case we call $\kappa=(\kappa_i)_{i=1}^n$ an unbiased (many-dimensional) estimation scheme for $G$. 
We denote $\wh{G}_{\kappa,i}=\phi_{\kappa_i,D_{G_i}}$, $i \in I_n$. For 
$G$ being known from the context and  $G_i \neq G_j, i \neq j$, $i, j \in I_n$, 
we call $\kappa_i$ the subscheme of $\kappa$ for estimation of $\lambda=G_i$ and 
denote $\wh{G}_{\kappa,i}$ as $\wh{\lambda}_{\kappa}$, $i \in I_n$. 

We further need the following theorem, generalizing Theorem 6 in \cite{badowski2011}. 
\begin{theorem}\label{thCond} 
Let us consider random variables $X = (X_1,X_2)$ and $Y_2$ such that  $Y_2\sim X_2$ and $Y_2$ is independent of $X$. 
Let $g$ and $h$ be measurable real-valued functions such that $g(X),h(X)$, and $g(X)h(X_1,Y_2)$ are integrable. Then it holds 
\begin{equation}\label{condthcond}
\E(g(X)h(X_1,Y_2)|X_1) = \E(g(X)|X_1)\E(h(X)|X_1).
\end{equation}
In particular, applying expected values to both sides of (\ref{condthcond}) and using the iterated expectation property, we have 
\begin{equation}
\E(g(X)h(X_1,Y_2)) = \E(\E(g(X)|X_1)\E(h(X)|X_1)).
\end{equation}
Using this for $g(X) = h(X)$ we receive the well-known formula \cite{Saltelli_2002} 
\begin{equation}\label{ggxy}
\E(g(X)g(X_1,Y_2)) = \E((\E(g(X)|X_1))^2),
\end{equation}
and the fact that
\begin{equation}
\Cov(g(X),g(X_1,Y_2)) = \Var(\E(g(X)|X_1)). 
\end{equation}
\end{theorem}
\begin{proof}
It holds 
\begin{equation} 
\begin{split}
\E(g(X)h(X_1,Y_2)|X_1) &= (\E(g(x_1, X_2)h(x_1,Y_2)))_{x_1 = X_1} \\
& = (\E(g(x_1, X_2)))_{x_1 = X_1}(\E(h(x_1,Y_2)))_{x_1 = X_1}\\
& = \E(g(X)|X_1)\E(h(X)|X_1),
\end{split}
\end{equation}
where in the first and last equality we used Theorem \ref{indepCond} and in the second independence
of $X_2$ and $Y_2$ and that from Fubini's theorem functions under the expectations are integrable for $\mu_{X_1}$ a. e. $x_1$. 
\end{proof}
From the above theorem it easily follows that for $X$ and $Y_2$ as in it and 
$f(X) \in L^2_n$ with some scalar product as in Section (\ref{secOrthog}), we have 
\begin{equation}\label{thCondVect}
\begin{split}
(f(X) ,f(X_1,Y_2))_n  = ||\E (f(X)|X_{1})||_n^2
\end{split}
\end{equation}
(see (3.41) in \cite{badowski2011} for a proof). 

For example for the estimand $V_1^{tot}$ we introduced earlier in this section for $N=2$, the unbiased estimation 
scheme $a2=(t, A)$ was defined
in \cite{badowski2011} by taking $A= \{(1,1), (2,1)\}$ and
\begin{equation}
t(x_{(1,1)}, x_{(2,1)}) = x_{(1,1)}^2  - x_{(1,1)}x_{(2,1)}.
\end{equation}
Using notation (\ref{clikegprev}), the estimator given by $a_2$ can be written as
\begin{equation}\label{V1a2totg}
\widehat{V}_{1,a2}^{tot} = g[0][0](g[0][0] - g[1][0]).
\end{equation}
The fact that this is an unbiased estimation scheme for $V_1^{tot}$ is a consequence of Theorem \ref{thCond} and the fact that
observable of this estimator corresponding to function $f$ and observable $\wt{X} = (\wt{X}_{1}[j]_{j=0}^{1}, \wt{X}_{2}[0]) \sim \mu^{n_A}$
is
\begin{equation}\label{obsVitot}
f(\wt{X}_1[0],\wt{X}_2[0])(f(\wt{X}_1[0],\wt{X}_2[0]) - f(\wt{X}_1[1],\wt{X}_2[0])).
\end{equation}
Similarly as in \cite{badowski2011} we shall often use formulas for estimators 
like (\ref{V1a2totg}) to concisely define previously undefined schemes, in particular for 
the mentioned formula retrieving scheme $a2$. 
Scheme given by a formula like (\ref{V1a2totg}) for estimator $\wh{\lambda}_{\kappa}$ 
of a certain estimand $\lambda$ on some admissible pairs $\mc{V}$,
is a pair $\kappa = (t, A)$, where $A$ 
consists of indices $v$ of different $g_{v}$ appearing on the 
rhs of the formula, and $t$ acts on its arguments in the same way as the function of different $g_{v}$ given by the rhs of the formula 
does. By estimator defined by such a formula we mean $\phi_{\kappa,\mc{V}}$. 
We can group such received schemes from many formulas for estimators of different estimands in a sequence to get a many-dimensional
estimation scheme for 
a sequence of estimands, an example of which we shall see in the next section. 
In the next section and further on we often 
define estimands $F_i, i \in I_n,$ 
and unbiased estimation schemes $\gamma_i$ for $F_i$, $i \in I_n$, 
and say that many dimensional scheme $\kappa$ consisting of $\gamma_i, i \in I_n,$ 
is unbiased for estimation of a sequence of estimands $G$ consisting of $F_1, \ldots, F_n,$  
without specifying the order of $F_i$ or $\gamma_i$, $i \in I_n$, in sequences $\kappa$ and $G$, so that 
one can assume that for some arbitrary permutation $\pi$ of $I_n$, we have 
$G=(F_{\pi(i)})_{i=1}^n$ 
and $\kappa = (\gamma_{\pi(i)})_{i=1}^n$.


\section{\label{secMany}Schemes for sensitivity indices of conditional expectations}
We recall here the unbiased estimation schemes for sensitivity indices of conditional expectation 
from Section 4.5 of \cite{badowski2011}, 
which will be needed to derive certain new schemes in Section \ref{secPolynEst}. 
Suppose that for $N_P \in \N_+$, 
$P= (P_i)_{i=1}^{N_P}$ is a random vector with independent coordinates and $R$ is a random variable independent of $P$. 
Let us consider a measurable function $f$ from the product measurable space of the image of $(P,R)$ to $\R$. 
$f(P,R)$ can be for instance an 
output of an MR corresponding to some of its constructions as discussed in Section \ref{genParSec}. 
Let us consider quantities $V_k=\wt{V}_k$, $\wt{V}^{tot}_{k}$, $k \in I_{N_P}$, $D$, $V_P$, and $AveVar = V_R^{tot}$ defined for 
$f(P,R) \in L^2$, and $Ave=AveE$ for $f(P,R) \in L^1$, 
in the same way as at the end of Section \ref{secVBSA} treating $Z=f(P,R)$ as an output of an MR. 
Let $\mc{V}$ be admissible pairs consisting of $\alpha_{\mu_P,\mu_R,f} = ((\mu_i)_{i=1}^{N_P+1}, f),$ such that 
$\mu_i \sim P_i, i \in I_{N_P},$ and $\mu_{N_P+1} \sim R$ for different $f,P,$ and $R$ as above. 
We will from now on interpret each of the above sensitivity indices or averages  
as estimands on $\mc{V}$, whose values on each $\alpha_{\mu_P,\mu_R,f}$ as above are the same as previously for the corresponding 
$f$, $P$, and $R$. 
For $i, j \in \N$, we denote $s[i][j] = g[v_1]\ldots[v_{N_P+1}]$ where $v_{N_P +1} = j$ and $v_n = i$ for $n \in I_{N_P}$. 
For $i \in \{0,1\}, j \in \N,$ and $k \in I_{N_P}$, we denote 
$s_{k}[i][j] =g[v_1]\ldots[v_{N_P+1}]$, where $v_{N_P +1} = j$ and for $n \in I_{N_P}$, $n \neq k$, $v_n = i$, while for 
$n = k$, $v_n = 1-i$. 
For some $f$, $P$ and $R$ as above, let $\wt{P}= (\wt{P}_k)_{k=1}^{N_P}$ have independent coordinates, where 
$\wt{P}_k= (\wt{P}_{k,i})_{i=1}^2 \sim \mu_{P_k}^2, k \in I_{N_P}$. We denote 
$\wt{P}[i] = (\wt{P}_{k,i})_{k=1}^{N_P}$, $i \in \{0,1\}$. 
Let further for $k \in I_{N_P}$, $\wt{P}_{(k)}[i]$ be equal to vector $\wt{P}[i]$ with $k$th coordinate replaced by $\wt{P}_{k,1-i}$, and 
let $\wt{R} \sim \mu_R^2$ be independent of $\wt{P}$. Assuming admissible pairs 
$\mc{V}$ as for some of the above estimands and the set of evaluation vectors $A$ equal to set of all $v$ from evaluations 
$g_v$ equal to $s[i][j]$ 
and $s_k[i][j]$, $i,j \in \{0,1\}, k \in I_{N_P}$, we have, identifying $(\wt{P}_1,\ldots,\wt{P}_{N_P},\wt{R})$ with $(\wt{P},\wt{R})$, 
\begin{equation}\label{sij}
s[i][j](f)(\wt{P},\wt{R}) = f(\wt{P}[i],\wt{R}[j]), 
\end{equation}
and
\begin{equation}\label{skij}
s_k[i][j](f)(\wt{P},\wt{R}) = f(\wt{P}_{(k)}[i],\wt{R}[j]). 
\end{equation}
Formulas below, defining unbiased estimators of the above estimands are taken from Section 4.5 in \cite{badowski2011}, 
and the fact they are unbiased is 
an easy consequence of formula (\ref{ggxy}) in Theorem \ref{thCond} and formulas (\ref{tVCJ}) and (\ref{tildeVP}). 
We call the scheme these formulas yield scheme $SE$ (in \cite{badowski2011} we called it scheme $E$ but the new name 
is needed for consistency with notations introduced in Section \ref{secPolynEst}). 
\begin{equation}\label{estVkE}
\widehat{V}_{k,SE}  = \frac{1}{4}\sum_{i=0}^{1}(s[i][0] - s_k[i][0])(s_k[1-i][1] - s[1-i][1]),
\end{equation}
\begin{equation}\label{estVkTotE}
\widehat{V}^{tot}_{k,SE} = \frac{1}{4} \sum_{i=0}^{1}(s[i][0] - s_{k}[i][0])(s[i][1] - s_{k}[i][1]),
\end{equation}
\begin{equation}\label{VEst}
\begin{split}
\widehat{D}_{SE} &= \frac{1}{4(N_P + 1)}\sum_{i=0}^{1}\sum_{j=0}^{1}(s[i][j](s[i][j] - s[1-i][1-j]), \\
&+ \sum_{k=1}^{N_P}s_k[i][j](s_k[i][j] - s_k[1-i][1-j])),
\end{split}
\end{equation}
\begin{equation}\label{VP}
\begin{split}
\widehat{V}_{P,SE} &= \frac{1}{4(N_P + 1)}\sum_{i=0}^{1}\sum_{j=0}^{1}(s[i][j](s[i][1-j] - s[1-i][1-j]) \\
& + \sum_{k=1}^{N_P}s_k[i][j](s_k[i][1-j] - s_k[1-i][1-j])),
\end{split}
\end{equation}
\begin{equation}\label{VRtot}
\widehat{AveVar}_{SE} = \widehat{V}_{R,SE}^{tot} = \widehat{D}_{SE} - \widehat{V}_{P,SE},
\end{equation}
\begin{equation}
\widehat{Ave}_{SE} = \frac{1}{4(N_P + 1)}\sum_{i=0}^{1}\sum_{j=0}^{1}(s[i][j] + \sum_{k=1}^{N_P}s_k[i][j]).
\end{equation}
Using the same evaluations we can also construct estimation schemes for many further indices, among others 
for $\wt{V}_{(P_{i},P_{j})}$ and $\widetilde{V}_{(P_{i},P_{j})}^{tot}$ (see \cite{badowski2011}), $i, j \in I_{N_P}$, $i \neq j$.
It is easy to see using Schwartz inequality that it is sufficient that 
$f(P,R) \in L^4$ for the above estimators and further ones in this section to have finite second moments and thus variances 
when applied to the corresponding $f,\wt{P}$, and $\wt{R}$. For $\widehat{Ave}_{SE}$ it is even sufficient that $f(P,R) \in L^2$.
In \cite{badowski2011} we also introduced scheme $EM$ consisting of 
subschemes for estimation of $V_k$, for $k \in I_{N_P}$, 
\begin{equation}\label{estVkEM}
\wh{V}_{k,EM} = \frac{1}{2}(s[0][0] - s_k[0][0])(s_k[1][1] - s[1][1]).
\end{equation}
Similarly as in \cite{badowski2011}, we define scheme $ET$ containing subschemes given by formulas 
\begin{equation}\label{estVkTotET}
\widehat{\widetilde{V}}^{tot}_{k,ET} = \frac{1}{2}(s[0][0] - s_k[0][0])(s[0][1] - s_k[0][1]),\ k \in I_{N_P}.
\end{equation}
As discussed in \cite{badowski2011} schemes in this section 
can be generalized to variables $f(P,R)\in L_n^2$ like 
conditional histograms by using appropriate scalar 
product of vectors instead of function multiplication in the formulas for estimators, 
which is a consequence of expression (\ref{thCondVect}) after the proof of Theorem \ref{thCond}. 

\section{\label{secMCIneff}Inefficiency constants of MC procedures}
See Appendix \ref{appStatMC} for an introduction to Monte Carlo method and associated notations we use, like $Var_s$
and $\Var_f(n)=\Var_f$ for the variances of singles step and final $n$-step MC estimators, respectively, fulfilling
\begin{equation}
 \Var_f = \frac{\Var_s}{n}.
\end{equation}
Let us consider a sequence of MC procedures estimating $\lambda \in \R$, indexed by $n \in \N_+$,  
such that the $n$-th one is an $n$-step MC procedure and its average duration $\tau_f(n)$, e. g. when run on a computer, fulfills
\begin{equation}\label{nfs}
\tau_f(n) = n\tau_s,\quad n \in \N_+,
\end{equation}
where $\tau_s \in \R_+$ is called the average duration of a single MC step of this sequence.
Assumption (\ref{nfs}) is a good approximation for many sequences of MC procedures run on a computer,
especially ones for which the $n$-th procedure consists of $n$ repeated computationally identical
single MC steps, each lasting on average $\tau_s$,
$n \in \N_+$. 
Similarly as in Section 4.2 in \cite{badowski2011} 
we define the inefficiency constant of a sequence of MC procedures as above by formula 
\begin{equation}\label{cdef}
c = \tau_{s} \Var_{s},
\end{equation}
so that from (\ref{varfsn}) and (\ref{nfs}), for each $n \in \N_+$,
\begin{equation}\label{cntau}
c =  \tau_{f}(n) \Var_{f}(n).
\end{equation}
For two different sequences of MC procedures as above for estimating $\lambda$, 
their inefficiency constants can be used for comparing their efficiency \cite{asmussen2007stochastic, badowski2011},
which can be justified by different interpretations of these constants. 
We shall provide below a correction of an interpretation  
from Section 4.2 in \cite{badowski2011} in 
which we used an incorrect asymmetric definition of $\delta$-approximate inequality.
Two new interpretations shall be provided in Section \ref{secStat}. 
See Chapter 3, Section 10 in \cite{asmussen2007stochastic} for yet another interpretation. 
For $x,y\in \R_+$ and $\delta \geq 0$, 
we say that $x$ and $y$ are $\delta$-approximately equal, which we denote as $x \approx_{\delta} y$, if 
$\frac{|x-y|}{\min(|x|,|y|)}\leq \delta$; in particular for $\delta=0$ this is equivalent to  $x = y$. 
If for some sequence of MC procedures as above and another one, 
also for estimating $\lambda$, for which we have the same assumptions and use the same notations but with a prim, 
we have $\delta$-approximate equality of their respective average duration times for some $n$ and $n'$, that is 
\begin{equation}
\tau_{f}(n) \approx_{\delta} \tau_{f}'(n'), 
\end{equation}
then from (\ref{cntau}) the ratio of variances of their respective final MC estimators is $\delta$-approximately equal to the ratio 
of their inefficiency constants, i. e. 
\begin{equation}\label{varAiRatio} 
\frac{\Var_{f}(n)}{\Var_{f}'(n')} =  \frac{c\tau_{f}'(n')}{c'\tau_{f}(n)} \approx_{\delta} \frac{c}{c'}. 
\end{equation} 

\section{\label{secIneffSchemes}Inefficiency constants of schemes}
Let us reformulate the theory of inefficiency constants from sections 4.3 and 4.5 in \cite{badowski2011}
in a more precise way. 
Let us consider a sequence of estimands $G=(G_i)_{i=1}^n$ such that $\mc{V} = \bigcap_{i=1}^n D_{G_i} \neq \emptyset$,
called estimands on common admissible pairs $\mc{V}$ with $N$ distributions. 
We denote $D_G=\mc{V}$. 
Suppose that $\kappa=(\kappa_i)_{i=1}^n= (t_i,A_i)_{i=1}^n$ is an unbiased 
estimation scheme for $G$. $A_\kappa = \bigcup_{i=1}^n A_i$ is called the set of evaluation vectors of $\kappa$.
$\kappa$ can be used to generate estimates of coordinates of $G(\alpha)$ for some $\alpha =(\mu,f) \in \mc{V}$ as follows. 
For a random vector $X \sim \mu^{n_A}$, one 
computes the quantities $g_{\mc{V},A_i,v}(f)(X_{n_{A_i}})=f(X_v)$, $i \in I_n$, $v \in A_{i}$, 
considering that they are equal for the same $v$ and different $i$ so that they are computed only once, and then one 
evaluates $t_i$ on $g_{\mc{V},A_i}(f)(X_{n_{A_i}})$ to get an estimate of $G_i(\alpha)$, $i \in I_n$. 
$|A|$ is the total number of evaluations of $f$ in such a computation.  
If for some $\alpha \in \mc{V}$ we have 
$\Var_{\alpha}(\phi_{\kappa_i,\mc{V}})<\infty$, $i \in I_n,$ then the above computation can be performed to get unbiased estimates of 
coordinates of $G(\alpha)$ in a single step of a MC procedure. 
We define an inefficiency constant $d_{G,i,\kappa}$ of $\kappa$ for estimating $G_i$ as a function $\mc{V}\rightarrow \overline{\R}$
such that
\begin{equation}\label{dIneffprev} 
d_{G,i,\kappa}(\alpha) = \Var_{\alpha}(\phi_{\kappa_i,\mc{V}})|A_{\kappa}|. 
\end{equation} 
Let $\kappa'$ be an unbiased estimation scheme for estimands $G'=(G')_{i=1}^{n'}$ on some common admissible pairs $\mc{V}'$
for which we shall use the same notations as for $\kappa$, $G$, and $\mc{V}$ but with prims. 
Let for some $\alpha=(\mu,f)\in \mc{V}$ , $\alpha' = (\mu',f') \in \mc{V}'$, $i \in I_n$, and $i' \in I_n'$, it hold
$G_{i}(\alpha)=G_{i'}(\alpha')$. 
Suppose that the ratio of positive average durations $\tau_{s}$ to $\tau'_{s}$ of single steps of sequences of 
MC procedures (see Section \ref{secMCIneff})
using $\kappa$ and $\kappa'$, computing $G(\alpha)$ and $G'(\alpha')$ as above 
is for some $\delta>0$, $\delta$-approximately equal to ratio of positive numbers of evaluation vectors $A_\kappa$ and 
$A_\kappa'$ in these schemes, that is
\begin{equation}\label{tauAratio} 
\frac{\tau_s}{\tau'_s} \approx_\delta \frac{|A_{\kappa}|}{|A_{\kappa'}|}. 
\end{equation} 
This can be the case for small $\delta$ e. g. when the most time-consuming part of both sequences of MC procedures 
are calls to implementations of $f$ and $f'$, respectively, taking on average approximately the same time 
to compute. 
As we demonstrate in Section \ref{secImpl}, such approximate proportionality and even its more general 
version discussed in Section \ref{secineffgen}  holds in our numerical experiments 
using different estimation schemes for variance-based sensitivity indices and some further estimands, in which 
$f(P,R)$ and $f'(P,R')$ for some parameters $P$ and noise variables $R$ and $R'$, 
stand for some outputs of an MR, constructed using the GD and RTC methods  or two times one of them (see (\ref{obsForm})). 
From (\ref{tauAratio}), the ratio of inefficiency constant  
\begin{equation}\label{cscheme}
c = \Var_{\alpha}(\phi_{\kappa_i,\mc{V}})\tau_s  
\end{equation}
of a sequence of MC procedures estimating quantities $G_i(\alpha)$, performing computations with 
scheme $\kappa$ (see (\ref{cdef})) to an analogous constant $c'$ for $\kappa'$, computing $G_{i'}(\alpha')$, fulfills (assuming 
both constants are finite), 
\begin{equation} 
\frac{c}{c'} = \frac{\Var_{\alpha}(\phi_{\kappa_i,\mc{V}})\tau_s}{\Var_{\alpha'}(\phi_{\kappa'_{i'},\mc{V}})\tau_s'} 
\approx_{\delta} \frac{d_{G,i,\kappa}(\alpha)}{d_{G,i',\kappa'}(\alpha')}, 
\end{equation} 
which we already noticed in \cite{badowski2011} but with equality rather than 
$\delta$-approximate equality in (\ref{tauAratio}). 
Similarly as for inefficiency constants of sequences of MC procedures in Section \ref{secMCIneff}, one proves that the ratio of positive
 real values of 
inefficiency constants (\ref{dIneffprev}) of $\kappa$ and $\kappa'$ for 
estimating $G_i(\alpha)$ and $G'_i(\alpha')$ as above
is $\delta$-approximately equal to the ratio of variances of the appropriate final MC estimators  
for $\delta$-approximately the same number of $i$th and $i'$th functions evaluations made in the respective MC procedures.
If $G$ is known from the context, and $G_l \neq G_m$ for $l \neq m$, $l,m \in I_n$,
then we denote $d_{G,i,\kappa}$ 
simply as $d_{G_i,\kappa}$ 

\section{\label{secSymIneq}Symmetrisation of schemes and inequalities between inefficiency constants}
Let us recall some definitions and facts from Section 4.4 of \cite{badowski2011} on symmetrisation of schemes, 
changing them for compatibility with future generalizations in sections \ref{secAveSchemes}
and \ref{secGenIneq}. 
Let $\Theta$ be the group of all bijections of $\N_+$. 
For $N \in \N_+$, 
we define $\Theta^N = \{(\pi_i)_{i=1}^N:\pi_i \in \Theta, i \in I_N\}$ and endow it with a structure of a direct product group by defining 
for each $\pi, \pi'\in \Theta^N$ their product as $\pi\pi'=  (\pi_i(\pi_i'))_{i=1}^N$. 
For $\pi \in \Theta^N$ we define $\wh{\pi}:\N_+^N\rightarrow\N_+^N: \wh{\pi}(v)=(\pi_i(v_i))_{i=1}^N$.
Let  $\Pi$ be a subgroup of $\Theta^N$. For $A \subset \N_+^N$, we denote its symmetrisation given by $\Pi$ as
\begin{equation}
\wh{\Pi}[A] = \bigcup_{\pi\in \Pi} \wh{\pi}[A] = \{\wh{\pi}(v): \pi \in \Pi, v \in A\}.
\end{equation}

For a function $t: \R^{|A|} \rightarrow \R$,
its symmetrisation given by $\Pi$ and $A$ is defined to be a function
$\ave_{A,\Pi}(t): \R^{|\wh{\Pi}[A]|}\rightarrow\R$ such that for each $z=(y_j)_{|j \in \wh{\Pi}[A]} \in \R^{|\wh{\Pi}[A]|}$
\begin{equation} \label{symOp}
\ave_{A,\Pi}(t)(z) =  \frac{1}{|\Pi|} \sum_{\pi \in \Pi} t((y_{\wh{\pi}(v)})_{|v \in A})
\end{equation}
(see \ref{ordnot}).
Let $\kappa =(t,A)$ be a scheme for $N$. Its symmetrisation given by $\Pi$ is defined as
\begin{equation}\label{avePiKO}
\ave_{\Pi}(\kappa) = (\ave_{A,\Pi},\wh{\Pi}[A]). 
\end{equation}
Let $\mc{V}$ be admissible pairs with $N$ distributions, $\alpha \in \mc{V}$, $\alpha = (\mu,f) \in \mc{V}$,
and $X \sim \mu^{n_{\wh{\Pi}[A]}}$. Then the corresponding 
observable of and estimator given by a symmetrised scheme fulfills
\begin{equation}
\phi_{\ave_{\Pi}(\kappa),\mc{V}}(f)(X) =  \frac{1}{|\Pi|} \sum_{\pi \in \Pi} t((f(X_\pi(v)))_{v \in A}),
\end{equation}
which for $Y \sim \mu^{n_A}$ is equal to a sum of random variables with the same distribution as $\phi_{\kappa,\mc{V}}(f)(Y)$.
Therefore if $\kappa$ is unbiased for estimation of some estimand $G$ on $\mc{V}$, then so is $\ave_{\Pi}(\kappa)$
and from Lemma \ref{lemVarAve} in Appendix \ref{appStatMC} we have for each $\alpha \in \mc{V}$,
\begin{equation}
 \Var_\alpha(\phi_{\ave_{\Pi}(\kappa),\mc{V}}) \leq  \Var_\alpha(\phi_{\kappa,\mc{V}}).
\end{equation}
Let $I$ be a nonempty subset of  $I_N$. For $\theta \in \Theta$, we define 
$\pi_{N,I,\theta} \in \Theta^N$ to be such that $\pi_{N,I,\theta,i} = \theta$, $i \in I$, and $\pi_{N,I,\theta,i}= \id_{\N_+}$, 
$i \in I_N\setminus I$. 
For $m \in \N$, we call 
\begin{equation}\label{defthetam}
\Theta_m=\{\theta\in\Theta: \theta(i) = i \text{ for } i > m \} 
\end{equation} 
the subgroup of $\Theta$ of permutations of the first $m$ indices. 
Let 
\begin{equation}\label{thetanim}
 \Theta_{N,I,m} =  \{\pi_{N,I,\theta}: \theta \in \Theta_m\}.
\end{equation}
Symmetrisation of a scheme for $N$ w. r. t. 
 $\Theta_{N,I,m}$ is called its symmetrisation in the argument given by $I$,
in $m$ dimensions. If $n_{A,i} = m$, $i \in J$, then we call it simply symmetrisation in 
the argument given by $I$ and if  $n_{A,i} = 1$, $i \in J$, we call 
it symmetrisation from one to $m$ dimensions.
For $I = \{j\}$ we say of symmetrisation in the $j$th argument, in which case
we write $j$ instead of $I$ in the subscript.

After symmetrising the scheme given by (\ref{V1a2totg}) in the first argument in two dimensions as in \cite{badowski2011}
we receive a scheme given by 
\begin{equation}\label{vTotMin} 
\wh{V}_{1,s2}^{tot} = \frac{1}{2}(g[0][0] - g[1][0])^2, 
\end{equation} 
and we conclude that 
\begin{equation}\label{ineqs2a2} 
\Var_{\alpha}(\wh{V}_{1,s2}^{tot}) \leq \Var_{\alpha}(\wh{V}_{1,a2}^{tot}),\  \alpha \in D_{V_{1}}.   
\end{equation}
As both schemes use the same number of function evaluations, we also have 
\begin{equation} 
d_{V_{1}^{tot},s2} \leq d_{V_{1}^{tot},a2}, 
\end{equation} 
which should be understood as holding for each $\alpha \in D_{V_1^{tot}}$. 
From expression (\ref{ggxy}) in Theorem \ref{thCond} we receive that the following formula defines an unbiased estimator 
of the main sensitivity index $V_1$ for $N=2$, already mentioned in \cite{badowski2011}, 
\begin{equation} 
\widehat{V}_{1,a3} = g[0][0](g[0][1] - g[1][1]). 
\end{equation} 
Similarly as in \cite{badowski2011}, from symmetrising its scheme in the first argument we receive a scheme given by the formula 
\begin{equation}\label{V1s4} 
\widehat{V}_{1,s4} =  \frac{1}{2}(g[0][0] - g[1][0])(g[0][1] - g[1][1]), 
\end{equation} 
which uses 4 rather than 3 evaluation vectors, so that their respective inefficiency constants fulfill 
\begin{equation}\label{ineqV1} 
d_{V_1,s4} \leq \frac{4}{3} d_{V_1,a3}. 
\end{equation} 
The following theorem is a slight generalization of Theorem 12 from \cite{badowski2011},  
the proof of which is analogous as in \cite{badowski2011},
and which shall follow from a more general Theorem \ref{thIneqds} in Section \ref{secGenIneq}.  
\begin{theorem}\label{thineqOld}
Let $\kappa_1$ be unbiased estimation scheme of some estimand $G$ on adissible pairs $\mc{V}$ with $N$ distributions,
and let the scheme $\kappa_2$ be created from $\kappa_1$ by its symmetrisation in the argument given by 
$I \subset I_N$ from one to two dimensions. Then 
\begin{equation}\label{d12}
d_{G,\kappa_1} \leq d_{G, \kappa_2} \leq 2 d_{G,\kappa_1}.
\end{equation}
which should be understood as holding for each $\alpha \in \mc{V}$.
\end{theorem}
For an illustration let us consider an estimation scheme for $V^{tot}_1$, for $N=2$,  
given by the formula 
\begin{equation}\label{VTot1s4} 
\widehat{V}^{tot}_{1, s4} =   \frac{1}{4}\sum_{i=0}^1(g[0][i] - g[1][i])^2. 
\end{equation} 
As in \cite{badowski2011} let us notice that 
scheme (\ref{VTot1s4}) is received from (\ref{vTotMin}) by symmetrisation in the second argument from 
one to two dimensions, so from the above theorem
\begin{equation}\label{compDi}
d_{V^{tot}_{1},s2} \leq  d_{V^{tot}_{1},s4}  \leq 2d_{V^{tot}_{1},s2}.
\end{equation}
As we noticed in \cite{badowski2011}, 
scheme given by formula (\ref{estVkTotE}) for $\wh{V}_{k,SE}^{tot}$ 
is symmetrisation of the one given by formula (\ref{estVkTotET}) for 
$\wh{\wt{V}}_{k,ET}^{tot}$ from $1$ to $2$ dimensions in the argument corresponding to $P_{\sim k}$, 
so from the above theorem and the fact that the ratio of number of evaluation 
vectors used by the individual subschemes of $SE$ and $ET$ for the total sensitivity indices 
and the whole schemes is the same (and equal 2), we receive
\begin{equation}\label{VitotComp}
d_{\wt{V}_i^{tot},ET} \leq d_{\wt{V}_i^{tot},SE} \leq 2d_{\wt{V}_i^{tot},ET}.
\end{equation}
In \cite{badowski2011} we also proved the following Theorem 13 which more precisely than originally can be formulated as follows. 
\begin{theorem}\label{thdEMEComp}
For each $\alpha = \alpha_{\mu_P,\mu_R,f}  \in \mc{D}_{V_{k}}$ for some $f$, $P$, $R$ as in Section \ref{secMany} such that 
 $f(P,R) \in L^4$, the 
inefficiency constants of schemes $EM$ and $SE$ for estimation of $V_k$ fulfill, for $N_P > 2$, 
\begin{equation}\label{EMComp}
d_{V_k,EM}(\alpha) \leq d_{V_k,SE}(\alpha) \leq 2d_{V_k,EM}(\alpha). 
\end{equation} 
\end{theorem} 

\section{\label{secVarDiff}Variances of estimators using the GD and RTC methods} 
Let $h(p,R)$ denote certain construction of an MRCP corresponding to parameters $p\in \R^n$ and a reaction network $RN$ 
(see (\ref{MRCPFun})). For $f(p,R) = g(h(p,R))\in L^2$ denoting the 
number of particles of certain species at a given moment of time, in \cite{Rathinam_2010} and \cite{Anderson2013} the mean finite difference 
\begin{equation}\label{fidi} 
\frac{1}{s} \E(f(p + se_i, R) - f(p, R)) 
\end{equation}
for some $i \in I_n$ and $s \in \R_+$, 
approximating the partial derivative $\partial_i\E(f(p,R))$, was estimated in a MC procedure evaluating an 
independent copy of $f(p + s e_i,R) - f(p, R)$ in each step. 
Using the RTC construction of the output in such a MC procedure for finite difference has been called common reaction path method, 
while using the GD construction - common random number method \cite{Rathinam_2010}, both names stressing the fact that the same noise variable $R$ is 
used to construct the initial and perturbed outputs in each MC step. 
The estimates of variance of $f(p + se_i, R) - f(p, R)$ determining the variance of the final MC estimators of (\ref{fidi}) 
for the models considered in \cite{Rathinam_2010}
was much lower when performing the simulations with the RTC than the GD method, however in 
Section E of \cite{Anderson2013} an example was provided with an opposite inequality. 
Note that the variance mentioned 
being higher for one of the above methods than the other is equivalent to the following mean squared difference 
\begin{equation}\label{errSqr} 
\msd(p_1, p_2) = \E((f(p_1, R) - f(p_2, R))^2) 
\end{equation} 
being higher for $p_1 =p$, $p_2 = p+se_i$, or $\E(f(p_1, R)f(p_2, R))$ or $\Cov(f(p_1, R),f(p_2, R))$  
 being lower for such $p_1$ and $p_2$ for one method than the other. 
Let us consider an output $f(P,R)=g(h(P,R))$ of an MR $(P,h(P,R))$ constructed using the RTC or the GD method. 
In the numerical experiments in \cite{badowski2011} we observed that for such outputs being particle numbers 
of some species at a given time, the estimates of the variances 
of estimators of main and total sensitivity indices given by schemes $EM$, $ET$, and $SE$, corresponding 
to admissible pairs $\alpha_{\mu_P,\mu_R,f}$ (see Section \ref{secMany}) were in some cases much higher 
when using the RTC than the GD method. The variances of these estimators also varied with the order of 
reactions in the GD method. In an experiment for many births - many deaths model 
which we describe in Section \ref{MBMDPrev}, grouping reactions with similar effect on the considered output together in the 
sequence of reactions 
resulted in lower estimates of variance of the above estimators using GD method than when reactions with different effects 
appeared one after another in the sequence. 
 Note that reordering the reactions in the RTC method results in reordering of the Poisson processes 
used in the construction, which causes no change of variance of the above discussed estimators using this method. 
For some of the above estimators one can show that, given $f(P,R)\in L^2$, 
if for some measurable set $A$ such that $\mu_P(A) = 1$, for all pairs of parameter values $(p_1,p_2) \in A^2,$ 
$\msd(p_1, p_2)$ is not higher for one construction of MR than another, e. g. for the GD method than the RTC method 
or for GD methods but with  different orders of reactions, 
then variance of the considered estimator corresponding to such $f$, $P$, $R$ should be either not lower or not higher 
for one construction than the other. 
For instance, as discussed in \cite{badowski2011}, for estimators $\wh{V}_{k,EM}$ and $\wh{V}^{tot}_{k,ET}$, 
using notations as in Section \ref{secMany} and $\alpha=\alpha_{\mu_P,\mu_R,f}$,  from the equalities 
 \begin{equation} 
 4\E_{\alpha}((\wh{V}_{k,EM})^2) = \E(\msd(\wt{P}_{(k)}[0], \wt{P}[0])\msd(\wt{P}[1], \wt{P}_{(k)}[1]))\\
  \end{equation}
 and 
 \begin{equation}
 4\E_{\alpha}((\wh{V}^{tot}_{k,ET})^2) = \E(\msd^2(\wt{P}_{(k)}[0], \wt{P}[0]))
 \end{equation}
it follows that the inequalities between the variances of estimators $\wh{V}_{k,EM}$ and $\wh{V}^{tot}_{k,ET}$ when using different 
constructions   
should be the same as the inequalities between the quantities $\msd(p_1, p_2)$ for all $p_1$, $p_2$ as above. 

\section{\label{secExSoft}Existing software used}
In \cite{badowski2011} we run experiments using a program written in the C++ programming language and   
a personal computer with 1GB RAM, 2-core 2.10 GHz processor, and Linux operating system. 
We shortly describe this program below, see Section 4.7 in \cite{badowski2011} for details. For the numerical experiments in this work
we made some extensions to this program, as described in Section \ref{secImpl}. 
For specifying the reaction network, distribution of parameters, and output of the MR considered in computations we 
used SBML (Systems Biology Markup Language) \cite{SBML} configuration files.  
We used     
GNU Scientific Library \cite{Galassi_2003} implementation of the 
Mersenne twister random number generator \cite{MATSUMOTO_NISHIMURA_1998}
 and simple implementations of the RTC and GD constructions. 
Values of the model output were generated by running a given simulation algorithm starting with the selected 
parameters and reusing the values of independent copies of the noise term $R$, e. g. for the schemes from Section \ref{secMany} 
when evaluating  $s[i][j]$ and $s_{(k)}[i][j]$ for the same $j$. 
This reusing was implemented 
by storing the noise variables in lists, 
in the GD method using a single list for each noise variable, while in the RTC method a different one for each Poisson process in 
the construction.

\section{Models used}
In this work we use certain mathematical models from the literature, which we briefly describe below.
\subsection{\label{SBPrev}Simple birth (SB)  model}
SB model is a very simple model from \cite{badowski2011}, for which, as opposed to the further models,  
many sensitivity indices and coefficients of our interest can be computed analytically. 
The reaction network of the model consists of a single reaction involving only one species $X$ 
\begin{equation} 
R_1:\ \emptyset \rightarrow  X. 
\end{equation}
The only kinetic rate of this reaction fulfills $a_1(K,x) = K_1 + K_2 + K_3$, where $K = (K_1,K_2, K_3)$ is a random vector 
with independent coordinates with respective distributions $U(0.3,0.9)$, $U(0.85, 1.15)$, and 
$U(0.07, 0.13)$. Random initial number $C$ 
of particles of $X$ is independent of $K$ and has uniform discrete distribution $U_d(30, 90)$. 
As the output for analysis we take the number of particles of species $X$ at time $t=100$. 

\subsection{\label{GTSPrev}Genetic toggle-switch (GTS) model}
Let us consider a model of a genetic toggle switch which is 
a simplified stochastic version of the model from \cite{Gardner2000}, first analyzed in \cite{Rathinam_2010}
and later also in \cite{badowski2011}. 
In the model, two species $U$ and $V$ are produced and degraded in the following four reactions 
\[R_1:\ \emptyset \rightarrow  U, \quad R_2:\ U \rightarrow \emptyset,\] 
\[R_3:\ \emptyset \rightarrow V,\quad R_4:\ V \rightarrow \emptyset.\] 
For $x = (x_1,x_2)$ being the vector of numbers of species $U$ and $V$, and the rate constants vector equal to 
$K = (\alpha_1, \alpha_2, \beta, \gamma)$, the rates of the above reactions are 
\[a_{1}(K,x) = \frac{\alpha_1}{1 + x_2^{\beta}}, \quad a_{2}(K,x) = x_1, \] 
\[a_{3}(K,x) = \frac{\alpha_2}{1 + x_1^{\gamma}}, \quad a_{4}(K,x) = x_2. \] 
The second and fourth rates describe degradation with a speed proportional to the current number of particles of a given species. 
The first and third rates 
describe the fact that each species is a repressor of the promoter transcribing the other species, that is 
it inhibits the production of the opposing repressor by attaching itself to the DNA sequence preceding the region coding the other repressor. 
The value of the rate constants vector $K$ in \cite{Rathinam_2010} was deterministic and equal to $(50, 16, 2.5, 1)$. 
However, similarly as in \cite{badowski2011}, we consider $K$ to be a random vector whose each coordinate which  in \cite{Rathinam_2010} had 
a fixed value $v$ is considered to be a random variable with distribution U($0.8v, 1.2v$) and independent 
of the other coordinates. 
Similarly as in \cite{Rathinam_2010}, the initial particle numbers of both species are zero and we thus consider the vector of random 
parameters to be equal to $K$. As in \cite{Rathinam_2010} and \cite{badowski2011}, the 
model output considered for sensitivity analysis is the number of particles of the species $U$ at time $t=10$. 
Using this model in \cite{badowski2011}
we observed lower variances for the RTC than the GD method for all estimators of main and total sensitivity 
indices of conditional expectation from Section \ref{secMany}.

\subsection{\label{MBMDPrev}Many births - many deaths (MBMD) model}
Let us now consider the MBMD model from \cite{badowski2011}, whose reaction network contains 
one species $X$ and the following $5$ different birth and death reactions can occur 
\begin{equation}
R_{bi}:\ \emptyset \rightarrow  X,\ R_{di}:\   X\rightarrow \emptyset, \quad i \in I_5.  
\end{equation}
The order of these reactions is 
\begin{equation} 
R_i = R_{bi},\ R_{5 + i} = R_{di}, \quad i \in I_5. 
\end{equation} 
The rate constants vector is $K = (K_{b1},\ldots,K_{b5},K_{d1},\ldots,K_{d5})$, and has independent coordinates  
with $K_{di} \sim \U(0.010,\ 0.040)$ and $K_{bi} \sim \U(0.10,\ 0.40)$, $i \in I_5$. 
The rates of birth reactions are $a_{bi}(K, x) = K_{bi}$, and of death reactions $a_{di}(K,x) = K_{di}x$, $i \in I_5$. 
The initial number $C$ of particles of $X$ has distribution $U_d(5, 15)$ and all the parameters are independent. 
The considered output is as in \cite{badowski2011} the number of particles of species $X$ at time $t=5$. 
In the numerical experiments in \cite{badowski2011} we considered three different constructions. 
The first two are the RTC and GD methods for the above initial reaction network and distribution of parameters, 
abbreviated shortly as RTC and GDI methods, while the third construction is 
the GD construction for such distribution of parameters but for a reaction network with reordered sequence of reactions 
\begin{equation}\label{newOrder} 
R_{2i-1} = R_{bi},\ R_{2i} = R_{di},\quad i \in I_5, 
\end{equation} 
abbreviated as GDR method. 
The intuition behind such reordering in \cite{badowski2011} was to increase the frequency of 
switching between the birth and death reactions in a given step of 
the GD construction for different values of the model parameters $p_1$, $p_2$ and the same noise variable so as to 
increase the value of the function $\msd(p_1,p_2)$ defined by (\ref{errSqr}). 
In \cite{badowski2011} the estimates of variances of estimators from the schemes 
$EM$ for the main and and $ET$ for the total sensitivity indices were lowest when using the RTC method, followed 
by the GDI and GDR methods.  

\chapter{\label{chapOwnRes}Own research}
\section{\label{secMRNew}New theorems for MRCP and MR}
Below we provide two new theorems giving criteria for
existence of MRCP and MR and finiteness of higher moments of particle numbers in the latter model. 
For $a, b \in \R^N$, we denote 
$ab = \sum_{i=1}^Na_ib_i$. Let us consider a reaction network $RN$ as
in Section \ref{secCRN}. 
The following theorem, which we prove in Appendix \ref{appProc}, gives a useful sufficient condition for the existence of an MRCP. 
\begin{theorem}\label{thNonexpl} 
Let $k \in B_{RN}$. If there exists a vector  $m =(m_i)_{i=1}^N \in \R^N$ with positive coordinates, such that 
\begin{equation}\label{supmx}
\sup_{x \in E} (\sum_{l=1}^L(ms_la_l(k,x)) - mx) < \infty, 
\end{equation}
then for each $c \in E$, an MRCP corresponding to $RN$ and $p=(c,k)$ exists. 
\end{theorem}
Vector $m$ can often be chosen such that for $i \in I_N$, $m_i$ is mass of the $i$th species, 
in which case $ms_l$ is the mass increase in the $l$th reaction and 
$mx$ is the total mass of all species in the system in state $x$. 

For real-valued outputs $g(M)$ in this work 
we often encounter the requirement that $\E(|g(M)|^n)<\infty$ for some $n \in \N_+$. 
For $Y_{t,i}$ denoting the number of particles of the $i$th species at the moment $t$ in a MR, we have a following criterion. 
\begin{theorem}\label{thMomsExist}
Let $\nu$ be a probability distribution on $\mc{S}_{RN,E}$. 
Suppose that there exists a vector $m =(m_i)_{i=1}^N \in \R_+^N$, such that for 
$L_m = \{l \in I_L: s_{l}m > 0\}$, for the function 
\begin{equation} 
A(k) = \max\{ \sup_{x \in E, l \in L_m}(a_l(k, x)), 0\}, 
\end{equation} 
for each $(K, C) \sim \nu$, and certain $n \in \N_{+}$, it holds 
$\E(A(K)^n) < \infty$ and $\E(C_i^n) < \infty$, $i \in I_N$. Then for $\nu$ a. e. $p$, MRCP given by $RN$ and $p$ exists. 
Moreover, for a process $Y$ of an MR corresponding to $RN$ and $\nu$, for each $i \in I_N$ and $t \in T$, it holds $\E(Y_{t,i}^n) < \infty$. 
\end{theorem}
Proof of the above theorem is provided in Appendix \ref{appProc}.
All the moments of each parameter in the SB, GTS, and MBMD models from sections \ref{SBPrev}, \ref{GTSPrev}, and 
\ref{MBMDPrev} exist, so given the form of reaction rates of these models 
assumptions of Theorem \ref{thMomsExist} are satisfied for all $n \in \N_+$ if we take $m_i$ equal to one for each 
$i$th species. 
Thus all moments of each particle numbers at each time instant in these models exist, 
which makes it possible to use the MC method 
for estimating the sensitivity indices and various coefficient defined further on 
using the schemes from the previous and further sections for output being the particle numbers as above.

\section{\label{secCondMoms}Functions of conditional moments}
For a real-valued random variable $Z$ on the probability space with a probability measure $\mu$ 
and $n \in \N_+$, we define the $n$th moment of $Z$ for $\mu$ to be the 
element of $\overline{\R}$ 
defined as  
\begin{equation}\label{Mn}
M_{n}(\mu,Z) = \E_{\mu}(Z^n) 
\end{equation}
and the $n$th central moment of $Z$ for $\mu$ the element of $\overline{\R}$ defined as  
\begin{equation}\label{CMn}
CM_{n}(\mu,Z) = \E_{\mu}((Z - \E_\mu(Z))^n), 
\end{equation}
whenever these expressions make sense (that is in the second case 
$\E_\mu (|Z|)< \infty$ and in both cases the functions appearing under the outer expectations must have 
their positive or negative parts Lebesgue integrable with respect to $\mu$). 
We shall consider the $n$th moment, denoted as $M_n$, or such 
central moment $CM_n$ to be a certain function $Q$ 
whose domain $D_Q$ are pairs $(\mu,Z)$ for which respective expression (\ref{Mn}) or (\ref{CMn})
makes sense, 
and for each $\alpha=(\mu,Z) \in D_Q$, $Q(\alpha)$, also denoted $Q_\mu(Z)$, is given by the rhs of 
(\ref{Mn}) or (\ref{CMn}), respectively. Both the $n$th moment and central moment $Q$ restricted to the class  
\begin{equation}\label{defMCRn}
\mc{T}_n = \{(\mu,Z):\text{$\mu$ is a probability measure and } Z \in L^n(\mu)\} 
\end{equation}
 is a real-valued function, equal to some measurable function $f_Q:\R^n\rightarrow \R$ 
 applied to a vector of the first $n$ moments of $Z$, each restricted to $\mc{T}_n$, that is 
\begin{equation}\label{qfq}
Q_{|\mc{T}_n}= f_Q((M_{k|\mc{T}_n})_{k=1}^n), 
\end{equation}
or equivalently
\begin{equation}\label{qz}
Q_\mu(Z)  = f_Q((\E_{\mu}(Z^i))_{i=1}^{n}),\ (\mu,Z) \in \mc{T}_n. 
 \end{equation}
The first moment is expectation for which $Q=\E$, $n=1$, and $f_E = \id_\R$. 
The second central moment is variance for which $Q=\Var$, and for each $(\mu,Z) \in \mc{T}_1$, 
 \begin{equation}\label{varDef}
 \Var_{\mu}(Z) = \E_{\mu}((Z - \E_{\mu}(Z))^2) = \E_{\mu}(Z^2) - \E_{\mu}^2(Z).
 \end{equation}
We have (\ref{qfq}) for $n=2$ and 
\begin{equation}\label{deffvar}
f_{Var}(x_1, x_2) = x_2 - x^2_1
\end{equation}
(note that we write $E$ and $Var$ instead of $\E$ and $\Var$ in the subscripts). 
In general, let $Q$ be some function whose domain $D_Q$ contains $\mc{T}_n$ and 
there exists a measurable function $f_Q:\R^n\rightarrow\R$ such that 
(\ref{qfq}) holds. 
This is the case e. g. for $Q$ equal to the $k$th moment or central moment, $k \leq n$, or 
arbitrary product or linear combination of such moments restricted to $\mc{T}_n$.  
Similarly as above for moments, for $\alpha = (\mu, Z) \in D_Q$, $Q(\alpha)$ is also denoted as $Q_{\mu}(Z)$ or simply $Q(Z)$ if $\mu = \PR$. 
Note that the function $f_Q$ for $Q$ as above is unique for $n=1$ while for $n\geq 2$
it is not 
since from $\sqrt{\E(X^2)}\geq\E(|X|) \geq \E(X)$ for $X \in L^2$, 
(see Theorem \ref{leqpq}) the value of $f_Q$ can be changed on some $x \in \R^n$ such that  $x_2< x_1^2$ with (\ref{qfq}) still being true. 
We denote $f_Q$ one of the possible choices of the required function for $Q$, taking for $Q=\Var$, $f_{Var}$ as in (\ref{deffvar}). 
If $Z \in L^n$ and $X$ is a random variable, then 
we define the function of the first $n$ conditional moments of $Z$ given $X$ and corresponding to $Q$ as 
 \begin{equation}\label{qxz}
 Q(Z|X) = f_Q((\E(Z^i|X))_{i=1}^{n}).
 \end{equation}
 In particular the function of conditional moments of $Z$ given $X$ and corresponding to $\Var$ is equal to 
\begin{equation}\label{defCondVar}
 \begin{split}
 \Var(Z|X) &= \E (Z^2|X) - \E^2(Z|X)\\
 &= \E((Z - \E(Z|X))^2|X), 
\end{split}
\end{equation}
 and we call it conditional variance of $Z$ given $X$. 
If $\mu_{Z|X}$ is conditional distribution of $Z$ given $X$ 
and $\phi(Z) \in L^1$ for some measurable function $\phi$, then 
from (\ref{emu}) and (\ref{condCond}) it follows that 
\begin{equation} 
\E(\phi(Z)|X) = \E_{\mu_{Z|X}(X, \cdot)}(\phi). 
\end{equation} 
Thus if $\alpha_i(Z) \in L^1$ for some measurable functions $\alpha_i$, $i \in I_n$, then for any function $\beta:\R^n \rightarrow \R$ 
we have a. s.
\begin{equation}\label{fezx}
\beta((\E(\alpha_i(Z)|X))_{i=1}^{n}) = \beta((\E_{\mu_{Z|X}(X,\cdot)}(\alpha_i))_{i=1}^n).
\end{equation}
In particular for $\phi(Z)\in L^n$, $\beta = f_Q$, and $\alpha_i = \phi^i$, $i \in I_n$, we receive from (\ref{qxz}) and (\ref{fezx}) that  
\begin{equation}\label{condFun}
Q(\phi(Z)|X)= Q_{\mu_{Z|X}(X,\cdot)}(\phi). 
\end{equation}
The formula (\ref{varmuzx}) from the end of Section \ref{secOrthog} 
is obtained for $\alpha_i = f_i$ for $i \in I_n$, $\alpha_{n+1} = |f|^2$, 
and $\beta((x_i)_{i=1}^{n+1}) = x_{n+1}^2 - |(x_i)_{i=1}^n|^2$, 
using (\ref{fezx}) and the last terms in (\ref{varGen}) and (\ref{condGen}). 
When $f(P,R) \in L^n$ for some $P$ and $R$ independent and $f$ measurable (from appropriate product measurable 
space to $\R$), like for stochastic outputs corresponding to certain constructions of MR in Section \ref{genParSec},
 we have an intuitive formula 
 \begin{equation} \label{CQKP}
 \begin{split}
 Q(f(P,R)|P) &= f_Q((\E(f^i(P,R)|P))_{i=1}^n)\\
 &= f_Q(((\E(f^i(p,R)))_{p=P})_{i=1}^n)\\
 &= (Q(f(p,R)))_{p=P},
 \end{split}
 \end{equation}
where in the second equality we used Theorem \ref{indepCond} from Appendix \ref{appMath},
and in the third the fact that from Fubini's theorem \cite{rudin1970}, $f(p,R) \in L^n$ for $\mu_P$ a. e. $p$.
Note that from (\ref{condFun}), for $Z \in L^1$, $\phi = \id_\R$, and when 
$\mu_{Z|X}(x,\cdot)$ exists and is uniquely determined for $\mu_X$ a. e. $x$,
 which holds for a large class of random variables $Z$ and $X$ (see Appendix \ref{appMath}), or from 
(\ref{CQKP}) when $Z=f(P,R)$ and $X=P$ for some $f$,  $P$, and $R$ as above, it follows that 
for different choices of the function $f_Q$ corresponding to $Q$ which we used to define $Q(Z|X)$ in (\ref{qxz}), 
the resulting $Q(Z|X)$ are a. s. equal. 

Let $Z = g(M) \in L^n$ be an output of an MR $M=(P,Y)$ with parameters 
$P = (P_i)_{i=1}^{N_P}$ and corresponding to a reaction network $RN$. 
Since conditional distribution of $M$ given $P$ is specified by Definition \ref{MRdef}, 
from formula (\ref{condFun}) it follows that distributions of functions of certain $n$ first conditional 
moments $Q(Z|P)$, like conditional variance, are determined by $RN$, $g$, and $\mu_P$. 
Therefore, if $Q(Z|P) \in L^1$, then the values of its mean 
\begin{equation}\label{defAveQ}
AveQ = \E(Q(Z|P))
\end{equation}
(which for $Q=E$ is equal to $Ave=\E(Z)$ by iterated expectation property),
and if $Q(Z|P) \in L^2$, also the values of the main sensitivity indices
\begin{equation} 
VQ_{P_J} = \Var(\E(Q(g(Y)|P)|P_J)) 
\end{equation} 
of these functions of conditional moments are determined by this data, and so are the total sensitivity indices 
\begin{equation} 
VQ_{P_J}^{tot} = VQ_{P} - VQ_{P_{\sim J}},
\end{equation} 
where $\sim J = I_{N_P} \setminus J$, $J \subset I_{N_P}$. 
The Sobol's main and total sensitivity indices, created by dividing the above indices by 
$VQ_P$, are denoted as $SQ_{P_J}$ and $SQ_{P_J}^{tot}$, respectively. 
Similarly as for the special case of $Q=\E$ in Section \ref{secVBSA},
for $J = \{i\}$, we usually write $i$ in place of $P_J$ in the above notations.

\section{Covariance and some properties of variance of random vectors}
Covariance of random vectors $U, Z \in L^2_n$ is defined as 
\begin{equation}\label{covGen}
\Cov(U,Z) = (U - \E(U),Z - \E(Z))_n. 
\end{equation}
Let for some $m \in \N_+$, $X_i \in L^2_n,$ $i \in I_m$. We have an easy to prove 
formula
\begin{equation}\label{sumcov}
\Var(\sum_{i=1}^m X_i) = \sum_{i=1}^{m}\Var(X_i) + 2\sum_{1\leq i<j \leq n}\Cov(X_i,X_j),
\end{equation}
which is well-known for $n=1$.
If  $X_i\in L^2_n$, $i \in I_m,$ are i. i. d., then from (\ref{scalarFun}) and (\ref{covGen}),
$\Cov(X_i,X_j)=0$, $i\neq j,$ so from  (\ref{sumcov}) we receive   
\begin{equation}\label{varAveVec}
\Var(\frac{1}{m}\sum_{i=1}^m X_i) = \frac{1}{m}\Var(X_1). 
\end{equation}


\section{Output approximations, correlations, and nonlinearity coefficients}\label{secAppr} 
Let us consider a set $\Phi = \{v_i\}_{i=1}^l$ of $l \in \N_{+}$ linearly independent elements of a Hilbert space $H$ 
with some scalar product $(,)$, inducing norm $||\cdot||$ and distance $d$. The linear subspace 
\begin{equation}
V = \text{span}(\Phi) = \{\sum_{i=1}^l a_iv_i: a_i \in \R, i \in I_l \} 
\end{equation}
is closed in $H$ (\cite{rudin1970} Section 4.15), and thus 
for each $x \in H$ there exists a unique element of $V$ minimizing the distance from $x$ - the orthogonal projection 
$P_V(x)$ of $x$ onto $V$ (see Appendix \ref{appHilb}). Denoting $y_i = (x,v_i)$
and $g_{ij} =(v_i,v_j)$, the coefficients $(b_i)_{i=1}^l$ such that 
\begin{equation}
P_V(x) = \sum_{i=1}^l b_i v_i,
\end{equation}
can be computed from the following set of equations \cite{rudin1970}
\begin{equation}\label{gijEqu} 
\{\sum_{j=1}^l g_{ij}b_j = y_i\}_{i=1}^l. 
\end{equation} 
In particular, if elements of $\Phi$ are orthonormal (see Appendix \ref{appHilb}), 
then from (\ref{gijEqu}) it holds $b_i = y_i$, $i \in I_l$. 
In such case $(b_i)_{i=1}^l$ are known as Fourier's coefficients \cite{rudin1970} of $x$ relative to the elements of 
$\Phi$, and distance between $x$ and  $P_V(x)$ fulfills 
\begin{equation}\label{dxp}
d(x, P_V(x))^2 = ||x||^2 - \sum_{i=1}^l {b_i^2}. 
\end{equation}

Let us consider the special case of $H= L^2_n$ with some scalar product $(,)_n$ corresponding to a 
scalar product $<,>$ on $\R^n$ as in Section \ref{secOrthog},
and let $e_j$, $j \in I_n$, be the elements of some orthonormal base of $\R^n$ with respect to $<,>$, e. g.  
for the standard scalar product we can take the standard base of $\R^n$. For some $k \in \N_+$, let $l=n+k$ and  
$\{v_i \in L^2_n: i \in I_{k+n}\}$ be a nonzero orthogonal set (see Appendix \ref{appHilb}) with $v_{k+i} = e_i$, $i \in I_n$. Then 
\begin{equation}\label{eequzero}
\E(v_i) = 0,\quad i \in I_k,  
\end{equation}
since $(v_i,e_j)_n = <\E(v_i),e_j> = 0$, $j \in I_n$.
We normalize $\Phi$ to get an orthonormal set 
$\Phi'= \{v_i'\}_{i=1}^{k+n}$, $v_i' = \frac{v_i}{\sigma(v_i)}$, $i \in I_k$, $v'_{k+i} = v_{k+i},i \in I_n$.  
Let $(b_i)_{i=1}^{k+n}$ denote the coefficients of
$P_V(x)$ as above relative to $\Phi$, and  $(c_i)_{i=1}^{k+n}$ relative 
to $\Phi'$. We have 
 \begin{equation}\label{hieqp}
c_i = (x,v_i')_n,\quad i \in I_l,
 \end{equation}
\begin{equation}\label{hieq}
b_i = \frac{(x,v_i)_n}{\Var(v_i)}= \frac{c_i}{\sigma(v_i)},\quad i \in I_k,
\end{equation}
and $b_{k+i} = c_{k+i}$, $i \in I_n$.
Let $U, Z \in L^2_n$ have nonzero variances. We define their correlation as 
\begin{equation}
\corr(U,Z) = \frac{\Cov(U,Z)}{\sigma(U)\sigma(Z)}.
\end{equation}
Correlation is a popular measure of strength of the linear relationship between $U$ and $Z$ for $n=1$, due to its properties 
which we discuss and prove below for arbitrary $n$. 
Using (\ref{eequzero}), we have 
$(x,v_i)_n = \Cov(x,v_i), i \in I_k,$
and thus from (\ref{hieqp}), $c_i = \Cov(x,v_i')$, $i \in I_k$. Furthermore, if $\Var(x) > 0$, then
\begin{equation}\label{correqu}
\corr(x,v_i) =\corr(x, v_i') = \frac{c_i}{\sigma(x)} = \frac{b_i\sigma(v_i)}{\sigma(x)},\quad i \in I_k.
\end{equation}
From discussion in Section \ref{secOrthog}, $\E(x)$ is orthogonal projection of $x$ onto span of constant random vectors, so that from 
Lemma \ref{lemP1P2} it easily follows that 
\begin{equation}
\sum_{i=1}^n b_{k+i}e_i = \E(x).
\end{equation}
Thus, from (\ref{dxp}) and (\ref{correqu}) we receive 
\begin{equation}\label{dcorr} 
d(x, c_iv_i' + \E(x))^2 =\Var(x) - c_i^2=  \Var(x)(1 - \corr(x,v_i)^2)\geq 0,\quad i \in I_k.
\end{equation}
In particular, for $U$ and $Z$ as above, taking $k=1$, $x=Z$, and $v_1 = U - \E(U)$, 
and using the fact that $\corr(Z,U) = \corr(x,v_1)$, we receive 
\begin{equation}
-1\leq \corr(Z,U)\leq 1,
\end{equation}
with equality in either of the above inequalities implying the linear relationship 
\begin{equation}
Z= b_1v_1 +\E(Z) = b_1U  - b_1\E(U) + \E(Z),
\end{equation}
with the sign of $b_1$ being due to (\ref{correqu}) the same as of the correlation. 
 In Section \ref{secApprCoeff} we discuss some general methods for estimating the coefficients in the above projection and correlations for 
 the case of $x=f(X)$  
 corresponding to different functions of conditional moments of functions of two independent variables given the first variable, like  
conditional variances of stochastic model outputs given the model parameters, 
and $v_i = \phi_i(X)$, $i \in I_k$, as above, being some functions of the first variable.    
However, in the numerical experiments and the discussion below 
we consider only the coefficients of orthogonal projection 
of $f(X)$ onto span of constant vectors and independent coordinates of $X$, 
which describe the linear part of the relationship of $f(X)$ and the coordinates.  
Let us assume that $Z= f(X) \in L^2_n$, $X \in L^2_N$, and $\Var(X_i) > 0$, $i \in I_N$. 
Elements of the set $\Phi = \{(X_i-\E(X_i))e_j\}_{i\in I_N, j\in I_n}$ are orthogonal, and for 
$X'= \left(\frac{X_i-\E(X_i)}{\sigma(X_i)}\right)_{i=1}^N$, elements of $\Phi' = \{X'_ie_j\}_{i\in I_N, j\in I_n}$ are 
orthonormal with respect to $(,)_n$.  Denoting by $W$ the space of constant $\R^n$-valued random vectors, 
we define  $V = \text{span}(\Phi\cup W) = \text{span}(\Phi' \cup W)$.  
Coefficients of the respective elements of $\Phi$ in the orthogonal projection of $Z$ onto $V$ (in $L^2_n$) fulfill
\begin{equation}
b_{i,j} = \frac{(Z,\Phi_{i,j})_n}{\Var(X_i)} = \frac{\Cov(Z, X_ie_j)}{\Var(X_i)}, 
\end{equation}
and for coefficients of elements of $\Phi'$ in this projection we have
\begin{equation}
c_{i,j} =  (Z,\Phi_{i,j}')_n = \frac{\Cov(Z, X_ie_j)}{\sigma(X_j)}.
\end{equation}
We denote for $i \in I_N$, $c_i = \sum_{j=1}^n c_{i,j}e_j$, for $J \subset I_N$, $c_J = (c_j)_{j\in J}$,  $c = c_{I_N}$, 
\begin{equation}
c_J^2 =  \sum_{i\in J, j \in I_n}c^2_{i,j},
\end{equation}
and analogously for coefficients $b_{i,j}$. 
For $J \subset I_N$, we define the space of functions of $X$ linear in $X_J$ to be 
$V_J = \overline{\text{span}}(L^2_{n,X_{\sim J}} \cup \{\Phi_{i,j}\}_{i\in J,j \in I_n})$,
so that $V_{I_N}=V$. One can easily verify that the orthogonal projection of $Z$ onto functions linear in $X_J$ 
is equal to 
\begin{equation}
P_{V_J}(Z) = \E(Z|X_{\sim J}) + c_{J}X'_J,
\end{equation}
where $c_{J}X'_J = \sum_{i \in J}c_{i}X'_i$. 
We define the nonlinearity coefficient of $Z$ in $X_J$ as  
\begin{equation}\label{DNJ}
\begin{split}
DN_J &= ||Z - P_{V_J}(Z)||_n^2\\
&= ||Z||_n^2 - ||\E(Z|X_{\sim J})||_n^2 - c_{J}^2 \\
&= V_{X_J}^{tot}- c_{J}^2.
\end{split}
\end{equation}
It holds $0 \leq DN_J \leq V_{X_J}^{tot}$, equality on the left meaning that $Z$ is linear in $X_J$ and on the right that $c_J^2 = 0$,
that is knowledge of the linear part of dependence of $Z$ on $X_J$ is of no help in approximating it. 
For $V_{X_J}^{tot}>0$ one can also consider the normalized nonlinearity coefficient 
\begin{equation}
dN_J = \frac{DN_J}{V_{X_J}^{tot}},
 \end{equation} 
which fulfills $0 \leq dN_J \leq 1$, and is equal to the ratio of squared errors of the 
best approximation of $Z$ with functions linear in $X_J$ and another one with 
functions of $X_{\sim J}$. The nonlinearity coefficient of $Z$ in $X$,  
\begin{equation}{\label{dnc}} 
DN = DN_{I_N} = \Var(Z^2) - c^2, 
\end{equation} 
is equal to the squared error of the best approximation of $f(X)$ in $V$, and 
\begin{equation} 
dN = \frac{DN}{\Var(f(X))} 
\end{equation} 
tells what its ratio is to the squared error of the best approximation of $f(X)$ using constant vectors. 
We call $dN$ the relative error of the best linear approximation of $f(X)$. 
We have focused on nonlinearity coefficients, because they appear directly in our estimates of probabilities of 
localizations of functions values changes due to perturbations of their independent arguments, discussed in the next section, 
but similarly one can define linearity coefficients, like such normalized coefficient
\begin{equation}
dL_J = 1 - dN_J = \frac{c_J^2}{V_{X_J}^{tot}}.
 \end{equation} 
Let us define, for $J \subset I_N$, $g_J$ to be a measurable function such that
\begin{equation} 
\begin{split} 
g_{J}(X) &= \sum_{K \subset I_N : K \cap J \neq \emptyset} f_K(X_K) \\ 
&= f(X) - \E(f(X)|X_{\sim J}) 
\end{split}
\end{equation} 
(for $J = \{j\}$ we simply write $j$ in the subscript), where we have used ANOVA decomposition (\ref{anovaDec}). It holds 
\begin{equation} 
||g_{J}(X)||^2_n = V_{X_J}^{tot}. 
\end{equation} 
If $V_i^{tot}\neq 0$ and $n=1$, then let us define a coefficient which we 
call linear correlation of $f(X)$ in $X_i$, and which 
could be used as a measure of strength of linearity of $f(X)$ in $X_i$,
\begin{equation} 
\corrL_i = \corr(g_{i}(X),X_i') = \frac{c_i}{\sqrt{V_i^{tot}}}. 
\end{equation}
We have $\corrL_i^2 = dL_J$ and  $-1 \leq\corrL_i \leq 1$, with either of the equalities in the inequalities meaning that 
$f(X)$ is linear in $X_i$ and it holds
\begin{equation}
f(X)=c_iX_i' + \E(f(X)|X_{\sim J})
\end{equation}
with the sign of $c_i$ being the same as of $\corrL_i$. 

\section{\label{secInterv}Interventions into systems with uncertain parameters} 
Let for some $N,n \in \N_+$ $X= (X_i)_{i=1}^N$ be an $\R^N$-valued random vector and $f$ be a measurable function from $\R^N$ to $\R^n$. 
The change of $f(X)$ due to a perturbation $\Delta \in \R^N$ of $X$ is defined as
\begin{equation}\label{hX} 
h(X)=f(X + \Delta) - f(X),
\end{equation} 
for any measurable function $h$ from $\R^N$ to $\R^n$ such that this equality holds. 
$X$ may be for instance 
uncertain parameters of some model and $f(X)$ can be some its output, like 
a vector of concentrations of some species at some moment of time for a deterministic chemical model, or 
vector of certain conditional moments of different particle numbers at a given time or their conditional histogram 
given the model parameters for a stochastic model. 
Perturbation $\Delta$ of the model parameters 
can imitate adding a given amount of some species to the chemical system, e. g. as a pharmaceutical intervention. 
When planning which uncertain parameters of a model to perturb to receive a desirable effect on the output  
it might be useful to know the 
probability that the change of output will belong to a given area , e. g. be 
positive or negative. We describe here a method for obtaining lower bounds on certain 
such probabilities for appropriate $f$ and $X$, using only total sensitivity indices and orthogonal projection coefficients.  

Let us assume that the coordinates of $Y\in L^2_N$ are independent 
and have uniform or uniform discrete distributions on $\R$, 
and let $B_Y$ be the support of $\mu_Y$ (see Appendix \ref{appMath}).
We assume that $f(Y) \in L^2_n$, random vector $X$ takes values in a 
measurable set $B_X\subset B_Y$ satisfying $\mu_Y(B_X)>0$, and for  
each measurable $D\subset \R^n$, 
\begin{equation}\label{muAssum} 
\mu_{X}(D) = \frac{\mu_Y(D\cap B_X)}{\mu_Y(B_X)}. 
\end{equation}
In particular if $B_X = B_Y$, then $\mu_{X} = \mu_{Y}$. 
For a measurable function 
$s$ such that $s(X)$ is integrable, one can easily prove 
that 
\begin{equation}
\frac{\E(\I_{Y\in B_X}s(Y))}{\mu_Y(B_X)} = \E(s(X)).
\end{equation}
If $s$ is further nonnegative, then we receive
\begin{equation}\label{geqmuy}
\frac{\E(s(Y))}{\mu_Y(B_X)} \geq \E(s(X)).
\end{equation}
For a perturbation $\Delta = (\Delta_i)_{i=1}^N \neq 0$, 
let $J = \{i \in I_N: \Delta_{i} \neq 0\}$ and $\Delta_J = (\Delta_i)_{i\in J}$. 
We denote $A = \{X + \Delta \in B_Y\}$, which is the event that the perturbed arguments are in $B_Y$. 
In particular, if $\Delta + B_X = \{x +\Delta:x\in B_X\} \subset B_Y$, then $\PR(A) = 1$. 
We further use notations introduced in the previous section, like coefficients $b_{i,j}$ and $c_{i,j}$, variables $Y'$,  sequence $\Phi'$, 
product $Y_Jb_J$, nonlinearity coefficient $DN_J$, $g_J(Y)$ etc. defined identically but with $X$ replaced by $Y$ in the definitions. 
Let
\begin{equation}
 \delta(x) = h(x) - b\Delta. 
\end{equation} 
We have 
\begin{equation} 
\I_A \delta(X) = \I_A(g_{J}(X + \Delta) - g_J(X) - b\Delta). 
\end{equation}
For a function 
\begin{equation}\label{axdef}
a(x) = g_J(x) - b_J(x_J - \E(Y_J)),
\end{equation}
it holds
\begin{equation}\label{diffA}
\I_A\delta(X) = \I_A(a(X + \Delta) - a(X)). 
\end{equation}
Using Lemma \ref{lemP1P2} it is easy to prove that $b_J(Y_J-\E(Y_J)) = c_JY'_J$ is an orthogonal projection of $g_J(Y)$ onto 
span($\{Y'_i e_j\}_{i \in J, j\in I_n}$) and $||c_JY'_J||^2_n = c_J^2$, so that from (\ref{dxp}) we have 
\begin{equation}\label{aY2}
||a(Y)||_n^2 = V_{Y_J}^{tot} - c^2_{J} = DN_J.
\end{equation}
We have the following easy generalization of Chebyshev's inequality \cite{billingsley1979}. 
\begin{lemma}\label{Chebyshev}
For  $Z \in L^2_n$, $\epsilon \in \R_{+}$, and each event $B$ it holds 
\begin{equation}
||\I_BZ||_n^2 = \E (\I_B|Z|^2) \geq \PR(B,|Z| \geq \epsilon)\epsilon^2.
\end{equation}
\end{lemma}
Using it we obtain 
\begin{equation}\label{prAdelta} 
\PR(A,|\delta(X)| \geq \epsilon) \leq \frac{||\I_A\delta(X)||^2_n}{\epsilon^2}. 
\end{equation} 
Applying triangle inequality \cite{rudin1970} 
to (\ref{diffA}) we receive 
\begin{equation}\label{deltaXm} 
||\I_A\delta(X)||_n^2 \leq (||\I_A a(X)||_n + ||\I_A a(X + \Delta)||_n)^2. 
\end{equation} 
We estimate 
\begin{equation}\label{Aax} 
\begin{split} 
||\I_A a(X)||^2_n &\leq \E(|a(X)|^2) \leq \frac{\E(|a(Y)|^2)}{\mu_Y(B_X)} \\ 
&  = \frac{||a(Y)||^2_{n}}{\mu_{Y}(B_X)}, 
\end{split} 
\end{equation} 
where in the second inequality we used (\ref{geqmuy}). 
Furthermore, 
\begin{equation}\label{xdbaxd} 
\begin{split}
||\I_{X + \Delta \in B_Y} a(X + \Delta)||^2_n  &\leq \frac{\E(\I_{Y + \Delta \in B_Y}|a(Y + \Delta)|^2)}{\mu_Y(B_X)}\\ 
& \leq \frac{||a(Y)||^2_{n}}{\mu_{Y}(B_X)}, 
\end{split} 
\end{equation} 
where in the first inequality we used (\ref{geqmuy}) and in the last one the assumption of independence and uniform distributions 
of coordinates of $Y$. 
From (\ref{deltaXm}), (\ref{Aax}), and (\ref{xdbaxd}), we receive 
\begin{equation}\label{deltaXmFin} 
||\I_A\delta(X)||_n^2 \leq \frac{4||a(Y)||^2_n}{\mu_Y(B_X)} = \frac{4DN_J}{\mu_Y(B_X)}. 
\end{equation} 
For $p,\ r \in \R^n$, 
we define a ball with center $p$ and radius $r$ as 
\begin{equation}
B(p,r) =  \{x\in\R^N:\wt{d}(x,p)<r \}. 
\end{equation}
We have the following lower bound on the probability that the effect of perturbation lies in a ball 
with center $b_J\Delta_J$ and radius $\epsilon >0$
\begin{equation}\label{longhB}
\begin{split}
\PR(h(X) \in B(b_J\Delta_J,\epsilon)) &= \PR(|\delta(X)| < \epsilon) \leq \PR(A,|\delta(X)| < \epsilon) \\
&= \PR(A) - \PR(A, |\delta(X)| \geq \epsilon) \\
& \geq \PR(A) - \frac{4DN_{J}}{\mu_Y(B_X)\epsilon^2}, 
\end{split}
\end{equation} 
where in the last equality we used (\ref{prAdelta}) and (\ref{deltaXmFin}). In particular   
if $P(A)=1$ and $f(Y)$ is linear in $Y_J$, so that $DN_J=0$, then we receive $h(X) = b_J\Delta_J$, 
which also follows from the fact that in such case $f(Y) = b_JY_J + \E(f(Y)|Y_{\sim J})$. 
If $n = 1$ and $b\Delta$ is positive (negative), then the probability that the effect of perturbation $\Delta$ 
on the output is positive (negative) is bounded from below by 
\begin{equation}\label{probChange}
\PR(h(X) \in B(b\Delta,b\Delta)) \geq \PR(A) - \frac{4DN_J}{ \mu_Y(B_X)(b_J\Delta_J)^2}.
\end{equation}
We apply the above theory to the GTS model at the end of Section \ref{secGTS}. 

\section{\label{secStat}Statistics, Monte Carlo procedures, and inefficiency constants - some new definitions, generalizations 
and interpretations} 
If $\mathcal{P}$ consists of all probability distributions on $\R$ 
with finite $n$th moments for some $n \in \N_+$, then 
for $Q$ whose restriction to $\mc{T}_n$ is a function of the first $n$ so restricted moments as in Section \ref{secCondMoms}, 
we define estimand $G_Q$ on $\mc{P}$ to be such that for each $\mu \in \mc{P}$, 
\begin{equation}\label{defgqmu}
G_Q(\mu) = Q_\mu(\id_\R), 
\end{equation}
or equivalently $G_Q(\mu)=Q(X)$, $X \sim \mu$. 
In particular, for $Q= \E$ and $\Var$ we receive estimands  $G_E$ and $G_{Var}$ from Appendix \ref{appStatMC}.
Degree of an estimand $G$ is defined as the smallest $n \in \N_+$
for which there exists an unbiased estimator of $G$ in $n$ dimensions (see Appendix \ref{appStatMC}), assuming that for some $n$ such 
estimator exists \cite{lehmann1998theory, Halmos_1946}. 
In other words, it is the minimum value of $n$ for which there exists a measurable real-valued $\phi$ on $\mc{S}^n$ such that
 for each $\mu \in \mc{P}$ and $X_1, \ldots, X_n$ i. i. d., $X_1 \sim \mu$, it holds 
\begin{equation}
G(\mu) = \E(\phi(X_1, \ldots, X_n)). 
\end{equation}

It was proved in \cite{Halmos_1946} 
that for admissible distributions $\mathcal{P}$ on $\R$ containing all 
finite discrete distributions on $\{0,1\}$ (see Appendix \ref{appStatMC}) and possibly some other distributions 
with finite $n$th moments for some $n \in \N_+$, for $Q$ being 
the $n$th moment or central moment, $G_Q$ restricted to $\mathcal{P}$ has degree 
exactly $n$. 

Let $n \in \N_+$, and $G_i$ be an estimand for $\mc{P}$, $i \in I_n.$ 
Then we call $G=(G_i)_{i=1}^n$ an $n$-dimensional or if $n$ is left unspecified simply vector-valued estimand. If 
$\phi_i$ is an [unbiased] estimator of $G_i$, $i \in I_n,$ then we call 
$\phi=(\phi_i)_{i=1}^n$ an [unbiased] estimator of $G$, where the words in square brackets in a sentence 
should be either all read or omitted. 
Error of approximation of $G$ by its unbiased estimator $\phi$ for some $\mu \in \mc{P}$ can be quantified by 
$\Var_\mu(\phi)$ for some variance for random vectors as in Section \ref{secOrthog}. 
The suitable scalar product in the definition of such a variance can depend on the estimation 
problem at hand. In Section \ref{secApprErr} we shall discuss a problem for which the standard scalar product is a natural choice. 

Let us assume that similarly as for $m=1$ in Appendix \ref{appStatMC},
to estimate some $\lambda_1,\ldots,\lambda_m \in \R$ for some $m \in \N_+$
we carry out $n$-step MC procedures using the same variable $X\sim \mu$ and 
single-step MC estimators $\phi_i$ of $\lambda_i$ for $\mu$, $i \in I_m$. 
Then we say that these quantities are estimated in the same MC procedure. 
For $i \in I_m$, for the subprocedure estimating $\lambda_i$ we use notations 
analogous as in Appendix \ref{appStatMC} but with a subscript $i$, 
like $W_{i,j}$ for the $j$th observable of the $i$th single step estimator as well as 
$\overline{W}_i$ for the observable and $\wt{\lambda}_i$ for its observed value,
 $\Var_{f,i}$ for variance and $\sigma_{f,i}$ for the standard deviation of 
the $i$th final estimator $\phi_{f,i}$.
Then $\phi=(\phi_i)_{i=1}^m$ is called a single step MC estimator of $\lambda=(\lambda_i)_{i=1}^m$ for $\mu$, and 
$\phi_f = (\phi_{f,i})_{i=1}^m$ the final or $n$-step one. 
We define the variance $\Var_s$ of a single step MC estimator and such variance $\Var_f$ of the final MC estimator 
using the same formulas as for $m=1$ in Appendix \ref{appStatMC} but with 
$\Var$ symbol denoting some variance for random variables 
as in Section \ref{secOrthog}. Note that from (\ref{varAveVec}) we still have 
\begin{equation}\label{varfsn} 
\Var_{f} = \frac{\Var_s}{n}.
\end{equation} 
We can define inefficiency constants for sequences of MC procedures for estimating 
$\lambda \in \R^m$ identically as in Section \ref{secMCIneff} for $m=1$ and thanks to (\ref{varfsn})  they enjoy the same interpretation as
in this section - if we have $\delta$-approximate equality of average duration times of two MC procedures 
then the ratio of the final MC variances is $\delta$-approximately equal to the ratio of their inefficiency constants.
Let us notice two further interpretations of the inefficiency constants using notations
as in Section \ref{secMCIneff}. 
The first is that if we have $\delta$-approximate equality of variances of the final estimators of the MC procedures
\begin{equation}\label{vdeltav}
\Var_{f}(n) \approx_{\delta} \Var_{f}'(n'),
\end{equation}
then the ratio of their average durations is $\delta$-approximately the same as of the inefficiency constants
\begin{equation}\label{varTiRatio}
\frac{\tau_{f}(n)}{\tau_{f}'(n')} \approx_{\delta} \frac{c}{c'}.
\end{equation}
Secondly, consider the approach to estimating $\lambda$ using a sequence of MC procedures 
in which for some target accuracy threshold $\epsilon>0$, one carries out the MC procedure 
with the smallest number $n(\epsilon)$ of MC steps 
for which variance of the final MC estimator $\Var_f(n(\epsilon))$ is below $\epsilon$. 
In practice one usually does not know $\Var_f(n(\epsilon))$, but can approximate it using values of estimator (\ref{varEst}).
For $x \in \R$, let $\lceil x\rceil$ denote the smallest integer $l$ such that $x \leq l$. 
It holds $n(\epsilon) = \left\lceil \frac{\Var_s}{\epsilon} \right\rceil$
and $\Var_f = \frac{\Var_s}{n({\epsilon})}$, and similarly for the primed sequence. We have  
\begin{equation}
\frac{\Var_{f}(n(\epsilon))}{\Var_{f}'(n'(\epsilon))} = \frac{\Var_s\left\lceil\frac{\Var_s'}{\epsilon}\right\rceil}{{\Var_s'}\left\lceil \frac{\Var_s}{\epsilon} \right\rceil},
\end{equation}
which tends to one as $\epsilon$ goes to zero, and thus from (\ref{cntau}) the ratio $\frac{\tau_f(n(\epsilon))}{\tau_f'(n'(\epsilon))}$ 
of average durations of these procedures tends to $\frac{c}{c'}$. 

\section{Testing methodology}
We shall use what we call $k$-$\sigma$ test for each 
of the null hypotheses that for some $b,\lambda \in \R$, 
$\lambda = b$, $\lambda \geq b$, or $\lambda \leq b$, in which for $\wt{\lambda}$ denoting observed value of 
the final MC estimator and $\wt{\sigma}_f$ estimate of its standard deviation as in Appendix \ref{appStatMC}, one 
rejects the hypothesis if $|\wt{\lambda} -b| > k\wt{\sigma}_f$, $b-\wt{\lambda} > k\wt{\sigma}_f$, 
or $\wt{\lambda} -b > k\wt{\sigma}_f$, respectively. For sufficiently large $n$  
the significance level (upper bound on the probability of rejecting wrongly the hypothesis if it is correct) 
for such $k$-$\sigma$ test can be chosen arbitrarily close to  
$2(1 - \Phi(k))$ for the equality and $1 - \Phi(k)$ for the inequalities hypotheses for $\Phi(k)$ 
being the  cumulative distribution function of the standard normal distribution (see Appendix \ref{appStatMC}).
Such significance levels are called asymptotic. 
Let the coordinates of $\lambda \in \R^2$ be estimated in the same $n$-step MC procedure and   
let $\overline{W}_d = \overline{W}_1 - \overline{W}_2$ and $\sigma_d = \sigma(\overline{W}_d)$. 
From the inequality $\sigma(X + Y) \leq \sigma(X) + \sigma(Y)$ for $X,Y \in L^2$, which follows from triangle inequality \cite{rudin1970}, 
we have $\sigma_d \leq \sigma_{f,1} + \sigma_{f,2}$. Furthermore, from CLT applied to the sequence $W_{1,j} - W_{2,j}, j \in I_n,$ for $n$
 going 
to infinity $\sqrt{n}\overline{W}_d$ converges in distribution to $\ND(\lambda_1 - \lambda_2, \Var(W_{1,1} - W_{2,1}))$. 
Thus if for the estimand $\lambda_i$ we obtained a final MC estimate $\wt{\lambda}_i \pm \wt{\sigma}_{f,i}$, $i \in I_2$, 
one can use $k$-$\sigma$ test rejecting the hypothesis $\lambda_1 = \lambda_2$ 
if $|\wt{\lambda}_1 - \wt{\lambda}_2| > k(\wt{\sigma}_{f,1} + \wt{\sigma}_{f,2})$ 
or the hypothesis $\lambda_1 \geq \lambda_2$ if $\wt{\lambda}_2 - \wt{\lambda}_1 > k(\wt{\sigma}_{f,1} + \wt{\sigma}_{f,2})$, 
with the same asymptotic significance levels as for the equalities and inequalities hypotheses discussed above. 
For two independently run $n_i$-step MC procedures estimating $\lambda_i$ and with 
observables of the final MC estimators $\overline{W}_i$ with variances $\Var_{f,i}, i \in I_2$, 
from the Lindeberg CLT \cite{billingsley1979}, 
\begin{equation}
\frac{(\overline{W}_1 - \lambda_1) + (\overline{W}_2 -\lambda_2)}{\sqrt{\Var_{f,1} + \Var_{f,2}}}  
\end{equation}
converges in distribution to $\ND(0,1)$ for $n_1$ and $n_2$ going to infinity. Thus using analogous notations as above 
one can use $\sqrt{\wt{\sigma}_{f,1}^2 + \wt{\sigma}_{f,2}^2}$ instead of 
$\wt{\sigma}_{f,1} + \wt{\sigma}_{f,2}$ in the above tests with the same asymptotic 
significance levels for $n_1$ and $n_2$ going to infinity 
as above for the same $k$. 
We often make statements about the results of our numerical experiments 
like that the estimate $\wt{\lambda}_1 \pm \wt{\sigma}_{f,1}$ is (statistically significantly) greater 
than $\wt{\lambda}_2 \pm \wt{\sigma}_{f,2}$ 
by which we mean that the null hypothesis $\lambda_1 \leq \lambda_2$ can be rejected in a $k$-$\sigma$ test 
as above for some $k \geq 3$. 
\newline

\section{\label{secUnbiased}Generalization of estimands on pairs and their estimation schemes to many functions case} 
In this section we among others generalize the concepts from Section \ref{secSchemesPrev}, like of admissible pairs, estimands, 
statistics, estimators, and estimation schemes, so that they can be used for problems of estimation of certain quantities 
defined for several functions of different sequences of random arguments.
These concepts shall be used in their full generality in Section \ref{secApprCoeff} e. g. when dealing with
orthogonal projection coefficients onto orthogonal functions of the first variable of functions of conditional moments 
given the first variable, like conditional variance,  of functions of two independent random variables.
Unfortunately, giving only the number of distributions as before 
is not sufficient to specify the type of the more general admissible pairs 
we need so we introduce a helper definition of signature containing such specification.
\begin{defin}\label{defSign} 
We call $Sg=(N,k,J,\mc{H})$ a signature (of some admissible pairs) if 
$N, k \in \N_{+}$, sequence $J = (J_i)_{i=1}^k$ consists of nonempty subsets of $I_N$ such that 
\begin{equation}\label{inkji} 
I_N = \bigcup_{i=1}^kJ_i, 
\end{equation} 
and coordinates of $\mc{H} = (\mc{H}_i)_{i=1}^k$ are measurable spaces $\mc{H}_i = (C_i, \mc{C}_i)$, $i \in I_k$. 
\end{defin} 
\begin{defin}\label{defAPairs} 
We call $\mc{V}$ admissible pairs with signature $Sg$ as in Definition \ref{defSign} or 
equivalently admissible pairs of $N$ distributions and $k$ functions with values spaces $\mc{H}$ 
and sets of arguments' indices 
$J$ as in this definition if it is a nonempty class consisting of pairs $(\mu, f) = ((\mu_i)_{i=1}^N,(f_i)_{i=1}^k)$ 
such that $\mu_i$ is a probability measure, $i \in I_N$, 
and $f_i$ is a measurable function from  $\bigotimes_{i \in J_i}\mc{S}_{\mu_i}$ to $\mc{H}_i$, $i \in I_k$. 
\end{defin} 
We identify each one-element sequence $(x)$ with $x$ (see Appendix \ref{appMath}), so that 
for $N=1$ the first coordinate in each pair from $\mc{V}$ in the above definition is a measure,
while for $k=1$, its second coordinate is a function and from (\ref{inkji}) we have $J = I_N$. 
Thus, for $k=1$ and $\mc{H}= \mc{S}(\R)$, 
the above definition reduces to definition of admissible pairs with $N$ distributions from Section \ref{secSchemesPrev}. 
Note that the class $\mc{T}_n$ (see (\ref{defMCRn})) is an example of admissible pairs of single distributions and single 
real-valued functions. 
Similarly as in Section \ref{secSchemesPrev},
an estimand on admissible pairs $\mc{V}$ is any real-valued function on it.
For instance for $Q$ such that restricted to $\mc{T}_n$ it is a real-valued function of the first $n$ so restricted moments 
as in Section \ref{secCondMoms}, e. g. for the $n$th moment or central moment, $Q_{|\mc{T}_n}$ is an estimand on $\mc{T}_n$. 
We define estimand $PR$ on the admissible pairs $\mc{V}$ of single distributions and two real-valued functions consisting of all 
possible $\alpha = (\mu,(f_1,f_2))$ such that $f_1f_2 \in L^1(\mu)$, in which case $PR(\alpha) = \E_{\mu}(f_1f_2)$. 
We now describe and illustrate by example a method for obtaining vectors of estimands, which 
will be frequently used in Section \ref{secApprCoeff}. 
Let us consider a signature $Sg$ as in Definition \ref{defSign}, signatures $Sg' =(Sg'_i)_{i=1}^n$, 
such that $Sg_i' = (N,k_i,(\mc{H}_{i,j})_{j=1}^{k_i}, (J_{i,j})_{j=1}^{k_i}), i \in I_n,$ and 
$\psi=(\psi_i)_{i=1}^n$, where $\psi_i:I_{k_i} \rightarrow \N_+$ 
are $1$-$1$ functions, $i \in I_n$. We say that such 
$Sg$ is received from $Sg'$ using $\psi$ if $J_{\psi_i(j)}=J_{i,j}$, $\mc{H}_{\psi_i(j)}=H_{i,j}$, $j \in I_{k_i}$, $i \in I_{n}$, 
and $\bigcup_{i=1}^n\psi_i[I_{k_i}]=I_k$. 
Let $G'=(G'_i)_{i=1}^n$ be such that $G'_i$ is an estimand on admissible pairs $\mc{V}'_i$ with signature $Sg'_i$, $i \in I_n,$ 
and $Sg$ be received from $Sg'$ using some $\psi$ as above. 
We say that $G=(G_i)_{i=1}^n$ are trivial extensions of $G'$ using $\psi$ if 
for each $i \in I_n$, $G_i$ is an estimand on pairs $\mc{V}_i$ with signature $Sg$ and consisting of 
all possible $(\mu,f) = ((\mu_j)_{j=1}^N,(f_i)_{i=1}^k)$ 
such that for some $\beta =(\mu, (g_1, \ldots, g_{k_i})) \in  \mc{V}_i'$, 
it holds  $f_{\psi(j)} = g_j, j \in I_{k_i}$, in which case $G_i(\alpha) = G_i'(\beta)$. 
As an example of the above construction 
we define estimands $PR^n = (PR_i)_{i=1}^n$ (identifying $PR^1$ with $PR$) to be 
trivial extensions of $(PR)_{i=1}^n$ using $\psi$ such that $\psi_i(1)=i$ and $\psi_i(2)=n+1$, $i \in I_n$. The resulting 
$PR^n$ are estimands on common admissible pairs (defined as at the beginning of
 Section \ref{secIneffSchemes}) consisting of 
$\alpha= (\mu,(f_i)_{i=1}^{n+1})$ such that  $(\mu,(f_i,f_{n+1}))\in D_{PR}$, $i \in I_n$, for which   
$PR^n_i(\alpha) = \E_\mu(f_if_{n+1})$, $i \in I_n$. 
Note that if $f_{n+1}\in L^2(\mu)$ and  $\Phi = \{f_i\}_{i=1}^n$ is an orthonormal set in $L^2(\mu)$, then for $\alpha$ as above,
$PR^n_i(\alpha)$ is the coefficient of $f_i$ in the orthogonal projection of $f_{n+1}$ onto span$(\Phi)$, $i \in I_n$. 

For $N \in \N_+$, 
let us consider a nonempty finite set $K \subset I_N \times \N_+$, 
called arguments' indices for $N$. For a sequence of measurable spaces $\mc{S} = (\mc{S}_i)_{i=1}^N$, we define 
$\mc{S}^K = \bigotimes_{(i,j)\in K} \mc{S}_i$, 
 of sets $B = (B_i)_{i=1}^N$, 
 $B^K = \prod_{(i,j)\in K} B_i$, 
and of probability distributions $\mu = (\mu_i)_{i=1}^N$, 
$\mu^K = \bigotimes_{(i,j)\in K}\mu_i$. 
Note that $\wt{X} \sim \mu^K$ means that $\wt{X} = (\wt{X}_{i,j})_{(i,j)\in K}$, random variables 
$\wt{X}_{i,j} \sim \mu_i,$ $(i,j) \in K$ being independent. 
Let $v \in \N_+^N$. We define $K_v = \{(i,j): i \in I_N, j \in I_{v_i}\}$. We identify sequences 
$((x_{i,j})_{j=1}^{v_i})_{i=1}^N$ and $(x_{\beta})_{\beta \in K_v}$. In particular for $K=K_v$, 
$\wt{X}$ as above is identified with  $((\wt{X}_{i,j})_{j=1}^{v_i})_{i=1}^N$, while for $B$, $\mc{S}$, and $\mu$ as above,
 $B^{K}$ is identified with $B^v$, $\mc{S}^{K}$ with $\mc{S}^v$ , and $\mu^{K_v}$ with $\mu^v$,
defined in Section \ref{secSchemesPrev}.  
Let $\mc{V}$ be some admissible pairs as in Definition \ref{defAPairs} and $K$ be arguments' indices for $N$. 
Sets $\mc{V}_1$ and $\mc{V}_2$ are defined analogously as in Section \ref{secSchemesPrev}.
For a measurable space $\mc{H}$, a $\mc{H}$-valued 
statistic $\phi$ for $\mc{V}$ with (arguments) indices $K$ 
is a function on $\mc{V}_2$ such that for each $(\mu, f) \in \mc{V}$, $\phi(f)$ 
is a measurable function from $\mc{S}_{\mu}^K$ to $\mc{H}$. 
For $k=1$, $K=K_v$ for some $v$, and $\mc{H}=\mc{S}(\R)$ 
this coincides with the definition of statistic for $\mc{V}$ with dimensions of arguments $v$
 from Section \ref{secSchemesPrev}. 
Analogously as in Section \ref{secSchemesPrev},
for a real-valued statistic $\phi$ for $\mc{V}$ with indices $K$, and some $Q$ as in Section \ref{secCondMoms} 
like variance $\Var$ or expectation $\E$, we denote for $\alpha=(\mu,f) \in \mc{V}$, 
\begin{equation}\label{Qmuf} 
Q_{\alpha}(\phi) =  Q_{\mu^{K}}(\phi(f)), 
\end{equation} 
whenever the expression on the right makes sense. If $\phi$ is an $\R^n$-valued statistic for $\mc{V}$ with indices $K$
then we shall also use notation (\ref{Qmuf}) for $Q=\E$ when $\phi(f) \in L^1_n(\mu^{K})$ or for 
$Q=\Var$ for some variance for random vectors as in Section \ref{secOrthog} and $\phi(f) \in L^2_n(\mu^{K})$. 
Let $G$ be an estimand on $\mc{V}$. 
We call any real-valued statistic $\phi$ for $\mc{V}$ with indices $K$ estimator of $G$ if 
for each $\alpha = (\mu,f)\in \mc{V}$, we consider values of 
$\phi(f)(X)$ for each $X \sim \mu^K$ to be certain approximations of $G(\mu,f)$, and 
analogously as in Section \ref{secSchemesPrev} such $\phi$ is further called unbiased if 
\begin{equation}\label{genUnb} 
\E_{\alpha}(\phi) = G(\alpha), \quad \alpha \in \mc{V}.
\end{equation}
We shall now introduce a number of notations needed to define estimation schemes for the above 
estimands.
Let us consider some signature $Sg$ as in Definition \ref{defAPairs}. A sequence of finite sets 
$A =(A_i)_{i=1}^k$ such that $A_i \subset \N_{+}^{J_i}$, $i \in I_k$, and at least one of these sets 
is nonempty is called sets of evaluation vectors. For $k =1$this reduces to evaluation vectors for $N$
from Section \ref{secSchemesPrev}.
We define the arguments' indices of $A$ as
\begin{equation}\label{pal} 
\begin{split} 
p_A &= \{(i,j) \in I_N\times\N_+: \text{ for some } l \in I_k \text{ such that } i \in J_l,\\ 
 &\text{ there exists } v \in A_l \text{ such that } v_i = j \}. 
\end{split} 
\end{equation} 
Let $\mc{V}$ be admissible pairs with signature $Sg$. 
For each $i \in I_k$ and $v \in A_i$, we define evaluation operator  
$g_{\mc{V},A,i,v}$ to be a $\mc{H}_i$-valued statistic for $\mc{V}$ with indices $p_A$ 
such that for each $(\mu,f) \in \mc{V}$ and $x \in B_\mu^{p_{A}}$, it holds 
\begin{equation}
g_{\mc{V},A,i,v}(f)(x) = f_i(x_v), 
\end{equation}
where 
\begin{equation}\label{xveq} 
x_v = (x_{l,v_l})_{l\in J_i}.  
\end{equation}  
For $k=1$ we omit subscript $i$ in the above or below alternative notations for evaluation operators, so that 
if further $p_A = K_w$ for some $w \in \N_+^N$ and $\mc{H}=\mc{S}(\R)$, 
the new $g_{\mc{V},A,v}$ coincides with the definition of evaluation
operator from Section \ref{secSchemesPrev}. Similarly as in Section \ref{secSchemesPrev},  
$\mc{V}$ and $A$ in the subscripts are omitted when known from the context.
If for some $l \in \N_+, l \leq N$, it holds $J_i = I_l$, then we use a C-array-like notation  
\begin{equation}\label{clikeg} 
g_i[v_1-1]\ldots[v_l-1] = g_{i,v}, 
\end{equation} 
while for $J_i = \{i\}$ we use notation 
\begin{equation}\label{ridef} 
r_i[v_i-1] = g_{i,v}. 
\end{equation}
For each $i \in I_k$ for which $A_i$ is nonempty, we define the following $\mc{H}_i^{|A_i|}$-valued statistic for $\mc{V}$ with indices $p_A$, 
$g_{\mc{V},A,i} = (g_{\mc{V},A,i,v})_{|v \in A_i}$ (see \ref{ordnot}). 
Let 
$\delta(A) = |\{i\in I_k: A_i \neq \emptyset\}|$, that is the number of nonempty coordinates of $A$, 
and for each $i \in I_{\delta(A)}$, let
$\gamma_A(i)$ be the index of the $i$th nonempty coordinate of $A$. 
Let 
\begin{equation}
\mc{H}_A = (C_A, \mc{C}_A) =\bigotimes_{i=1}^{\delta(A)}\mc{H}_{\gamma_A(i)}^{|A_{\gamma_A(i)}|}. 
\end{equation} 
We define the following $\mc{H}_A$-valued statistic for $\mc{V}$ with indices $p_A$, 
$g_{\mc{V},A}=(g_{\mc{V},A,\gamma_A(i)})_{i=1}^{\delta(A)}$. 
For a signature $Sg$, let $A$
be sets of evaluation vectors for $Sg$ and $t$ be a measurable real-valued function on $\mc{H}_A$.
Let $\kappa=(t, A)$, which we call a scheme for $Sg$. 
This coincides with the previous definition of a scheme for $k=1$ and $\mc{H}= \mc{S}(\R)$.
We define arguments' indices of $\kappa$ as $p_\kappa = p_A$. 
The statistic $\phi_{\kappa,\mc{V}}$ given by $\kappa$ and $\mc{V}$
is defined using the same formula (\ref{phiAF}) as in Section \ref{secSchemesPrev}.
Let $G$ be an estimand on $\mc{V}$. 
Similarly as in Section \ref{secSchemesPrev} $\kappa$ is called an [unbiased] (estimation) scheme for $G$  
if $\phi_{\kappa,\mc{V}}$ is an [unbiased] estimator of $G$. 

Let now for some $n \in \N_+$, $\kappa = (\kappa_i)_{i=1}^n=(t_i,A_i)_{i=1}^n$ be a sequence of schemes for $Sg$, called 
an ($n$-dimensional) scheme for $Sg$. We define the vector of sets of evaluation vectors of $\kappa$ as  
\begin{equation}\label{Akappa}
A_\kappa = (\bigcup_{i=1}^n A_{i,j})_{j=1}^k. 
\end{equation}
Let $N \in \N_+$, $K$  be arguments' indices for $N$, $L \subset K$, $L\neq \emptyset$, and 
$B=(B_1,\ldots,B_N)$ be a sequence of nonepmty sets. For $x \in B^K$, we define 
\begin{equation}\label{xL}
x_L = (x_\beta)_{\beta \in L}, 
\end{equation}
while for $x \in B^K$ and $L=\emptyset$, we define $x_L=\emptyset$. 
We also define arguments' indices $p_\kappa$ of $\kappa$ to be equal to $p_{A_\kappa}$ defined as in (\ref{pal}). 
A statistic given by $\kappa$ and $\mc{V}$, denoted as $\phi_{\kappa,\mc{V}}$, 
 is defined as an $\R^n$-valued statistic for $\mc{V}$ with indices $p_{A_\kappa}$
such that for each $(\mu,f) \in \mc{V}$ and $x \in B_\mu^{p_{A_\kappa}}$
\begin{equation} \label{phikappamany}
\phi_{\kappa,\mc{V}}(f)(x) = (\phi_{\kappa_i,\mc{V}}(f)(x_{p_{\kappa_i}}))_{i=1}^n, 
\end{equation} 
which for $n=1$ coincides with the previous definition. 
Let $G = (G_i)_{i=1}^n$ be a sequence of estimands, each on some (possibly different) admissible pairs but all with the same 
signature $Sg$.  
Let us assume that $\kappa_i$ is an [unbiased] estimation scheme for $G_i$, $i \in I_n$, 
in which case we call the above $\kappa$ an [unbiased] estimation scheme for $G$. 
Similarly as in Section \ref{secSchemesPrev}
we denote $\wh{G}_{\kappa,i}=\phi_{\kappa_i,D_{G_i}}$, $i \in I_n$, and use for it 
analogous shorthand notations in that section in analogous situations.

Let us now move on to examples. For an estimand $Q_{|\mc{T}_n}$ as above, if there exists an estimator $\phi_{Q}$ of $G_Q$ 
in $m$ dimensions (see Section \ref{secStat}), 
then an unbiased estimation scheme $SQR=(t,A)$ for $Q_{|\mc{T}_n}$ is given by $t = \phi_Q$ and $A = I_m$. Using 
notation (\ref{clikeg}), the estimator of this scheme can be written as 
\begin{equation}\label{gqsq} 
\wh{G}_{Q,SQR}= \phi_Q((g[i])_{i=0}^{m-1}). 
\end{equation}
The fact that this estimator is unbiased follows from the fact that for each $(\mu,f) \in \mc{T}_n$, 
\begin{equation}
\begin{split}
\E_{\mu^m}(\wh{G}_{Q,SQR}(f))&= \E_{(\mu f^{-1})^m}(\phi_Q) = G_Q(\mu f^{-1}) \\
&= Q(\mu f^{-1},\id_\R) = Q(\mu,f), \\
\end{split}
\end{equation}
where in the first equality we used the change of variable Theorem \ref{thchvar}, 
in the second and third the definitions of $\phi_Q$ and $G_Q$ (see (\ref{defgqmu})), respectively, and in the last 
(\ref{qz}) and again Theorem \ref{thchvar}. 
An unbiased estimation scheme $SPR=(t,(A_i)_{i=1}^2)$ for $PR$ is given by $A_{1} = A_{2}=\{1\}$ and 
$t(x_1,x_2) = x_1x_2$, so that, using notation (\ref{clikeg}), its estimator can be written as 
\begin{equation}\label{defspr} 
\wh{PR}_{SPR} = g_1[0]g_{2}[0], 
\end{equation} 
and for each $(\mu, f) \in \mc{V}$ and  $X \sim \mu$, it holds 
\begin{equation} 
\wh{PR}_{SPR}(f)(X) = f_1(X)f_{2}(X). 
\end{equation} 
If $Sg$ is received from $Sg'$ using $\psi$ as above and we are given 
schemes $\kappa' = (\kappa'_i)_{i=1}^n$ such that 
$\kappa'_i$ is a scheme for $Sg'_i$, $i \in I_n$, then trivial extensions of $\kappa'$ 
using $\psi$ are defined as a scheme $\kappa = (\kappa_i)_{i=1}^n$ for $Sg$ 
such that for each $i \in I_n$, $t_{\kappa_i} = t_{\kappa'_i}$ and for $j \in I_k$, if $j \in \psi_i[I_{k_i}]$, then 
$A_{\kappa_i,j} = A_{\kappa'_i,\psi_i^{-1}(j)}$, and otherwise  $A_{\kappa_i,j} = \emptyset$. It is easy to check that 
if $\kappa'_i$ is an unbiased scheme for estimation $G'_i$, $i \in I_n$, as above, 
and $G$ are trivial extensions of $G'$ using $\psi$, then $\kappa$ is an 
unbiased estimation scheme for $G$. An unbiased estimation scheme $SPR^n$ for $PR^n$ is defined as trivial extensions of 
$(SPR)_{i=1}^n$ using the same $\psi$ as when extending $(PR)_{i=1}^n$ to $PR^n$. 
With the help of notation (\ref{clikeg}), estimator of its $i$th subscheme can be written as 
\begin{equation}\label{E1g} 
\wh{PR}^n_{i,SPR^n} = g_i[0]g_{n+1}[0]. 
\end{equation} 
We shall use formulas for estimators like (\ref{gqsq}) and (\ref{E1g}) to define previously undefined schemes
analogously as in Section \ref{secSchemesPrev}.

\section{\label{secineffgen}Generalization of the inefficiency constants of schemes}
Let us make some generalizations of the definitions of inefficiency constants of schemes 
from Section \ref{secIneffSchemes} so that they can be used for the more general schemes from the previous section 
and for quantifying the inefficiency of estimation of several estimands in the same sequence
of Monte Carlo procedures using a given scheme.
If $\kappa=(\kappa_i)_{i=1}^n= (t_i,A_i)_{i=1}^n$ is an estimation scheme for 
estimands $G=(G_i)_{i=1}^n$ on some common admissible pairs $\mc{V}$ as in Definition \ref{defAPairs}, then
$\kappa$ can be used to generate estimates of coordinates of $G(\alpha)$ for some $\alpha =(\mu,f) \in \mc{V}$ as follows. 
For a $X \sim \mu^{p_\kappa}$, one 
computes the quantities $g_{\mc{V},A_i,j,v}(f)(\wt{X}_{p_{A_i}})=f_j(X_v)$, $i \in I_n$, $j \in I_k$, $v \in A_{i,j}$, 
bearing in mind that they are equal for the same $j$ and $v$ and different $i$, so that they are computed only once, and one 
evaluates $t_i$ on $g_{\mc{V},A_i}(f)(\wt{X}_{p_{A_i}})$ to obtain an estimate of $G_i(\alpha)$, $i \in I_n$. 
Note that this time, for each $i \in I_k$, $|A_{\kappa,i}|$ (see (\ref{Akappa})) 
is the number of all evaluations of $f_i$ made in such a computation. 
Let further $\kappa$ be unbiased for estimation of $G$ and $\Var_{\alpha}(\phi_{\kappa_i,\mc{V}})<\infty$, $i \in I_n$. Then
we can use the above estimate of $G(\alpha)$ in a single step of a MC procedure. 
Let $J \subset I_n$ be nonempty. We define subvector of $G$ consisting of its estimands with indices in $J$, 
as $G_J=(G_j)_{|j\in J}$ and an analogous subvector of $\kappa$ as $\kappa_J=(\kappa_j)_{|j \in J}$.  
Note that from (\ref{phikappamany}) and discussion below (\ref{Qmuf}), quantity 
$\Var_{\alpha}(\phi_{\kappa_J,\mc{V}}) \in \overline{\R}$ is well-defined for $|J|=1$ for all $\alpha \in \mc{V}$, 
while for $|J|>1$, for which symbol $\Var$ in this quantity is some variance for random vectors as in Section 
\ref{secOrthog}, it is well-defined only for $\alpha \in \mc{V}$ for which  $\Var_{\alpha}(\phi_{\kappa_i,\mc{V}})<\infty$, 
$i \in J$. 
We define an inefficiency constant $d_{G,J,i,\kappa}$ of $\kappa$ with respect to the
$i$th function for estimating the subvector of $G$ with indices in $J$ 
to be an $\overline{\R}$-valued function defined for each $\alpha \in \mc{V}$ for which $\Var_{\alpha}(\phi_{\kappa_J,\mc{V}})$ 
is well-defined, in which case it is given by formula
\begin{equation}\label{dIneffgen} 
d_{G,J,i,\kappa}(\alpha) = \Var_{\alpha}(\phi_{\kappa_J,\mc{V}})|A_{\kappa,i}|. 
\end{equation} 
This is an extension of the definition from Section \ref{secIneffSchemes}
which coincides with the above one for $k = 1$ and $|J|=1$. When $|J|=1$ and the index $i$ of the function
 is known from the context and omitted in the subscript, we shall use the same simplified notations 
as in Section \ref{secIneffSchemes}. 
The above defined inefficiency constants have analogous interpretation as the less general 
ones in Section \ref{secIneffSchemes}. However, using notations as in this section, 
one now needs to assume that for estimands $G$ and $G'$ it holds $(G_{J})_j(\alpha)=(G_{J'}')_{j}(\alpha'),$ $j=1,\ldots,|J|$, and that 
the ratio of positive average durations $\tau_{s}$ to $\tau'_{s}$ of single steps of sequences of MC procedures 
using $\kappa$ and $\kappa'$, computing $G(\alpha)$ and $G'(\alpha')$ fulfills
\begin{equation}\label{tauAratiogen} 
\frac{\tau_s}{\tau'_s} \approx_\delta \frac{|A_{\kappa,i}|}{|A_{\kappa',i'}|}, 
\end{equation} 
which can be the case for small $\delta$ e. g. when the most time-consuming part of both sequences of MC procedures 
are computations of only the $i$th and $i'$th functions.
Similarly as in Section \ref{secSchemesPrev} in our numerical experiments these functions will 
be constructions of outputs of MRs.
Then we receive that the ratio of inefficiency constant $c = \Var_{\alpha}(\phi_{\kappa_J,\mc{V}})\tau_s$
for estimation of $G_J(\alpha)$ (see Section \ref{secStat}) using $\kappa_J$ to an analogous constant 
for the primed procedure, fulfills
\begin{equation} 
\frac{c}{c'} \approx_{\delta} \frac{d_{G,J,i,\kappa}(\alpha)}{d_{G,J',i',\kappa'}(\alpha')}. 
\end{equation}

Similarly as for the inefficiency constants of sequences of MC procedures in Section \ref{secStat}, the ratio of positive
real values of 
inefficiency constants (\ref{dIneffgen}) of $\kappa$ and $\kappa'$ for 
estimating the subvectors of $G(\alpha)$ and $G'(\alpha')$ with indices $J$ and $J'$ as as above, 
is $\delta$-approximately equal to the ratio of variances of the appropriate final MC estimators  
for $\delta$-approximately the same number of $i$th and $i'$th functions evaluations made in the respective MC procedures 
or 
to the ratio of the 
number of these functions evaluations in the MC procedures for $\delta$-approximately equal variances of the final MC estimators, 
and it is also equal to the limit of ratios of minimum numbers of respective functions evaluations needed for the variances 
of the final MC estimators to be below $\epsilon$ for $\epsilon$ tending to zero.

\section{The possibility of a better performance of translation-invariant estimators}
In this section we provide certain criteria for verifying that some estimators of estimands on pairs which are in a sense invariant 
under translations can 
in some situations significantly outperform their certain counterparts without this property. 
Let $\mc{V}$ be some admissible pairs as in Definition \ref{defAPairs} 
such that $\mc{H}_i = \mc{S}(\R)$ for some $i \in I_k$. For $f=(f_j)_{j=1}^k \in \mc{V}_2$ and $c\in \R$, 
we denote $\tr_{i}(f,c)=(f_1,\ldots f_i+c,\ldots,f_k)$.
\begin{defin}\label{definvestimand} 
We say that an estimand $G$ on $\mc{V}$ 
is translation-invariant in the $i$th function (or simply translation-invariant if $k=1$), 
if for each $(\mu,f) \in \mc{V}$ and $c \in \R$ such that 
$(\mu,\tr_{i}(f,c)) \in \mc{V}$, it holds $G(\mu,f) = G(\mu,\tr_{i}(f,c))$. 
\end{defin} 
\begin{lemma}\label{invLem} 
For an estimand $G$ on $\mc{V}$, translation-invariant in the $i$th function, 
suppose that there exists $\alpha = (\mu,f) \in \mc{V}$ and a real sequence 
$(c_l)_{l=1}^\infty$, $\lim_{l \to \infty}|c_l| = \infty$ such that 
for each  $l \in \N_+$, $(\mu,\tr_{i}(f,c_l)) \in \mc{V}$. Suppose further that 
for some unbiased estimator $\phi$ of $G$ with indices $K$ and each $\wt{X} \sim \mu^K$, there 
exist $n \in \N_+$ and $Z_j \in L^2$, $j = 0,\ldots, n$, where $\E(Z_n^2)>0$,
such that 
for each $l \in \N_+$, 
\begin{equation} 
 R(c_l) = \phi(\tr_{i}(f,c_l))(\wt{X}) =\sum_{j=0}^n c^j_l Z_j 
\end{equation} 
a. s. Then 
\begin{equation}\label{limnVar} 
\lim_{l \to\infty} \Var_{\mu,\tr_{i}(f,c_l)}(\phi) = \infty. 
\end{equation} 
\end{lemma} 
\begin{proof} 
For certain random variables  $W_1, \ldots, W_{2n-1} \in L^2$, it holds a. s. 
\begin{equation} 
R^2(c_l) =  c^{2n}_l Z_n^2 + \sum_{j=0}^{2n -1} c^j_lW_j. 
\end{equation} 
Thus, from $\E(Z_n^2)>0$, we receive  
\begin{equation} 
\lim_{l \to \infty}\E(R(c_l)^2) = \infty 
\end{equation} 
and (\ref{limnVar}) follows from the fact that  
\begin{equation} 
\Var_{\mu,\tr_{i}(f,c_l)}(\phi) = \E(R^2(c_l)) - G^2(\alpha). 
\end{equation} 
\end{proof} 
In all situations in which we use the above lemma its assumptions are satisfied for each 
unbounded real sequence $(c_l)_{l=1}^{\infty}$, so we further only specify the required $\alpha$. 
\begin{defin}\label{definvestimator} 
A statistic for $\mc{V}$ with indices $K$ 
is translation-invariant in the $i$th function or simply translation-invariant 
if $k=1$, if for each $(\mu,f) \in \mc{V}$ and $c \in \R$ 
such that $(\mu,\tr_{i}(f,c)) \in \mc{V}$, and each $\wt{X} \sim \mu^K$, it holds 
\begin{equation}\label{thesame} 
\phi(f)(\wt{X}) = \phi(\tr_{i}(f,c))(\wt{X}). 
\end{equation} 
\end{defin} 
Note that if an unbiased estimator of an estimand $G$ is translation-invariant in the $i$th function, then $G$ must also be 
translation-invariant in this function. 
\begin{theorem}\label{thtransl} 
Let $G$ be an estimand on $\mc{V}$.
Let $\phi$ be an unbiased estimator of $G$, translation-invariant in the $i$th function,
and let the unbiased estimator $\phi'$ of $G$ 
satisfy the assumptions of Lemma \ref{invLem}. Then for each $(c_l)_{l=1}^{\infty}$ and $\alpha = (\mu,f)$  as
in this lemma for which further $\Var_{\alpha}(\phi)$ is finite, for each $B > 0$, there exists $n\in \N_+$ such that
for $\alpha = (\mu,\tr_{i}(f,c))$, 
\begin{equation}\label{isworse} 
\Var_{\alpha_n}(\phi') > \Var_{\alpha_n}(\phi) + B. 
\end{equation} 
In particular, both the difference and ratio of variances of $\phi'$ and $\phi$ can be arbitrarily large. 
\end{theorem} 
\begin{proof} 
From (\ref{thesame}),  $\Var_{\alpha_n}(\phi)=\Var_{\alpha}(\phi)$, $n \in \N_+$,
while from Lemma \ref{invLem}, as $n$ goes to infinity, the lhs of (\ref{isworse}) goes to infinity. 
\end{proof} 
Let us apply the above theory to certain estimators defined in Section \ref{secSymIneq}.
Estimator $\widehat{V}_{1,s2}^{tot}$ is translation-invariant. For estimator $\wh{V}_{1,a2}^{tot}$ let us take 
$\alpha=(\mu,f) \in D_{V_1^{tot}}$ such that for each random variables $X_i \sim \mu_i, i \in I_2,$ we have $X_i \in L^2$, $i \in I_2$, 
$\Var(X_1)>0$, $\E(X_2^2)>0$, and $f(X_1,X_2) = X_1X_2$. Then for $\wt{X} \sim \mu^{p_{a2}}$, 
the assumptions of Lemma \ref{invLem} are satisfied for $n = 1$ and 
$Z_1 = (\wt{X}_1[0]-\wt{X}_1[1])\wt{X}_2$, since  $\E(Z_1^2) = 2\Var(X_1)\E(X^2_2) > 0$. Thus 
$\widehat{V}_{1,a2}^{tot}$ can have much higher variance 
than $\widehat{V}_{1,s2}^{tot}$ in the sense of Theorem \ref{thtransl},
or equivalently  $d_{V_1^{tot},a2}$ can be much higher than $d_{V_1^{tot},s2}$ (in the above sense). 
Notice that $\wh{V}_{1,s4}$ is translation-invariant and 
$\wh{V}_{1,a3}$ satisfies the conditions of Lemma \ref{invLem} for some $(\mu,f) \in D_{V_1}$ such that for each $X_i \sim \mu_i$, 
$i \in I_2$, 
$f(X_1,X_2) = X_1$, $X_1 \in L^2$, and $\Var(X_1)> 0$, since then for $\wt{X} \sim \mu^{p_{a3}}$ 
we have in Lemma \ref{invLem}, $n = 1$ and $\E(Z^2_1) = \E((\wt{X}_1[0] - \wt{X}_1[1])^2) = 2\Var(X_1)>0$. 
Thus $d_{V_1,a3}$ can be much higher than $d_{V_1,s4}$. 

\section{\label{secAveSchemes}Averaging of estimators and schemes} 
Let $\mc{V}$ be some admissible pairs with a signature $Sg$ as in Definition \ref{defAPairs}. 
Let $\pi \in \Theta^N$ (see Section \ref{secSymIneq}). We define a function 
$\wt{\pi}$ on $I_N\times\N_+$ by formula $\wt{\pi}(i,j) = (i,\pi_i(j))$. 
Let $K$ be some arguments' indices for $N$. 
The image under $\wt{\pi}$ (see Appendix \ref{appMath}) of $K$ is 
\begin{equation} 
\wt{\pi}[K] =  \{\wt{\pi}(\beta):\beta \in  K\}. 
\end{equation}
Let $B = (B_i)_{i=1}^N$ be a sequence of nonempty sets and 
the function $\sigma_{B,K,\pi}:B^{K}\rightarrow B^{\wt{\pi}[K]}$ be such that for each $x \in B^K$ and $\beta \in K$,  
\begin{equation}\label{sigmadef} 
(\sigma_{B,K,\pi}(x))_{\wt{\pi}(\beta)} = x_{\beta}. 
\end{equation} 
Note that for each $\mu \in \mc{V}_1$ and $X \sim \mu^K$, we have 
\begin{equation}\label{sigmapix} 
\sigma_{B_\mu,K,\pi}(X) \sim \mu^{\wt{\pi}[K]}. 
\end{equation} 
For a statistic $\phi$ for $\mc{V}$ with indices $K$, a permutation of $\phi$ given by 
$\pi$, denoted as $\A_\pi(\phi)$, is defined as a statistic for $\mc{V}$ with 
indices $\wt{\pi}[K]$ such that for each $(\mu,f) \in \mc{V}$ and $x \in B_\mu^{\wt{\pi}[K]}$, 
\begin{equation}\label{apidef} 
\A_\pi(\phi)(f)(x) = \phi(f)(\sigma_{B_\mu,K,\pi}^{-1}(x)). 
\end{equation} 
From (\ref{sigmapix}) and (\ref{apidef}) it follows that for each $(\mu,f) \in \mc{V}$, $X \sim \mu^K$, and $Y \sim \mu^{\wt{\pi}[K]}$, 
\begin{equation}\label{sameDistr} 
\A_\pi(\phi)(f)(Y) \sim \phi(f)(X). 
\end{equation} 
For a function $h:\R^n\rightarrow\R$, like e. g. summation $h(x)= \sum_{i=1}^nx_i$, 
and real-valued statistics $\phi_i$ for $\mc{V}$ with indices $K_i$, $i \in I_n$,   
we define $h((\phi_i)_{i=1}^n)$ to be a real-valued statistic for $\mc{V}$ with indices $K = \bigcup_{i=1}^n K_i$ 
such that for each $(\mu,f) \in \mc{V}$ and $x \in B_\mu^K$, 
\begin{equation} \label{hstat} 
h((\phi_i)_{i=1}^n)(f)(x)  = h((\phi_i(f)(x_{K_i}))_{i=1}^n) 
\end{equation} 
(see (\ref{xL})).
Let $K$ be some arguments' indices for $N$ and $\Pi$ be a nonempty finite subset of $\Theta^N$. We define
\begin{equation}
\wt{\Pi}[K] = \bigcup_{\pi \in \Pi}\wt{\pi}[K]. 
\end{equation}
Let $\phi$ be an $\R^n$-valued statistic for $\mc{V}$ with indices $K$. 
We define an average of $\phi$ given by $\Pi$ as the following statistic for $\mc{V}$ with indices $\wt{\Pi}[K]$, 
\begin{equation}\label{defAvestat} 
\A_{\Pi}(\phi) = \frac{1}{|\Pi|}\left(\sum_{\pi \in \Pi} \A_\pi(\phi)\right). 
\end{equation} 
From (\ref{sameDistr}) it follows that for each $(\mu,f) \in \mc{V}$, $X \sim \mu^{\wt{\Pi}[K]}$, and 
$Y \sim \mu^{K}$, $\A_{\Pi}(\phi)(f)(X)$ is an average of $|\Pi|$ random variables with the same distribution as $\phi(f)(Y)$. 
In particular, if $\phi$ is an estimator of some estimand $G$ on $\mc{V}$, then so is $\A_{\Pi}(\phi)$. 
For $p > 0$, we write $\phi \in L^p(\mc{V})$ if $\phi(f) \in L^p(\mu^K)$ for each 
$(\mu,f) \in \mc{V}$. From Lemma \ref{lemVarAve} in Appendix \ref{appStatMC} it follows that for each $\phi \in L^1(\mc{V})$, 
$\A_{\Pi}(\phi)$ has uniformly not higher variance than $\phi$, that is for each $\alpha \in \mc{V}$, 
\begin{equation}\label{uniformalpha} 
\Var_{\alpha}(\A_{\Pi}(\phi)) \leq \Var_{\alpha}(\phi). 
\end{equation} 
Let $\pi \in \Theta^N$. For each nonempty $I \subset I_N$, 
we identify each sequence $v = (v_i)_{i \in I} \in \N_+^{I}$ with the set $\{(i,v_i):i \in I\} \subset I_N \times \N_+$,  
so that $\wt{\pi}[v] = (\pi_i(v_i))_{i \in I}$. In particular for 
$v \in \N_{+}^N$ we receive $\wt{\pi}[v] = \wh{\pi}(v)$ (see Section \ref{secSymIneq}). 
Let $W \subset \mc{P}(I_N\times\N_+)$ (see Appendix \ref{appMath}). 
For $\wt{\pi}^{\rightarrow}$ denoting the image function of $\wt{\pi}$ (see Appendix \ref{appMath}), we have 
\begin{equation} 
\wt{\pi}^{\rightarrow}[W] = \{\wt{\pi}[v]:v \in W\}. 
\end{equation} 
For $\Pi$ as above we define 
\begin{equation} 
\wt{\Pi}^{\rightarrow}[W] = \bigcup_{\pi \in \Pi} \wt{\pi}^{\rightarrow}[W]. 
\end{equation} 
In particular for $W \subset \N_{+}^N$ we receive $\wt{\Pi}^{\rightarrow}[W] = \wh{\Pi}[W]$ (see Section \ref{secSymIneq}).
Let $A=(A_i)_{i=1}^k$ be some sets of evaluation vectors for $Sg$. 
We define 
\begin{equation} 
\wt{\Pi}^{\rightarrow}[A] = (\wt{\Pi}^{\rightarrow}[A_i])_{i=1}^k. 
\end{equation} 
Let $\pi \in \Theta^N$. We denote $\wt{\{\pi\}}^{\rightarrow}[A]$ simply as $\wt{\pi}^{\rightarrow}[A]$. 
For convenience we shall write $\delta$ and $\gamma$ instead of $\delta(A)$ and $\gamma_A$ defined in the previous section.
We define function $\rho_{C,A,\pi}:C_A \rightarrow C_{A}$ to be such that for each 
\begin{equation}\label{z1example} 
z = ((y_{i,v})_{|v \in (\wt{\pi}^{\rightarrow}[A])_{\gamma(i)}})_{i=1}^\delta \in C_{\wt{\pi}^{\rightarrow}[A]} =C_{A}, 
\end{equation} 
it holds
\begin{equation}\label{defRho} 
\rho_{C,A,\pi}(z) = ((y_{i,\wt{\pi}[v]})_{|v \in A_{\gamma(i)}})_{i=1}^\delta. 
\end{equation} 
Let $t: C_{A} \rightarrow \R$. We define function $\ave_{C,A,\Pi}(t): C_{A} \rightarrow \R$,
called permutation of $t$ given by $\pi$ and $A$, to be such that for each $z$ as in (\ref{z1example}),
\begin{equation}
\ave_{C,A,\pi}(t)(z) = t(\rho_{C,A,\pi}(z)). 
\end{equation} 
Let further $\eta_{C,A,\Pi,\pi}:C_{\wt{\Pi}^{\rightarrow}[A]} \rightarrow C_A$ be such that for each 
\begin{equation}\label{zexample}
z = ((y_{i,v})_{|v \in (\wt{\Pi}^{\rightarrow}[A])_{\gamma(i)}})_{i=1}^\delta \in C_{\wt{\Pi}^{\rightarrow}[A]}, 
\end{equation} 
it holds
\begin{equation} 
\eta_{C,A,\Pi,\pi}(z) = ((y_{i,v})_{|v \in (\wt{\pi}^{\rightarrow}[A])_{\gamma(i)}})_{i=1}^\delta. 
\end{equation} 
We define $\ave_{C,A,\Pi}(t):C_{\wt{\Pi}^{\rightarrow}[A]} \rightarrow \R$, called average of $t$ given by $\Pi$ and $A$, 
to be such that for each $z$ as in (\ref{zexample}),
\begin{equation}\label{tcap}
\ave_{C,A,\Pi}(t)(z) =  \frac{1}{|\Pi|} \sum_{\pi \in \Pi} \ave_{C,A,\pi}(t)(\eta_{C,A,\Pi,\pi}(z)). 
\end{equation}
For the special case of $k=1$, $\mc{H} = \mc{S}(\R),$ and $\Pi$ being a subgroup of $\Theta^N$,
$\ave_{C,A,\Pi}(t)$ is equal to $\ave_{A,\Pi}(t)$ given by formula (\ref{symOp}) from Section \ref{secSymIneq}.
Let $\kappa = (t,A)$ be a scheme for $Sg$. Its average given by $\Pi$ is defined as a scheme 
\begin{equation} 
\ave_{\Pi}(\kappa) = (\ave_{C,A,\Pi}(t),\wt{\Pi}^{\rightarrow}[A]). 
\end{equation} 
This coincides with definition (\ref{avePiKO}) from Section \ref{secSymIneq}
for the same special case as discussed below (\ref{tcap}). 
When $\Pi = \{\pi\}$, $\ave_{\Pi}(\kappa)$ is denoted as $\ave_{\pi}(\kappa)$ 
and called permutation of $\kappa$ given by $\pi$. 
We have a following theorem, which we prove in Appendix \ref{appcoeffsLem}. 
\begin{theorem}\label{thEquAves} 
Under the preceding assumptions, 
\begin{equation}\label{avephikappa} 
\phi_{\ave_{\Pi}(\kappa),\mc{V}} = \A_{\Pi}(\phi_{\kappa,\mc{V}}). 
\end{equation} 
\end{theorem}
For an $n$-dimensional scheme $\kappa = (\kappa_i)_{i=1}^n$ for $Sg$, 
we define its average as $\ave_{\Pi}(\kappa) = (\ave_{\Pi}(\kappa_i))_{i=1}^n$. 
If $\Pi$ is a subgroup of $\Theta^N$, then an average of a scheme or a statistic given by $\Pi$ 
is called their symmetrisation. 
From Theorem \ref{thEquAves} and a similar fact concerning estimators stated above, it 
follows that an average of an unbiased estimation scheme for some estimand $G$ remains an unbiased scheme for its estimation and its 
estimator has uniformly not higher variance. 
Let us consider an $n$-dimensional scheme $\kappa = (t_i,A_i)_{i=1}^n$ for $\mc{V}$, and $m \in \N_+$. 
We define an $m$-step MC scheme $\kappa(m)$ using scheme $\kappa=(\kappa_i)_{i=1}^n$  
to be an average of $\kappa$ given by any  $\Pi \subset \Theta^N$, $|\Pi| = m$, 
such that for each $\pi_1, \pi_2 \in \Pi, \pi_1 \neq \pi_2,$ schemes $\A_{\pi_i}(\kappa), i \in I_2$, have disjoint arguments' indices, that is 
$\wt{\pi_1}(p_{\kappa}) \cap \wt{\pi_2}({p_\kappa}) = \emptyset$. Note that $|A_{\kappa(m),i}| = m |A_{\kappa,i}|, i \in I_k$, 
and for each $\alpha = (\mu,f) \in \mc{V}$, $i \in I_n$, 
$X \sim \mu^{p_{\kappa_i}}$, and $Y \sim \mu^{p_{\kappa(m)_i}}$, $\phi_{\kappa(m)_i,\mc{V}}(f)(Y)$ 
is an average of $m$ independent random variables with the same distribution 
as $\phi_{\kappa_i,\mc{V}}(f)(X)$. Let us further assume that $\Var_\alpha(\phi_{\kappa_i,\mc{V}}) < \infty$, $i \in I_n$, 
so that $\Var_\alpha(\phi_{\kappa(m)_i,\mc{V}}) = \frac{\Var_\alpha(\phi_{\kappa_i,\mc{V}})}{m},$ $i \in I_n$. 
If $\kappa$ is further an 
unbiased estimation scheme for estimands $G=(G_i)_{i=1}^n$ with common admissible pairs $\mc{V}$, 
then $\phi_{\kappa_i,\mc{V}}(f)$ and $\phi_{\kappa(m)_i,\mc{V}}(f)$ 
can be identified with the single-step and final MC estimators of $G_i(\alpha), i \in I_n,$ respectively, and 
we have equality of inefficiency constants of the schemes 
\begin{equation}\label{equdIneffMC} 
d_{G,i,j,\kappa} = d_{G,i,j,\kappa(m)},\ i \in I_n,\ j \in I_k. 
\end{equation}

\section{\label{secGenIneq}Some general inequalities between variances of estimators and inefficiency constants of schemes}
We will now prove some general inequalities between variances of estimators of 
estimands on pairs and inefficiency constants of schemes, the latter including as a special 
case the inequality from Theorem \ref{thineqOld},
but first we need some helper facts and definitions. 
\begin{lemma}\label{lemcomb}
For $N \in \N_+$, let $X= (X_i)_{i=1}^N$ be a random vector with independent coordinates, and let us consider independent random variables 
$Y_{i,j} \sim X_i, i \in I_N, j \in \N_+$. For $v \in \N_+^N$, let $Y_v=(Y_{i,v_i})_{i=1}^N$. For some measurable function 
$f$ such that $Z = f(X) \in L^2$ and a finite nonempty set $A \subset \N_+^N$, let 
\begin{equation}
\overline{Z}= \frac{1}{|A|}\sum_{v \in A}f(Y_v).
\end{equation}
Then it holds 
\begin{equation}\label{ineqaz}
\frac{1}{|A|}\Var(Z) \leq \Var(\overline{Z}) \leq \Var(Z).
\end{equation}
\end{lemma}
\begin{proof}
We have 
\begin{equation}
\Var(\overline{Z}) =  \left(\frac{1}{|A|^2}\sum_{v,w \in A}\Cov(f(Y_v),f(Y_w)\right).
\end{equation}
For $v,w \in \N_+^N$, let $c(v,w) = \{i \in I_N:v_i = w_i\}$, and $Y_{v,w} = (Y_{i,v_i})_{i \in c(v,w)}$. Then from Theorem 
\ref{thCond}, for each $v,w \in A$, $\Cov(f(Y_v),f(Y_w)) = \Var(\E(f(Y_v)|Y_{v,w}))$. Inequalities (\ref{ineqaz}) follow from the fact 
that $\Var(\E(f(Y_v)|Y_{v,w}))$ is nonnegative and from (\ref{aveVarError}) it is not higher than $\Var(Z)$ and equal to it for $v=w$.
\end{proof}

Let $\mc{V}$ be some admissible pairs as in Definition \ref{defAPairs}. 
For each arguments' indices $K$ for $N$, we define $n_K$ to be a vector from $\N^N$ whose $i$th coordinate is 
\begin{equation} 
n_{K,i} = \max(\{j: (i,j) \in K\} \cup \{0\}). 
\end{equation} 
For sets of evaluation vectors $A$ or a scheme $\kappa$ for $Sg$, we define $n_{A} = n_{p_A}$ and $n_{\kappa} = n_{p_\kappa}$, 
and for a statistic $\phi$ for $\mc{V}$ with indices $K$, $n_\phi = n_K$. 
Let $I$ be a nonempty subset of $I_N$ and $m\in \N_+$ 
For symmetrisations given by $\Theta_{N,I,m}$ (see (\ref{thetanim})) e. g. of some scheme for $Sg$ or a statistic for $\mc{V}$, 
we use the same nomenclature as for the less general schemes in Section \ref{secSymIneq}.
For some finite subgroup $\Pi \subset \Theta^N$, we say that a scheme for $Sg$ or a statistic for $\mc{V}$ 
is $\Pi$-symmetric if it is equal 
to its symmetrisation given by $\Pi$. 
Suppose that $\psi$ is a statistic for $\mc{V}$ or 
a scheme for $Sg$  such that 
$n_{\psi,i}=n$, $i \in I$, and $\psi$ is $\Theta_{N,I,n}$-symmetric. 
Then for each $n'\in \N_+$,  $n'\geq n$, symmetrisation of $\psi$ given by $\Theta_{N,I,n'}$ is called 
its symmetrisation from $n$ to $n'$ dimensions (simply symmetrisation if $n=n'$) 
in the argument given by $I$ (or in the $i$th argument if $I= \{i\}$). 
For some arguments' indices $K$ for $N$, sequence of sets $B= (B_i)_{i=1}^n$, 
and $x \in B^K$, for 
$L \subset I_N$ and $K_L = \{(i,j)\in K:i \in L\}$, we denote $x_L = x_{K_L}$ (see (\ref{xL})),
while for $L \subset \N_+$ and $K:L = \{(i,j)\in K:j \in L\}$, 
we denote $x:L = x_{K:L}$. For arguments' indices $K_i$ for $N$, $i \in I_2$, 
$K = K_1\cup K_2$, and $K_1\cap K_2 = \emptyset$, in the proof of the below theorem we identify $x_K$ with $(x_{K_1},x_{K_2})$. 
\begin{theorem}\label{thIneqEsts} 
Let $\phi'$ be a symmetrisation of a statistic $\phi \in L^1(\mc{V})$ from $n$ to $n'$ dimensions in the argument given by some $I$
as above.
Then for each $\alpha \in \mc{V}$ such that $\Var_{\alpha}(\phi)< \infty$, 
\begin{equation}\label{VarN1N2} 
\frac{1}{{n' \choose n}} \Var_{\alpha}(\phi) \leq \Var_{\alpha}(\phi') \leq \Var_{\alpha}(\phi). 
\end{equation} 
\end{theorem} 
\begin{proof} 
Since $\phi$ is $\Theta_{N,I,n}$-symmetric statistic, 
for each $\theta_1, \theta_2 \in \Theta_{n'}$ such that $\theta_1[I_n] = \theta_2[I_n]$, or equivalently $\theta_2^{-1}\theta_1[I_n]= I_n$, 
it holds 
$\A_{\pi_{N,I,\theta_2^{-1}\theta_1}}(\phi)=\phi$, and thus  $\A_{\pi_{N,I,\theta_2^{-1}}}\A_{\pi_{N,I,\theta_1}}(\phi)=\phi$ and 
$\A_{\pi_{N,I,\theta_1}}(\phi)=\A_{\pi_{N,I,\theta_2}}(\phi)$.  We denote $\sim I= I_N \setminus I$. 
Let $\alpha = (\mu,f) \in \mc{V}$, $K' = \wt{\Theta}_{N,I,n'}[K]$ be the arguments 
indices of $\phi'$, $X \sim \mu_{K'}$, 
$U = X_{\sim I}$, and $V = X_I$, so that $X= (U,V)$. 
We denote $\wt{\phi}'= \phi'(f)$. 
Let further $\mc{W} = \{L \subset I_{n'}, |L|= n\}$. For each $L \in \mc{W}$, let us choose certain 
$\theta_L \in \Theta_{n'}$ such that $\theta_L[I_n] = L$, and denote 
$\wt{\phi}_L= \A_{\pi_{N,I,\theta_L}}(\phi)(f)$. From (\ref{defAvestat}) and the above remarks we have 
\begin{equation}\label{phiprimform} 
\wt{\phi}'(X) = \frac{1}{{n' \choose n}} \sum_{L \in \mc{W}}\wt{\phi}_L(U, V:L). 
\end{equation} 
Thus (\ref{VarN1N2}) follows from Lemma \ref{lemcomb}. 
\end{proof}
\begin{theorem}\label{thIneqds}
Let $G=(G_i)_{i=1}^m$ be estimands with common admissible pairs $\mc{V}$. 
If scheme $\kappa'$ is created from an unbiased estimation scheme $\kappa = (\kappa_i)_{i=1}^m$ for $G$ 
by its symmetrisation from $n$ to $n'$ dimensions in the argument corresponding to some $I$ 
then for each $j \in I_m$, 
$i \in I_k$ such that $J_i \cap I\neq \emptyset$, and $\alpha \in D_{G_j}$ for which $d_{G,j,i,\kappa}(\alpha)< \infty$, it holds 
\begin{equation}\label{d1d2n}
\frac{\frac{n'}{n}}{{n' \choose n}} d_{G,j,i,\kappa}(\alpha) \leq  d_{G,j,i,\kappa'}(\alpha) \leq \frac{n'}{n} d_{G,j,i,\kappa}(\alpha). 
\end{equation} 
\end{theorem}
\begin{proof}
Since $\kappa$ is $\Theta_{N,I,n}$-symmetric, it holds for each $l \in I\cap J_i$,
\begin{equation}\label{A1A2n}
\frac{|A_{\kappa,i}|}{n} = |\{v \in A_{\kappa,i}: v_l = 1\}| = \frac{|A_{\kappa',i}|}{n'}. 
\end{equation}
Now (\ref{d1d2n}) follows from (\ref{A1A2n}), the fact that for each $j \in I_m$, $\phi_{\kappa_j',\mc{V}}$ is a symmetrisation of 
$\phi_{\kappa_j,\mc{V}}$ from $n$ to $n'$ dimensions, Theorem \ref{thIneqEsts}, and formula
(\ref{dIneffgen}) defining an inefficiency constant. 
\end{proof}
Taking $k=m=n=1$ and $n'=2$, we receive the thesis of Theorem \ref{thineqOld}.

\section{\label{secPolynEst}Schemes for the sensitivity indices of functions of conditional moments}
For some $n \in \N_+$, let us consider a function $f$ and independent random variables $P=(P_i)_{i=1}^{N_P}$ and $R$ as 
at the beginning of Section \ref{secMany} but with 
$f(P,R) \in L^n$.  
Let $Q$ restricted to $\mc{T}_n$ be a function of 
the first $n$ so restricted moments as in Section \ref{secCondMoms}, like $Q = \Var$ for $n=2$. 
Suppose that there exists an unbiased estimator $\phi_Q$ of $G_Q$ in 
$m \in \N_+$ dimensions (see Appendix \ref{appStatMC}). 
For $\wt{R} \sim \mu_R^m$ and independent of $P$, let 
\begin{equation}\label{hQ}
h_Q(f)(P,\wt{R}) = \phi_Q((f(P,\wt{R}_i))_{i=1}^m)  
\end{equation}
and let us assume that $h_Q(f)(P,\wt{R}) \in L^1$.
Then it holds 
\begin{equation}\label{phiqcq}
\begin{split}
Q(f(P,R)|P) &= (Q(f(p,R)))_{p=P}\\
&=(\E(\phi_Q((f(p,\wt{R}_i))_{i=1}^{m})))_{p=P}\\
&= \E(h_Q(f)(P,\wt{R})|P),
\end{split}
\end{equation}
where in the first equality we used expression (\ref{CQKP}), in the second the fact that $\phi_Q$
is an unbiased estimator of $G_Q$
and that from Fubini's theorem \cite{rudin1970} 
$f(p,\wt{R}_i) \in L^n$ for $\mu_P$ a. e. $p$, and in the last Theorem \ref{indepCond} and (\ref{hQ}).
In particular, expected values and variance-based sensitivity 
indices of $Q(f(P,R)|P)$ and $\E(h_Q(f)(P,\wt{R})|P)$ coincide 
(whenever both are well-defined). Since the latter is a conditional expectation of the function $h_Q(f)$ of independent 
random variables $P$, $\wt{R}$ given the first variable, its sensitivity indices can be estimated with the help of 
estimators from Section \ref{secMany}, e. g. in a way we describe below. 
Let $r_Q$ be the degree of $G_Q$ and as $\phi_Q$ let us take the unique
symmetric estimator of $G_Q$ in $r_Q$ dimensions (see Section \ref{secStat}). 
For instance for $Q = \Var$, we have  $r_{Var}=2$ and 
\begin{equation}
\phi_{Var}(x_1,x_{2}) = \frac{1}{2}(x_1 - x_2)^2,
\end{equation}
so that 
\begin{equation}\label{hvar}
h_{Var}(f)(P,\wt{R}) = \frac{1}{2}(f(P,\wt{R}_1) -f(P,\wt{R}_2))^2.
\end{equation}
Let us now reinterpret
different quantities from the end of Section \ref{secCondMoms} 
like $AveQ$, $VQ_k$, or $VQ_k^{tot}, k \in I_{N_P}$, as estimands 
on admissible pairs $\alpha_{\mu_P,\mu_R,f}$ defined analogously as
 in Section \ref{secMany}, but for $f(P,R) \in L^n$ and   
$h_Q(f)(P,\wt{R})\in L^p$, 
where $p =1$ for $AveQ$ and $p=2$ for other estimands (this condition will be needed for our estimators to be integrable).
The values of such estimands on such $\alpha_{\mu_P,\mu_R,f}$ 
are defined identically as in Section \ref{secCondMoms} treating $Z=f(P,R)$ as output an MR. 
 For $l \in \N_+$, we call a pair $\pi = (J_1, J_2)$ equal partition of the set $I_{2l}$, if 
 for $i \in I_2$, we have $J_i \subset I_{2l}$,  $|J_1| = |J_2| = l$, 
$J_1 \cap J_2 = \emptyset$, and $1 \in J_1$ (note that $J_1\cup J_2 = I_{2l}$). 
 Let $\Psi_Q$ be the set of all equal partitions of 
 $I_{2r_Q}$. We have  $|\Psi_Q| = \frac{{2r_Q \choose r_Q}}{2}$.  
 Consider $\wt{P}$ corresponding to $P$ as in Section \ref{secMany}, 
 and let $\wt{R} \sim \mu_R^{2r_Q}$ be independent of $\wt{P}$. 
 For a partition $\psi = (\psi_1, \psi_2) \in \Psi_Q$, 
 we denote $\wt{R}_{\psi} = (\wt{R}_{\psi_1}, \wt{R}_{\psi_2})$, where $\wt{R}_{\psi_i} = (\wt{R}_{j})_{j \in \psi_i}, i \in I_2$. 
We shall now define a scheme $SQ$ whose subschemes yield estimators $\wh{\lambda Q}_{SQ}$ for different estimands $\lambda Q$ for $Q$, 
like $AveQ$, $VQ_k$, and $VQ^{tot}_k$, corresponding to such estimands $\lambda E$ for $E$. These estimators 
 evaluated on each appropriate function $f$ and random vector 
$(\wt{P}, \wt{R})$ as above are equal to the average over $\psi \in \Psi_Q$ 
of the corresponding estimators $\wh{\lambda E}_{SE}$ from Section \ref{secMany} 
evaluated on the function $h_Q(f)$ and random vector $(\wt{P}, \wt{R}_{\psi})$, that is
 \begin{equation}\label{lQ} 
 \widehat{\lambda Q}_{SQ}(f)(\wt{P},\wt{R}) = \frac{1}{|\Psi_Q|} \sum_{\psi \in \Psi_Q}\widehat{\lambda E}_{SE}(h_Q(f))(\wt{P}, \wt{R}_{\psi}). 
 \end{equation}
For instance for the main sensitivity index and $Q = \Var$ we have  
 \begin{equation} 
 \widehat{VVar}_{k,SVar}(f)(\wt{P},\wt{R}) = \frac{1}{3} \sum_{\psi \in \Psi_{Var}}\widehat{V}_{k,SE}(h_{Var}(f))(\wt{P}, \wt{R}_{\psi}). 
 \end{equation} 
Formulas like (\ref{lQ}) for different estimands $\lambda Q$ for some $Q$ can be easily expanded in terms of  
evaluation operators $s[i][j]$ and $s_k[i][j], i \in I_2, j \in I_{2r_Q}$, from Section \ref{secMany}, in which form they define  
the sought scheme $SQ$ in the sense discussed at the end of Section \ref{secSchemesPrev}. 
From Schwartz inequality, it is sufficient that $h_Q(f)(P,\wt{R}) \in L^4$ for the estimators of subschemes of $SQ$ to have finite variance.  
In particular, from (\ref{hvar}), for scheme $SVar$ it is sufficient that $f(P,R) \in L^8$. 
Such defined scheme $SQ$ uses 
together $4r_Q(N_P +1)$ evaluation vectors for $N_P > 2$. For scheme $SVar$ this is $8(N_P + 1)$, that is two times more than for 
scheme $SE$ for the same $N_P$. 
By analogy to discussion in Appendix C of \cite{badowski2011} for scheme $SE$, 
for $N_P=3$ one can construct schemes with lower inefficiency constants  
for estimation of sensitivity indices of $Q(f(P,R)|P)$ than for the subschemes of $SQ$. 
For some $Q$ and $Q'$ as above, such that $r = r_{Q} \leq r' = r_{Q'}$, an 
unbiased $n$-dimensional estimation scheme $SQ = (SQ_i)_{i=1}^n$ and an $n'$-dimensional one $SQ'= (SQ'_i)_{i=1}^{n'}$ for  
sequences of estimands 
$G=(G_i)_{i=1}^n$ and $G'=(G_i)_{i=1}^{n'}$, respectively, one can add 
symmetrisation of certain subscheme $SQ_i$ of $SQ$ from 
 $2r$ to $2r'$ dimensions in the argument given by $R$ as the $n+1$st subscheme to $SQ'$ and $G_i$ as such $n+1$st estimand to $G'$.  
We then have the following inequality of inefficiency constants of schemes in the sense of Theorem \ref{thIneqds}, 
\begin{equation}\label{CQd}
\frac{\frac{r'}{r}}{{2r' \choose 2r}}d_{G,i,SQ} \leq  d_{G',n+1,SQ'} \leq \frac{r'}{r} d_{G,i,SQ}
\end{equation}
and analogously for the inefficiency constants of the subschemes 
due to proportionality of the number of evaluation vectors used by the subschemes and the whole schemes. 
Let us add in this way to scheme $SVar$ all subschemes of $SE$, like ones for  
estimation of $Ave$, $AveVar$, as well as $V_k$ and $\wt{V}_k^{tot}$ for 
$k \in I_{N_P}$, symmetrised from two to four dimensions in the argument given by $R$. 
Then for each of such 
estimands $\lambda$, we have from (\ref{CQd}) for $r = 1$ and $r' =  2$, 
\begin{equation}\label{ineqDVar}
\frac{1}{3}d_{\lambda,SE} \leq d_{\lambda,SVar} \leq 2 d_{\lambda,SE}.
\end{equation} 
Considering in addition to relations (\ref{ineqDVar}) also inequalities (\ref{EMComp}) and (\ref{VitotComp}), we receive for 
$N_P >2$ and $\lambda$ equal to $V_k$ (and arguments for which these relations were proved), 
\begin{equation}\label{compVkVar}
\frac{1}{3}d_{V_k,EM} \leq d_{V_k,SVar}  \leq 4 d_{V_k,EM},
\end{equation} 
while for $\lambda=\wt{V}_k^{tot}$ we obtain  
\begin{equation}\label{compVkTotVar}
\frac{1}{3}d_{\wt{V}_k^{tot},ET} \leq d_{\wt{V}_k^{tot},SVar}  \leq 4 d_{\wt{V}_k^{tot},ET}.
\end{equation}  
We shall compute numerical estimates of Sobol's sensitivity indices $SQ_{k}$ and $SQ_{k}^{tot}$ for $Q$ equal to $\Var$ and $\E$, defined in 
Section \ref{secVBSA}, 
using scheme 
$SVar$ 
 by inserting the final MC estimates obtained using the above defined subschemes for estimands like 
$VQ_k$, $VQ_k^{tot}$, and $VQ_P$ 
instead of exact values into appropriate definitions. 

\section{\label{secApprCoeff}Schemes for products, covariances, and orthogonal projection coefficients}
For some $n \in \N_+$, 
let $\psi=(\psi)_{i=1}^{n+1}$ be functions such that $\psi_i: I_2 \rightarrow\N_+,\psi_i(1)=i$, $\psi_i(2)=n+1$, $i \in I_n$, 
and $\psi_{n+1}:\{1\}\rightarrow\N_+:\psi_{n+1}(1)=n+1$. 
For the estimand $PR$ defined in Section \ref{secUnbiased} and $Ave = \E_{|\mc{T}_1}$ (that is estimand $Ave$ from Section 
\ref{secSchemesPrev} for $N=1$), 
estimands $PRA = (PRA)_{i=1}^{n+1}$ are defined as trivial extensions (see Section \ref{secUnbiased}) 
of $n+1$ estimands $(PR,\ldots,PR, Ave)$  ($PR$ appearing $n$-times at the beginning of this sequence) 
using $\psi$. Informally, this means that $PRA$ is 
equal to $PR^n$ from Section \ref{secUnbiased} extended by adding to it average of the $n+1$st function as the last estimand. 
Let $\wt{Cov}$ be an estimand on admissible pairs $\mc{V}$ of single distributions and two functions 
consisting of all possible $\alpha= (\mu, (f_1, f_2))$ such that $f_1$, $f_2$, $f_1f_2 \in L^1(\mu)$, in which case for any $X \sim \mu$,
$\wt{Cov}(\alpha) = \Cov(f_1(X),f_2(X))$. We define $CovA=(CovA_i)_{i=1}^{n+1}$ as trivial extensions of $n+1$ estimands 
$(\wt{Cov}, \ldots, \wt{Cov},Ave)$ using $\psi$. 
We define estimands $b=(b_i)_{i=1}^{n+1}$ 
to be equal to $CovA$ or equivalently $PRA$ with each coordinate restricted to admissible pairs of  
single distributions and $n+1$ functions consisting of $\alpha = (\mu,(f_i)_{i=1}^{n+1}),$ 
such that for $X \sim \mu$,  $\{f_i(X)\}_{i=1}^n$ is nonzero orthogonal in $L^2$, $\E(f_i(X))=0$, $i \in I_{n}$, and $f_{n+1}(X) \in L^2$. 
From discussion in Section \ref{secAppr}, $b_i(\alpha)$ is the coefficient of 
$f_i'(X) = \frac{f_i(X)}{\Var(f_i(X))}$, $i \in I_n$, and 
$b_{n+1}(\alpha)$ of $\I$ (see Appendix \ref{appMath}), 
in the orthogonal projection of $f(X)$ onto span($\{f_i'(X)\}_{i=1}^n\cup\{\I\}$). 
We define $c$ to be a restriction of $b$ to admissible pairs $\alpha$ as above, except that for each above $X$ the set 
$\{f_i(X)\}_{i=1}^n$ is orthonormal in $L^2$. 
Note that each unbiased estimation scheme for $PRA$ or for $CovA$ is also an unbiased estimation scheme for $b$ and $c$. 
Let us consider an unbiased estimation scheme $SAve$ for $Ave$ given by the formula for estimator 
\begin{equation}\label{save} 
\wh{Ave}_{SAve} = g[0] 
\end{equation}
and the following formula for estimator giving an $N$-step MC scheme $SAve(N)$ using $SAve$ 
\begin{equation}\label{saven}
\wh{Ave}_{SAve(N)} = \frac{1}{N}\sum_{i=0}^{N-1}g[i].
\end{equation}
For $N \in \N, N>1,$ let us consider scheme $SCov(N)$ given by the following formula for estimator of $\wt{Cov}$
\begin{equation}\label{scov}
\wh{\wt{Cov}}_{SCov(N)} = \frac{1}{N-1}\sum_{i=0}^{N-1}g_{1}[i]g_2[i] - \frac{1}{N(N-1)} (\sum_{i=0}^{N-1} g_{1}[i]\sum_{i=0}^{N-1}g_2[i]). 
\end{equation}
We define an unbiased estimation scheme $P1$ for $PRA$ as trivial extensions of $n+1$ schemes $(SPR,\ldots,SPR,SAve)$ 
using $\psi$, and an unbiased estimation scheme $C1(N)$ for $CovA$ as trivial extensions of 
$n+1$ schemes $(SCov(N),\ldots, SCov(N), SAve(N))$ also using $\psi$. 
Each estimand $CovA_k$ and estimator $\wh{CovA}_{k,C1(N)}, k \in I_{n},$ is translation-invariant in all functions 
(see Definitions \ref{definvestimand} and \ref{definvestimator}), while for each $k \in I_{n},$ estimator 
$\wh{c}_{k,P1}$ satisfies the assumptions of Lemma \ref{invLem} (and so does $\wh{b}_{k,P1}$) in the $n+1$st 
function for $n=1$ in this lemma and each $(\mu,f) \in D_{CovA}$, since for $X \sim \mu$, we have $Z_1 = f_k(X)$ and $\E(Z_1^2) = 1$. 
Thus $P1$ can have much higher inefficiency constant than $C1(N)$ for 
estimation of $c_k$ and $b_k,$ $k \in I_{n}$, in the sense of Theorem \ref{thtransl}. 
However, as we shall now show, for each $k \in I_{n}$ and $N \in \N_+, N > 2$, there exists $\alpha \in \mc{V}$ such that 
\begin{equation}\label{dckpless} 
d_{c_k,P1}(\alpha) < d_{c_k,C1(N)}(\alpha). 
\end{equation} 
For a MC scheme $P1(N)$ using scheme $P1$ in $N$ steps, we have from (\ref{equdIneffMC}) that (\ref{dckpless}) is equivalent to 
$d_{c_k,P1(N)}(\alpha) < d_{c_k,C1(N)}(\alpha)$, and since both schemes use the same number of evaluation vectors for the $n+1$st function 
and both are unbiased, this is equivalent to 
\begin{equation}\label{MCleqC1N}
\E_{\alpha}(\wh{c}^2_{k,P1(N)}) < \E_{\alpha}(\wh{c}^2_{k,C1(N)}).
\end{equation}
We will need the following lemma which we prove in Appendix \ref{appcoeffsLem}. 
\begin{lemma}\label{lemPrBetterCov} 
For a random variable  $X \in L^4$, $\E(X) = 0$, $0<\E(X^2)$, 
let $Y \sim \mu_X^N$. Let us denote for $l \in I_2$,
\begin{equation} 
\overline{Y^l} = \frac{1}{N}\left(\sum_{i=1}^NY_i^l\right). 
\end{equation} 
Then it holds 
\begin{equation} 
\E(\overline{Y^2}^2) < \left(\frac{N}{N-1}\right)^2\E((\overline{Y^2} - \overline{Y}^2)^2). 
\end{equation} 
\end{lemma} 
Thus for (\ref{MCleqC1N}) to hold it is sufficient to take $\alpha= (\mu,(f_i)_{i=1}^{n+1}) \in D_c$ such that 
for $X$ as in the above lemma, for which further $\E(X^2)=1$ (e. g. $\PR(X=1)=\PR(X=-1)= \frac{1}{2}$), 
it holds $X \sim \mu$ and $f_{k}(X)=f_{n+1}(X)= X$. 

For $m \in \N_+$, let us consider a $Q$ whose restriction to $\mc{T}_m$ is a function of the first $m$ so restricted moments as in Section 
\ref{secCondMoms} and such that 
$G_Q$ has degree $r_Q$. We define estimand $PRQ$ on 
admissible pairs of two distributions and two real-valued functions with sets of 
arguments' indices $(\{1\},\{1,2\})$, consisting 
of all possible $\alpha = ((\mu_1,\mu_2), (f_1,f_2))$, such that for $P \sim \mu_1$ and $R \sim \mu_2$ it holds 
$f_2(P,R) \in L^m$  and for 
$h_Q$ corresponding to  
the symmetric unbiased estimator of $G_Q$ in $r_Q$ dimensions as in (\ref{hQ}), for $\wt{R} \sim \mu_R^{r_Q}$ and 
independent of $P$, it holds 
\begin{equation}\label{hql2} 
h_Q(f_2)(P,\wt{R})f_1(P), h_Q(f_2)(P,\wt{R}) \in L^1 
\end{equation} 
(this will be needed for our estimators to be integrable), 
in which case 
\begin{equation}
PRQ(\alpha) = \E(f_1(P)Q(f_2(P,R)|P)).
\end{equation}
We also define estimand $CovQ$ on pairs $\alpha$ as above for which additionally for the above $P$ it holds $f_1(P) \in L^1$, in which case 
\begin{equation}
CovQ(\alpha)= \Cov(f_1(P),Q(f_2(P,R)|P)). 
\end{equation}
Let the estimand $AveQ$ be defined as in the previous section, but for $N_P=1$, for which  
it is an estimand on admissible pairs of two distributions and single functions. 
We define $n+1$ estimands $PRAQ$ as trivial extensions of 
$n+1$ estimands $(PRQ,\ldots,PRQ,AveQ)$ using the above $\psi$, and $n+1$ estimands $CovAQ$ 
as trivial extensions of $n+1$ estimands  $(CovQ,\ldots,CovQ,AveQ)$ also using $\psi$. 
We define $n+1$ estimands $bQ$ and $cQ$ whose coordinates are equal to these of coordinates of $PRAQ$ or equivalently of $CovAQ$, restricted to  
$\alpha=((\mu_1,\mu_2),(f_i)_{i=1}^{n+1})$ such that for $P\sim \mu_1$ and $R\sim\mu_2$, the set 
$\{f_i(P)\}_{i \in I_{n}}$ is nonzero orthogonal in $L^2$ for $bQ$ or orthonormal for $cQ$, $\E(f_i(P)) = 0$, $i \in I_{n}$, and 
$Q(f_{n+1}(P,R)|P) \in L^2$, 
so that $bQ_i(\alpha)$ is the coefficient of 
$f_i'(P) = \frac{f_i(P)}{\Var(f_i(P))}$, $i \in I_n$, and 
$bQ_{n+1}(\alpha)$ of $\I$, in the orthogonal projection of $Q(f_{n+1}(P,R)|P)$ onto span($\{f'_i(P)\}_{i=1}^n\cup\{\I\}$). 
Let us consider an unbiased estimation scheme $SCovE(N)$ for $CovE$, such scheme
$SPRE$ for $PRE$, as well as $SAveE$ and $SAveE(N)$ for $AveE$, 
which are counterparts of the above schemes 
$SCov$, $SPR$, $SAve$, and $SAve(N)$ and whose formulas for their respective estimators are analogous as for their counterparts 
but with $g_{2}[i]$ on the rhs of (\ref{scov}) and (\ref{defspr}) replaced by $g_2[i][i]$ for $SCovE(N)$ and $SPRE$, respectively, and with 
each $g[i]$ on the rhs of (\ref{save}) and (\ref{saven}) replaced by $g[i][i]$ for $SAveE$ and $SAveE(N)$. 
The fact that such schemes are unbiased is 
an easy consequence of Theorem \ref{condexpX} and (\ref{hql2}) (note that $h_E(f_2)(P,R) = f_2(P,R)$). 
We also define counterparts of schemes $P1$ and $C1(N)$ - an 
unbiased estimation scheme $P1E$ for $PRAE$ defined as trivial extensions of $n+1$ schemes $(SPRE,\ldots,SPRE,SAveE)$ 
and scheme $C1E(N)$ for $CovAE$ as such extensions of $n+1$ schemes $(SCovE(N),\ldots,SCovE(N),SAveE(N))$, both using $\psi$. 
We define another unbiased estimation scheme $SCov2E$ for $CovE$ given by the formula for estimator  
\begin{equation}\label{CovC2} 
\wh{CovE}_{SCov2E} = \frac{1}{2}(g_2[0][0]-g_2[1][0])(g_1[0] - g_1[1]),  
\end{equation} 
and a scheme $SAve2E$ for $AveE$ given by formula
\begin{equation}\label{Ave2E} 
\wh{AveE}_{SAve2E} = \frac{1}{2}(g[1][0] + g[0][0]). 
\end{equation} 
Scheme $SAve2E$ is a symmetrisation of scheme $SAveE$ in the first argument from one to two dimensions 
and thus from Theorem \ref{thIneqds}, 
\begin{equation}\label{dAve2E} 
d_{AveE,SAveE} \leq d_{AveE,SAve2E} \leq 2d_{AveE, SAveE}. 
\end{equation} 
We define an unbiased estimation scheme $C2E$ for $CovAE$ as trivial extensions of $n+1$ schemes $(SCov2E,\ldots,SCov2E,SAve2E)$ 
using $\psi$. 
Analogously as above for schemes $P1$ and $C1(N)$, by arguments based on Theorem \ref{thtransl} one shows that scheme $P1E$ can have much 
higher inefficiency constants for estimation of $cE_k$ (and thus also $bE_k$) than schemes $C1E(N)$ and $C2E$ do, 
and also by an analogous argument as for the former schemes there exist $\alpha \in D_{cE_k}$ such that 
\begin{equation}
d_{cE_k,P1E}(\alpha) < d_{cE_k,C1E(N)}(\alpha). 
\end{equation} 
We will now prove that scheme $C1E(N)$ can have arbitrarily higher inefficiency constant for estimation of $cE_k$ (and thus also $bE_k$ 
and $CovAE$), $k \in I_n$, than scheme $C2E$, from which it also follows that scheme 
$SCovE(N)$ can have arbitrarily higher inefficiency constant for estimation 
of $CovE$ than $SCov2E$. We have the following lemma, the proof of which is given in Appendix \ref{appcoeffsLem}.
\begin{lemma}\label{lemC1C2} 
For some $k \in I_n$, let us consider random variables $P$ and $R$, function $f_P \in L^2(\mu_P)$, and functions 
$f_{R,l} \in L^2(\mu_R), l\in \N_+$, such that $\lim_{l \to \infty}\Var(f_{R,l}(R)) = \infty$ 
and for each $l \in \N_+,$ there exists 
$\alpha_l = ((\mu_P, \mu_R), f_l) \in D_{cE}$,  
such that $f_{l,k} = f_P$ and 
$f_{l,n+1}(P,R) = f_P(P) + f_{R,l}(R)$. 
Then we have $\lim_{l \to \infty}\E_{\alpha_l}(\wh{cE}_{k,C1E(N)}^2) = \infty$. 
\end{lemma} 
For notations as in the above lemma we have from independence of $f_P(P)$ and $f_{R,l}(R)$ and $\E(f_P(P))=0$ that 
$cE_{k}(\alpha_l)=\E(f_P^2(P))$ and it does not depend on $l$, so that 
\begin{equation}
\lim_{l \to \infty}\Var_{\alpha_l}(\wh{cE}_{k,C1E(N)})=\infty.
\end{equation}
On the other hand the value of 
$\Var_{\alpha_l}(\wh{cE}_{k,C2E})$ does not depend on $l$ as the evaluations of $f_{R,l}$ cancel out when evaluating its estimator. 

For some $Q$ as above, distributions $\mu_1$, $\mu_2$, 
and some $m_1, m_2 \in \N_+$, let us define the corresponding independent random vectors 
with i. i. d. coordinates $\wt{P} \sim \mu_1^{m_1}$ and 
$\wt{R} \sim \mu_2^{m_2r_Q}$, and denote 
$\wt{R}_{Q} = ((\wt{R}_{r_Q i+l})_{l=1}^{r_Q})_{i=0}^{m_2-1}$. Analogously as when defining the 
subschemes of $SQ$ in the previous section, let us define an unbiased estimation scheme $SPRQ$ for  
$PRQ$ such that for each $\wt{P}$, $\wt{R}$, and $\wt{R}_Q$ 
corresponding to  $m_1=1$, $m_2=1$ and each $((\mu_1,\mu_2),(f_1,f_2)) \in D_{PRQ}$,  
the estimator given by $SPRQ$ fulfills 
\begin{equation}\label{gencq} 
\wh{PRQ}_{SPRQ}(f_1,f_2)(\wt{P},\wt{R}) = \wh{Cov}_{k,SPRE}(f_1,h_Q(f_2))(\wt{P},\wt{R}_Q). 
\end{equation} 
We analogously define unbiased scheme $SCovQ(N)$ for $CovQ$ but for $m_1 = m_2=N \in \N_+,$ $N >1$, and using 
$SCovE$ on the rhs of condition analogous to (\ref{gencq}), scheme $SCov2Q$ for $CovQ$, for  $m_1 = 2, m_2=1$, and 
using $SCov2E$ on the rhs of such condition, and the following schemes for $AveQ$ - 
scheme $SAveQ$ for $m_1 = m_2 = 1$ and using $SAveE$, 
$SAveQ(N)$ for $m_1 = m_2 = N$ and using $SAveE(N)$, and $SAve2Q$ for $m_1 =2, m_2 = 1$ and using $SAve2E$ in the condition. 
The fact that the 
above defined schemes are unbiased for estimation of $PRQ$ or $CovQ$ is an easy consequence of 
(\ref{phiqcq}), Theorem \ref{condexpX}, and (\ref{hql2}), while for $AveQ$ it is consequence of (\ref{phiqcq}) and the 
iterated expectation property. 
Unbiased schemes $P1Q$ for $PRAQ$, 
and such schemes $C1Q(N)$ and $C2Q$ for $CovAQ$ are defined as trivial extensions using $\psi$ of 
$n+1$ schemes $(SPRQ, \ldots, SPRQ, SAveQ)$,  $(SCovQ(N), \ldots, SCovQ(N), SAveQ(N))$, and $(SCov2Q, \ldots, SCov2Q, SAve2Q)$, respectively. 
We have a generalization of inequality of inefficiency constants analogous to (\ref{dAve2E}) and with the same justification 
\begin{equation}\label{dAve2Q} 
d_{AveQ,SAveQ} \leq d_{AveQ,SAve2Q} \leq 2d_{AveQ, SAveQ}. 
\end{equation} 

For some $Q$ as above let us now consider a random vector with independent coordinates $P = (P_i)_{i=1}^{N_P}$, $N_P \in \N_+$,
$0< \Var(P_i) < \infty$, random variable $R$ independent of $P$, $f$ measurable with $f(P,R) \in L^m$, 
and $h_Q(f)(P,\wt{R}) \in L^2$ for $\wt{R} \sim \mu_R^{r_Q}$ and 
independent of $P$(\ref{hql2}). 
%
In our numerical experiments we will be using different schemes defined below for estimation of  
coefficients of elements of the orthogonal set $\Phi = \{P_i - \E(P_i)\}_{i=1}^N \cup \{\I\}$ and the orthonormal one 
$\Phi' = \left \{\frac{P_i - \E(P_i)}{\sigma(P_i)}\right\}_{i=1}^{N_P} \cup \{\I\}$ 
in the orthogonal projection 
of $Q(f(P,R)|P)$ onto span($\Phi$) for $Q$ equal to $\E$ and $\Var$. As these schemes are unbiased for estimation of some more 
general estimands we shall start by introducing them. 
Let us define $N_P+1$-dimensional vectors of estimands $\wt{PRAQ}$ and $\wt{CovAQ}$  
whose each $i$th coordinate $\wt{\lambda}$ corresponding to such $i$th coordinate $\lambda$ of 
$CovAQ$ or $PRAQ$, respectively, for $n = N_P$, is such that $\wt{\lambda}$ is defined on all 
admissible pairs $\wt{\alpha} = ((\mu_{1,1},\ldots,\mu_{1,N_P},\mu_2),(f_i)_{i=1}^{N_P+1})$ 
of $N_P+1$ distributions and $N_P+1$ functions with sets of arguments' indices $(\{1\},\{2\},\ldots,\{N_P\},I_{N_P+1})$ 
such that for $\mu_1 = \bigotimes_{i=1}^{N_P} \mu_{1,i}$ 
and $\pi_i$ denoting projection from $B_{\mu_1}\times\ldots\times B_{\mu_N}$ onto the $i$th coordinate, $i \in I_{N_P}$,
it holds $\alpha = ((\mu_1,\mu_2),(f_1(\pi_1),\ldots,f_{N_P}(\pi_{N_P}), f_{N_P+1})) \in D_{\lambda}$, 
in which case 
\begin{equation} \label{wtlambda} 
\wt{\lambda}(\wt{\alpha}) = \lambda(\alpha). 
\end{equation} 
We analogously define estimands $\wt{bQ}$ and $\wt{cQ}$ corresponding to $bQ$ and $cQ$, respectively. 
For $\mu_i \sim P_i$, $i \in I_{N_P}$, and $\mu = (\mu_i)_{i=1}^{N_P}$, the coefficients of respective elements 
 of $\Phi$ as above in the orthogonal 
projection of $Q(f(P,R)|P)$ 
onto span($\Phi$) are equal to the consecutive coordinates of 
$\wt{bQ}(\alpha)$ for $\alpha = (\mu,(\phi_1, \ldots, \phi_{N_P},f))$, 
$\phi_i(x) = \frac{x - \E(P_i)}{\Var(P_i)}, i \in I_{N_P}$, 
and the coefficients of such elements $\Phi'$ in this projection are equal to the coordinates of $\wt{cQ}(\alpha')$ for 
\begin{equation}\label{alphaprim} 
\alpha' = (\mu,(\phi_1', \ldots, \phi_{N_P}',f)), 
\end{equation} 
where $\phi_i(x)' = \frac{x - \E(P_i)}{\sigma(P_i)}, i \in I_{N_P}$. 
We define evaluation vectors $s_{(l)}[i][j]$ and $s_{(l),k}[i][j]$ 
as  $s[i][j]$ and $s_{k}[i][j]$ in Section \ref{secMany} but using $g_l[v_1]\ldots[v_{N_P+1}]$
rather than $g[v_1]\ldots[v_{N_P+1}]$ for the same $v$ on the right hand sides of expressions defining them. 
We define unbiased estimation schemes $\wt{P1Q}$ for $\wt{PRAQ}$ as well as 
$\wt{C1Q}(N)$ and $\wt{C2Q}$ for $\wt{CovAQ}$ (and thus all three also unbiased 
for $\wt{bQ}$ and $\wt{cQ}$) as obvious modifications of the 
schemes $P1Q,$ $C1Q(N)$, and $C2Q$, respectively, whose formulas for estimators have 
each occurrence of $g_l[i]$ replaced by $r_{l}[i]$ (see (\ref{ridef})), $l \in I_{N_P}$, and $g_{N_P+1}[i][j]$ by 
$s_{(N_P+1)}[i][j]$ (see Section \ref{secMany}). For instance for some $i \in I_{N_P}$, the estimator of $CovAQ_i$ given by $\wt{C2Q}$ is 
\begin{equation} 
\wh{\wt{CovAQ}}_{i, \wt{C2Q}} = \frac{1}{2}(r_{i}[1]-r_{i}[0])(s_{(N_P + 1)}[1][0]-s_{(N_P + 1)}[0][0]). 
\end{equation} 
We shall now introduce a new unbiased estimation scheme $SQCov$ for $(\wt{CovAQ}_i)_{i=1}^{N_P}$, 
that is the first $N_P$ coordinates of $\wt{CovAQ}$. 
Let $\phi_{Q,t}$ be the unbiased symmetric estimator of $G_Q$ in $t = 2r_Q$ dimensions,
where $r_Q$ denotes the degree of $G_Q$ as in the previous section. 
For $k \in I_{N_P}$, let 
\begin{equation}
\wh{\wt{CovAQ}}_{k,0} = \frac{1}{2}(\phi_{Q,t}((s_{(N_P + 1)}[0][j])_{j=1}^{t})
-\phi_{Q,t}((s_{(N_P + 1)}[1][j])_{j=1}^{t}))(r_k[0] - r_k[1])
\end{equation}
and introducing a C language-like notation 
\begin{equation}\label{CNot} 
(a == b)?c:d = \begin{cases} 
      c & \text{ if $a=b$,} \\ 
      d & \text{otherwise,} \\ 
      \end{cases} 
\end{equation}
for $l \in I_{N_P}$, let  
\begin{equation} 
\wh{\wt{CovAQ}}_{k,l} = 
\frac{1}{2}(\phi_{Q,t}((s_{(N_P + 1),l}[0][j])_{j=1}^{t}) -\phi_{Q,t}((s_{(N_P + 1),l}[1][j])_{j=1}^{t}))(r_k[0] - r_k[1])(k==l?-1:1). 
\end{equation} 
The unbiased subscheme of $SQCov$ for estimation of $CovAQ_{k}$, $k \in I_{N_P}$, is given by the formula for estimator 
\begin{equation} 
\wh{\wt{CovAQ}}_{k,SQCov} = \frac{1}{N_P + 1}\sum_{l=0}^{N_P}\wh{\wt{CovAQ}}_{k,l}. 
\end{equation} 
We define scheme $\wt{SQ}$ as a one consisting of trivial extensions 
of subschemes from $SQ$ from the previous section, for which coordinates of $\psi$ defining the extensions are equal 
to $\phi_1:\{1\}\rightarrow \N_+:\phi_1(1) = N_P+1$, and also of subchemes of $SQCov$ for which such coordinates are equal to  
$\phi_2=\id_{I_{N_P+1}}$. 
Intuitively, scheme $\wt{SQ}$ is created by adding 
to $SQCov$ subschemes of $SQ$ applied to the $N_P+1$st function. Such scheme is unbiased for estimation
of estimands $\wt{\lambda}$ also created by 
trivial extensions using the above $\psi$ of the corresponding estimands $\lambda$ of scheme $SQ$ 
for which coordinates of $\psi$ are  $\phi_1$ and estimands 
$(\wt{CovAQ}_{k})_{k=1}^{N_P}$ for which these coordinates are $\phi_2$. Similarly as in the previous section, let us further 
add to $\wt{SVar}$ subschemes from $SECov$ 
for estimation of $\wt{CovAE}_l$, $l \in I_{N_P},$ symmetrised from two to four dimensions in the $N_P +1$st argument. 
Let us consider the following set 
of symmetries in different first $N_P$ arguments in two dimensions $\Pi_1 = \bigcup_{j=1}^{N_P}\Theta_{{N_P+1},j,2}$ 
and set of symmetries in the $N_P+1$st argument in four dimensions $\Pi_2 = \Theta_{{N_P+1},{N_P+1},4}$ (see 
definitions below (\ref{defthetam})). 
Subschemes of $\wt{SVar}$ for estimation of $\wt{CovAE}_l, l \in I_{N_P+1}$, (note that $\wt{CovAE}_{N_P+1} = \wt{Ave}$) 
are averages of subschemes of $\wt{C2E}$ with respect to $\Pi_1\Pi_2 = \{\pi_1\pi_2:\pi_1\in \Pi_1, \pi_2\in\Pi_2\}$, 
and they use both individually and together $4(N_P + 1)$ times more evaluation vectors for the last function 
than the latter, so that we have 
\begin{equation}\label{ineqCovd} 
d_{\wt{CovAE}_l,\wt{SVar}} \leq 4(N_P + 1)d_{\wt{CovAE}_l, \wt{C2E}},\quad l \in I_{N_P +1}. 
\end{equation} 
Subscheme of $\wt{SVar}$ for estimation of $\wt{Ave}$ uses $8(N_P + 1)$ times more evaluation vectors for the last function 
than the $N_P+1$st subscheme of $\wt{P1E}$ using one such vector, and it is also an
average of the latter with respect to $\Theta_{{N_P+1},I_{N_P},2}\Pi_1\Pi_2$, so that from Lemma \ref{lemcomb} it easily follows that 
\begin{equation}\label{ineqCovd2} 
d_{\wt{Ave}, \wt{P1E}} \leq d_{\wt{Ave},\wt{SVar}} \leq 8(N_P + 1)d_{\wt{Ave}, \wt{P1E}}. 
\end{equation} 
An estimand corresponding to the 
nonlinearity coefficient (\ref{DNJ}) of $Q(f(P,R)|P)$ in $P_k$, $k \in I_{N_P}$, for $Q = \Var, \E$, is 
defined as  
\begin{equation}
DNQ_k = \wt{VQ}_k^{tot} - \wt{cQ}_k^2  
\end{equation}
for arguments $\alpha'$ as in (\ref{alphaprim}).
We use for its estimation an unbiased scheme 
which can be treated as an additional subscheme of a scheme $\wt{SVar}(2)$ using  $\wt{SVar}$ in two independent steps, 
$\wt{SVar}_i = \A_{\pi_i}(\wt{SVar})$  (see Section \ref{secAveSchemes}), $i \in I_2$, 
and which is given by formula 
\begin{equation}
\wh{DNQ}_{k,\wt{SVar}(2)} = \frac{1}{2}\sum_{i=1}^2\wh{VQ}^{tot}_{k,\wt{SVar}_i} - \Pi_{i=1}^2\wh{cQ}_{k,\wt{SVar}_i}. 
\end{equation}
For the estimand corresponding to the nonlinearity coefficient (\ref{dnc}) of $Q(f(P,R)|P)$
in all coordinates of $P$, $DNQ = \wt{VQ}_P - \sum_{i=1}^{N_P}\wt{cQ}_i^2$ for $Q = \Var, \E$, 
we use a scheme given by
\begin{equation} 
\wh{DNQ}_{\wt{SVar}(2)} = \frac{1}{2}\sum_{i=1}^2\wh{VQ}^{tot}_{P,\wt{SVar}_i} - \sum_{k=1}^{N_P}\Pi_{i=1}^2\wh{cQ}_{k,\wt{SVar}_i}. 
\end{equation} 
In our numerical experiments the above schemes for nonlinearity coefficients were 
used to obtain estimates once per each two steps of a MC procedure using scheme $\wt{SVar}$ and thus the 
final MC estimator was computed by averaging over two times fewer estimates than for subschemes of $\wt{SVar}$. 
Correlations between a given function of conditional moments $Q(f(P,R))$ and coordinates of $P$
are $\N_P$ estimands $corrQ=(corrQ_i)_{i=1}^{N_P}$ on common admissible pairs, such that (see (\ref{correqu})) 
\begin{equation}
corrQ_i = \frac{\wt{bQ}_i}{\wt{VQ_P}}, \quad i \in I_{N_P},
\end{equation}
is defined on the intersection of domains of the divided estimands. One can compute 
the estimates of it for $Q=\E$ or $\Var$ e. g. by dividing the final 
MC estimates of $bQ_i$ and $DQ$ obtained using scheme $\wt{SVar}$ and one can use 
analogously defined schemes $\wt{SQ'}$ for estimating $corrQ_i$ for other $Q$ and $Q'$. 
Note that similarly as for schemes for variance-based sensitivity indices of conditional expectation in Section \ref{secMany}, 
estimation schemes for estimands like $PRA$, $CovA$, $PRE$, $CovE$, and $DNE_k$ can be easily generalized to functions 
with values in $\R^m$, $m \in \N_+$, 
by using appropriate scalar product of vectors instead of function multiplication in the formulas for estimators. 
When this should not cause any misunderstandings, 
to simplify notations we often drop the tilde sign over the symbols
of schemes or estimands introduced in this section, e. g. write $SVar$ instead of $\wt{SVar}$.  

\section{\label{secApprErr}Schemes for the mean squared error of approximation} 
Let us consider a Hilbert space $H$ with some scalar product $<,>$, inducing norm $|\cdot|$. 
Let $v \in H$,  $\Psi = \{\psi_i\}_{i=1}^l$ be an orthonormal set in $H$, $V = \text{span}(\Psi)$, $P_{V}$ 
be an orthogonal projection from $H$ onto $V$, 
and $c=(c_i)_{i=1}^l$ be the Fourier's coefficients of $v$ relative to $\Psi$, that is 
\begin{equation} 
P_V(v) = \sum_{i=1}^lc_i\psi_i. 
\end{equation} 
For example we can have $H= L^2$, with scalar product (\ref{scall2}), $v= Q(f(P,R)|P)$ for $f(P,R)$ being some construction of 
an output of an MR and $Q$ being a function which restricted to $\mc{T}_m$ is a function of the first $m$ so restricted moments 
as in Section \ref{secCondMoms}, and $\psi_i = \phi_i(P)$ 
for some functions $\phi_i$, $i \in I_l$, orthonormal in $L^2(\mu_P)$. 
Let $h=(h_i)_{i=1}^l \in \R^l$. Squared error of the 
approximation of $v$ using $\Psi\cdot h = \sum_{i=1}^{l}h_i\psi_i$ in $H$, denoted as $\err(h)$, fulfills 
\begin{equation}\label{errh}
\begin{split}
\err(h) &= |v - \Psi\cdot h|^2 \\
&= |P_V(v) - \Psi\cdot h|^2 + |v - P_V(v)|^2, 
\end{split} 
\end{equation} 
where in the second equality we used the fact that $v - P_V(v)$ is orthogonal to $V$ 
and in the last the fact that $\Psi$ is orthonormal. 
Let us consider an unbiased estimator 
$w = (w_i)_{i=1}^l$ of $c$ for some distribution $\nu$ so that 
$\E_\nu(w)=c$ 
and let us further assume that $w_i \in L^2(\nu), i \in I_n$. For instance for the above example, 
some unbiased estimation scheme $\kappa = (\kappa_i)_{i=1}^{l}$, $l >1$,  for $cQ$ for $n = l-1$ as in the previous section, 
$\mc{V}=D_{cQ}$, and for 
$(\mu,g) \in \mc{V}$ such that $\mu=(\mu_P,\mu_R)$ and $g =(\phi_1,\ldots,\phi_n,f)$, 
we can take $\nu = \mu^{p_{\kappa}}$ and $w = \phi_{\kappa,\mc{V}}(g)$ (see (\ref{phikappamany})). 

From (\ref{errh}), the average squared error of approximation of $v$ using estimates of $c$ given by $w(X)$, $X \sim \nu$, fulfills 
\begin{equation}
\begin{split}
\E(\err(w(X))) &= \sum_{i=1}^l\Var(w_i(X)) + |v - P_V(v)|^2 \\
&= \Var(w(X)) + |v - P_V(v)|^2, 
\end{split}
\end{equation}
where by variance in the last term we mean variance for random vectors defined as in Section \ref{secOrthog} using the 
standard scalar product in $\R^l$. Since for a fixed $v$ and orthonormal set 
$\Psi$, $|v - P_V(v)|$ is constant, we get lower mean squared approximation error when using 
estimator of orthogonal projection coefficients onto $\Psi$ with lower sum of variances of its coordinates.  
Thus standard scalar product is here a natural choice for defining variance used to quantify error of approximation of 
$c$ by $w$ for $\nu$. 

We define an estimand $ErrQ$ on all admissible pairs $\alpha = (\mu,s)=((\mu_1,\mu_2), (s_1,s_2))$ 
of two distributions and two functions with sets of arguments' indices $(\{1,2\},\{1\})$, 
such that for each $P \sim \mu_1$ and $R \sim \mu_2$, 
it holds $s_1(P,R) \in L^m$, $s_2(P) \in L^2$, and for $\wt{R} \sim \mu_2^{r_Q}$ and independent of $P$, it holds 
$h_Q(s_1)(P,\wt{R})\in L^2$, in which case 
\begin{equation} 
ErrQ(\alpha) = \E((Q(s_1(P,R)|P) - s_2(P))^2). 
\end{equation} 
For the above discussed example in which $v= Q(f(P,R)|P)$, for $s_1 = f$, $h \in \R^l$, and $s_2(P) = \sum_{i=1}^lh_i\phi_i(P)$, 
$ErrQ(\alpha)$ is equal to $\err(h)$. 
Let us consider an unbiased estimation scheme $SErrE$ for $ErrE$, defined by formula 
\begin{equation}
\wh{ErrE}_{SErrE} = (g_1[0][0] - g_2[0])(g_1[0][1] - g_2[0]).
\end{equation}
The fact that it is unbiased follows from formula (\ref{ggxy}) in Theorem \ref{thCond} since for 
$\alpha = (\mu,s)\in D_{ErrE}$, $X \sim \mu^{p_{SErrE}}$, and $p(x_1,x_2) = s_1(x_1,x_2) -s_2(x_1)$, we have
\begin{equation}
\begin{split}
\wh{ErrE}_{SErrE}(f)(X) &= \E(p(X_1,X_{2}[0])p(X_1,X_{2}[1])) \\ 
&=\E((\E(p(X_1,X_{2}[0])|X_1))^2)\\
&= \E((\E(s_1(X_1,X_{2}[0])|X_1) - s_2(X_1))^2).  
\end{split}
\end{equation}
Analogously in Section \ref{secPolynEst}, we define scheme $SErrQ$ 
giving an unbiased estimator of $ErrQ$ such that for $(\mu,s) \in D_{ErrQ}$, $P\sim\mu_1$, $R\sim\mu_2$, 
and $\wt{R}\sim \mu_2^{2r_Q}$ and independent of $P$, it holds 
\begin{equation} 
\widehat{ErrQ}_{SErrQ}(s)(P,\wt{R}) = \frac{1}{|\Psi_Q|} \sum_{\pi \in \Psi_Q}\widehat{ErrE}_{SErr}(h_Q(s_1),s_2)(P, \wt{R}_{\pi}), 
\end{equation} 
and we add to scheme $SErrVar$ 
subschemes for estimation of $ErrE$ which are created by symmetrisation of subschemes from 
$SErrE$ from $2r_{E} = 2$ to $2r_{Var} = 4$ dimensions in the second argument. 

\section{\label{secVarDiffNew}Variances of the new estimators for the RTC and GD methods}
In our numerical experiments which we describe in the further sections,  
the estimates of variances of various unbiased estimators of main and total sensitivity indices of conditional 
variance from Section \ref{secPolynEst} as well as such estimators of
orthogonal projection coefficients of conditional variance and expectation onto the span of model parameters and constants 
significantly depended on whether the GD or the RTC method was used and on the order of reactions 
in the GD method.
Using notations as in Section \ref{secVarDiff},  
for $s$ such that $\alpha = ((\mu_P,\mu_R),(s,f)) \in D_{CovE}$ and  
 $\wt{P}\sim \mu_P^2$ and independent of $R$, it holds 
\begin{equation} 
\begin{split} 
4\E_{\alpha}(\wh{CovE}^2_{SCov2E}) &= \E((f(\wt{P}[0],R)- f(\wt{P}[1],R))^2(s(\wt{P}[0]) - s(\wt{P}[1]))^2)\\ 
&= \E(\msd(\wt{P}[0],\wt{P}[1])(s(\wt{P}[0]) - s(\wt{P}[1]))^2), \\ 
\end{split} 
\end{equation} 
so the inequalities between the variances $\Var_\alpha(\wh{CovE}_{SCov2E})$ 
should also be the same as for $\msd(p_1, p_2)$ for all appropriate $p_1$, $p_2$ 
as in Section \ref{secVarDiff} depending on the method used. 
By an analogous argument the same applies to the variances of estimators $\wh{CovAE}_{k,C2E}$, 
$k \in I_n$ for appropriate admissible pairs. 
For some  $l, n, m \in \N_+$, let us now consider two functions $\phi: I_l\rightarrow I_n$ and $\psi: I_l\rightarrow I_m$, 
a random vector with not necessarily independent coordinates  $P' = (P'_i)_{i=0}^{n}, P'_i \sim P, i \in I_n$, 
and a random vector $\wt{R} \sim \mu_R^m$, independent of $P'$. We have 
\begin{equation}
\E((\sum_{i=1}^{l}f(P_{\phi(i)}',R_{\psi(i)}))^2) = \sum_{i, j \in I_l}\E(f(P_{\phi(i)}',R_{\psi(i)})f(P_{\phi(j)}', R_{\psi(j)})).
\end{equation}
By a proof similar as of Theorem \ref{thCond}, for $\psi(i)\neq\psi(j)$, it holds 
\begin{equation}\label{efpij} 
\E(f(P_{\phi(i)}',R_{\psi(i)})f(P_{\phi(j)}', R_{\psi(j)})) = \E(\E(f(P_{\phi(i)}',R)|P_{\phi(i)}')\E(f(P_{\phi(j)}',R)|P_{\phi(j)}')), 
\end{equation} 
which, given $f(P,R)=g(h(P,R))$,  
is determined by the distribution of $(P_{\phi(i)}',P_{\phi(j)}')$, $g$, and the reaction network $RN$ used 
in the definition of MR, and thus its value should not depend on the construction of MR being used. 
For $\psi(i)=\psi(j)$, 
\begin{equation} 
\E(f(P_{\phi(i)}',R_{\psi(i)})f(P_{\phi(j)}', R_{\psi(j)}) = \E(f(P,R)^2) - \frac{1}{2}\msd(P_{\phi(i)}',P_{\phi(j)}'). 
\end{equation} 
From the above calculations it easily follows that the inequality between the variances of estimators like 
$\wh{AveE}_{SAve2E}$, $\wh{CovAE}_{n+1,C2E}$, and $\wh{Ave}_{\wt{SVar}},$ 
for the appropriate admissible pairs corresponding to $f$, $P$, $R$, for different constructions  
should be opposite than the inequalities between $\msd(p_1,p_2)$ for all $p_1$ and $p_2$ as in Section \ref{secVarDiff}. 
In our numerical experiments discussed in the next section 
the estimates of variances of estimators $\wh{cE}_{k,C2E}$ were often much lower and of estimators 
$\wh{CovAE}_{n+1,C2E}$ and $\wh{Ave}_{\wt{SVar}}$ higher when using the RTC than the GD method. However, for the 
MBMD model, as discussed in Section \ref{secMBMD},
some estimates of variance of $\wh{cE}_{k,C2E}$ in our experiments were statistically significantly higher while of 
$\wh{CovAE}_{n+1,C2E}$ (denoted there as $\wh{Ave}_{C2E}$) smaller 
for the RTC than the GD method, from which it follows that for this model and its output, 
similarly as for the ones  from Section E of 
\cite{Anderson2013}, there exist parameters 
$p_1$, $p_2$, for which $\msd(p_1,p_2)$ is higher for the RTC than the GD method. 

\section{Numerical experiments}\label{chapNumExp}
\subsection{Implementation extensions and tests of validity of the inefficiency constants of schemes}\label{secImpl}
The numerical experiments in this work were run using the same hardware and operating system as described in Section 
\ref{secExSoft}. The program described in that section was extended by adding implementations of 
MC procedures using the new estimators described in sections \ref{secPolynEst}, \ref{secApprCoeff}, and \ref{secApprErr}.
Figure \ref{pic1} describes basic specification process 
and the corresponding results of computations with our extended program. 
We carried out a numerical experiment comparing the average execution times of MC procedures using schemes 
$\wt{SVar}$, $SE$, $C1E$, and $C2E$, 
and the same number of simulations of the RTC or GD methods for the outputs of the SB, GTS, and
MBMD models defined in sections \ref{SBPrev}, \ref{GTSPrev}, and \ref{MBMDPrev}. 
For $N_P$ denoting the number of parameters of a given model, $k=50$ for the GTS model and $k=500$ for the other models, 
for each model we performed a $50$-step MC procedure measuring in each step the execution time of $k$ MC-steps 
using scheme $\wt{SVar}$, $2k$ MC-steps of scheme $SE$, one MC-step of $\wt{C1E}(8(N_P+1)k)$, and $4(N_P +1)k$ MC-steps of scheme $\wt{C2E}$. 
The computed mean execution times are presented in Table \ref{tabTimeDiff}.  
From the table we can see that the mean execution times of our implementations of 
the procedures using different schemes and simulation methods and a given model 
for the same number of process simulations 
are comparable. 
For this reason and to make our analysis independent of the implementation or computer architecture used, rather than comparing 
the estimates of inefficiency constants of sequences of MC procedures, in our numerical experiments
we shall focus on comparing the estimates of variances of the final MC estimators for the same number of process simulations carried out 
in the MC procedures, 
the ratio of such variances being equal to the ratio of appropriate 
inefficiency constants of the schemes used, as discussed in sections \ref{secIneffSchemes} and \ref{secineffgen}. 
\begin{figure}
\input{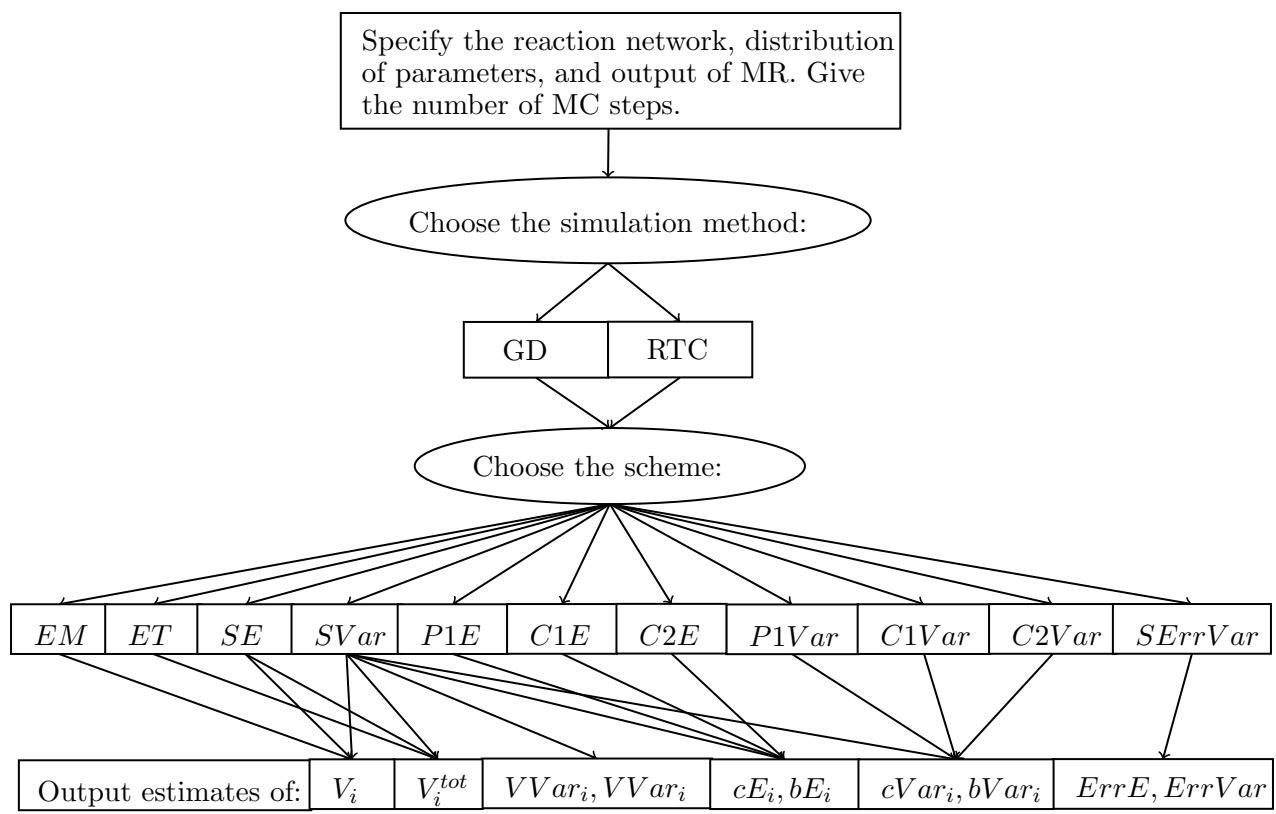}
\caption{\label{pic1} Diagram describing the basic specification process and the corresponding results of computations carried 
out with our program. 
See Section \ref{secExSoft} for details on specification of MR and sections 
\ref{secMany}, \ref{secPolynEst}, \ref{secApprCoeff}, and \ref{secApprErr} for definitions of the above schemes and estimands.}
\end{figure}

\begin{table}[h]
\begin{tabular}{|l|c|c|c|c|}
\hline
{\multirow{2}{*}{MR}}
& $SE$ & $SVar$ & $C1E$& $C2E$\\
\cline{2-5}
& \multicolumn{4}{c|}{RTC} \\
\cline{2-5}
\hline
SB &$1.8845 \pm 0.0014$&$1.9893 \pm 0.0020$&$2.1302 \pm 0.0015$&$1.9872 \pm 0.0015$\\ 
GTS &$2.616 \pm 0.014$&$2.612 \pm 0.010$&$2.6934 \pm 0.0045$&$2.6419 \pm 0.0058$\\ 
MBMD &$2.1798 \pm 0.0068$&$2.5835 \pm 0.0048$&$2.2885 \pm 0.0039$&$2.2357 \pm 0.0049$\\ 
\hline
MR & \multicolumn{4}{c|}{GD} \\
\hline
SB &$1.9782 \pm 0.0029$&$2.1013 \pm 0.0023$&$2.4037 \pm 0.0028$&$2.1594 \pm 0.0028$\\
GTS & $2.734 \pm 0.010$&$2.705 \pm 0.016$&$2.8295 \pm 0.0070$&$2.7709 \pm 0.0079$\\
MBMD & $2.1564 \pm 0.0090$&$2.6173 \pm 0.0088$&$2.2450 \pm 0.0081$&$2.2058 \pm 0.0082$\\
\hline
\end{tabular}
\caption{\label{tabTimeDiff} 
Estimates of mean execution times in seconds computed from 50 runs of 
the MC procedures involving the same number of simulations of the GD or RTC constructions for each model and using different schemes 
as explained in Section \ref{secImpl}.} 
\end{table}

\subsection{SB model}\label{secSB} 
Let us consider the SB model and its output from Section \ref{SBPrev}. 
See \cite{badowski2011} and Appendix \ref{appd} 
for derivation of some analytical expressions for the sensitivity indices and orthogonal projection coefficients  
in this model. Some values obtained from these expressions are presented in Table \ref{exactVal} and the main 
Sobol's indices of conditional expectation and variance are also shown on pie charts in Figure \ref{pieSB}. 
For computations with this model we used only the RTC method 
since for a reaction network with one reaction 
there is no difference in variance of our estimators using the GD and RTC methods. 
We performed a one-million-step MC procedure using scheme $SVar$. The computed sensitivity indices, orthogonal projection coefficients, 
and nonlinearity coefficients are presented in Table \ref{tabSBRTC}, while the 
mean value and average variance of the model output are given in Table \ref{tabDisp}. The results  
of computations are in good agreement with the analytically computed values in Table \ref{exactVal} and Appendix \ref{appd}. 
We performed a ten-million-step MC procedure using scheme $SErrVar$ to estimate the mean squared error of approximation of 
the conditional expectation and conditional variance of the output using linear combinations of elements of the set 
of centered model parameters and constant one, that is the set 
$\Phi = \{P_i - \E(P_i)\}_{i=1}^{N_P}\cup\{\I\}$, 
 using as coefficients the estimates of 
 $(bE_1, \ldots, bE_{N_P}, Ave)$ from Tables \ref{tabSBRTC} and \ref{tabDisp} when approximating the  conditional expectation, and 
 estimates of $(bVar_1, \ldots, bVar_{N_P}, AveVar)$ from these tables when approximating the conditional variance. 
We obtained estimates of mean squared approximation error $0.027 \pm 0.022$ for the conditional expectation and $1 \pm 15$ for the variance, both 
being in good agreement with the values of these errors we computed analytically, approximately equal to $0.00051$ and $0.0062$, respectively. 
We also performed a numerical experiment comparing the estimates of variances of the final MC estimators of different indices using  
scheme $SVar$ in $25000$ steps, scheme $SE$ in $50000$, and 
schemes $EM$ and $ET$ in $100000$ steps, 
so that each above MC procedure used the same number of one million process simulations. 
We ran each above procedure five times collecting in each run the estimate of variance of the final MC estimator (\ref{varEst}), 
and finally computing the estimate of mean and standard deviation of the estimates of variances as described in Appendix \ref{appStatMC}. 
The results are presented in Table \ref{tabCompareMB} and in Figure \ref{barSBV2},
from which we can see that the estimates of variances of the final MC estimators given by 
scheme $EM$ are approximately two times lower than for scheme $SE$ and four times lower than for scheme $SVar$ for the main sensitivity
indices of all parameters except $K3$
and analogously for schemes $ET$, $SE$, and $SVar$ for the total sensitivity indices of these parameters. 
Such proportions correspond to equalities in the rhs 
inequalities of  
relations (\ref{EMComp}), (\ref{ineqDVar}), and (\ref{compVkVar}) for the main as well as in the 
relations (\ref{VitotComp}), (\ref{ineqDVar}), and (\ref{compVkTotVar}) 
for the total sensitivity indices of conditional expectation. 
For $\Phi' = \{\frac{P_i - \E(P_i)}{\sigma(P_i)}\}_{i=1}^{N_P}\cup\{\I\}$, that is the set of normalized centred parameters and constant one, we 
also performed a numerical experiment comparing the variances of MC methods estimating the coefficients of orthogonal projection 
of the conditional expectation and conditional variance onto span($\Phi'$). 
We used scheme $SVar$ in $k = 100$ steps for estimating the orthogonal projection coefficients 
of both conditional expectation and conditional variance. For $N_P = 4$ 
denoting the number of parameters, we also carried out MC procedures using $4(N_P+1)k$ steps of scheme $C2E$, $8(N_P+1)k$ 
steps of $P1E$, and a single step of $C1E(8(N_P+1)k)$ for the conditional expectation. For the conditional variance we applied besides 
scheme $SVar$ also $4(N_P+1)k$ steps of scheme $P1Var$, single step of $C1Var(4(N_P+1)k)$, and $2(N_P+1)k$ steps of $C2Var$.  
The same number of $4000$ process evaluations was used in each above method. 
We performed a $200$ step procedure to compute the mean variances of the final MC estimators. For schemes 
$C1E$ and $C1Var$, the variance in each step was computed using unbiased estimator of variance computed 
from a sample of ten runs of the method, while for the 
other methods this was an estimate of variance of the mean computed in the method using estimator (\ref{phiAveVar}). 
Let $\Sigma Q$ for $Q = Var,E$ be defined as a sum of variances 
of the final MC estimators of all the coefficients of orthogonal projection of $Q(f(P,R)|P)$ onto span($\Phi'$) given by certain scheme.
From discussion in Section \ref{secApprErr}, using the coefficients computed with a scheme with lower value of $\Sigma Q$ should 
lead to lower average error of approximation of $Q(f(P,R)|P)$. Furthermore, from discussion in sections \ref{secineffgen} 
and \ref{secApprErr} the ratio of values of $\Sigma Q$ when using different schemes and the same number $l$ of 
process simulations is equal to the ratio of inefficiency constants
of these schemes for estimating the vector of projection 
coefficients, with variances in definitions of the constants being given by standard scalar product. 
In each of the above $200$ steps we also obtained estimates of $\Sigma Q$ for different schemes by 
summing the estimates of variances of the estimators of the coefficients and then computed the mean from all steps. 
The results of the above numerical experiment are presented in Table \ref{tabCovSB} and Figure \ref{figSBBar}.
From the table we can see that $\Sigma E$ is lowest for scheme $C1E$, followed by $C2E$, $SVar$, and $P1E$, while 
$\Sigma Var$ is lowest for $SVar$ followed by $P1Var$, $C2Var$, and $C1Var$. 
The reader can easily confirm that the results in Table \ref{tabCovSB} are in good agreement with various inequalities 
between variances of estimators of orthogonal projection coefficients given in Section \ref{secApprCoeff}. 

\begin{table}[h]
\begin{tabular}{|l|c|c|c|c|c|}
\hline
$i $ & $\widetilde{V}_i$ & $\widetilde{V}_i^{tot} $ &$\widetilde{S}_i$ &$\widetilde{S}_i^{tot}$ & $bE_i$\\
\hline
$C $ &$ 310 $& $310$  &$0.451$&$0.451$ & $ 1 $\\
$K_1 $ &$ 300 $& $300$ & $0.436  $ & $0.436 $&$ 100$\\
$K_2 $ &$ 75 $& $75$ & $0.109  $ & $0.109  $&$ 100 $\\
$K_3 $ &$ 3 $& $3$ & $0.0044 $ & $0.0044 $&$ 100 $\\
\hline
$i$& $V_i$ & $V_i^{tot} $ & $S_i$ & $S_i^{tot}$ &\multicolumn{1}{c}{} \\
\cline{1-5}
$P $ & $ 688 $& & $0.80$ & &\multicolumn{1}{c}{} \\
$R $ & &$170$  & &$0.20$ &\multicolumn{1}{c}{}\\
$P,R $ &$ 858 $& $ 858 $ & $1$ & $0$ &\multicolumn{1}{c}{}\\
\hline
$i$ &  $VVar_i$ &$VVar^{tot}_i$ & $SVar_i$ &$SVar_i^{tot}$ & $bVar_i$\\
\hline
$C $& $0$ & $0$ &$0$ & $0$ & $0$ \\
$K_1 $&$300$ & $300$ & $0.794$&$0.794$&$ 100 $\\
$K_2 $&$75$& $75$& $0.198 $& $0.198 $ &$ 100$\\
$K_3 $&$3$& $3$&$ 0.00794$&$0.00794$ & $ 100$\\
\hline
$K_3 $&$3$& $3$&$ 0.00794$&$0.00794$ & $ 100$\\
\hline
$P $ & $ 378 $& $1$& $1$ & $1$&\multicolumn{1}{c}{} \\
\cline{1-5}
\end{tabular}
\caption{\label{exactVal}Values of sensitivity indices and orthogonal projection coefficients 
onto model parameters of conditional expectation and variance in the SB model, obtained using analytical expressions
 from Appendix \ref{appd} and \cite{badowski2011}.} 
\end{table}

\begin{figure}[h]%
  \centering
  \subfloat[]{\includegraphics[width=0.46\textwidth]{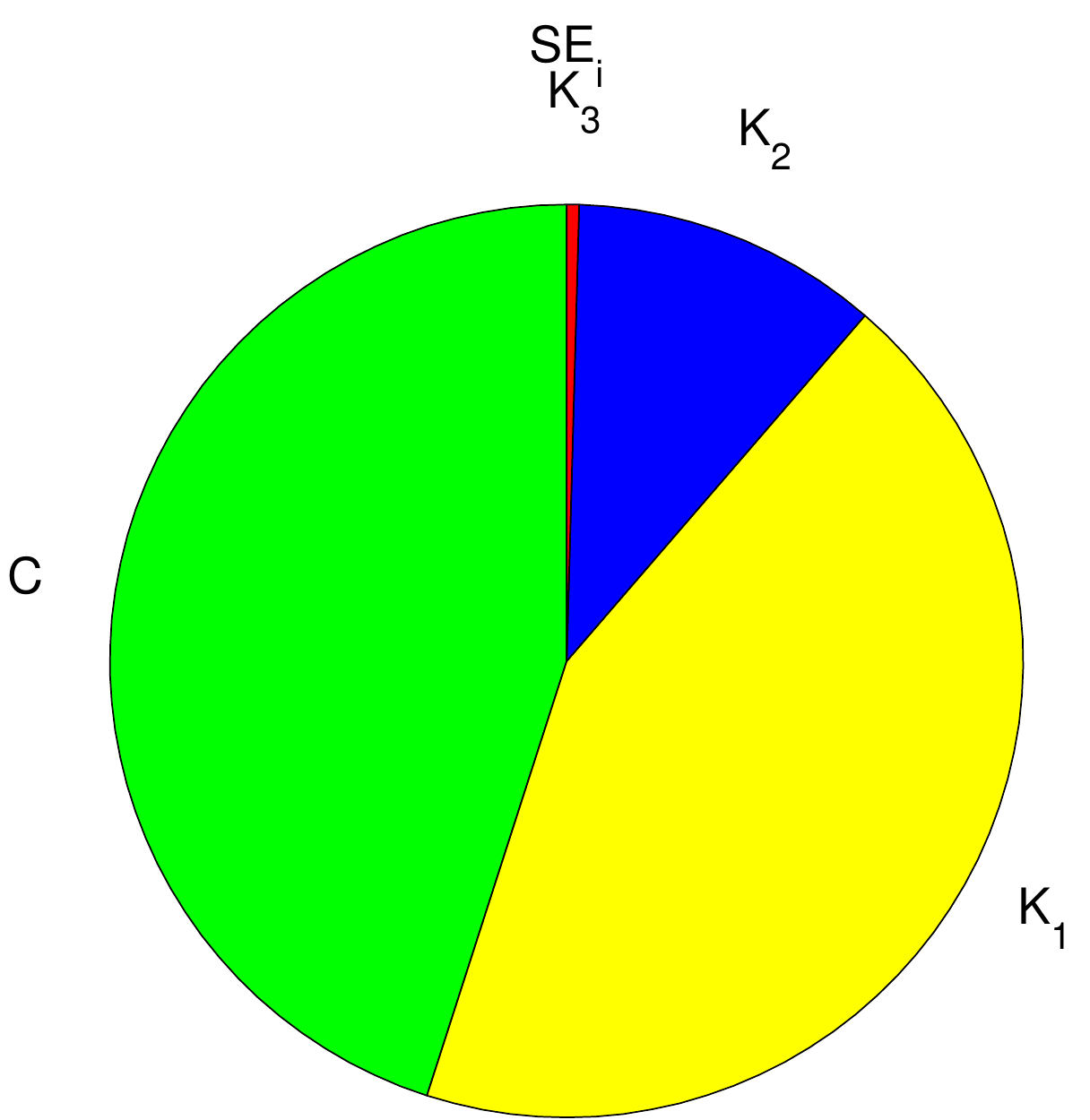}}
\qquad  
\subfloat[]{\includegraphics[width=0.39\textwidth]{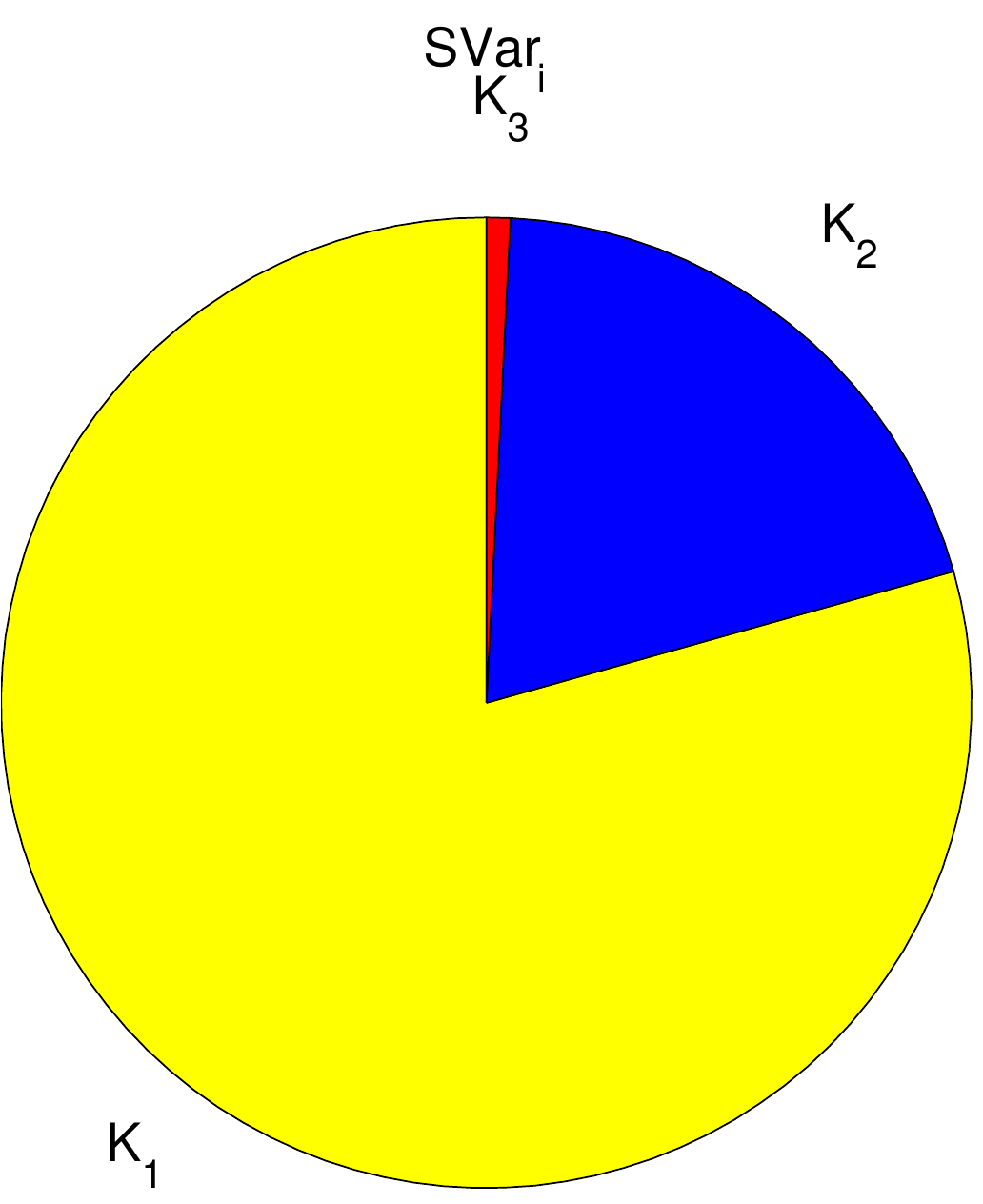}} 
\caption{The proportion of the total arc length occupied by a sector of a pie chart is equal to the main Sobol's 
sensitivity index of conditional expectation in (a) or conditional variance in (b) of the output of the SB model with respect to its 
given parameter.}
\label{pieSB}
\end{figure}

\begin{table}[h]
\begin{tabular}{|l|c|c|c|c|c|c|}
\hline
$i $ & $\widetilde{V}_i$ & $\widetilde{V}_i^{tot} $ &$\widetilde{S}_i$ &$\widetilde{S}_i^{tot}$& $bE_i$ & $ DNE_i$\\ 
\hline
$C $& $309.83 \pm 0.37$ & $309.83 \pm 0.37$ &$ 0.45 $&$ 0.45$& $0.9994 \pm 0.0012$ & $-0.45 \pm 0.72$ \\  
$K1 $& $299.80 \pm 0.36$ & $299.80 \pm 0.36$ &$ 0.44 $&$ 0.44$& $99.93 \pm 0.12$ & $-0.05 \pm 0.70$ \\  
$K2 $& $74.895 \pm 0.091$ & $74.895 \pm 0.091$ &$ 0.11 $&$ 0.11$& $99.83 \pm 0.13$ & $-0.16 \pm 0.18$ \\  
$K3 $& $2.9991 \pm 0.0039$ & $2.9999 \pm 0.0039$ &$ 0.0044 $&$ 0.0044$& $100.15 \pm 0.33$ & $-0.0173 \pm 0.0073$ \\  
\hline
$i $ & $V_i$ & $V_i^{tot} $ & $S_i$ & $S_i^{tot}$ &$ DNE $& $-0.27 \pm 0.56$\\
\hline
$P $& $687.43 \pm 0.54$ & $698.59 \pm 0.55$ &$ 0.8 $&$ 0.81$&\multicolumn{2}{c}{}\\
$R $& $158.79 \pm 0.13$ & $169.95 \pm 0.13$ &$ 0.19 $&$ 0.2$&  \multicolumn{2}{c}{} \\
$P, R$& $857.38 \pm 0.56$ & $857.38 \pm 0.56$ &$ 1 $&$ 1$ & \multicolumn{2}{c}{} \\
\hline
$i$ &  $VVar_i$ & $VVar^{tot}_i$ & $SVar_i$ &$SVar_i^{tot}$& $bVar_i$ & $DNVar_i$\\
\hline
$C $& $0$ & $0$ &$ 0 $&$ 0$& $-7.2 \pm 6.5\cdot 10^{-4}$ & $0$ \\  
$K1 $& $301.5 \pm 5.6$ & $301.1 \pm 5.6$ &$ 0.79 $&$ 0.79$& $99.62 \pm 0.36$ & $2.9 \pm 7.8$ \\  
$K2 $& $76.2 \pm 2.4$ & $75.5 \pm 2.4$ &$ 0.2 $&$ 0.2$& $100.34 \pm 0.44$ & $0.7 \pm 3.1$ \\  
$K3 $& $3.27 \pm 0.42$ & $3.08 \pm 0.43$ &$ 0.0086 $&$ 0.0081$& $101.85 \pm 0.90$ & $-0.08 \pm 0.53$ \\  
\cline{6-7}
$P$& $380.7 \pm 6.1$ & $380.7 \pm 6.1$ &$ 1 $&$ 1$& $DNVar $ & $4.1 \pm 4.3$ \\
\hline
\end{tabular}
\caption{\label{tabSBRTC} Estimates of various indices and coefficients
 for the SB model computed in a one-million-step MC procedure using the RTC algorithm 
and scheme $SVar$.}
\end{table}
\begin{table}[h]
\begin{tabular}{|l|c|c|c|c|}
\hline
{\multirow{2}{*}{i}}
 &$SE$ &  $EM$  & $ET$ & $SVar$ \\
\cline{2-5}
& \multicolumn{4}{c|}{$V_i$ } \\
\cline{2-5}
\hline
$C$ &$2.6922 \pm 0.0061$ & $1.3437 \pm 0.0019$ & &$5.356 \pm 0.024$\\
$K_1$  &$2.696 \pm 0.012$ & $1.3687 \pm 0.0053$ & &$5.221 \pm 0.023$\\
$K_2$  &$0.1754 \pm 0.0016$ & $0.09212 \pm 0.00037$ & &$0.33392 \pm 0.00043$\\
$K_3$  &$3.839 \pm 0.037\cdot 10^{-4}$ &$2.490 \pm 0.010\cdot 10^{-4}$ & &$6.279 \pm 0.024\cdot 10^{-4}$\\
\hline
$i$ & \multicolumn{4}{c|}{$\widetilde{V}^{tot}_i$ } \\
\hline
$C$ &$2.6922 \pm 0.0061$ &  &$1.3474 \pm 0.0039$ &$5.356 \pm 0.024$\\
$K_1$ &$2.696 \pm 0.012$ &  &$1.3684 \pm 0.0030$ &$5.221 \pm 0.024$\\
$K_2$  &$0.1754 \pm 0.0016$ &  &$0.09266 \pm 0.00016$ &$0.33393 \pm 0.00046$\\
$K_3$  &$3.842 \pm 0.033\cdot 10^{-4}$ & & $2.498 \pm 0.012\cdot 10^{-4}$ &$6.276 \pm 0.031\cdot 10^{-4}$\\
\hline
\end{tabular}
\caption{\label{tabCompareMB} Estimates of variances of the final MC estimators of the sensitivity indices of conditional expectation
 using different schemes and the SB model.}
\end{table}

\begin{figure}[h]%
  \centering
  \subfloat[]{\includegraphics[width=0.45\textwidth]{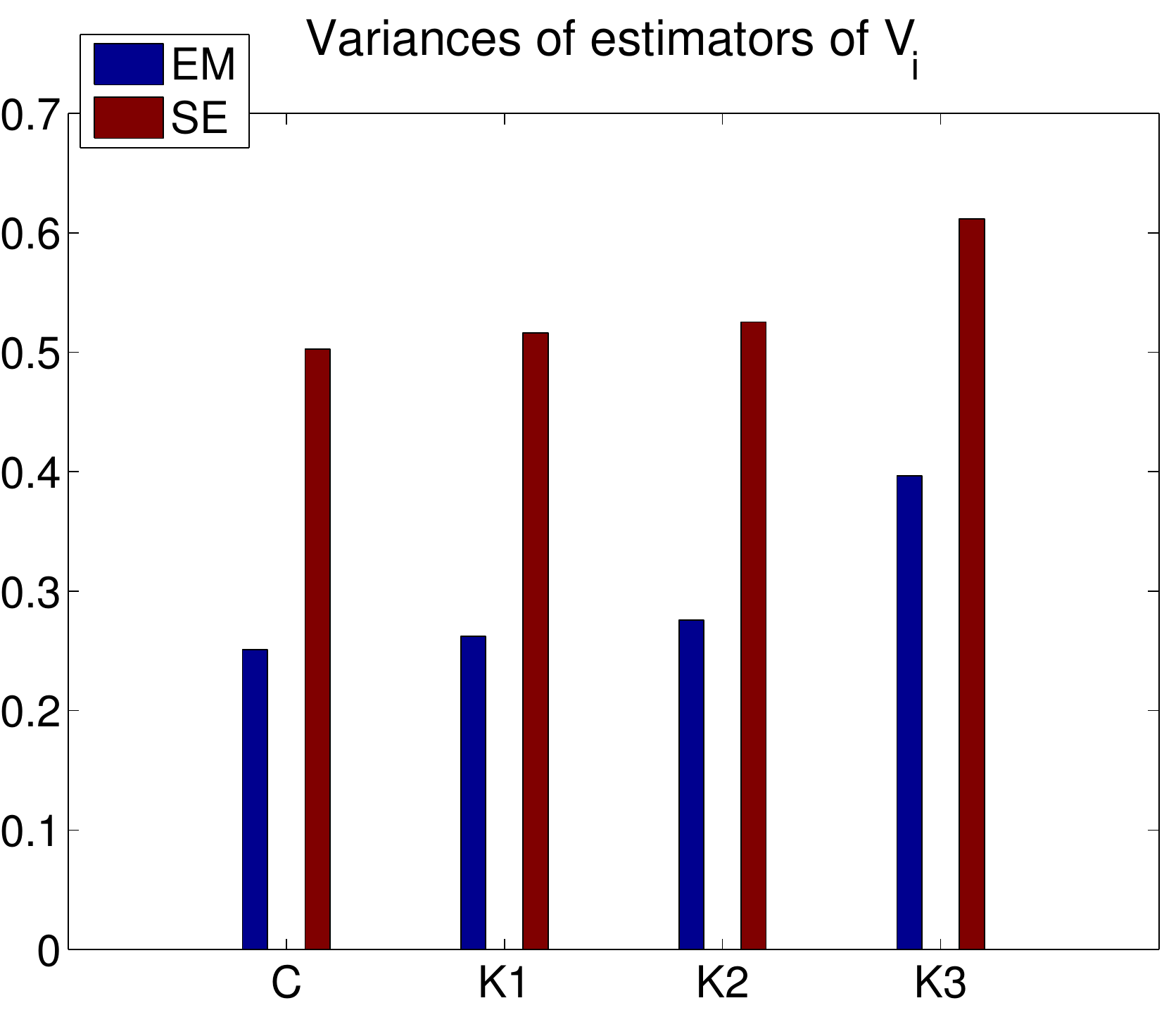}}
\qquad  
\subfloat[]{\includegraphics[width=0.45\textwidth]{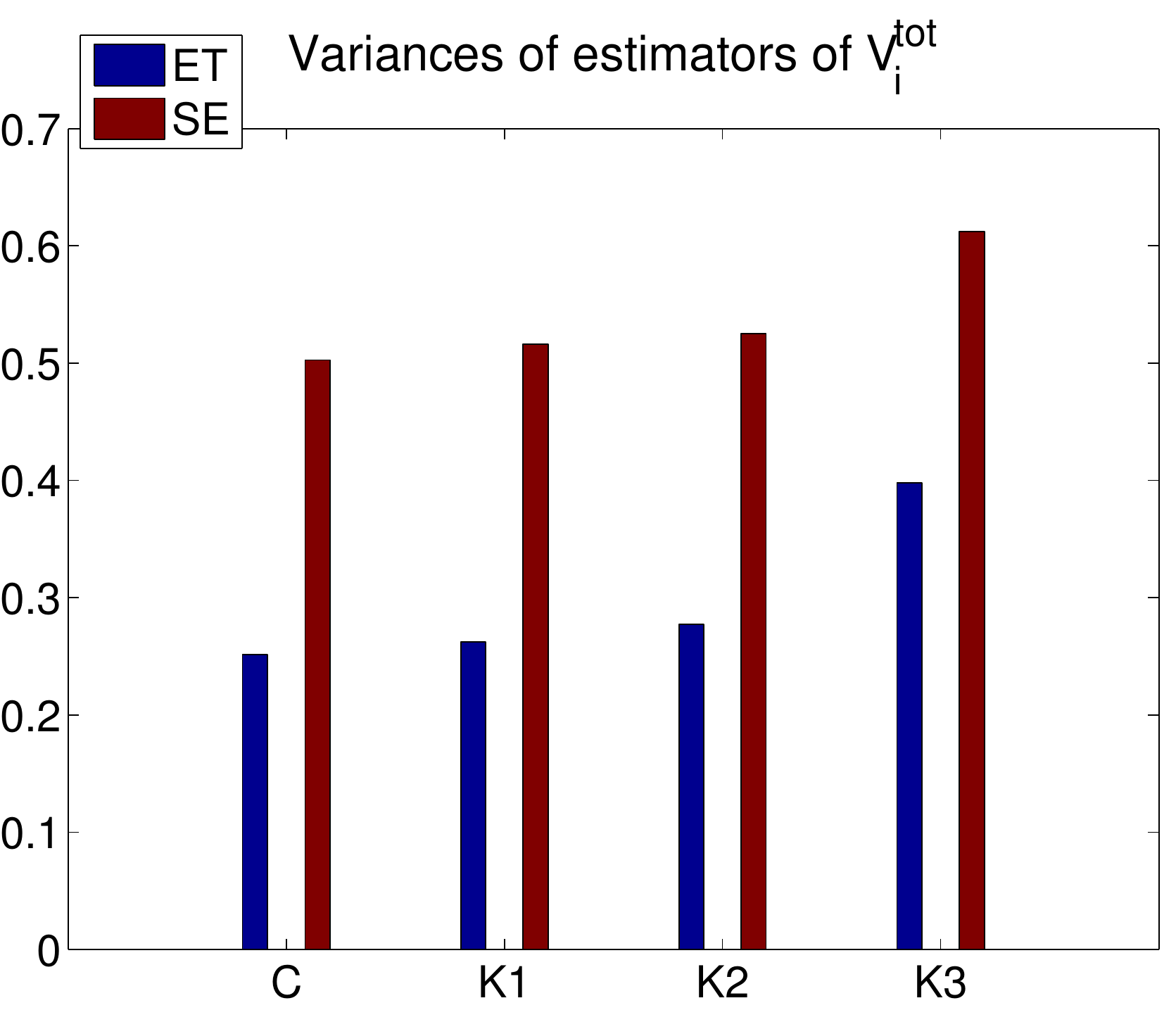}}
\caption{Chart (a) shows the ratios of estimates of variances of estimators of the main sensitivity indices 
given by schemes $EM$ and $SE$ to the estimate of variance of such estimator given by scheme $SVar$ for the SB model output,
and chart (b) of estimates of variances of estimators of the total sensitivity indices given by schemes $ET$ and $SE$ 
to such estimate for scheme $SVar$. See Section \ref{secSB} for details.}
\label{barSBV2}
\end{figure}

\begin{table}[h]
\begin{tabular}{|l|c|c|c|c|}
\hline
{\multirow{2}{*}{i}}
 & \multicolumn{4}{c|}{$cE_i$ } \\
\cline{2-5}
 & $P1E$ & $C1E$ & $C2E$ & $SVar$ \\
\hline
$ C $ &$13.414 \pm 0.014$ &$0.2039 \pm 0.0073$ & $0.4107 \pm 0.0015$ & $4.473 \pm 0.065$\\
$ K1 $ &$13.417 \pm 0.014$ &$0.1913 \pm 0.0068$ & $0.4118 \pm 0.0015$ & $4.366 \pm 0.060$\\
$ K2 $ &$13.456 \pm 0.014$ &$0.2158 \pm 0.0079$ & $0.3635 \pm 0.0014$ & $1.318 \pm 0.022$\\
$ K3 $ &$13.436 \pm 0.015$ &$0.2198 \pm 0.0075$ & $0.3491 \pm 0.0013$ & $0.3232 \pm 0.0057$\\
$Ave$ &$2.1403 \pm 0.0033\cdot 10^{-1}$ &$0.1930 \pm 0.0064$ & $2.5423 \pm 0.0055\cdot 10^{-1}$ & $3.851 \pm 0.034$\\
$\Sigma E$ &$53.936 \pm 0.033$ &$1.024 \pm 0.016$ & $1.7892 \pm 0.0043$ & $14.33 \pm 0.11$\\
\hline
{\multirow{2}{*}{i}}
& \multicolumn{4}{c|}{$cVar_i$ } \\
\cline{2-5}
& $P1Var$  & $C1Var$ & $C2Var$ & $SVar$ \\
\hline
$ C $ &$232.2 \pm 1.3$ &$4.60 \pm 0.15\cdot 10^{4}$ &$279.1 \pm 2.0$ &$1.308 \pm 0.038$\\
$ K1 $ &$145.06 \pm 0.94$ &$2.753 \pm 0.095$ &$137.8 \pm 1.2$ &$38.1 \pm 1.1$\\
$ K2 $ &$144.96 \pm 0.86$ &$0.669 \pm 0.023$ &$134.1 \pm 1.1$ &$13.88 \pm 0.43$\\
$ K3 $ &$144.48 \pm 0.91$ &$2.820 \pm 0.092\cdot 10^{-2}$ &$133.7 \pm 1.1$ &$2.443 \pm 0.076$\\
$AV$ &$92.01 \pm 0.44$ &$90.1 \pm 2.8$ &$116.94 \pm 0.95$ &$184.6 \pm 3.0$\\
$\Sigma V$ &$758.7 \pm 3.8$ &$4.61 \pm 0.15\cdot 10^{4}$ &$801.7 \pm 4.5$ &$240.3 \pm 3.4$\\
\hline
\end{tabular}
\caption{\label{tabCovSB} Estimates of variances of the final MC estimators of orthogonal projection coefficients 
of conditional expectation and conditional 
variance for the SB model as explained in Section \ref{secSB}. $AV$ is an abbreviation for $AveVar$, and $\Sigma V$ for $\Sigma Var$.}
\end{table}

\begin{figure}[h]%
\centering
\subfloat[]{\includegraphics[width=0.45\textwidth]{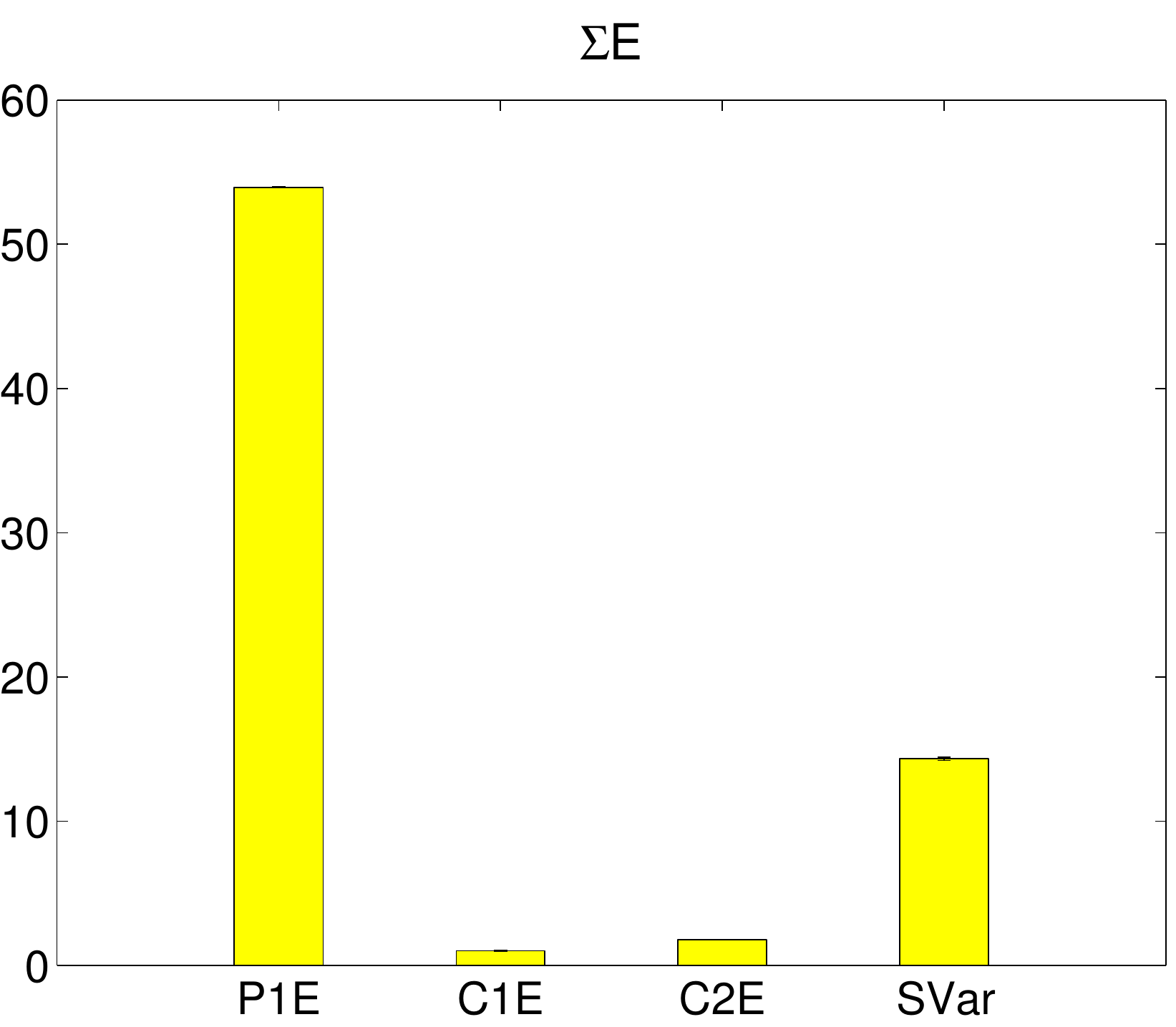}}
\qquad  
\subfloat[]{\includegraphics[width=0.45\textwidth]{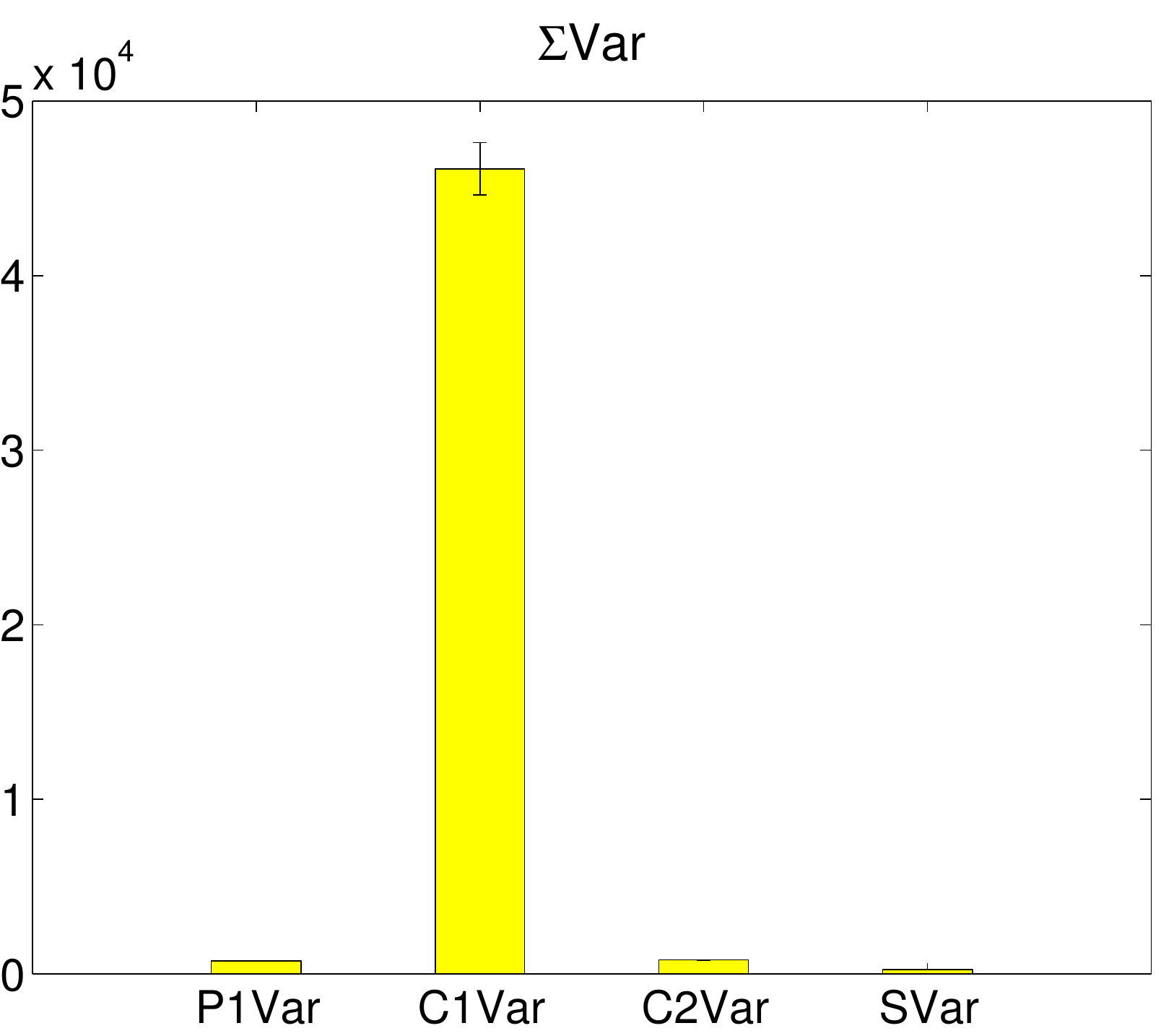}} 
  \caption{Estimates for the SB model of quantities $\Sigma E$ in chart (a) and $\Sigma Var$ in chart (b) for different 
estimation schemes 
  of orthogonal projection coefficients of conditional expectation in (a) and variance in (b) 
for the computations described in Section \ref{secSB}.}
\label{figSBBar}
\end{figure}


 \begin{table}[h]
 \begin{tabular}{|l|c|c|}
 \hline
 MR& Ave & AveVar \\
 \hline
 SB & $230.004 \pm 0.020$ & $169.95 \pm 0.13$ \\
 MBMD& $10.0298 \pm 0.0019$ & $7.1035 \pm 0.0051$ \\
 GTS & $30.244 \pm 0.021$ & $370.05 \pm 0.40$\\
 \hline
 \end{tabular}
 \caption{\label{tabDisp} Estimates of the means and average variances 
of outputs of the SB and MBMD models computed in a one-million-step and of the 
GTS model in a $250000$-step MC procedure using scheme $SVar$ and the RTC algorithm.}
 \end{table}

\subsection{GTS model}\label{secGTS} 
Let us consider the GTS model and its output from Section \ref{GTSPrev}. 
We performed a $250000$ step MC procedure using scheme $SVar$ and the RTC method. 
The estimates of different coefficients and sensitivity indices obtained in this procedure are given in Table 
\ref{tabGTS} and figures \ref{pieGTS} and \ref{barGTSVVtot}. 
The estimates of mean and average variance of the output from the  procedure are given in Table \ref{tabDisp}. 
Note that the sum of Sobol's interaction indices, equal to the proportion of total arc lenght of each pie chart 
occupied by the empty sector in Figure \ref{pieGTS} (see also formula \ref{sumSob}), 
is much higher for the conditional variance than the conditional expectation. From Table \ref{tabGTS} and Figure \ref{barGTSVVtot}
we can also see that the total sensitivity indices of conditional variance are significantly higher than the main sensitivity 
indices, especially for the parameters $\gamma$ and $\alpha_2$,
and that the order of the parameters with respect to the total indices of conditional variance is different than 
with respect to its main indices. 

We performed a $2.5$-million-step MC procedure using scheme $SErrVar$ to estimate the mean squared error of approximation of 
conditional expectation and variance of the output using linear combinations of centered parameters and constant one as in the previous section, 
taking as coefficients the estimates of $bE_i$ and $bVar_i$ from Table \ref{tabGTSComp} and estimates of mean and mean variance 
from Table \ref{tabDisp}. 
We obtained estimates of error for conditional expectation $1.779 \pm 0.100$ and for variance $2.27 \pm 0.14\cdot 10^{3}$, 
both being significantly higher than zero 
and not significantly different from the estimates of squares of 
the best possible linear approximation errors, equal to the values of 
$DNE$ and $DNVar$ given in Table \ref{tabGTS}. 
In Table \ref{tabGTSComp} we present estimates of variances of the final MC estimators of the procedures using the RTC and the GD methods 
and $1000$ steps of scheme $SVar$, $2000$ of $SE$, and $4000$ of $EM$ and $ET$, so that the variances are computed for the same number of 
process simulations used by the schemes. The mean estimates of variances for each method 
were computed analogously as in the previous section, except that 
fifty rather than five runs of each procedure were carried out to compute the means and standard deviations. 
Note that the estimates of variances of estimators from scheme $SVar$ for estimation 
of some main sensitivity indices in Table \ref{tabGTSComp} are significantly lower than these of the subschemes of scheme $EM$ introduced in 
\cite{badowski2011} and analogously for the 
total sensitivity indices and scheme $ET$. For instance the estimate of variance of the total sensitivity index with 
respect to the parameter $\beta$ computed using scheme $SVar$ is about $2.45$ times lower than the one from scheme $ET$, 
both using the GD method, which is not far from the theoretical bound of $3$ 
corresponding to equality in the lhs of relation (\ref{compVkTotVar}). 
From Table \ref{tabGTSComp} we can also see that  
the estimates of variances of estimators from scheme $SVar$ are lower for the RTC than the GD method for all the main and total indices of 
the conditional 
expectation. They are even over $4$ times lower for the total and main sensitivity index with respect to the parameter $\beta$. 

We carried out a numerical experiment comparing the variances of estimation schemes for orthogonal projection coefficients 
which was analogous as in the previous 
section, except that here for schemes $C2E$, $C2Var$, and $SVar$ we tested the GD and RTC methods separately. 
The results are presented in Table \ref{tabCovGTS} and values of $\Sigma E$ and $\Sigma Var$ also in Figure \ref{figGTSBar}. 
From Table \ref{tabCovGTS} 
we can see that the estimates of variances of orthogonal projection coefficients onto 
normalized centred parameters computed with the use of schemes $C2E$ and 
$SVar$ for conditional expectation and $C2Var$ 
for conditional variance are significantly 
lower for the RTC than the GD method for most coefficients. On the other hand, these variances are higher 
for the RTC than the GD method for all the coefficients of projections onto constant one (that is the averages $AveE$ and $AveVar$), 
for these schemes. 
From Table \ref{tabCovGTS} and Figure \ref{figGTSBar} 
we can also see that $\Sigma E$ is similarly as for the SB model lowest for the scheme $C1E$, 
followed by scheme $C2E$ using RTC and then GD methods, but in contrast to the SB model now next comes 
scheme $P1E$ and then  
$SVar$ using RTC and GD methods. As opposed to the SB model, for the GTS model $\Sigma Var$ is lowest for 
scheme $C2Var$ for the RTC method and for scheme $P1Var$, followed by $C2Var$ for the GD method, $SVar$ for the RTC and GD methods, 
and finally $C1Var$. 

Let us now illustrate the theory from Section \ref{secInterv}, using notations from there. 
Let the distribution of vector $Y$ be as of the parameter vector of the GTS model defined above and let the new parameter vector $X$ have
distributions of all coordinates as in $Y$, except for the $i$th coordinate, for certain
 $i \in I_4$, which has distribution U($0.8v_i, 1v_i$) for $v_i$ equal to the fixed value of that parameter 
in \cite{Rathinam_2010}. 
We have $\mu_Y(B_X) = \frac{1}{2}$ and $P(A)=1$. 
Using inequality (\ref{probChange})
for a perturbation $\Delta_i = 0.2v_i$ of only the $i$th parameter and values of $DNE_i$ and $bE_i$ 
from Table \ref{tabGTS}, we receive an estimate of the lower bound on the probability that the effect of 
this perturbation on the mean number of particles has the same sign as $bE_i$, equal to $95\%$, $84\%$, $89\%$, and $88\%$,
for the consecutive $i \in I_4$. 
Let now $Y$ and $X$ both have distributions as the parameter vector of the GTS model and consider a perturbation 
$\Delta_i = 0.1v_i$ only of the $i$th parameter. We now have $\mu_Y(B_X)=1, P(A) = \frac{3}{4}$, 
and the estimates of bounds on the probabilities as above are equal to $64\%$, $43\%$, $52\%$, and $50\%$, for 
the consecutive $i \in I_4$. 


\begin{figure}[h]%
  \centering
  \subfloat[]{\includegraphics[width=0.45\textwidth]{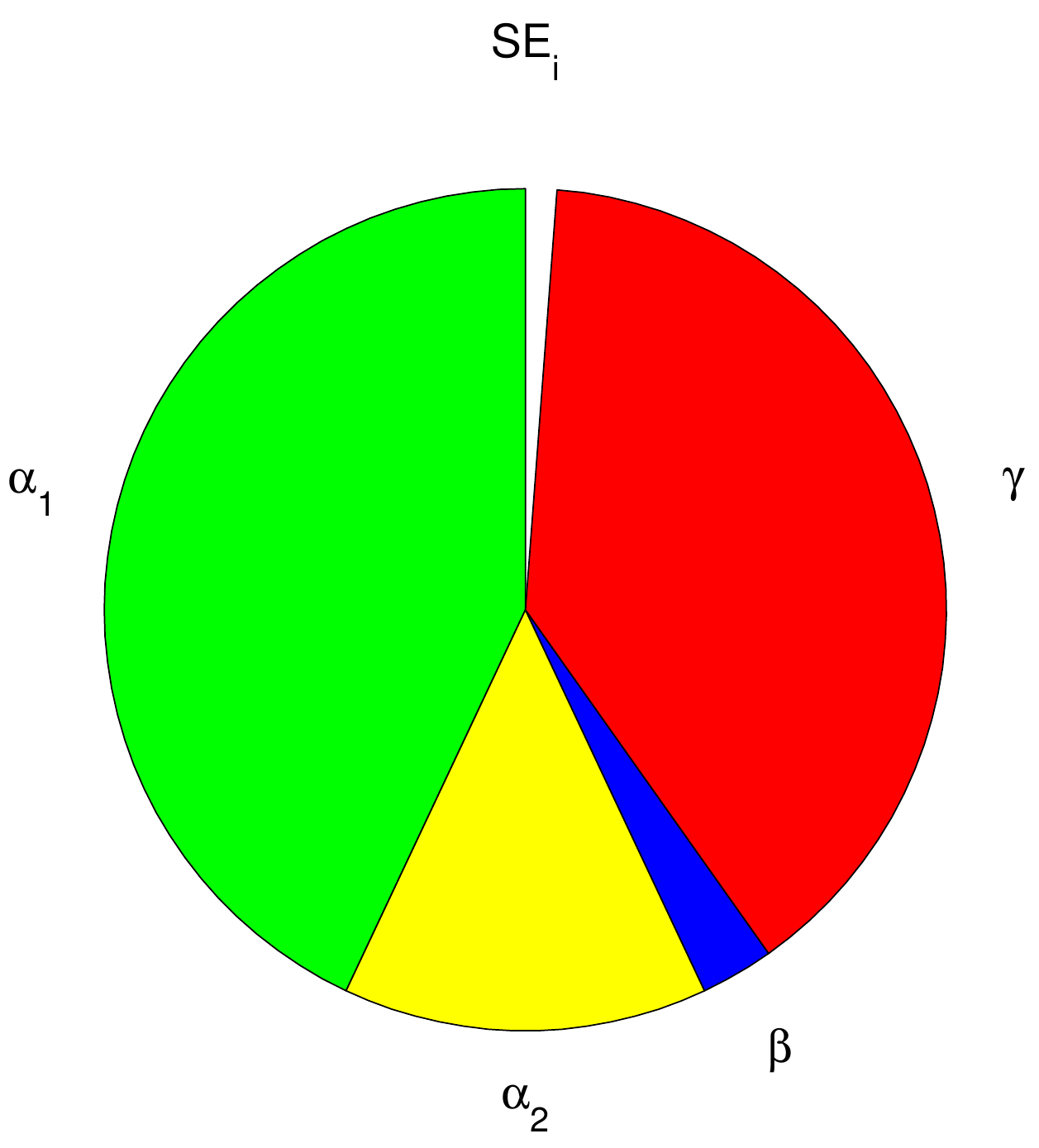}}
\qquad  
\subfloat[]{\includegraphics[width=0.45\textwidth]{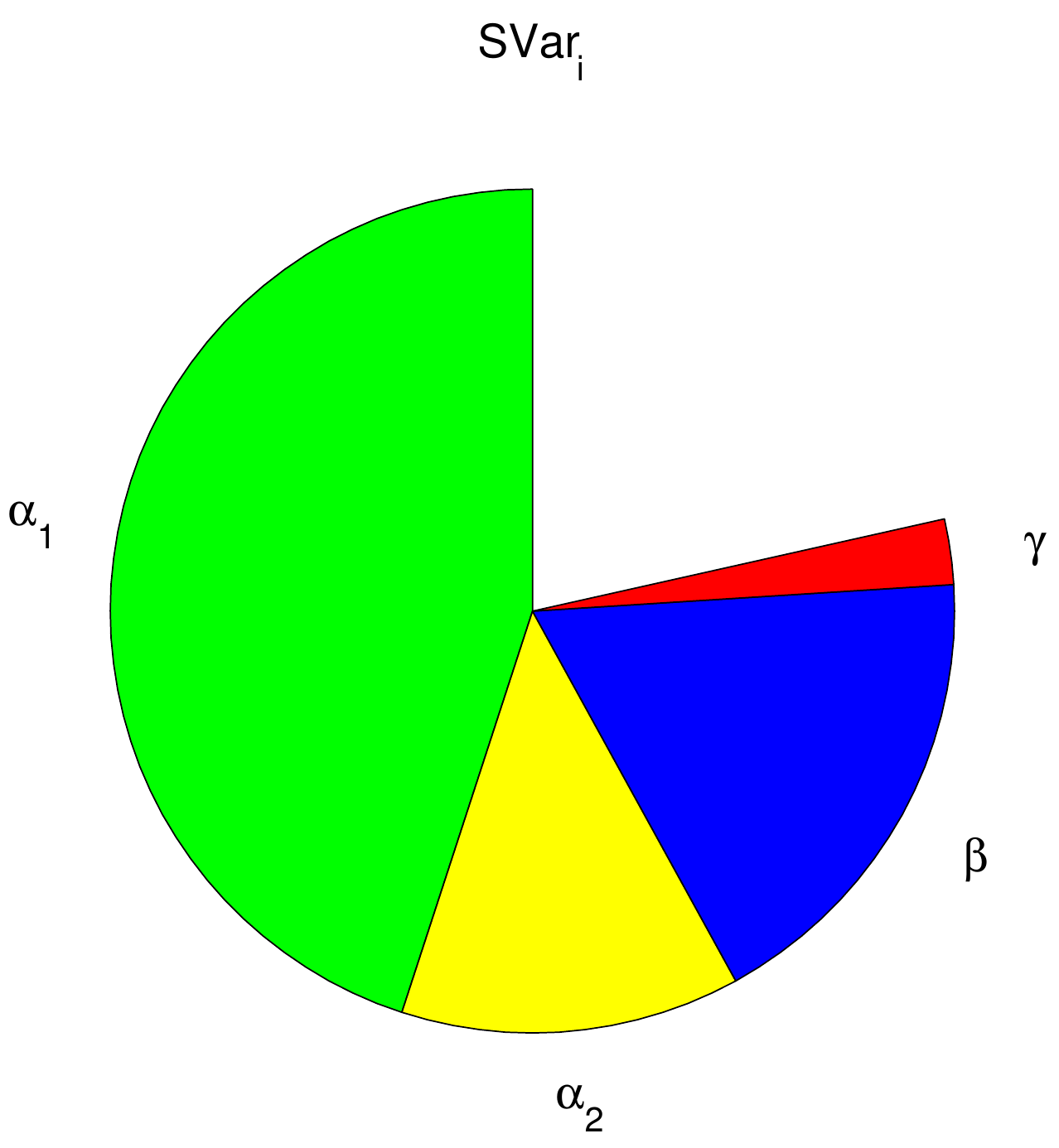}} 
\caption{Pie charts analogous as in Figure \ref{pieSB} but for the GTS model.}
\label{pieGTS}
\end{figure}

\begin{figure}[h]%
  \centering
  \subfloat[]{\includegraphics[width=0.45\textwidth]{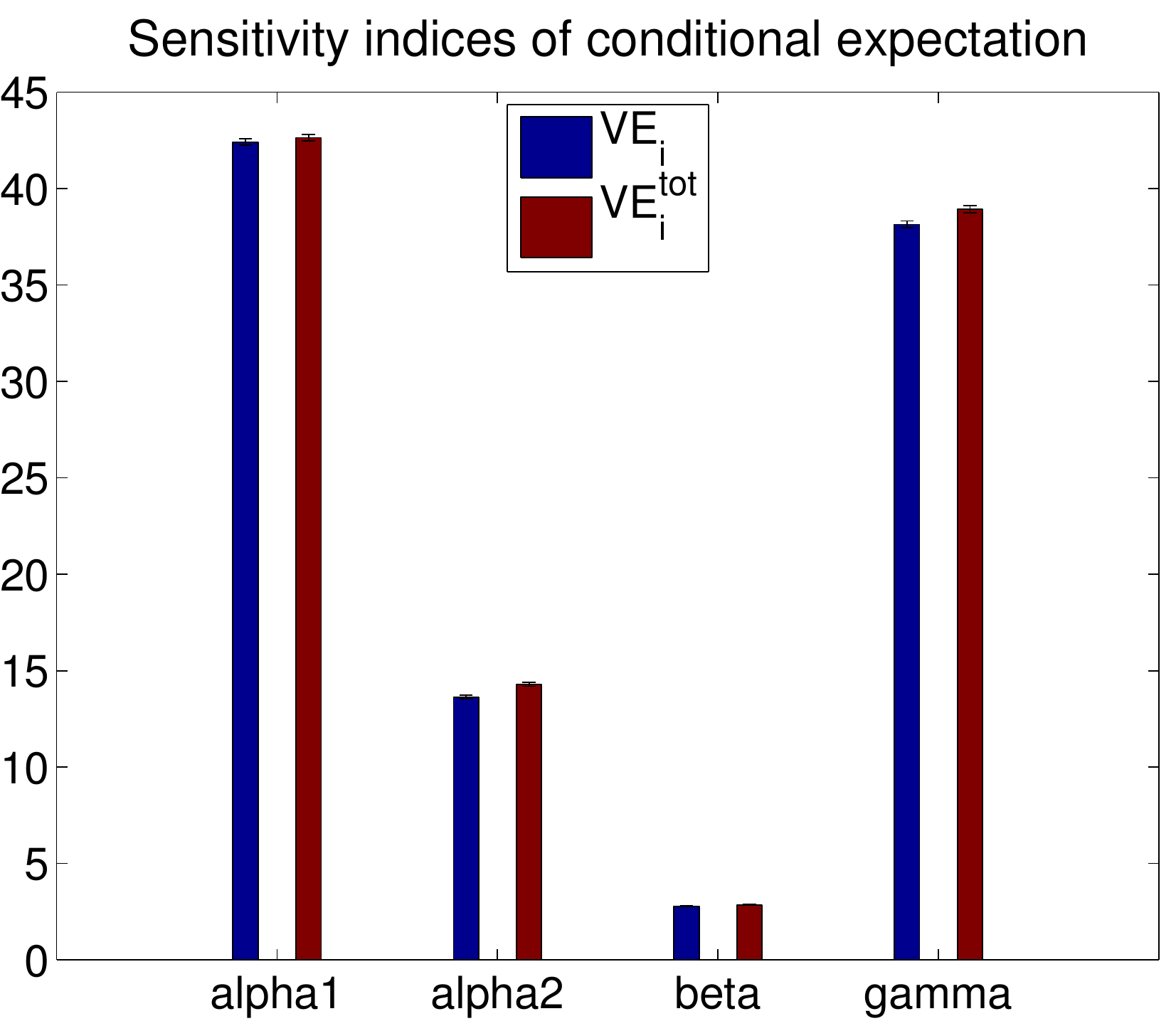}}
\qquad  
\subfloat[]{\includegraphics[width=0.47\textwidth]{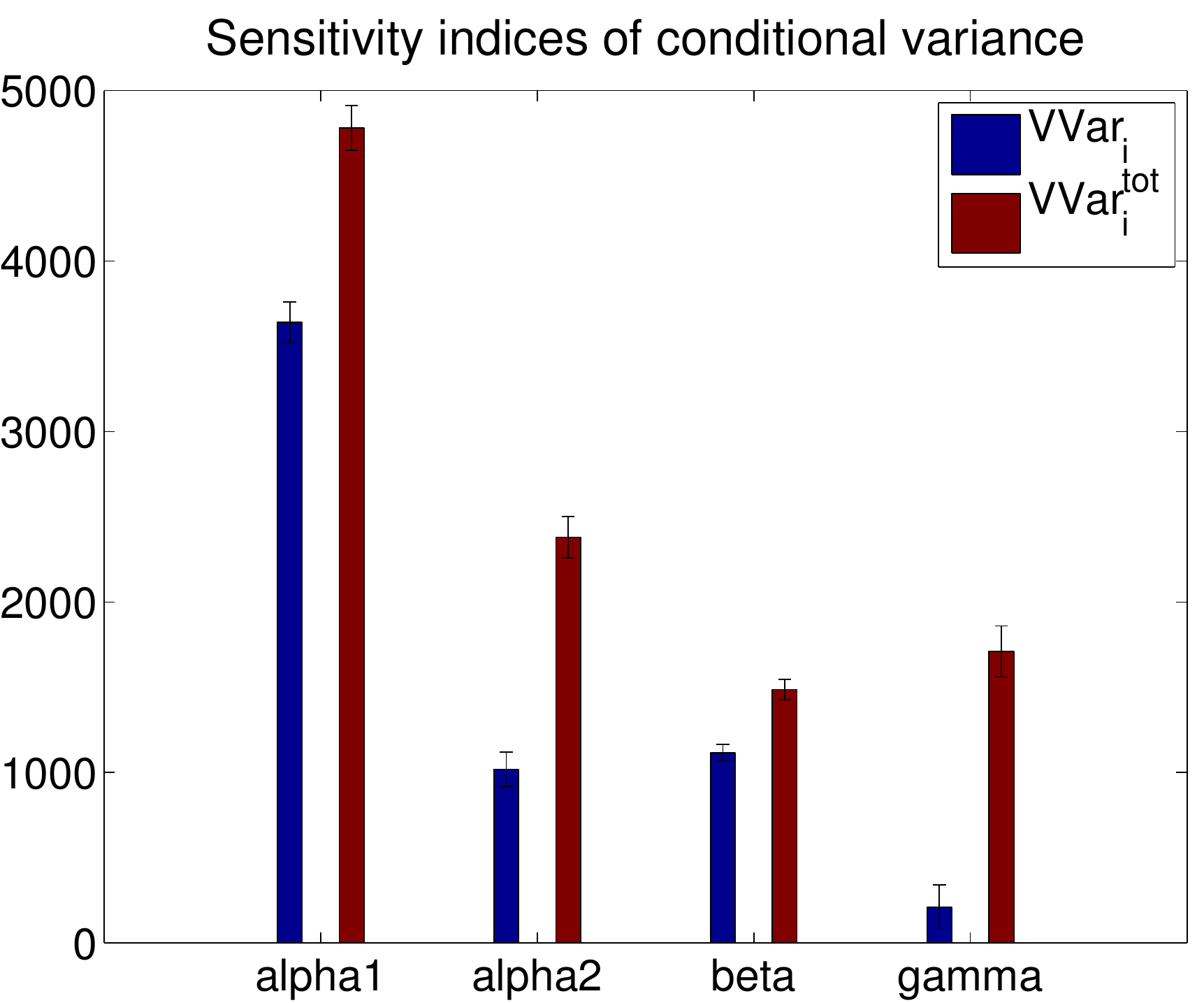}} 
  \caption{Estimates of the main and total sensitivity indices of conditional expectation in chart (a) and variance in chart (b)
of the GTS model output as discussed in Section \ref{secGTS}.}
\label{barGTSVVtot}
\end{figure}

\begin{table}[h]
\begin{tabular}{|l|c|c|c|c|c|c|}
\hline
$i $ & $\widetilde{V}_i$ & $\widetilde{V}_i^{tot} $ &$\widetilde{S}_i$ &$\widetilde{S}_i^{tot}$& $bE_i$ & $DNE_i $\\ 
\hline
$\alpha_1 $& $42.41 \pm 0.17$ & $42.64 \pm 0.17$ &$ 0.43 $&$ 0.44$& $1.1255 \pm 0.0031$ & $0.86 \pm 0.25$ \\  
$\alpha_2 $& $13.628 \pm 0.097$ & $14.289 \pm 0.100$ &$ 0.14 $&$ 0.15$& $-1.9959 \pm 0.0067$ & $0.81 \pm 0.13$ \\  
$\beta $& $2.780 \pm 0.034$ & $2.859 \pm 0.035$ &$ 0.028 $&$ 0.029$& $-5.738 \pm 0.028$ & $0.116 \pm 0.047$ \\  
$\gamma $& $38.14 \pm 0.18$ & $38.92 \pm 0.18$ &$ 0.39 $&$ 0.4$& $52.90 \pm 0.16$ & $1.68 \pm 0.26$ \\  
\hline
$i $ & $V_i$ & $V_i^{tot} $ & $S_i$ & $S_i^{tot}$ &$ DNE $& $1.45 \pm 0.13$\\
\hline
$P $& $97.73 \pm 0.26$ & $248.77 \pm 0.36$ &$ 0.21 $&$ 0.53$&  \multicolumn{2}{c}{} \\
$R $& $219.01 \pm 0.42$ & $370.05 \pm 0.40$ &$ 0.47 $&$ 0.79$&  \multicolumn{2}{c}{} \\
$P, R$& $467.79 \pm 0.44$ & $467.79 \pm 0.44$ &$ 1 $&$ 1$ & \multicolumn{2}{c}{} \\
\hline
$i$ &  $VVar_i$ & $VVar^{tot}_i$ & $SVar_i$ &$SVar_i^{tot}$& $bVar_i$ & $DNVar_i$\\
\hline
$\alpha_1 $& $3.64 \pm 0.12\cdot 10^{3}$ & $4.78 \pm 0.13\cdot 10^{3}$ &$ 0.45 $&$ 0.59$& $10.385 \pm 0.058$ & $1.25 \pm 0.15\cdot 10^{3}$ \\  
$\alpha_2 $& $1.02 \pm 0.10\cdot 10^{3}$ & $2.38 \pm 0.12\cdot 10^{3}$ &$ 0.13 $&$ 0.29$& $17.88 \pm 0.15$ & $1.29 \pm 0.13\cdot 10^{3}$ \\  
$\beta $& $1116 \pm 50$ & $1487 \pm 61$ &$ 0.14 $&$ 0.18$& $109.98 \pm 0.71$ & $467 \pm 67$ \\  
$\gamma $& $2.1 \pm 1.3\cdot 10^{2}$ & $1.71 \pm 0.15\cdot 10^{3}$ &$ 0.025 $&$ 0.21$& $-33.3 \pm 3.0$ & $1.75 \pm 0.18\cdot 10^{3}$ \\  
\cline{6-7}
$P$& $8.16 \pm 0.16\cdot 10^{3}$ & $8.16 \pm 0.16\cdot 10^{3}$ &$ 1 $&$ 1$& $DNVar $ & $2.63 \pm 0.10\cdot 10^{3}$ \\
\hline
\end{tabular}
\caption{\label{tabGTS} Estimates of different sensitivity indices and coefficients for the GTS model computed in a 
$250000$-step MC procedure using scheme $SVar$.}
\end{table}

\begin{table}[h]
\resizebox{17cm}{!} {
\begin{tabular}{|l|c|c|c|c|c|c|c|c|}
\hline
{\multirow{3}{*}{i}}& \multicolumn{2}{c|}{$SE$} & \multicolumn{2}{c|}{$EM$} 
& \multicolumn{2}{c|}{$ET$} & \multicolumn{2}{c|}{$SVar$} \\
\cline{2-9}
& GD & RTC & GD & RTC & GD & RTC & GD & RTC \\
\cline{2-9}
 & \multicolumn{8}{c|}{$\widetilde{V}_i$} \\
\hline
$ \alpha_1$ &$16.05 \pm 0.14$ &$7.697 \pm 0.086$ &$13.820 \pm 0.089$ &$6.307 \pm 0.053$ & & &$11.56 \pm 0.14$ &$7.086 \pm 0.087$ \\
$ \alpha_2$ &$5.350 \pm 0.068$ &$4.067 \pm 0.050$ &$5.098 \pm 0.053$ &$3.636 \pm 0.034$ & & &$2.876 \pm 0.051$ &$2.428 \pm 0.039$ \\
$ \beta$ &$3.118 \pm 0.058$ &$0.640 \pm 0.022$ &$2.960 \pm 0.043$ &$0.586 \pm 0.016$ & & &$1.239 \pm 0.021$ &$0.283 \pm 0.013$ \\
$ \gamma$ &$11.41 \pm 0.11$ &$9.950 \pm 0.098$ &$9.733 \pm 0.074$ &$7.635 \pm 0.051$ & & &$8.69 \pm 0.13$ &$8.28 \pm 0.10$ \\
 \hline
 i & \multicolumn{8}{c|}{$\widetilde{V}_i^{tot}$} \\
 \hline
$ \alpha_1 $ &$16.27 \pm 0.15$ &$7.721 \pm 0.085$ & & &$13.755 \pm 0.082$ &$6.349 \pm 0.058$ &$11.60 \pm 0.14$ &$7.110 \pm 0.090$\\
$ \alpha_2 $ &$5.720 \pm 0.075$ &$4.223 \pm 0.052$ & & &$5.463 \pm 0.054$ &$3.869 \pm 0.036$ &$3.076 \pm 0.053$ &$2.536 \pm 0.043$\\
$ \beta $ &$3.311 \pm 0.058$ &$0.658 \pm 0.025$ & & &$3.188 \pm 0.053$ &$0.582 \pm 0.019$ &$1.301 \pm 0.024$ &$0.303 \pm 0.014$\\
$ \gamma $ &$11.92 \pm 0.11$ &$10.315 \pm 0.099$ & & &$10.181 \pm 0.070$ &$7.942 \pm 0.061$ &$9.09 \pm 0.12$ &$8.53 \pm 0.11$\\
 \hline
\end{tabular}
}
\caption{\label{tabGTSComp} Estimates of variances of the final MC 
estimators of the sensitivity indices of conditional expectation for the 
GTS model, computed using the RTC and GD methods and different schemes as described in Section \ref{secGTS}.} 
\end{table}

\begin{table}[h]
\resizebox{17cm}{!} {
\begin{tabular}{|l|c|c|c|c|c|c|}
\hline
{\multirow{2}{*}{i}}
& \multicolumn{6}{c|}{$cE_i$ } \\
\cline{2-7}
&P1ERTC & C1ERTC & C2EGD & C2ERTC &  SVarGD & SVarRTC  \\
\hline
$ \alpha_1 $ &$3.4223 \pm 0.0085\cdot 10^{-1}$ &$0.1284 \pm 0.0082$ & $1.7937 \pm 0.0098\cdot 10^{-1}$ & $1.3828 \pm 0.0099\cdot 10^{-1}$ & $0.914 \pm 0.026$ & $0.762 \pm 0.019$\\
$ \alpha_2 $ &$3.4342 \pm 0.0082\cdot 10^{-1}$ &$0.1168 \pm 0.0085$ & $1.6615 \pm 0.0098\cdot 10^{-1}$ & $1.3328 \pm 0.0080\cdot 10^{-1}$ & $0.404 \pm 0.014$ & $0.402 \pm 0.011$\\
$ \beta $ &$3.4620 \pm 0.0083\cdot 10^{-1}$ &$0.1156 \pm 0.0086$ & $1.6423 \pm 0.0098\cdot 10^{-1}$ & $1.2710 \pm 0.0074\cdot 10^{-1}$ & $0.2342 \pm 0.0081$ & $0.1478 \pm 0.0051$\\
$ \gamma $ &$3.2824 \pm 0.0077\cdot 10^{-1}$ &$0.1047 \pm 0.0072$ & $1.7466 \pm 0.0094\cdot 10^{-1}$ & $1.4427 \pm 0.0072\cdot 10^{-1}$ & $0.795 \pm 0.022$ & $0.777 \pm 0.022$\\
$A $&$1.1706 \pm 0.0015\cdot 10^{-1}$ &$0.1060 \pm 0.0076$ & $1.5464 \pm 0.0037\cdot 10^{-1}$ & $1.7194 \pm 0.0031\cdot 10^{-1}$ & $1.015 \pm 0.015$ & $1.144 \pm 0.015$\\
$\Sigma E$&$1.4772 \pm 0.0025$ &$0.572 \pm 0.019$ & $0.8390 \pm 0.0025$ & $0.7149 \pm 0.0023$ & $3.362 \pm 0.036$ & $3.232 \pm 0.038$\\
\hline
{\multirow{2}{*}{i}}
& \multicolumn{6}{c|}{$cVar_i$} \\
\cline{2-7}
&P1VarRTC & C1VarRTC & C2VarGD & C2VarRTC &  SVarGD & SVarRTC \\
\hline
$ \alpha_1 $ &$186.2 \pm 1.4$ &$4.36 \pm 0.33\cdot 10^{3}$ &$210.5 \pm 2.7$ &$173.3 \pm 2.1$ &$335 \pm 11$ &$264.4 \pm 9.0$\\
$ \alpha_2 $ &$180.9 \pm 1.3$ &$438 \pm 31$ &$197.8 \pm 1.7$ &$163.7 \pm 1.7$ &$176.4 \pm 7.1$ &$196.2 \pm 6.0$\\
$ \beta $ &$181.5 \pm 1.2$ &$8.27 \pm 0.52$ &$192.9 \pm 2.2$ &$157.5 \pm 1.8$ &$154.6 \pm 5.7$ &$98.8 \pm 4.6$\\
$ \gamma $ &$177.9 \pm 1.2$ &$1.383 \pm 0.094$ &$199.7 \pm 2.3$ &$174.1 \pm 1.9$ &$262.8 \pm 8.2$ &$285.5 \pm 9.4$\\
$AV $&$117.70 \pm 0.59$ &$120.3 \pm 7.5$ &$141.94 \pm 0.82$ &$159.2 \pm 1.0$ &$340.0 \pm 4.9$ &$395.1 \pm 6.2$\\
$\Sigma V$ &$844.3 \pm 4.6$ &$4.93 \pm 0.33\cdot 10^{3}$ &$942.8 \pm 6.8$ &$827.8 \pm 6.0$ &$1269 \pm 17$ &$1240 \pm 16$\\
\hline
\end{tabular}
}
\caption{\label{tabCovGTS} Estimates of variances of the final estimators of the 
orthogonal projection coefficients of conditional expectation and 
conditional variance given the parameters using different schemes for the GTS model as described in Section \ref{secGTS}. The suffix 
RTC or GD of the scheme means that the RTC or the GD method was applied.} 
\end{table}
\begin{figure}[h]%
\centering
\subfloat[]{\includegraphics[width=0.45\textwidth]{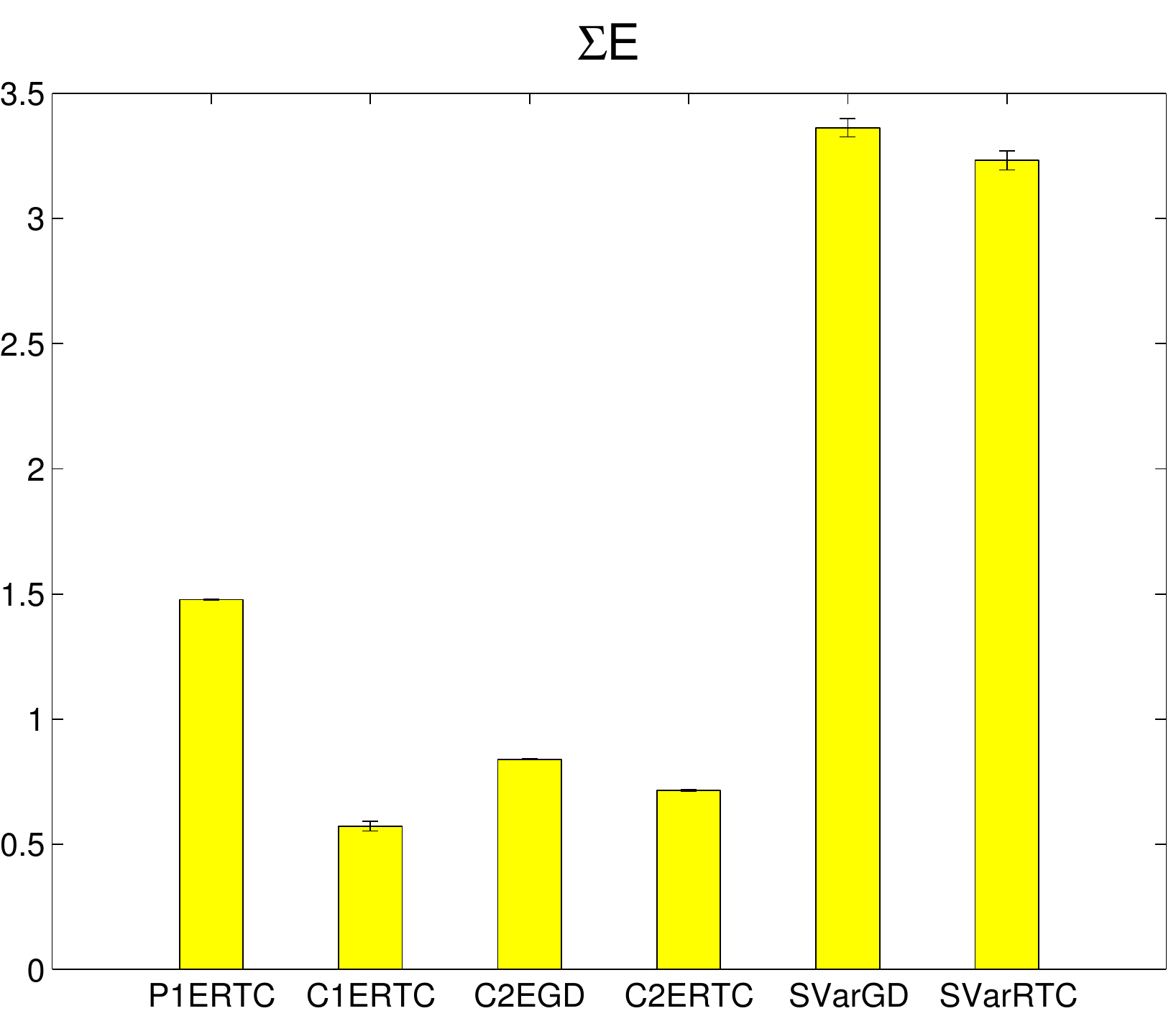}}
\qquad  
\subfloat[]{\includegraphics[width=0.45\textwidth]{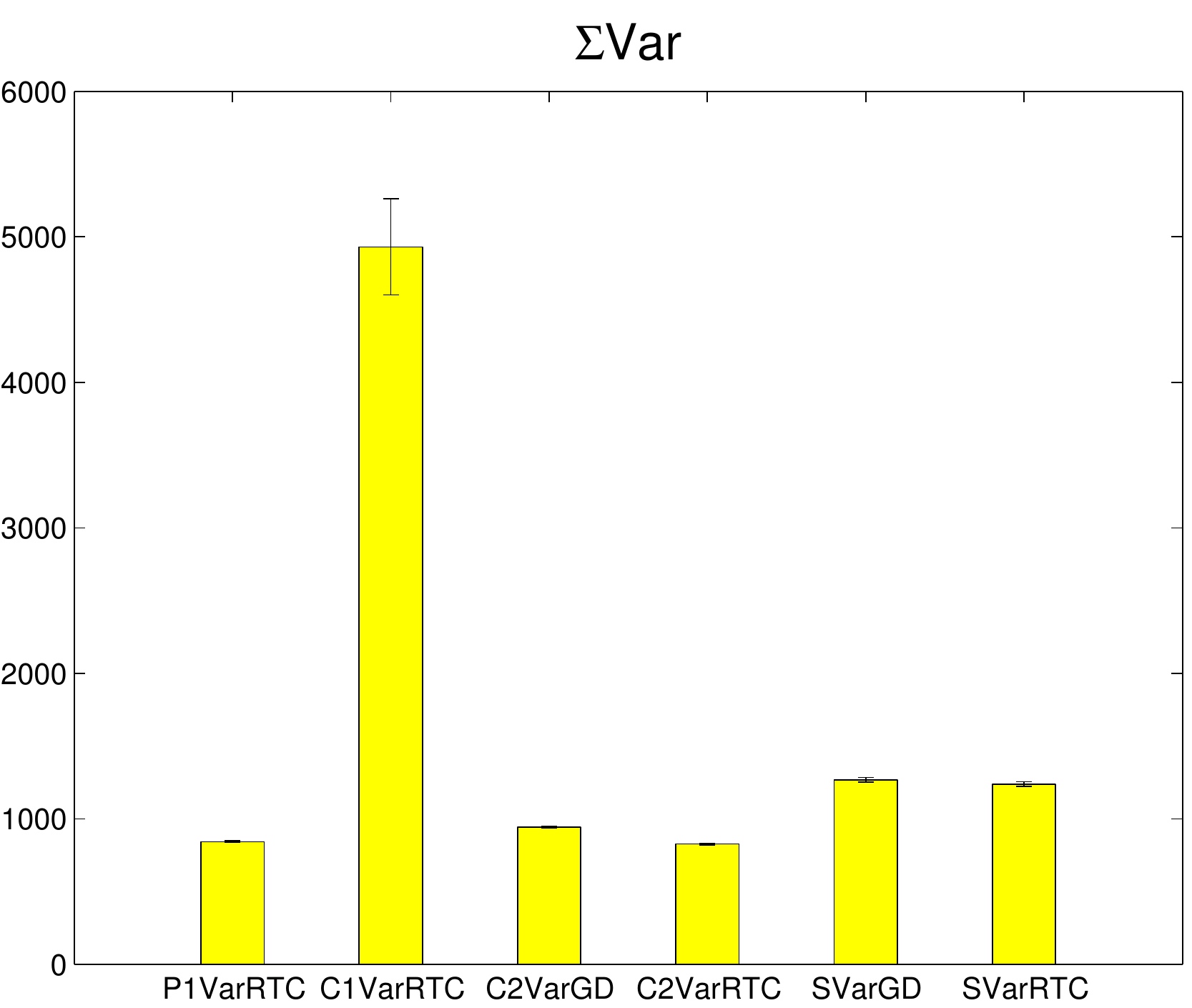}} 
  \caption{
Estimates for the GTS model of quantities $\Sigma E$ in chart (a) and $\Sigma Var$ in chart (b) for different 
estimation schemes 
  of orthogonal projection coefficients of conditional expectation in (a) and variance in (b) 
for the computations described in Section \ref{secGTS}. The suffix 
RTC or GD of the scheme means that the RTC or the GD method was applied.}
\label{figGTSBar}
\end{figure}

\subsection{MBMD model}\label{secMBMD} 
Let us finally consider the MBMD model and its output from Section \ref{MBMDPrev}. 
We performed a one-million-step MC procedure computing various indices and coefficients 
using the RTC method and scheme $SVar$. The results are presented in Table \ref{tabMBMD} and on Figure \ref{pieMBMD}. 
We carried out a ten-million-step MC procedure using scheme $SErrVar$ to estimate the mean squared error of approximation of 
conditional expectation and variance using linear combinations of centred parameters and constants and 
estimates of $bE_i$ and $bVar_i$ from Table \ref{tabMBMD} and mean and mean variance 
from Table \ref{tabDisp} as coefficients, analogously as in the previous sections. We obtained estimates of error for conditional expectation 
$3.091 \pm 0.094\cdot 10^{-2}$ and for variance $0.000 \pm 0.027$, which are 
not significantly different from estimates of the squared best theoretical errors $DNE$ and $DNVar$ in Table \ref{tabMBMD}. 
We carried out $500$ independent runs of $250$-step MC procedures using scheme $SVar$ 
and RTC, GDI, and GDR methods described in Section \ref{MBMDPrev} to get estimates of variances of the final MC  estimators 
of the sensitivity indices of conditional variances from this scheme, 
analogously as for the indices of conditional expectations in the previous sections. The results are presented in Table \ref{tabCompMB} 
and Figure \ref{fig3SVarMBMD}. 
For all the parameters except $C$ the estimates of variances  
are lowest for the RTC method, followed by the GDI, and then the GDR 
method, while for $C$ they are lower for the GDR than the GDI method, with the RTC method still 
yielding the smallest variance. The estimate of variance of the final MC estimator of the main sensitivity index 
of conditional variance with respect to $K_{d1}$ is even about 48 times higher for 
the RTC than the GDI method. Qualitatively the same results 
were obtained for variances of estimators of total sensitivity indices using this scheme (data not shown). 

We also performed an experiment comparing the variances of estimators of orthogonal projection coefficients from 200 independent 
runs of MC procedures using different above constructions of the MBMD model 
and $100$ runs of scheme $SVar$ and procedures using schemes $C2E$ and $C2Var$ using the same number of process evaluations, 
similarly as in the previous sections. 
The results are presented in Table \ref{tabCompMBCov} and Figure \ref{figSigmaMBMD}. 
We can see that for schemes $C2E$ and $SVar$ for the coefficients 
of conditional expectation, as well as for scheme $C2Var$ for the coefficients of conditional variance, the GDR method 
yields highest variance of the estimators of coefficients of orthogonal 
projection onto normalized centered parameters and the lowest variance for 
the averages for both conditional expectations and variances. For all of the schemes, using the 
GDR method leads to highest estimates of $\Sigma E$ and $\Sigma Var$, followed by 
the RTC method, and finally by the GDI method. 
Note that for the parameters $K_{b1}$ and $K_{d1}$ the estimates of variances of estimators of the orthogonal projection coefficients 
of conditional expectation from scheme $C2E$ are statistically significantly 
higher when using the RTC than the GDI method, while the opposite sharp inequality holds for 
the estimand $Ave$, which, as discussed in Section \ref{secVarDiff}, shows that for this model 
the value of $\msd(p_1,p_2)$ defined by (\ref{errSqr}) must be higher for certain parameter values 
when using the RTC than the GD method, both with the initial order of indices.

\begin{table}[h]
\resizebox{16cm}{!} {
\begin{tabular}{|l|c|c|c|c|c|c|}
\hline
$i $ & $\widetilde{V}_i$ & $\widetilde{V}_i^{tot} $ &$\widetilde{S}_i$ &$\widetilde{S}_i^{tot}$& $bE_i$ & $DNE_i $\\ 
\hline
$C $& $2.8913 \pm 0.0038$ & $2.9183 \pm 0.0038$ &$ 0.73 $&$ 0.74$& $5.3781 \pm 0.0070\cdot 10^{-1}$ & $0.0304 \pm 0.0069$ \\  
$K_{b1} $& $1.0351 \pm 0.0026\cdot 10^{-1}$ & $1.0371 \pm 0.0026\cdot 10^{-1}$ &$ 0.026 $&$ 0.026$& $3.724 \pm 0.017$ & $5.891 \pm 380.717\cdot 10^{-6}$ \\  
$K_{d1} $& $1.0507 \pm 0.0025\cdot 10^{-1}$ & $1.1129 \pm 0.0026\cdot 10^{-1}$ &$ 0.026 $&$ 0.028$& $-37.26 \pm 0.17$ & $6.10 \pm 0.36\cdot 10^{-3}$ \\  
\hline
$i $ & $V_i$ & $V_i^{tot} $ & $S_i$ & $S_i^{tot}$ &$ DNE $& $0.0306 \pm 0.0038$\\
\hline
$P $& $3.9668 \pm 0.0049$ & $6.0810 \pm 0.0052$ &$ 0.36 $&$ 0.55$&  \multicolumn{2}{c}{} \\
$R $& $4.9893 \pm 0.0050$ & $7.1035 \pm 0.0051$ &$ 0.45 $&$ 0.64$&  \multicolumn{2}{c}{} \\
$P, R$& $11.0703 \pm 0.0071$ & $11.0703 \pm 0.0071$ &$ 1 $&$ 1$ & \multicolumn{2}{c}{} \\
\hline
$i$ &  $VVar_i$ & $VVar^{tot}_i$ & $SVar_i$ &$SVar_i^{tot}$& $bVar_i$ & $DNVar_i$\\
\hline
$C $& $0.606 \pm 0.012$ & $0.602 \pm 0.012$ &$ 0.53 $&$ 0.53$& $0.2454 \pm 0.0011$ & $-0.012 \pm 0.016$ \\  
$K_{b1} $& $0.1066 \pm 0.0051$ & $0.1075 \pm 0.0052$ &$ 0.093 $&$ 0.094$& $3.750 \pm 0.037$ & $0.0109 \pm 0.0070$ \\  
$K_{d1} $& $0.0016 \pm 0.0041$ & $-0.0023 \pm 0.0044$ &$ 0.0014 $&$ -0.002$& $-7.88 \pm 0.36$ & $-0.0064 \pm 0.0057$ \\  
\cline{6-7}
$P$& $1.147 \pm 0.028$ & $1.147 \pm 0.028$ &$ 1 $&$ 1$& $DNVar $ & $0.043 \pm 0.016$ \\
\hline
\end{tabular}
}
\caption{\label{tabMBMD} Estimates of different sensitivity indices and coefficients in the MBMD model computed in a one-million-step
MC procedure using the RTC method and scheme $SVar$.}
\end{table}

\begin{figure}[h]
  \centering
  \subfloat[]{\includegraphics[width=0.41\textwidth]{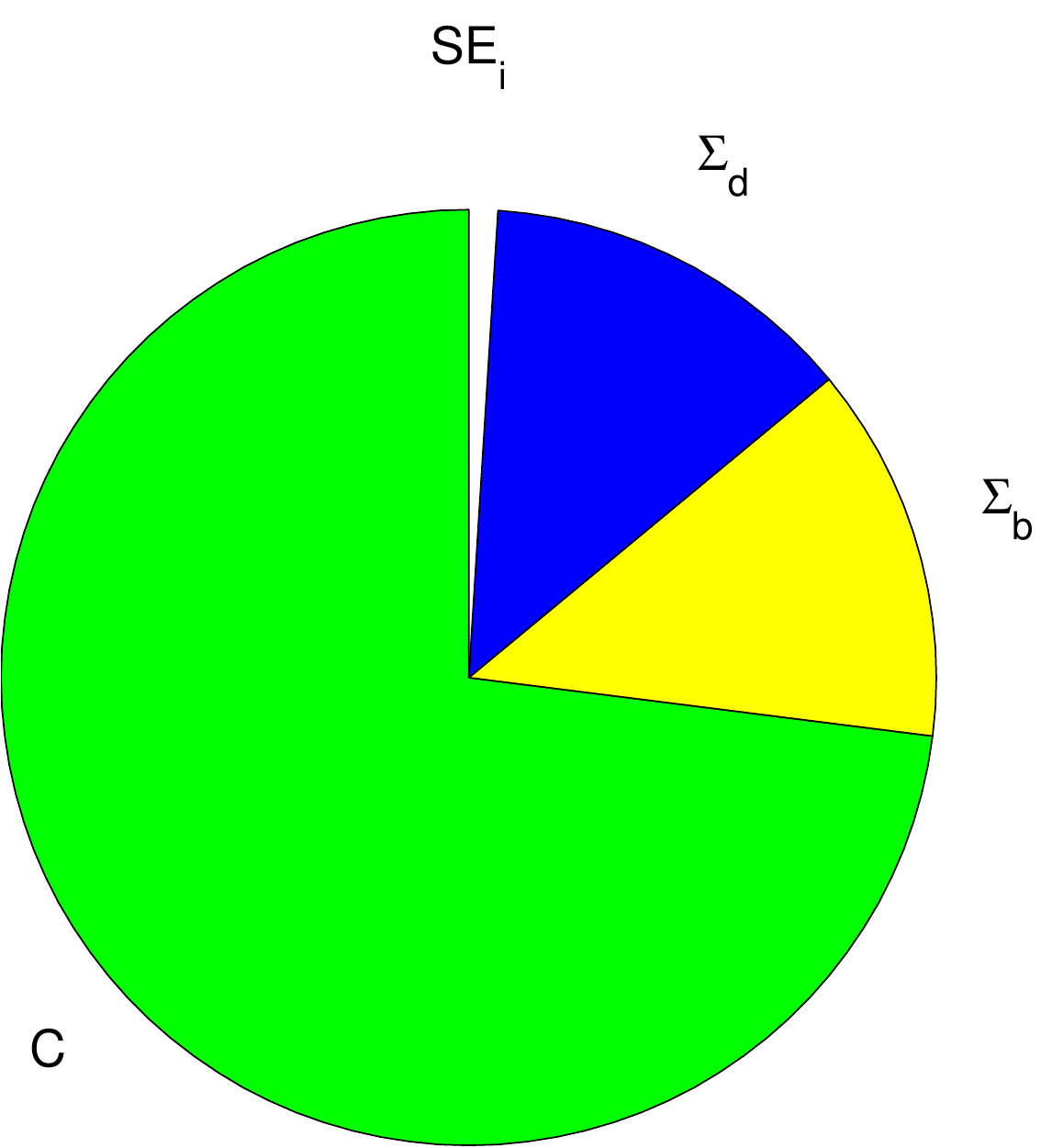}}
\qquad  
\subfloat[]{\includegraphics[width=0.46\textwidth]{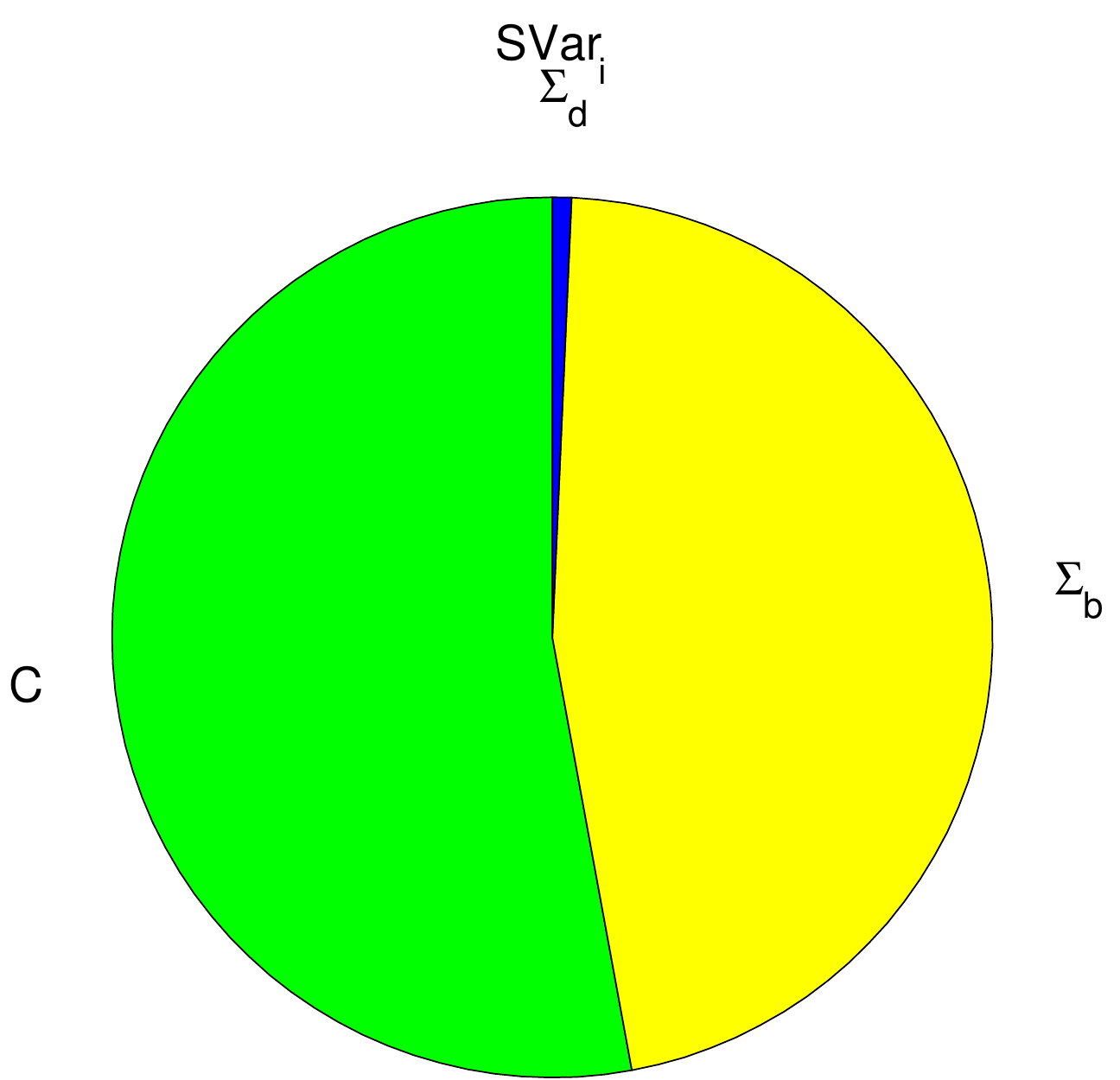}} 
  \caption{Pie charts analogous as in Figure \ref{pieSB} but for the MBMD model. The portions of total arc lengths 
occupied by segments with symbols $\Sigma_b$ and $\Sigma_d$ are equal to the sums of the 
Sobol's main sensitivity indices with respect to all parameters $K_{b,i}$, $i\in I_5$, for $\Sigma_b$ and 
$K_{b,i}$, $i\in I_5$, for $\Sigma_d$.}
\label{pieMBMD}
\end{figure}

\begin{table}[h]
\begin{tabular}{|l|c|c|c|}
\hline
{\multirow{2}{*}{i}} & GDR & GDI & RTC \\
\cline{2-4}
 & \multicolumn{3}{c|}{$VVar_i$, $SVar$ } \\
\hline
$C $ & $2.160 \pm 0.045$ & $2.523 \pm 0.078$ & $0.580 \pm 0.019$ \\ 
$K_{b1}$ & $3.573 \pm 0.084$ & $0.3063 \pm 0.0072$ & $0.1057 \pm 0.0037$ \\ 
$K_{d1} $ & $3.223 \pm 0.081$ & $0.2335 \pm 0.0077$ & $0.0669 \pm 0.0021$ \\ 
\hline
\end{tabular}
\caption{Estimates of variances of the final MC estimators of the main sensitivity indices of conditional variance 
using scheme $SVar$ and the GDR, GDI, and RTC methods as described in Section \ref{secMBMD}.}
\label{tabCompMB}
\end{table}

\begin{figure}[h]%
\includegraphics[width=0.45\textwidth]{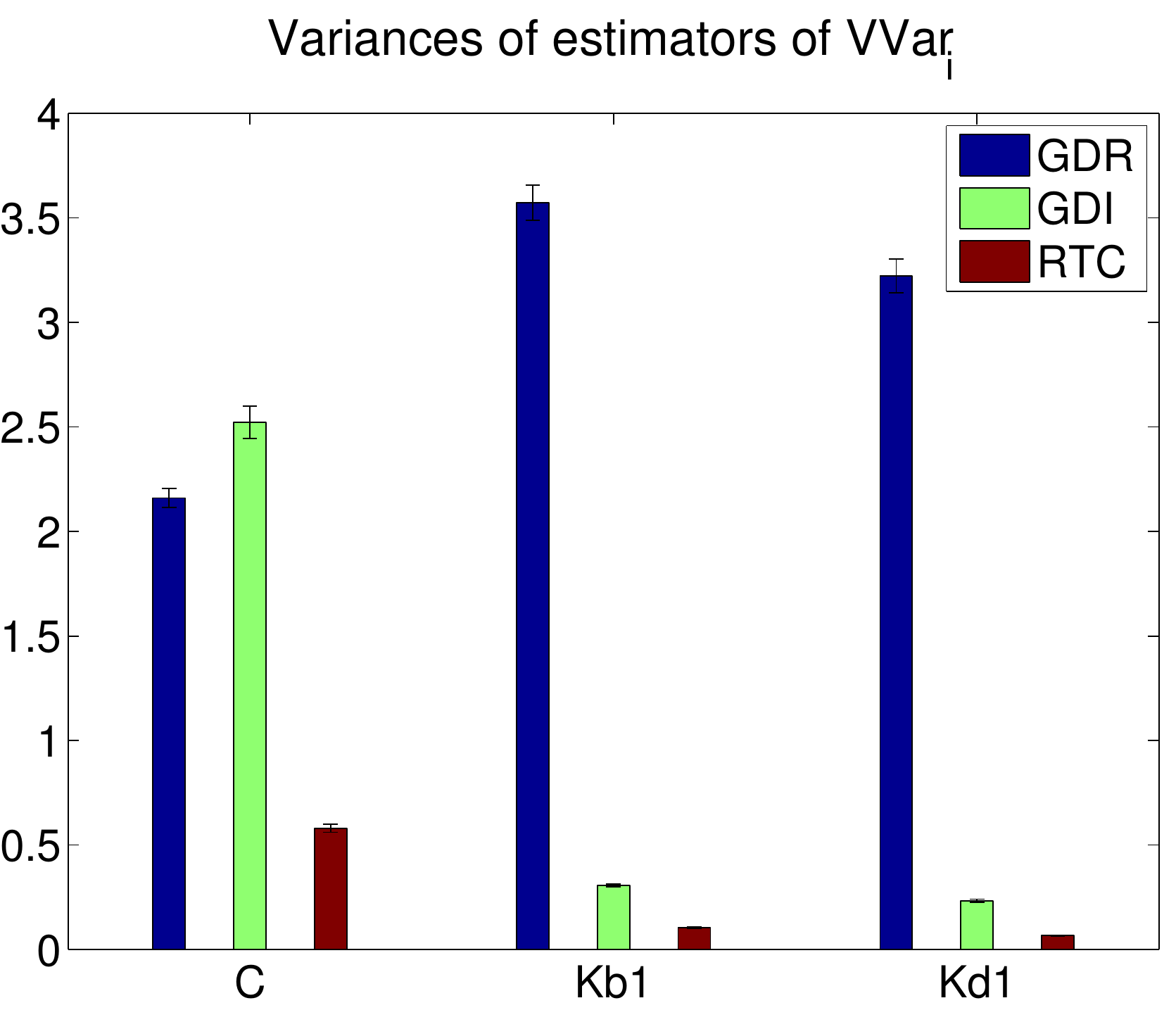}
  \caption{Chart illustrating data from Table \ref{tabCompMB}.}
\label{fig3SVarMBMD}
\end{figure}

\begin{table}[h]
\begin{tabular}{|l|c|c|c|}
\hline
{\multirow{2}{*}{i}} & GDR & GDI & RTC \\
\cline{2-4}
& \multicolumn{3}{c|}{$\wh{cE}_{i,C2E}$ } \\
\hline
$ C $& $ 2.3531 \pm 0.0012$&$1.64792 \pm 0.00088$ & $1.54710 \pm 0.00078$  \\
$ K_{b1} $ &$2.1750 \pm 0.0011$ &$1.27239 \pm 0.00069$& $1.30307 \pm 0.00067$ \\
$ K_{d1} $ &$2.1300 \pm 0.0011$ &$1.26781 \pm 0.00068$& $1.29614 \pm 0.00066$ \\
$Ave $ &$1.26760 \pm 0.00037$&$1.68266 \pm 0.00050$  & $1.67421 \pm 0.00050$\\
$\Sigma E $ &$24.8995 \pm 0.0078$&$16.0320 \pm 0.0050$ & $16.2203 \pm 0.0048$\\
\hline
i & \multicolumn{3}{c|}{$\wh{cE}_{i,SVar} $ } \\
\hline
$C$& $51.76 \pm 0.17$& $50.15 \pm 0.16$ & $48.73 \pm 0.15$\\
$K_{b1}$& $26.233 \pm 0.093$ &$20.747 \pm 0.074$& $21.669 \pm 0.076$\\
$K_{d1}$& $25.602 \pm 0.091$ & $20.655 \pm 0.070$& $21.434 \pm 0.073$\\
$Ave$& $27.023 \pm 0.052$&  $34.461 \pm 0.070$ & $34.654 \pm 0.069$\\
$\Sigma E$& $330.69 \pm 0.64$ & $291.88 \pm 0.56$ &$298.35 \pm 0.57$\\
\hline
i & \multicolumn{3}{c|}{$\wh{cVar}_{i,C2Var} $ } \\
\hline
$ C $ &$69.004 \pm 0.096$&$57.896 \pm 0.083$ &$59.382 \pm 0.083$ \\
$ K_{b1} $ &$59.941 \pm 0.082$&$36.303 \pm 0.051$ &$40.662 \pm 0.055$ \\
$ K_{d1} $ &$59.550 \pm 0.080$ &$36.014 \pm 0.051$&$40.422 \pm 0.055$ \\
$AV $&$32.483 \pm 0.029$ &$44.040 \pm 0.047$&$42.132 \pm 0.045$ \\
$\Sigma V$&$697.29 \pm 0.62$&$464.07 \pm 0.43$ &$507.02 \pm 0.46$ \\
\hline
i & \multicolumn{3}{c|}{$\wh{cVar}_{i,SVar} $ } \\ 
\hline 
$C$&$122.89 \pm 0.65$&$159.87 \pm 0.98$ &$125.83 \pm 0.71$\\ 
$K_{b1}$&$127.88 \pm 0.65$ &$75.79 \pm 0.45$&$102.64 \pm 0.56$\\ 
$K_{d1}$&$121.75 \pm 0.64$ &$69.65 \pm 0.41$&$95.29 \pm 0.52$\\ 
$AV$&$102.30 \pm 0.35$ &$250.70 \pm 0.98$&$258.81 \pm 0.97$\\ 
$\Sigma V$&$1400.4 \pm 4.1$ &$1140.6 \pm 3.8$&$1367.8 \pm 4.2$\\ 
\hline
\end{tabular}
\caption{\label{tabCompMBCov} Estimates of variances of the final MC estimators 
of orthogonal projection coefficients of conditional expectation and variance
using different schemes and the GDR, GDI, and RTC methods as described in Section \ref{secMBMD}.}
\end{table}

\begin{figure}[h]
  \centering
  \subfloat[]{\includegraphics[width=0.45\textwidth]{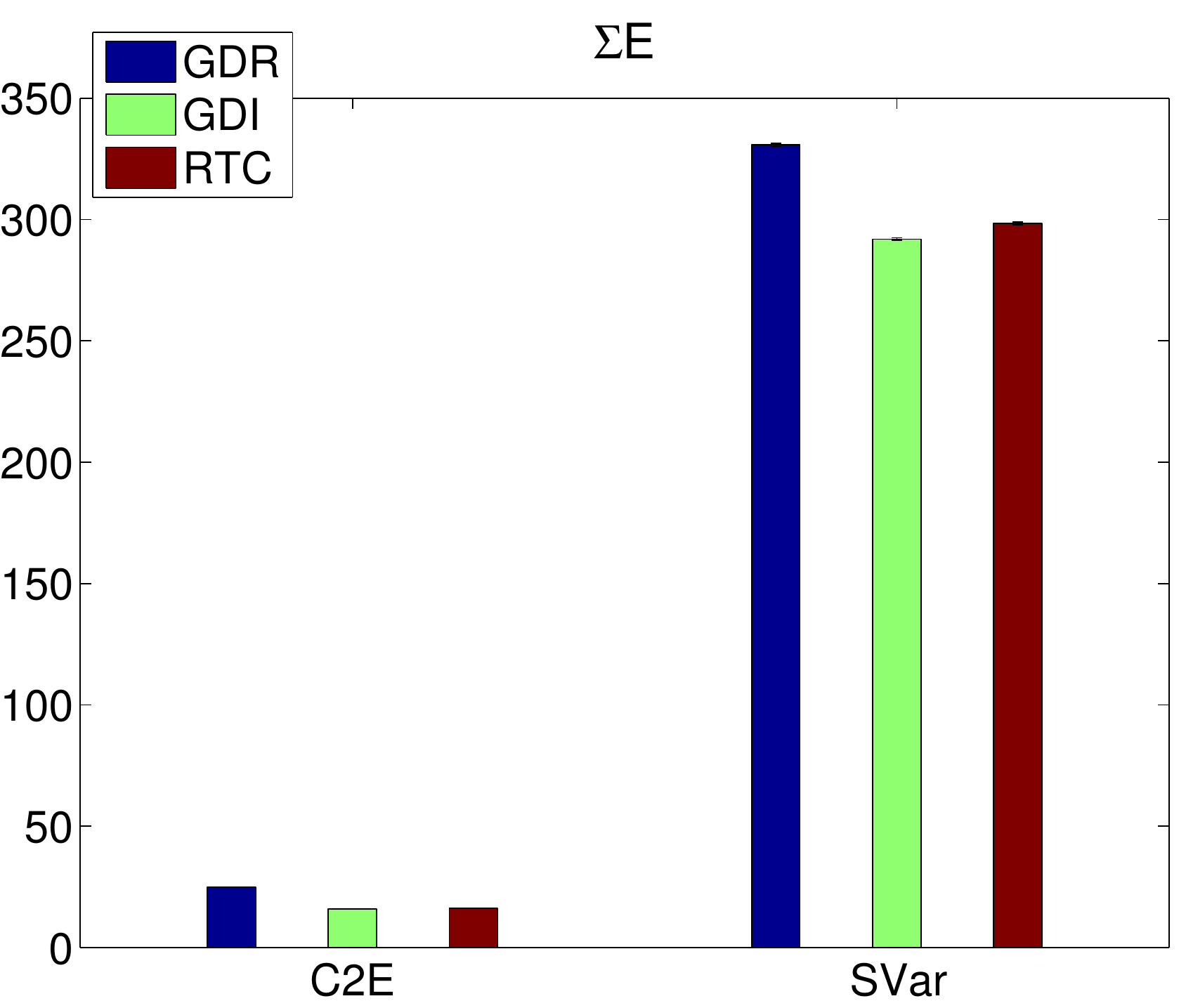}}
\qquad 
\subfloat[]{\includegraphics[width=0.45\textwidth]{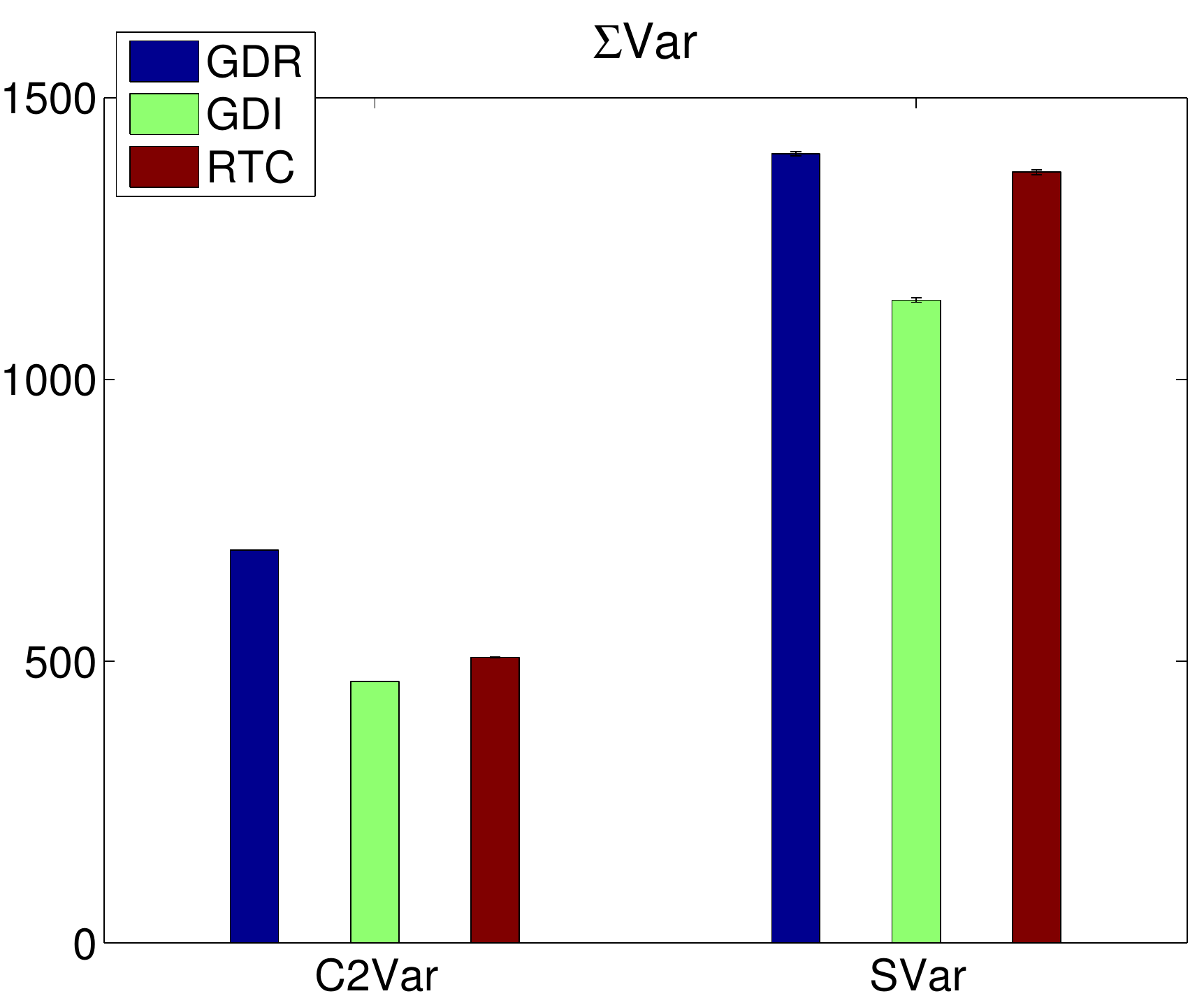}} 
  \caption{
Estimates for the MBMD model of quantities $\Sigma E$ in chart (a) and $\Sigma Var$ in chart (b) for different schemes 
and methods for the estimation of the orthogonal projection coefficients of conditional expectation in (a) and variance in (b) 
as described in Section \ref{secMBMD}.
}
\label{figSigmaMBMD}
\end{figure}

\chapter*{Conclusions}
\addcontentsline{toc}{chapter}{Conclusions}

In this work we formalized and generalized the former 
concept of an estimation scheme from our master's thesis in computer sciene \cite{badowski2011}, making it a convenient tool for defining 
estimators of vector-valued estimands depending on a number of functions. 
We also defined  inefficiency constant of such a scheme, 
which can be useful for comparing the efficiency of unbiased estimation
schemes when used in MC procedures. 
We developed new estimation schemes for various quantities defined for functions 
of two independent random variables, which can be outputs of stochastic models in function of the 
model parameters and a noise variable used to construct the random trajectories of the model process. 
In particular, we provided such first unbiased estimation 
schemes for the variance-based sensitivity indices of a large class of functions of conditional 
moments other than conditional expectation, like conditional variance, of functions of two independent random variables given the
 first variable, 
and developed some new schemes for the case of conditional expectation. We also provided first unbiased estimation schemes 
for covariances and products of functions of conditional moments 
and functions of the first variable, for
coefficients of orthogonal projection of functions of conditional moments 
onto orthogonal functions of the first variable, and of the mean squared error of approximation of functions 
of conditional moments using functions of this variable. Furthermore, we derived estimation schemes for
normalized variance-based sensitivity indices and correlations between functions of conditional moments and functions 
of the first variable. 
We defined a new nonlinearity coefficient which can be 
used for obtaining lower bounds on the probabilities of certain localizations of functions values changes, caused by 
perturbations of their independent arguments. We also 
provided unbiased estimation schemes for nonlinearity coefficients with respect to
all independent arguments and computed these coefficients numerically for the GTS model. 
One of the proposed schemes, called scheme $SVar$, allows to estimate most of the above mentioned 
indices and coefficients for conditional expectation and variance, such as 
variance-based sensitivity indices and coefficients of orthogonal projection 
onto linear combinations of coordinates of the first variable and constants. 
 It can be also easily extended to allow for the 
 estimation of coefficients of orthogonal projection onto higher polynomials of the first variable. 
Thus, it may be an efficient and diverse tool for the analysis of outputs of stochastic models. 
We derived a number of inequalities between the inefficiency 
constants of the proposed schemes.
We tested the introduced schemes and the relationships between 
their inefficiency constants using outputs of three continuous-time Markov chain models of the reaction network dynamics. 
In particular, we proved that the inefficiency constant of scheme $SVar$ for the estimation of 
the sensitivity indices of conditional expectation is no more than 
four times higher and three times lower, 
and in numerical experiments using the GTS model we showed that it can be more than two times lower
than for the best schemes introduced in \cite{badowski2011}. 
In our numerical tests the order of estimators of orthogonal projection coefficients 
with respect to the mean squared errors of the corresponding approximations of conditional expectation and variance 
was different for different models. 
We also demonstrated significant dependence of 
variances of the proposed estimators on the simulation algorithm used, as well as on the order of reactions in 
the GD method. We discussed the relationship of this effect 
with similar ones reported in \cite{badowski2011} and \cite{Rathinam_2010}. 
In practice, one can choose the simulation algorithm and the scheme adaptively  
using preliminary simulations to estimate 
the inefficiency constants of the corresponding MC sequences. 

An interesting topic for the future research is to compare the error 
when using different methods of approximation 
of functions of conditional moments of functions of two independent random variables given the first variable,
using orthogonal functions of the first variable, like 
approximation error of conditional expectation or variance of some output of an MR given the model parameters. 
One can consider the method of 
direct estimation of the coefficients of orthogonal projection proposed here 
and different methods based on double-loop sampling, or using the least squares method possibly with some regularization 
and constraints \cite{Li2002, Khammash_2012}. 
The estimates of mean squared error of approximation of functions of conditional moments 
for different methods 
could be computed using the corresponding schemes from Section \ref{secApprErr}, like $SErrVar$ in the case of conditional variances 
being approximated. 
Coefficients of orthogonal projections obtained using the above methods can be used to estimate the variance-based sensitivity 
indices similarly as in \cite{Li2002b, Li2002, blatman_2010, goutsias_2010}, and an interesting question is if 
the obtained estimates could be more accurate than the ones received using estimators from this work for the same 
computation time and for any stochastic model of practical importance.

\begin{appendices}

\chapter{\label{appMath}Mathematical background}
For a finite set $A$, we denote by $|A|$ the number of its elements. For a set $B$,
 we denote by $\id_{B}$ the identity function on $B$. 
We assume that the set of natural numbers $\N$ contains zero, and by 
$\N_{+}$ we denote the positive natural numbers. For $n \in \N_+$, we define $I_n = \{1, \ldots, n\}$ and for $n=0$, $I_n = \emptyset$. 
We denote by $\R_+$ positive real numbers, and by extended real line we mean $\overline{\R} = \R\cup\{-\infty\}\cup\{+\infty\}$. 
For $a,b \in \overline{\R}$, we write $a \leq b$ not only when $a, b \in \R$ and $a \leq b$, but also when 
 $b= \infty$ or $a = -\infty$. 
We assume an infimum over an empty set to be plus and supremum minus infinity. 
For sets $X$ and $Y$, we denote by $Y^X$ the set of all functions from $X$ to $Y$. Let $f \in Y^X$, which 
we also denote $f:X\rightarrow Y$. 
Domain of $f$, denoted as $D_f$, is the set $X$, and the image 
of some $A \subset X$ under $f$, denoted as $f[A]$, is the set $\{f(x): x \in A\}$. $f[X]$ is 
called the image of $f$. If $B\subset Y$, then preimage of $B$ under $f$, denoted as $f^{-1}[B]$, is the set $\{ x\in X: f(x)\in B\}$. 
Let for a set $A$, $\mc{P}(A)$ be its power set, that is the set of its all subsets. 
The image function of $f$ is a function $f^{\rightarrow}:\mc{P}(X)\rightarrow\mc{P}(Y)$  such that for each 
$C \in \mc{P}(X)$, $f^{\rightarrow}(C) = f[C]$. 
If $X$ is a subset of $\N$, we often write $f_l$ rather than $f(l)$ for $l \in X$, and use notation $(f_l)_{l \in X}$ for $f$.
When $X = I_n$ for some $n \in \N_+$, we often denote $f$ as $(f_1,\ldots,f_n)$. 
Measurable space is a pair $(B,\mc{B})$ consisting of a set $B$ 
and a $\sigma$-field $\mc{B}$ of its subsets. By default, the $\sigma$-field we 
associate with a set $B$ with default topology, like $B \subset \R^n$ for some $n \in \N_+$ with 
topology generated by the Euclidean distance or some countable space like $\N$ with discrete topology, is 
its Borel $\sigma$-field $\mc{B}(B)$, that is the smallest $\sigma$-field generated by open sets, and 
the default measurable space for $B$ is $\mc{S}(B)=(B,\mc{B}(B))$. 
For measurable spaces $S_i = (B_i,\mc{B}_i)$, 
$i \in I_2$, a function from $B_1$ to $B_2$ is said to be measurable from $\mc{S}_1$ to $\mc{S}_2$ if for each 
$A \in \mc{B}_2$, $f^{-1}[A] \in \mc{B}_1$. 
If $\mc{S}$ is the default measurable space for $B$, then we often use $B$ in place of $\mc{S}$, e. g. we say that 
a function is measurable from or to $B$. 
Suppose that $J$ is a countable nonempty set. 
For a family of sets $\{B_i \subset B: i \in J\}$, 
we define their Cartesian product $\prod_{i \in J}B_i$ to be the set of functions $f$ from $J$ to $B$ such that 
for each $i \in J$ it holds $f(i) \in B_i$. For  
$B_i$ all equal to $B$, it holds $\prod_{i \in J}B_i = B^J$. For $N \in \N_+$, we denote $B^{I_N}$ simply as as $B^N$, and informally 
identify $B^1$ with the set $B$. 
For some measurable spaces $\mc{S}_i = (B_i,\mc{B}_i), i\in J$, the product measurable space 
$\mc{S} = \bigotimes_{i\in J}\mc{S}_i $ is defined to be a measurable space $(\Pi_{i\in J} B_i, \bigotimes_{i\in J} \mc{B}_i)$, 
where the product $\sigma$-field $\bigotimes_{i \in J} \mc{B}_i$ is defined as the one generated by the family 
$\mc{T} = \{\Pi_{i\in J}A_i: A_i \in \mc{B}_i\text{ for } i \in J \text{, and only for finite number of $i \in J$, $A_i \neq B_i$}\}$. 
For probability distributions $\mu_i$ on $\mc{S}_i$, $i \in J$, their product $\nu = \bigotimes_{i \in J} \mu_i$ is 
defined as the unique probability distribution on $\bigotimes_{i \in J}\mc{S}_i$ 
such that for each $D= \Pi_{i\in J}A_i \in \mc{T}$, we have $\nu(D) = \prod_{i\in J}\mu_i(A_i)$. 
For a measurable space $\mc{S}$, let $\mc{F}(\mc{S})$ 
be the set of measurable functions $f$ from $\mc{S}$ to $\R$. 
For a measure $\mu$ on $\mc{S}$, let $[f]_{\mu}$ be the class of equivalence of $f \in \mc{F}(\mc{S})$ 
 with respect to relation $g \sim h$ if $f = g$, $\mu$ almost everywhere (a. e.). 
For $p > 0 $, by $L^p(\mu)$ we denote the linear space $\{[f]_\mu: f \in\mc{F}(\mc{S}),\ \int |f|^p \mathrm{d}\mu <\infty\}$ 
 (see \cite{rudin1970} Section 3.10 for more details). 
As common in the literature \cite{rudin1970}, 
for convenience we informally identify classes from $L^p(\mu)$ with their elements, e. g. by writing 
$f \in L^p(\mu)$ for $f \in \mc{F}(\mc{S})$, when it holds $[f]_\mu \in L^p(\mu)$. 
We say that that $f\in \mc{F}(\mc{S})$ is integrable with respect to a measure $\mu$ on $\mc{S}$ if
if $f \in L^1(\mu)$ and square-integrable if $f \in L^2(\mu)$.
For measurable spaces $\mc{S}_i = (B_i,\mc{B}_i)$, for $i \in I_2$, let the function $T$  be 
measurable from $\mc{S}_1$ to $\mc{S}_2$. For a measure $\mu$ on $\mc{S}_1$ we define measure $\mu T^{-1}$ on $\mc{S}_2$ by 
\begin{equation}\label{immes}
\mu T^{-1}(A) = \mu (T^{-1}(A)),\ A \in \mc{B}_2. 
\end{equation}
Below we present a change of variable theorem (\cite{billingsley1979}, Theorem 16.12) 
\begin{theorem}\label{thchvar}
$f$ is integrable with respect to $\mu T^{-1}$ if and only if $fT$ is integrable with respect to $\mu$, in which case 
\begin{equation}
\int_{B_1}\! f(T(x))\, \mu(dx) = \int_{B_2}\! f(y)\, \mu T^{-1}(dy). 
\end{equation}
\end{theorem}

Probability space is denoted by default as 
 $(\Omega, \mathcal{F}, \PR)$ \cite{Durrett}, and
 $L^p(\PR)$ is denoted simply as $L^p$.
For a measurable space $\mc{S}$, an $\mc{S}$-valued 
random variable is a measurable function from $(\Omega, \mathcal{F})$  
to $\mc{S}$. 
\begin{defin}\label{distrDef}
Distribution of an $\mc{S}$-valued random variable $X$, denoted as $\mu_X$, is 
a probability distribution on $\mc{S}$ defined as $\PR X^{-1}$. In other words, for each $A \in \mc{B}$, $\mu_X(A) = \PR(X \in A)$.
\end{defin}
For two random variables $X$ and $Y$, by $X \sim Y$ we mean that $\mu_X=\mu_Y$ and for a probability 
distribution $\Lambda$, $X \sim \Lambda$ denotes $\mu_X=\Lambda$. 
For $N \in \N_+$, random variable $X=(X_i)_{i=1}^N$ with values in a product $\mc{S}= \prod_{i=1}^N\mc{S}_i$ of 
measurable spaces $\mc{S}_i$, $i\in I_N$, is also called an $\mc{S}$-valued random vector. 
The expected value of a real-valued random variable $\phi$ on the probability space with probability $\mu$ is defined as  
\begin{equation}\label{emu} 
\E_{\mu}(\phi) = \int\! \phi(x)\, \mu(dx),
\end{equation} 
where the integral on the rhs is Lebesgue integral and the subscript $\mu$ in $\E_\mu$
is usually omitted if $\mu$ is the default $\PR$. 
We say that a real-valued random variable $Z$ is integrable if it is integrable with respect to $\PR$, 
and analogously for the square-integrability.  
U($a,b$) denotes uniform distribution on the interval $[a,b]$ and 
$\Exp(\lambda)$ is exponential distribution with parameter $\lambda$ \cite{Norris1998}. 
\begin{defin}\label{defUD}
We say that a random variable $X$ has uniform discrete distribution and denote it $X\sim U_d(a,b)$ if $a,b \in \Z$, $a \leq b$,
and for each $c \in \Z$, $a \leq c \leq b$, 
\begin{equation}
\PR(X = c) = \frac{1}{b - a + 1}. 
\end{equation}
\end{defin} 
\begin{defin}\label{defSupp} 
The support \cite{lehmann1998theory} of a probability measure $\PR$ on $(\R^n, \mathcal{B}(\R^n))$ 
is the set $\{ x\in R^n: \PR(A)>0 \text{ for every open rectangle $A$ containing }x\} $. 
\end{defin} 
For $A \in \mathcal{F}$, we denote by $\I_{A}$ the indicator of $A$, that is $\I_{A}(\omega) = 1$ if $\omega \in A$ and $0$ otherwise, 
and we denote $\I = \I_{\Omega}$. 
\begin{defin}\label{condDef} 
Let $X$ be a random variable taking values in a measurable space $(B,\mc{B})$. 
Conditional expectation $\E(Y|X)$ of an integrable random variable $Y$ given $X$ 
is a random variable such that (cf. \cite{Durrett} Section 4.1 
and \cite{billingsley1979}, Theorem 20.1 ii)) 
\begin{enumerate} 
\item $\E(Y|X)$ is equal to $f(X)$ for some measurable function $f$ from $(B,\mc{B})$ to  $\R$, 
\item for each ${A \in \mc{B}},\ \E(Y\I_A(X)) =  \E(\E(Y|X)\I_A(X))$. 
\end{enumerate} 
\end{defin} 
Conditional expectation always exists, however, function $f$ yielding it is uniquely defined only up to sets of measure $\mu_X$ zero. 
Thus equalities in the theorems below hold almost surely (a. s.) \cite{Durrett}, 
but for convenience we omit writing this, and so we do often in the main text. 
%
\begin{theorem}\label{indepCond}
For a measurable function $f$ and independent random variables $X$, $Y$  such that $f(X,Y) $ is integrable, 
it holds (\cite{Durrett} Section 4.1 Example 1.5)
\begin{equation}
\E(f(X,Y)|X) = (\E (f(x, Y)))_{x = X}.
\end{equation}
\end{theorem}
\begin{theorem}\label{condexpX}
For random variables $X$, $Z$, and a measurable function $f$, such that $Z$ and $f(X)Z$ are integrable, 
it holds (\cite{Durrett} Section 4.1 Theorem 1.3) 
\begin{equation}
\E(f(X)Z|X) = f(X)\E(Z|X). 
\end{equation}
In particular 
\begin{equation}\label{equexpconds}
\E(f(X)Z) = \E(f(X)\E(Z|X)). 
\end{equation}
\end{theorem}
For $s,t \in (0,\infty)$, $\frac{1}{s} + \frac{1}{t} = 1$, $X \in L^s$, and $Y \in L^t$, we have the following
 H\"{o}lder's inequality \cite{rudin1970} (called Schwartz inequality for $s = t  = 2$) 
\begin{equation}\label{hold}
\E(|XY|) \leq \sqrt[s]{\E(|X|^s)}\sqrt[t]{\E(|Y|^t)}.
\end{equation}
\begin{theorem}\label{leqpq}
For $p\geq q > 0$ and $Z \in L^p$, it holds 
\begin{equation}
\sqrt[q]{\E(|Z|^q)} \leq \sqrt[p]{\E(|Z|^p)}.
\end{equation}
In particular, $Z \in L^q$.
\end{theorem}
\begin{proof}
It is sufficient to take $X = 1$, $Y = Z^q$ and $t = \frac{p}{q}$ in H\"{o}lder's inequality. 
\end{proof}

\begin{theorem}\label{jensen}
If $\phi$ is convex and $\phi(Y)$ and $Y$ are integrable, then for each random variable $X$
we have the following Jensen's inequality for conditional expectations (\cite{Durrett} 4.1.1 (d)). 
\begin{equation}
\E(\phi(Y)|X) \geq \phi(\E(Y|X)).
\end{equation}
\end{theorem}
The following well-known theorem states that conditional expectation is contraction in $L^p$ for $p \geq 1$. 
\begin{theorem}\label{contrac}
For $Y^p$ integrable for $p \geq 1$ it holds
\begin{equation}
\E(|Y|^p) \geq \E(|\E(Y|X)|^p). 
\end{equation}
\end{theorem}
\begin{proof}
It follows from Theorem \ref{leqpq}, Theorem \ref{jensen} for $\phi(x) = |x|^p$, 
and the iterated expectation property (\ref{doubleCond}).
\end{proof}

Conditional probability of an event $B \subset \Omega$ given a random variable $X$ is defined as
\begin{equation}\label{PBX}
\PR(B|X) = \E(\I_{B}|X).
\end{equation}
Below we provide a definition of conditional distribution which is convenient for our needs 
(cf. \cite{borovkov1999mathematical}, Chapter 20, definitions 1 and 2). 
\begin{defin}\label{defMu} 
For two random variables $X_i, i \in I_2$,  
 with values in measurable spaces $\mc{S}_i =(B_i, \mathcal{B}_i), i \in I_2$, 
respectively, we call $\mu_{X_2|X_1}:B_1 \times \mathcal{B}_2 \rightarrow [0,1]$ conditional distribution of $X_2$ given $X_1$ 
if the following conditions are satisfied: 
\begin{enumerate} 
\item for each $x \in B_1$,  $\mu_{X_2|X_1}(x, \cdot)$ is a probability measure on $\mathcal{B}_2$, 
\item for each $A \in \mathcal{B}_2$, function $x \rightarrow \mu_{X_2|X_1}(x , A) $ is measurable from $\mc{S}_1$ to $\R$, 
\item for each $A \in \mathcal{B}_2$, $\mu_{X_2|X_1}(X_1 , A) $ is a version of $\PR(X_2 \in A|X_1)$. 
\end{enumerate} 
\end{defin} 
It turns out that for random variables $X_i, i \in I_2,$ with values in standard Borel spaces \cite{Ikeda1981} 
such as complete spaces (including $\R^n$ 
with Euclidean distance) with 
Borel $\sigma$-field, conditional distribution $\mu_{Y|X}$ of $Y$ given $X$ exists and $\mu_{Y|X}(x,\cdot)$ 
is uniquely determined for $\mu_X$ a. e. $x$, which follows from Theorem 3.3 in Chapter 1 in \cite{Ikeda1981}. 
For a random variable $Y$ with values in measurable space $\mc{S}$, a real-valued measurable function $g$ on $\mc{S}$ such that $g(Y) \in L^1$,  
and any random variable $X$ such that  
$\mu_{Y|X}$ exists, it holds (cf. \cite{borovkov1999mathematical}, Section 20, Theorem 1) 
\begin{equation}\label{condCond} 
\E(g(Y)|X) = \int\! g(y)\, \mu_{Y|X}(X,dy). 
\end{equation}

\chapter{\label{appHMC}Continuous-time Markov chains}
Let $T = [0,\infty)$ and $E$ be a countable set with discrete topology, called state space. 
Let $Y$ be a stochastic 
process on $E$ with times $T$, that is a sequence of random variables $(Y_t)_{t \in T}$ with values in $E$, 
where variable $Y_t$ describes the random state of the process at time $t$. 
By $\mc{B}(E^T)$ we denote the $\sigma$-field of subsets of $E^T$ generated by the family of sets 
$\{\{f \in \E^T: f(t) = i\}:t \in T, i \in E\}$. Process $Y$ can be identified with 
a random variable taking values in the measurable space 
$\mc{S}(E^T)=(E^T, \mc{B}(E^T))$, whose values  
$Y(\omega)$, known as trajectories of the process, 
are functions of time given by $Y(\omega)(t) = Y_t(\omega)$, $t \in T$, and they describe evolution of the process
 in time corresponding to the elementary event $\omega \in \Omega$. 
Distribution $\mu_Y$ of a process $Y$ is defined as for any random variable
(see Definition \ref{distrDef} in  Appendix \ref{appMath}). 
Let us assume that $Y$ is a right-continuous process,
which means that its trajectories are right-continuous functions of time for each $\omega \in \Omega$, so 
that we can define its jump times, jump chain, and holding times, the names being adopted from \cite{Norris1998}. 
See Section 1.2 of our previous work \cite{badowski2011} or \cite{Norris1998} for intuitive informal descriptions of these objects. 
We define jump times $J_0, J_1, \ldots$ of $Y$ inductively as 
\begin{equation} 
\begin{split} 
J_{0} &= 0, \\ 
J_{n+1} &= \begin{cases} 
 \inf\{t > J_n: Y_t \neq Y_{J_n}\} & \text{if} \ J_n < \infty, \\ 
 \infty & \text{otherwise,} \\ 
\end{cases} 
\end{split} 
\end{equation} 
its jump chain $Z_0, Z_1, \ldots$ as 
$Z_{n} = X_{J_{m(n)}}$, where $m(n) = \max\{k:  k\leq n,\ J_k < \infty\}$, 
and its holding times $S_1, S_2 \ldots$ as 
\begin{equation} 
\begin{split} 
S_{n} &= \begin{cases} 
 J_n - J_{n-1} & \text{if} \ J_n < \infty, \\ 
 \infty & \text{otherwise.} \\ 
\end{cases} 
\end{split} 
\end{equation} 
The moment of explosion $\zeta$ of $Y$  is defined as the moment when $Y$ makes infinitely many jumps for the first time, that is 
\begin{equation}\label{expTime} 
\zeta =\sup_{n} J_{n}. 
\end{equation} 
We say that $Y$ is nonexplosive if $\zeta = \infty$. 
We say that a matrix $Q = (q_{x,y})_{x,y \in E}$ is a $Q$-matrix (on $E$) if 
for each $x, y \in E, x \neq y$, $0 \leq q_{x,y} < \infty$, and for each $x \in E$, 
\begin{equation}\label{qxx} 
-q_{x,x}  = \sum_{y \in E}\ q_{x,y} < \infty. 
\end{equation} 
Entries of a $Q$-matrix are called intensities, and thanks to (\ref{qxx}) it is sufficient to specify the 
off-diagonal intensities to specify the whole $Q$-matrix. 
Continuous-time homogeneous Markov chain (HMC) \cite{pierre1999markov} $Y$ with $Q$-matrix $Q$ on $E$ with times $T$ 
and initial distribution $\Lambda$
is a right-continuous stochastic process with such $E$ and $T$, such that $Y_0 \sim \Lambda$, and for certain function $p$ fulfilling for each 
$x ,y \in E$ and $h \geq 0$, 
\begin{equation} 
p(x,y,h) = q_{x,y}h + o(h), 
\end{equation} 
 for each $h \geq 0$, $k \in \N_{+}$, $x_1, x_2, \ldots, x_{k+1} \in E$, and $0 \leq t_1 \leq \ldots \leq t_{k}$, it holds 
\begin{equation} 
\PR(Y_{t_k+h}= x_{k+1}|Y_{t_1}=x_1, \ldots, Y_{t_k}=x_k) =  p(x_k,x_{k+1},h), 
\end{equation} 
whenever the event we condition on has positive probability. 
Distribution of each nonexplosive HMC with a $Q$-matrix $Q = (q_{x,y})_{x,y \in E}$ and initial distribution $Y_0 \sim \Lambda$ 
is uniquely determined by $Q$ and $\Lambda$. 

Poisson process $N$ with rate $\lambda > 0$ is defined as a nonexplosive HMC on state space $\N$ whose jump chain fulfills 
$Z_n = n$ for $n \in \N$ and whose holding times $S_1, S_2, \ldots$ are i. i. d., $S_1 \sim \Exp(\lambda)$.

\chapter{\label{appProc}Proofs of new theorems for MRCP and MR}
To prove Theorem \ref{thNonexpl} we need the following easy consequence of Theorem 4.3.6 from \cite{stroock05markov}. 
\begin{theorem}\label{thStroock}
For a state space $E$, let $(F_N)_{N=1}^{\infty}$ be finite sets such that $F_N \subset E$, $F_N \subset F_{N+1}$ and 
$\bigcup_{N=1}^{\infty}F_N = E$. If there exists a nonnegative function $u$ on $E$, such that 
$\inf_{j \notin F_N}u(j) \rightarrow \infty$ as $N \rightarrow \infty$, and for some $\alpha > 0$, for a
$Q$-matrix $Q = (q_{ij})_{i,j \in E}$, for each $i \in E$, 
\begin{equation}
\sum_{j \in E, j \neq i}q_{ij}(u(j) - u(i)) \leq \alpha u(i), 
\end{equation}
then for each probability distribution $\Lambda$ on $E$ there exists a nonexplosive HMC 
with initial distribution $\Lambda$ and $Q$-matrix $Q$. 
\end{theorem}
Below we provide the proof of Theorem \ref{thNonexpl}.  
\begin{proof}
For $m$ as in Theorem \ref{thNonexpl}, for assumptions of Theorem \ref{thStroock} to be fulfilled it is sufficient to take 
$F_N = \{x \in E: mx \leq N\}$, $\alpha = 1$, and for $A$ denoting the lhs of (\ref{supmx}), $u(x) = \max\{A,0\} + mx$. 
\end{proof}

\begin{theorem}\label{aveSecFin}
Using notations as in Section \ref{genParSec}, if for 
$\mu_P$ a. e. $p = (k, c)$ we have 
\mbox{$h(p, R) \sim \mu_{MRCP}(RN(k), c)$}, then
$\wt{\mu}$ is conditional distribution of $h(P,R)$ given $P$.
\end{theorem}
\begin{proof}
Point 1 in Definition \ref{defMu} obviously holds. Let $\zeta(p,R)$ denote the initial explosion 
time of a process given by the considered construction 
of MRCP using noise variable $R$ and parameters $p \in B_{RN,E} $. 
The set $B \subset B_{E,RN}$ on which MRCP exists consists of $p$ such that 
 $\PR(\zeta(p,R) = \infty)=1$ and hence 
from measurability of $\zeta$ (which is measurable as a supremum of measurable initial jump times), we have $B  \in \mc{B}_{RN,E}$. 
Point 2 now follows from the fact 
that $h$ is measurable and for each $A \subset \mc{B}(E^T)$, it holds 
\begin{equation}
\wt{\mu}(p, A) = \I_B(p)\mu_0(A) + \I_{B_{RN, E}\setminus B}(p) \PR(h(p,R) \in A).
\end{equation}
Proof of point 3 is analogous as such proof of a less general Theorem 18 in \cite{badowski2011}. For each $A \in \mc{B}(E^T),$ 
\begin{equation}
\begin{split}
\PR(h(P,R) \in A |P) &= \E(\I_{A}(h(P,R))|P) \\
&= (\PR(h(p,R) \in A))_{p = P} \\
&= \wt{\mu}(P,A),\\
\end{split}
\end{equation}
where in the second equality we used Theorem \ref{indepCond} and in the third the assumption of this theorem.  
\end{proof}
Below we provide the proof of Theorem \ref{thMomsExist}. 
\begin{proof}
Let $P = (K,C) \sim \nu$. From $ \E(A(K)^n) < \infty$ it follows that $A(K)$ is finite a. s. 
The assumptions of Theorem \ref{thNonexpl} are satisfied for $\mu_K$ a. e. $k$ with the same $m$ as here, as the lhs of 
(\ref{supmx}) is bounded from above by 
\begin{equation}
L\max(\{0\}\cup\{ms_l:l \in I_L\})A(k). 
\end{equation}
Thus MRCP corresponding to RN and $p$ exists for $\nu$ a. e. $p$. 
For an MR $(P,Y)$ corresponding to RN, with $P$ as above and $Y$ 
 built with the help of the RTC construction it holds a. s. for each $t \in T$ and $i \in I_N$ (see formula \ref{intEqu}) 
\begin{equation}
m_iY_{t,i} \leq  mY_{t} \leq Cm + \sum_{l \in L_m}s_lm N_l(tA(K)).
\end{equation}
From Minkowski's inequality \cite{rudin1970},
 for $\E(Y_{t,i}^n) < \infty$ to hold it is therefore sufficient that $\E(C_i^n) < \infty$ for $i \in I_N$ and for any unit rate Poisson
process $N_1$,  $\E(N_1(tA(K)^n) < \infty.$
For $i \in \N_+$ we define polynomial $x^{\underline{i}} = x(x-1)\ldots(x-i+1)$ and
let the sequence $(b_i)_{i=1}^n$ be such that
\begin{equation}
x^n = \sum_{i=1}^nb_i x^{\underline{i}}.
\end{equation}
For each $\lambda \geq 0$, it holds
\begin{equation}
\begin{split}
\E(N_1^n(\lambda)) &= \sum_{k=0}^{\infty} k^n\frac{\lambda^k}{k!}e^{-\lambda} \\
&= \sum_{k=0}^{\infty} \sum_{i=1}^n(b_i k^{\underline{i}})\frac{\lambda^k}{k!}e^{-\lambda} \\
&=  \sum_{i=1}^n(b_i\lambda^i e^{-\lambda}\sum_{k=0}^{\infty}\frac{k^{\underline{i}}}{k!}\lambda^{k-i}) \\
&= \sum_{i=1}^n b_i\lambda^i,\\
\end{split}
\end{equation}
where in the fourth equality we used the fact that for $i \in \N_+$,
\begin{equation}
\begin{split}
\sum_{k=0}^{\infty}\frac{k^{\underline{i}}}{k!}\lambda^{k-i}) &= \sum_{k=i}^{\infty}\frac{1}{(k-i)!}\lambda^{k-i} \\
&= \sum_{l=0}^{\infty}\frac{1}{l!}\lambda^{l} \\
&= e^{\lambda}. \\
\end{split}
\end{equation} 
Thus, from $\E(A(K)^n) < \infty$ we have
\begin{equation}\label{Nln}
\begin{split}
\E(N_1(tA(K))^n) &= \E((\E(N^n_1(tA(k))))_{k=K}) \\
& \leq \E((\E(\sum_{i=1}^n |b_i| (tA(k))^i))_{k=K}) \\
&= \sum_{i=1}^n |b_i|\E((tA(K))^i) < \infty,
\end{split}
\end{equation}
where in the first and last equalities we used Fubini's theorem, and in the last inequality Theorem \ref{leqpq}. 
\end{proof}

\chapter{\label{appHilb}Hilbert spaces}
We introduce below some definitions and facts from Hilbert space theory, which are used in the main text 
(see \cite{rudin1970}  and \cite{KolmogorovFomin60} for proofs and more details) and prove some new facts. 
Hilbert space is a pair $(H, (,))$ consisting of a linear space $H$ and a scalar product $(,)$ in it, such that 
for metric $d$ and norm $||\cdot||$ defined as
\begin{equation}
d(x,y) = ||x-y|| = \sqrt{(x -y,x-y)},
\end{equation}
 $(H,d)$ is a complete metric space. For simplicity we also say that $H$ is a Hilbert space (with scalar product $(,)$).
We say that a set $\{v_i \in H: i \in I_n\}$ is orthogonal in $H$ 
if $(v_i,v_j) =0$, $i,j \in I_n$, $i \neq j$, nonzero orthogonal if further $v_i \neq 0$, $i \in I_n$, and 
orthonormal if it is orthogonal with $||v_i||=1$, $i \in I_n$. 
For linear subspaces $W_1,\ldots,W_n$ of a certain linear space, we define 
\begin{equation}
\sum_{i=1}^n W_i =  \{\sum_{i=1}^{n}w_i:\ \forall i \in I_n,\ w_i \in W_i\}.
\end{equation}
\begin{defin}\label{defHilb}
 Hilbert space $H$ is direct sum of its linear subspaces $H_1,\ldots, H_n$, which we denote
\begin{equation}
H = \bigoplus_{i=1}^n H_i = H_1 \oplus \ldots \oplus H_n
\end{equation}
 if the following conditions are fulfilled
\begin{enumerate}
  \item subspaces $H_1, \ldots, H_n$  are closed in $H$,
  \item 
\begin{equation}
H = \sum_{i=1}^n H_i,
\end{equation}
  \item these subspaces are mutually orthogonal, which means that for each $i,j \in I_n$, $i \neq j$ for each $v_i \in H_i$ and $v_j \in H_j$
\begin{equation}
 (v_i,v_j) = 0.
\end{equation}
\end{enumerate}
\end{defin}
From point 3 it follows that for $v \in H$, elements $v_i \in H_i$ for $i \in I_n$ such that
\begin{equation}
v = \sum_{i=1}^n v_i
\end{equation}
are uniquely determined.
Let $<,>$ be a scalar product in $\R^n$ and $(a_{ij})_{i,j \in I_n}$ be real numbers for which, for 
each $x,y \in \R^n$, it holds 
\begin{equation}\label{scpr}
<x,y> = \sum_{i,j \in I_n}a_{ij}x_iy_j. 
\end{equation}
We say that 2 norms $|\cdot|_1, |\cdot|_2$ on a linear space $V$ are equivalent, if there exist $\alpha$ and $\beta$ real positive such that 
for each $x \in V$
\begin{equation}
\alpha|x|_1 \leq |x|_2 \leq \beta|x|_1.
\end{equation}
\begin{theorem}\label{genDSpace}
For a Hilbert space $H$ with a scalar product $(,)$, and for a scalar product $<,>$ in $\R^n$ as in (\ref{scpr}), the Cartesian product space 
 $H^n = \{(v_i)_{i=1}^n: v_i \in H,\ i \in I_n\} $ 
with function $(,)_n:H^n\times H^n\rightarrow\R$  
\begin{equation}\label{scDir}
 (v,w)_n = \sum_{i,j \in I_n} a_{ij}(v_i,w_j)
\end{equation}
is a Hilbert space, which we call the direct sum of $H$ given by $<,>$ and denote by $\bigoplus_{<,>}H$. 
Norms $||\cdot ||_n$ induced by scalar products (\ref{scDir}) corresponding to different scalar products $<,>$ in $\R^n$ are equivalent. 
\end{theorem}
\begin{proof}
For $<,>$ equal to the standard scalar product on $\R^n$,  $\bigoplus_{<,>}H$ is the $n$-fold direct sum of
Hilbert spaces known from the literature \cite{KolmogorovFomin60}, which is a Hilbert space, and 
whose norm let us denote $||\cdot||_{st}$. 
For general $<,>$ function $(,)_n$ defined by \ref{scDir} is bilinear and symmetric so for the thesis to hold 
it is sufficient to show that it is positive definite and 
function $||\cdot||_n$ given by  $||x||_n = \sqrt{(x,x)}$ is a norm equivalent to $||\cdot||_{st}$.
Since the matrix $A = (a_{ij})_{i,j \in I_n}$ is real symmetric and positive definite, 
there exists an orthogonal matrix $B = (b_{ij})_{i,j \in I_n}$ and diagonal matrix $C = (c_{ij})_{i,j \in I_n}$ such that 
$c_{ii} > 0$ for $i \in I_n$  
and $A = B^TCB$ (\cite{strang2006linear}, sections 5.6 and 6.2). For each $v \in H^n$, we have 
\begin{equation}\label{vn2}
\begin{split}
(v,v)= \sum_{i,j \in I_n} a_{ij}(v_i,v_j) = \sum_{i,j,k \in I_n} b_{ki}c_{kk}b_{kj}(v_i,v_j) \\
= \sum_{k=1}^n c_{kk} ||\sum_{i=1}^n b_{ki} v_i ||^2 
 \end{split}
\end{equation}
From orthogonality of $B$,  $\sum _{k}b_{ki}b_{kj} = \delta_{ij}$, so that  
\begin{equation}\label{vst2}
\begin{split}
\sum_{k=1}^n ||\sum_{i=1}^n b_{ki} v_i ||^2 = \sum_{i,j,k \in I_n} b_{ki}b_{kj}(v_i,v_j) \\
= \sum_{i,j \in I_n}(v_i, v_j) = ||v||_{st}^2. 
\end{split}
\end{equation}
From (\ref{vn2}) and (\ref{vst2}) it holds 
\begin{equation}
 \min_{i \in I_n}(c_{ii}) ||v||^2_{st} \leq ||v||^2_n \leq \max_{i \in I_n}(c_{ii}) ||v||^2_{st}, 
\end{equation}
which completes the proof. 
\end{proof} 
If $M$ is a closed subspace of $H$, then the orthogonal complement of $M$ in $H$, defined as 
 $ M^{\perp} = \{v \in H: \forall w \in M\quad v \perp w\}$ is closed and it holds
\begin{equation}\label{mDirect}
H = M \oplus M^{\perp}.
\end{equation}
Projection $P$ onto $M$ in the above direct sum is called orthogonal. For $v \in H$, $P(v)$
is the unique element of $M$ minimizing the distance from $v$, which we also call error of approximation of $v$, 
\begin{equation}\label{infProp} 
d(v,P(v)) = \inf_{w \in M} ||v-w||.
\end{equation}
Furthermore, it holds  
\begin{equation}\label{distP}
||v||^2 = ||v - P(v)||^2 + ||P(v)||^2.
\end{equation}
\begin{lemma}\label{lemP1P2}
If $M_2\subset M_1$ are closed subspaces of $H$ and $P_i$ is orthogonal projection from $H$ onto $M_i$, $i \in I_2$, then
\begin{equation}
P_2P_1= P_1P_2 = P_2.
\end{equation}
 \end{lemma}
 \begin{proof}
 Denoting $M_2^{\perp_{1}}$ the orthogonal complement of $M_2$ in $M_1$, one can easily check that
$H= M_1^{\perp}\oplus M_2^{\perp_{1}} \oplus M_2$, from which the thesis easily follows. 
 \end{proof}
A well-known example of orthogonal projection is conditional expectation, which we prove below for the reader's convenience
 (cf. \cite{Durrett}, Section 4.1 Theorem 1.4).
\begin{lemma}\label{condort}
Conditional expectation given $X$ is an orthogonal projection from Hilbert space $L^2$ onto $L^2_X$ (defined in Section \ref{secOrthog}).
\end{lemma}
\begin{proof}
From the definition of conditional expectation and Theorem \ref{contrac}, $\E(Z|X) \in L^2_X$.
It is sufficient to prove that for each random variable $Z \in L^2$, $Z-\E(Z|X) \in (L^2_X)^{\perp}$.
For each $f(X) \in L^2$ for some measurable $f$, we have from Schwartz inequality $Zf(X),\E(Z|X)f(X)\in L^1$. 
Thus, from Theorem \ref{condexpX}, 
\begin{equation}
\E((Z -\E(Z|X))f(X)) = 0. 
\end{equation}
\end{proof}
Let $M$ be a closed subspace of $H$, then $\bigoplus_{<,>}M$ is a complete space,
 so it is a closed subspace of $\bigoplus_{<,>}H$.
\begin{theorem}\label{projDS}
If $P$ is orthogonal projection of $H$ onto $M$, then the function $P_n:H^n\rightarrow H^n,$ 
given by $P_n(v) = (P(v_i))_{i=1}^n$ is an orthogonal projection from $\bigoplus_{<,>}H$ onto $\bigoplus_{<,>}M$. 
\end{theorem}
\begin{proof} 
For each $v \in \bigoplus_{<,>}H$ we have $P_n(v) \in \bigoplus_{<,>}M$. Furthermore, for each $w \in \bigoplus_{<,>}M$, it holds
\begin{equation} 
(v - P_n(v),w)_n = \sum_{i,j\in I_n} a_{ij}(v_i - P(v_i),w_j) = 0, 
\end{equation} 
since for each $i \in I_n$ it holds $v_i - P(v_i)\in M^{\perp}$. Thus $v - P_n(v) \in (\bigoplus_{<,>}M)^{\perp}$. 
\end{proof}

\chapter{\label{appStatMC}Statistics and Monte Carlo background}
In this section we introduce certain definitions and facts from statistics 
and Monte Carlo simulations (cf. \cite{lehmann1998theory, asmussen2007stochastic, badowski2011}), which are used throughout the text. 
Let us consider a nonempty set of probability distributions $\mathcal{P}$ defined on the same measurable space $\mc{S}$, called 
(set of) admissible distributions (on $\mc{S}$, cf. \cite{kolmogorov1992selected}, Section 38).  
For a measurable space $\mc{H}$, a measurable function from $\mathcal{S}$ to $\mc{H}$ 
is called $\mc{H}$-valued (simply real-valued if $\mc{H} = \mc{S}(\R)$) statistic for $\mc{P}$. 
For a given $\mu \in \mathcal{P}$, random variable $X \sim \mu$ is called an observable. 
A real-valued function $G$ on $\mathcal{P}$ is called an estimand on $\mc{P}$. 
We say that a probability distribution $\mu$ on $\R$ has finite $n$-th moment, $n \in \N_+$, if 
\begin{equation}
\int |x|^n\! d\mu\, < \infty. 
\end{equation}
Let us define estimand $G_E$  
on all probability distributions $\mu$ on $\R$ with finite first moments, for which 
$G_E(X)= \E(X), X \sim \mu$, and 
estimand $G_{Var}$ on distributions $\mu$ on $\R$ with finite second moments, for which
$G_{Var}(\mu)=\Var(X)$, $X \sim \mu$.
We say that a real-valued statistic $\phi$ for $\mc{P}$ is an estimator of an estimand $G$ on $\mc{P}$ 
if for each $\mu \in \mathcal{P}$, for observables $X \sim \mu$, we think of 
random values of $\phi(X)$ as estimates of $G(\mu)$, that is its certain approximations. 
For each $\mu \in \mc{P}$, average error of this approximation can be measured using 
mean squared error 
\begin{equation}\label{meanSqrErr} 
\E_{\mu} ((\phi - G(\mu))^2). 
\end{equation} 
Let estimator $\phi$ of $G$ be unbiased, that is for each $\mu \in \mathcal{P}$, 
\begin{equation}\label{muPhi} 
\E_\mu(\phi) = G(\mu). 
\end{equation} 
Then from (\ref{muPhi}) we have that for $\mu \in \mc{P}$, variance $\Var_\mu(\phi)$ of  $\phi$ is  
equal to the mean squared error (\ref{meanSqrErr}). 
For $n \in \N_+$, we define 
\begin{equation}\label{mcpn}
\mc{P}^n = \{\mu^n:\mu \in \mc{P}\}, 
\end{equation}
where $\mu^n$ is the $n$-fold product of distribution $\mu$. 
For an estimand $G$ on $\mathcal{P}$ and $n \in \N_+$, we define 
estimand $G_n$ on $\mc{P}^n$ by formula $G_n(\mu^n) = G(\mu)$, and call it $G$ in $n$ dimensions. 
Let $n \in \N_+$ and $\Pi_n$ denote the group of all permutations of $I_n$. 
For some set $B$ and function $\phi$ from $B^n$ to $\R$, we define symmetrisation of $\phi$ to be 
a function from $B^n$ to $\R$ such that for each $x  \in B^n$, 
\begin{equation}
\Sym(\phi)(x) = \frac{1}{n!}\sum_{\pi \in \Pi_n}\phi((x_{\pi(i)})_{i=1}^n). 
\end{equation}
We say that $\phi$ as above is symmetric if it is equal to its symmetrisation. 
For some admissible distributions $\mc{P}$, for 
each $\mu \in \mc{P}$,  $X \sim \mu^n$, and $\phi$ being a real-valued statistic for $\mc{P}^n$, 
 $\Sym(\phi)(X)$ is an average of random variables with the same distribution as $\phi(X)$. 
In particular, if $\phi$ is an estimator of some estimand $G$ on $\mc{P}^n$, then so is 
$\Sym(\phi)$, and from the lemma below 
it immediately follows that it has uniformly not higher variance, that is for each $\mu \in \mc{P}$ it holds
\begin{equation}
\Var_{\mu^n}(\Sym(\phi)) \leq \Var_{\mu^n}(\phi). 
\end{equation}
We proved the below lemma as Theorem 11 in \cite{badowski2011}, but this time we provide a different simpler proof. 
\begin{lemma} \label{lemVarAve} 
For some $n \in \N_+$, let $X_1, \ldots, X_n$ be real-valued square-integrable 
random variables with the same distribution. Then 
\begin{equation}\label{varSum} 
\Var\left(\frac{1}{n}\sum_{i=1}^{n}X_i\right) \leq \Var(X_1), 
\end{equation} 
and equality in (\ref{varSum}) holds if and only if for each $i,j \in I_n$, $X_i=X_j$ a. s. 
\end{lemma} 
\begin{proof} 
For $x_1,\ldots,x_n$ real positive, from the well-known inequality between arithmetic and quadratic 
means we have 
\begin{equation}\label{avesqr} 
\left(\frac{1}{n}\sum_{i=1}^{n}x_i\right)^2  \leq  \frac{1}{n}\sum_{i=1}^{n}x_i^2, 
\end{equation} 
which is equivalent to 
\begin{equation} 
\sum_{1 \leq i<j \leq n}(x_i -x_j)^2  \geq  0, 
\end{equation} 
so equality in (\ref{avesqr}) holds only if all $x_i$ are equal. 
Replacing $x_i$ by $X_i$ in (\ref{avesqr}), taking expected value of both sides, and using the fact that
each $X_i$ has the same expected value as their average, we receive the thesis.   
\end{proof} 
We say that a distribution $\mu$ on a measurable space $(B, \mc{B})$ is finite discrete (on $D$) 
if for some finite set $D \in \mc{B}$, $\mu(D) = 1$. 
In \cite{Halmos_1946} it was proved that if $\mc{P}$ contains all finite discrete distributions on $\R$, and if $\phi$ is an 
unbiased estimator 
of some estimand $G$ on $\mc{P}^n$, then $\Sym(\phi)$ is the unique symmetric unbiased estimator of $G$. 
In particular for any other unbiased estimator $\phi'$ of $G$ we have $\Sym(\phi)=\Sym(\phi')$, 
so $\Sym(\phi)$ has uniformly not higher variance than $\phi'$. 
For instance, the unique symmetric unbiased estimator of $G_E$ 
in $n \geq 1$ dimensions is given for each $x \in \R^n$ by formula 
\begin{equation}\label{phiAve} 
\phi_{E,n}(x) = \frac{1}{n} \sum_{i=1}^n x_i, 
\end{equation} 
and of $G_{Var}$ in $n \geq 2$ dimensions by formula 
\begin{equation}\label{phiVar} 
\begin{split}
\phi_{Var,n}(x) = \frac{1}{n-1} (\sum_{i=1}^n x_i^2 -n(\phi_{E,n}(x))^2).
\end{split} 
\end{equation} 
For admissible distributions $\mc{P}$ consisting of all probability distributions on $\R$ having second moments 
and $n \geq 2$, we define estimand $G_{VarAve,n}$  of variance of the mean 
on $\mc{P}^n$ by formula $G_{VarAve,n}(\mu^n) = \Var_{\mu^n}(\phi_{E,n})$. Its symmetric unbiased estimator is 
given by formula 
\begin{equation}\label{phiAveVar} 
\phi_{VarAve,n}(x) =  \frac{\phi_{Var,n}(x)}{n}. 
\end{equation} 
For admissible distributions $\mc{P}=\{\mu\}$, 
let $\phi$ be an unbiased estimator of an estimand $G_{\lambda,\mu}$ on $\mc{P}$ defined by $G_{\lambda,\mu}(\mu)=\lambda$. 
We call such $\phi$ unbiased estimator of $\lambda$ (for $\mu$).
If further $\phi \in L^2(\mu)$,
we call it a single-step MC estimator of $\lambda$. 
For some $n \in \N_{+}$, for a random vector $X \sim \mu^n$, i. e. one with independent 
coordinates with distribution $\mu$, in each $i$th step of an $n$-step 
MC procedure one computes a value of a random variable 
$W_i = \phi(X_i)$, called the $i$th observable of the single-step MC estimator. 
For $W = (W_i)_{i=1}^n$, we use the values of 
\begin{equation}\label{meanMC}
\overline{W} = \phi_{E,n}(W)
\end{equation}
as final MC estimates of $\lambda$. Function given by formula 
\begin{equation} 
\phi_{f}(x) =  \phi_{E,n}((\phi(x_i))_{i=1}^n), \ x \in B^n,
\end{equation} 
for which we have $\overline{W}=\phi_{f}(X)$, is an unbiased estimator of $G_{\lambda,\mu}$ in $n$ dimensions, and 
we call it an $n$-step or final MC estimator of $\lambda$ (for $\mu$) and call (\ref{meanMC}) its observable. 
Let us denote variance of the single-step estimator as $\Var_s = \Var_\mu(\phi)$ and its standard deviation as $\sigma_s = \sqrt{\Var_s}$, 
while for the $n$-step estimator as $\Var_{f} = \Var_f(n) = \Var_{\mu^n}(\phi_{n})$ and $\sigma_{f} =\sigma_{f}(n)= \sqrt{\Var_f}$. 
It holds
\begin{equation}\label{varfs} 
\Var_{f} = \frac{\Var_s}{n}.
\end{equation} 
We use the values of 
\begin{equation}\label{varEst} 
\phi_{VarAve,n}(W) 
\end{equation} 
as estimates of $\Var_{f}$ for $n \geq 2$, and the values of 
\begin{equation}\label{sigmaEst}
\widehat{\sigma}_{f,n}(W)= \sqrt{\phi_{VarAve,n}(W)}
\end{equation}
as such estimates of $\sigma_f$. 
For some such obtained estimates $\wt{\lambda}=\overline{W}(\omega)$ of $\lambda$, 
and $\wt{\sigma}_f = \widehat{\sigma}_{f,n}(W)(\omega)$ of $\sigma_f$, 
we report the results of a MC procedure in form 
$\wt{\lambda} \pm \wt{\sigma}_f$ 
(cf. Chapter 3, Section 1 in \cite{asmussen2007stochastic}). 
From the central limit theorem (CLT) \cite{billingsley1979}, as $n$ goes to infinity in the above described  
MC procedure, $\sqrt{n}(\overline{W} - \lambda)$ converges in distribution to $\ND(\lambda, \Var_s)$, 
that is normal distribution with mean $\lambda$ and variance $\Var_s$, 
and from the law of large numbers $\sqrt{n}\widehat{\sigma}_{f,n}(W)$ converges a. s. and thus in probability to $\sigma_s$. 
In particular for 
$k > 0$, and $\Phi$ being the cumulative distribution function of standard normal distribution, i. e. $\Phi(x) = P(Z \leq x)$, 
$Z \sim N(0,1)$, the probability $\PR(|\overline{W} - \lambda| < k\widehat{\sigma}_{f,n}(W))$ converges to $2(1 - \Phi(k))$, 
which is approximately $68\%$ for $k=1$ and $99,73\%$ for $k=3$. 

\chapter{\label{appcoeffsLem}Proofs of Theorem \ref{thEquAves} and lemmas \ref{lemPrBetterCov} and \ref{lemC1C2}}
Below we provide a proof of Theorem \ref{thEquAves}. 
\begin{proof} 
Let $(\mu,f) \in \mc{V}$ and $x \in B_\mu^{\wt{\Pi}[p_A]}.$ We have 
\begin{equation} 
\begin{split} 
\phi_{\ave_{\Pi}(\kappa),\mc{V}}(x) &= \ave_{C,A,\Pi}(t)(g_{\mc{V},\wt{\Pi}^{\rightarrow}[A]}(f)(x))\\
&= \frac{1}{|\Pi|} \sum_{\pi \in \Pi} \ave_{C,A,\pi}(t)(\eta_{C,A,\Pi,\pi}(g_{\mc{V},\wt{\Pi}^{\rightarrow}[A]}(f)(x))).\\
\end{split}
\end{equation}
We denote $x_\pi = x_{\wt{\pi}[p_A]},$ $\pi \in \Pi$.
For $\pi \in \Pi$ we have 
\begin{equation}\label{equvec}
\begin{split}
\ave_{C,A,\pi}(t)(\eta_{C,A,\Pi,\pi}(g_{\mc{V},\wt{\Pi}^{\rightarrow}[A]}(f)(x))) 
&= \ave_{C,A,\pi}(t)(g_{\mc{V},{\wt{\pi}^{\rightarrow}[A]}}(f)(x_{\pi})) \\
&= t(\rho_{C,A,\pi}(g_{\mc{V},\wt{\pi}^{\rightarrow}[A]}(f)(x_\pi))) \\
&=t(((g_{\mc{V}, \wt{\pi}^{\rightarrow}[A],i,\wt{\pi}[v]}(f)(x_\pi))_{|v \in A_{\gamma(i)}})_{i=1}^\delta) \\
&=t(((f_i(x_{\wt{\pi}[v]}))_{|v \in A_{\gamma(i)}})_{i=1}^\delta). \\
\end{split}
\end{equation}
On the other hand, 
\begin{equation}
\A_{\Pi}(\phi_{\kappa,\mc{V}})(x) =\frac{1}{|\Pi|} \sum_{\pi \in \Pi} \A_{\pi}(\phi_{\kappa,\mc{V}})(f)(x_{\pi}) \\
\end{equation}
and for $\pi \in \Pi$ we have 
\begin{equation}\label{equapik}
\begin{split}
\A_{\pi}(\phi_{\kappa,\mc{V}})(f)(x_{\pi}) &= \phi_{\kappa,\mc{V}}(f)(\sigma_{B_\mu,\wt{\pi}[p_A],\pi}^{-1}(x_{\pi}))\\
&= t(((f_i((\sigma_{B_\mu,\wt{\pi}[p_A],\pi}^{-1}(x))_v))_{|v \in A_{\gamma(i)}})_{i=1}^\delta).
\end{split}
\end{equation}
Let $i \in I_k$ be such that $A_i \neq \emptyset$, and let $v \in A_i$.
Comparing the last terms in (\ref{equvec}) and (\ref{equapik}) we can see that it is sufficient to prove that 
\begin{equation}
(\sigma_{B_\mu,\wt{\pi}[p_A],\pi}^{-1}(x))_v = x_{\wt{\pi}[v]}.
\end{equation}
Indeed, for $l \in J_i$ it holds
\begin{equation}
\begin{split}
((\sigma_{B_\mu,\wt{\pi}[p_A],\pi}^{-1}(x))_v)_l &= (\sigma_{B_\mu,\wt{\pi}[p_A],\pi}^{-1}(x))_{(l,v_l)}\\
&= x_{(l,\pi_l(v_l))}\\
&= (x_{\wt{\pi}[v]})_l.\\  
\end{split}
\end{equation}
\end{proof}

Below we provide the proof of Lemma \ref{lemPrBetterCov}. 
\begin{proof}
The thesis of is equivalent to  
\begin{equation}\label{sharp0E1C1} 
0 < \frac{2N-1}{N^2} \E(\overline{Y^2}^2) - 2\E(\overline{Y^2}\overline{Y}^2) + \E(\overline{Y}^4). 
\end{equation} 
We have 
\begin{equation}\label{ey22} 
\E(\overline{Y^2}^2) = \frac{1}{N^2}(N\E(X^4) + N(N-1)\E^2(X^2)). 
\end{equation} 
Using the fact that $\E(X) = 0$ we receive 
\begin{equation}\label{ey2y2} 
\begin{split} 
\E(\overline{Y^2}\overline{Y}^2) &= \frac{1}{N^3}\E(NY_1^2(Y_1^2 +(N-1)(Y_2^2))) \\
&= \frac{1}{N^2}(\E(X^4) + (N-1)\E^2(X^2))
\end{split}
\end{equation}
and 
\begin{equation}\label{ey4}
 \E(\overline{Y}^4) = \frac{1}{N^4}(N\E(X^4) + N3(N-1)\E^2(X^2)),
\end{equation}
where the coefficient $N3(N-1)$ appears since to get a product of squares $Y_i^2Y_j^2$ for some $i \neq j$ when performing multiplication in  
$(\sum_{i=1}^N Y_i)^4$ one can choose some $i$th of $N$ summands from the first sum, the same summand from one of three other sums, 
and some $j$th of $N-1$ remaining summands in the two remaining sums. 
Substituting (\ref{ey22}), (\ref{ey2y2}), and (\ref{ey4}) into (\ref{sharp0E1C1}), we receive 
\begin{equation}
\begin{split}
0 &< \left(\frac{2N-1}{N^3} -\frac{2}{N^2} + \frac{1}{N^3}\right)\E(X^4) \\
&+ \left(\frac{2N-1}{N^2}\frac{N-1}{N} - \frac{2(N-1)}{N^2} + \frac{3(N-1)}{N^3}\right)\E^2(X^2), 
\end{split}
\end{equation}
which is equivalent to 
\begin{equation}
0<\frac{2(N-1)}{N^3}\E^2(X^2). 
\end{equation} 
\end{proof} 

 Below we provide the proof of Lemma \ref{lemC1C2}. 
\begin{proof} 
Let $\wt{P} \sim \mu_P^N$, $\wt{R} \sim \mu_R^N$, and for $i \in I_N$, 
 \begin{equation}\label{bpi} 
 B_{P,i} = f_P(\wt{P}_i)- \frac{1}{N}\sum_{j=1}^{N} f_P(\wt{P}_j),
 \end{equation} 
 and 
 \begin{equation}
 B_{R,i,l} = f_{R,l}(\wt{R}_i)- \frac{1}{N}\sum_{j=1}^{N} f_{R,l}(\wt{R}_j). 
 \end{equation} 
 We have 
 \begin{equation}
 \wh{cE}_{k,C1E(N)}(f_l)(\wt{P}, \wt{R}) = \frac{1}{N-1}\sum_{i=0}^{N-1} (B_{P,i} + B_{R,i,l})B_{P,i}, 
 \end{equation}
 and 
 \begin{equation}\label{covl2} 
 \begin{split} 
 \E(\wh{cE}_{k,C1E(N)}^2(\wt{P}, \wt{R})) &\geq \frac{1}{(N-1)^2}\E(\sum_{i,j \in I_N}B_{P,i}B_{R,i,l}B_{P,j}B_{R,j,l}) \\ 
 & = \frac{1}{N-1}\Var(f_P(P))\Var(f_{R,l}(R)), 
 \end{split} 
 \end{equation} 
 where in the first inequality we used the fact that   
 $\E(B_P[i]^2B_P[j]B_R[j]) = 0$, $i,j \in I_N,$ and in the last equality the easy to check equalities 
 $\E(B_P[i]^2) = \frac{N-1}{N}\Var(f_P(P))$, $i \in I_N$, and $\E(B_P[i]B_P[j]) = \frac{1}{N}\Var(f_P(P))$, $i,j \in I_N, i \neq j$. 
 If $\Var(f_{R,l}(R))\rightarrow \infty$ as $l \rightarrow \infty$, then so does the rhs of (\ref{covl2}). 
\end{proof}

\chapter{\label{appd}New analytical expressions for the SB model} 
As we justified in Appendix D in \cite{badowski2011}, in the SB model one can replace the considered 
one birth process with rate equal to the sum of coordinates of random vector $K = (K_i)_{i=1}^3$  with 
three birth processes with rates equal to its consecutive coordinates without changing the conditional distribution of the model output 
given the parameters and thus the quantities computed here. We will perform the computations 
using construction of a process of MR given by integral equation (\ref{intEqu}) but with random parameters  
\begin{equation}\label{sumIndep}
Y_t = C + \sum_{i=1}^3 N_i(K_it).
\end{equation}
The values of the main and total sensitivity indices of conditional expectation of output given the parameters were computed in Appendix D of 
\cite{badowski2011} and we provide them along with results of below computations in Table \ref{exactVal}. 
For each $\lambda > 0$,
\begin{equation}
\E(N(\lambda)) = \lambda 
\end{equation}
 and 
 \begin{equation}\label{Poiss2}
 \E(N(\lambda)^2) = \lambda^2 + \lambda. 
 \end{equation}
Furthermore, from Theorem \ref{indepCond}, for $i \in I_3$, 
\begin{equation}\label{NCond}
\E(N_i(K_i t)|K_i) =  (\E(N_i(k_it)))_{k_i = K_i} = K_it,
\end{equation}
and thus 
\begin{equation}\label{condexpSB}
\E(Y_t|P) = C + t\sum_{i=1}^3 K_i.
\end{equation}
We can see that the conditional expectation is linear in the model parameters, so its nonlinearity coefficients with respect to 
all subvectors of the parameter vector are zero. 
From the iterated expectation property we have 
\begin{equation}
\E(Y_t) = \E(C) + t\sum_{i=1}^3\E(K_i) = 60 + 100(0.6 + 1 + 0.1) = 230.
\end{equation}
From (\ref{condexpSB}), the coefficient of $C - \E(C)$ in the orthogonal projection of the mean output onto span of the 
centred parameters and constants fulfills 
\begin{equation}
 bE_{C} = \frac{\Cov(\E(Y_t|P),C)}{\Var(C)}  =1
\end{equation}
and for the kinetic rates we have
 \begin{equation}
 \Cov(\E(Y_t|P), K_i)  = \Cov(tK_i,K_i) = t\Var(K_i)
 \end{equation}
and thus 
\begin{equation}
bE_{K_i} = t =100,\ i \in I_3. 
\end{equation}
Due to (\ref{CQKP}) and (\ref{Poiss2}), the conditional variance of $Y_t$ given $P$ is equal to 
\begin{equation}
 \begin{split}
 \Var(Y_t|P) &= (\Var(c + \sum_{i=1}^3N_i(k_it)))_{p = P} \\
= \sum_{i=1}^3K_it.
 \end{split}
\end{equation}
Thus, similarly as for the conditional expectation, the
nonlinearity coefficients of the conditional variance with respect to all subvectors of the parameter vector 
are equal zero. 
Furthermore, $AveVar = \E(\Var(Y_t|P)) = 170$,  $VVar_C =VVar_C^{tot} = bVar_C = 0$, and from (\ref{condexpSB}), 
$VVar_{K_i} = VVar_{K_i}^{tot} = V_{K_i},\ i \in I_3$. Using the values of 
$\wt{V}_{K_i}, i \in I_3$,  
computed in \cite{badowski2011} (see Table \ref{exactVal}), we receive 
\begin{equation}
VVar_{P} = \sum_{i=1}^{3} V_{K_i} = 378.
\end{equation}
We also have $\Cov(K_i,\Var(Y_t|P)) = t\Var(K_i),$ and thus $bVar_{K_i} = 100$, $i \in I_3$.
\end{appendices}
\bibliographystyle{plain}
\bibliography{pub}

\end{document}